\documentclass{eptcs}

\def\proofsinappendix{}      %

\usepackage{cite}

\usepackage{amsthm}
\usepackage{thmtools}      %
\usepackage{amsmath,amssymb,amsfonts}
\usepackage{mathtools}
\usepackage{physics}
\usepackage{stmaryrd}
\usepackage{xfrac}        %

\usepackage{afterpage}
\usepackage{graphicx}
\usepackage{hyperref}
\usepackage[nameinlink,noabbrev]{cleveref}
\usepackage{xcolor}
\usepackage{transparent} %

\usepackage{scalerel}
\usepackage{stackengine,wasysym}

\usepackage{tikz}
\usetikzlibrary{decorations.markings}
\usetikzlibrary{decorations.pathreplacing}
\usetikzlibrary{calc,math,shadows}
\usetikzlibrary{intersections,fadings}

\usepackage{tikzit}
\input{style/tikz/DZX.tikzdefs}

\tikzstyle{box}=[fill=white, font={\small}, draw=black, shape=rectangle, inner sep=2.5pt]
\tikzstyle{wirelable}=[node on layer=labeltextlayer, font={\scriptsize}, scale=.9, text={wirelabelcolor}, fill=none, inner sep=1pt,  tikzit fill={rgb,255: red,255; green,191; blue,191}]

\tikzstyle{wn}=[style=spider, style=phase-thin-spider, fill = addspiderfillcolor, tikzit fill=white, tikzit draw=black, tikzit shape=circle, tikzit category=GLA]
\tikzstyle{bn}=[style=spider, style=phase-thin-spider, fill=copyspiderfillcolor, tikzit draw=black, tikzit shape=circle, tikzit category=GLA, tikzit fill={rgb,255: red,128; green,128; blue,128}]
\tikzstyle{gwn}=[style=wn, style=very thick, style=phase-thick-spider, tikzit shape=circle, tikzit draw=magenta]
\tikzstyle{gbn}=[style=bn, style=very thick, style=phase-thick-spider, tikzit shape=circle, tikzit draw=magenta, tikzit fill={rgb,255: red,128; green,128; blue,128}]
\tikzstyle{ggn}=[style=gbn]
\tikzstyle{grn}=[style=gwn]

\tikzstyle{wphase}=[bubble-base, bubble-white-node, tikzit category=ZX, tikzit fill=white, tikzit draw=blue, tikzit shape=rectangle]
\tikzstyle{whphase}=[style=wphase]
\tikzstyle{bphase}=[bubble-base, bubble-black-node, tikzit category=ZX, tikzit fill={rgb,255: red,128; green,128; blue,128}, tikzit draw=blue, tikzit shape=rectangle]
\tikzstyle{gphase}=[style=bphase]
\tikzstyle{rphase}=[style=wphase]
\tikzstyle{mphase}=[rectangle, rounded corners=1pt, draw=black!60, fill=black!15, inner sep=2pt, font={\scriptsize\boldmath}, minimum height=.3cm, minimum width=.3cm]

\tikzstyle{lmat}=[style=mat, signal to=west, signal from=east, adapt-horizontal-angle, inner xsep=\matHeadPad, tikzit fill=gray, tikzit category=GLA] %
\tikzstyle{rmat}=[style=mat, signal to=east, signal from=west, adapt-horizontal-angle, inner xsep=\matHeadPad, tikzit fill=gray, tikzit category=GLA] %
\tikzstyle{dmat}=[style=mat, signal to=south, signal from=north, outer sep=0pt, adapt-vertical-angle, inner ysep=\matHeadPad, tikzit fill=gray, tikzit category=GLA] %
\tikzstyle{umat}=[style=mat, signal to=north, signal from=south, outer sep=0pt, adapt-vertical-angle, inner ysep=\matHeadPad, tikzit fill=gray, tikzit category=GLA] %

\tikzstyle{glmat}=[style=lmat, very thick, tikzit draw=magenta, tikzit fill=white]
\tikzstyle{grmat}=[style=rmat, very thick, tikzit draw=magenta, tikzit fill=white]
\tikzstyle{gdmat}=[style=dmat, very thick, tikzit draw=magenta, tikzit fill=white]
\tikzstyle{gumat}=[style=umat, very thick, tikzit draw=magenta, tikzit fill=white]
\tikzstyle{lmatt}=[style=glmat]
\tikzstyle{rmatt}=[style=grmat]
\tikzstyle{umatt}=[style=gumat]
\tikzstyle{dmatt}=[style=gdmat]

\tikzstyle{gather}=[draw=none, fill=none, inner sep=0pt, minimum size=0pt]
\tikzstyle{divide}=[style=gather]
\tikzstyle{scalar}=[style=gather]
\tikzstyle{multiplexer}=[style=gather]

\tikzstyle{dashed}=[-, dashed]
\tikzstyle{very}=[style=thick]

\tikzstyle{had}=[fill=hadamardfillcolor, text=hadamardtextcolor, draw=black, shape=rounded rectangle, rounded rectangle east arc=none, rounded rectangle west arc=none, rounded rectangle arc length=80, inner sep=1.5pt, minimum height=5pt, minimum width=5pt, font={\scriptsize\boldmath}, tikzit category=ZX, tikzit fill=yellow, tikzit draw=black]
\tikzstyle{rhad}=[style=hadshape, path picture={\hadcapeast}, inner xsep=\hadHeadPad, tikzit category=ZX, tikzit fill=yellow, tikzit draw=black]
\tikzstyle{lhad}=[style=hadshape, path picture={\hadcapwest}, inner xsep=\hadHeadPad, tikzit category=ZX, tikzit fill=yellow, tikzit draw=black]
\tikzstyle{dhad}=[style=hadshape, path picture={\hadcapsouth}, inner ysep=\hadHeadPad, tikzit category=ZX, tikzit fill=yellow, tikzit draw=black]
\tikzstyle{uhad}=[style=hadshape, path picture={\hadcapnorth}, inner ysep=\hadHeadPad, tikzit category=ZX, tikzit fill=yellow, tikzit draw=black]
\tikzstyle{had-thick}=[outer sep=-\hadThickInset, /utils/exec={}]
\tikzstyle{ghad}=[style=had, very thick, tikzit category=ZX, tikzit fill=yellow, tikzit draw=magenta]
\tikzstyle{grhad}=[style=rhad, style=had-thick, tikzit category=ZX, tikzit fill=yellow, tikzit draw=magenta]
\tikzstyle{glhad}=[style=lhad, style=had-thick, tikzit category=ZX, tikzit fill=yellow, tikzit draw=magenta]
\tikzstyle{gdhad}=[style=dhad, style=had-thick, tikzit category=ZX, tikzit fill=yellow, tikzit draw=magenta]
\tikzstyle{guhad}=[style=uhad, style=had-thick, tikzit category=ZX, tikzit fill=yellow, tikzit draw=magenta]

\tikzstyle{groundout}=[style=groundoutshape]
\tikzstyle{groundin}=[style=groundinshape]

\tikzstyle{dwphase}=[-, pointer-white, phasetosouth, tikzit fill=white, tikzit draw=white]
\tikzstyle{uwphase}=[-, pointer-white, phasetonorth, tikzit fill=white, tikzit draw=white]
\tikzstyle{lwphase}=[-, pointer-white, phasetowest, tikzit fill=white, tikzit draw=white]
\tikzstyle{rwphase}=[-, pointer-white, phasetoeast, tikzit fill=white, tikzit draw=white]
\tikzstyle{dbphase}=[-, pointer-black, phasetosouth, tikzit fill=white, tikzit draw=gray]
\tikzstyle{ubphase}=[-, pointer-black, phasetonorth, tikzit fill=white, tikzit draw=gray]
\tikzstyle{lbphase}=[-, pointer-black, phasetowest, tikzit fill=white, tikzit draw=gray]
\tikzstyle{rbphase}=[-, pointer-black, phasetoeast, tikzit fill=white, tikzit draw=gray]

\tikzstyle{backgroundborder}=[-, fill=none, draw=backgroundbordercolour, tikzit draw=grey, on layer=midlayer, line width=.2pt]

\tikzstyle{backgroundborderless}=[-, fill=backgroundcolour, draw=none, tikzit fill=gray, on layer=backlayer]
\tikzstyle{background}=[-, preaction={fill=backgroundcolour, on layer=  backlayer}, draw=backgroundbordercolour, tikzit draw=black, tikzit fill=gray, on layer=midlayer, line width=.2pt]
\tikzstyle{cobackground}=[-, fill=white, draw=backgroundbordercolour, tikzit draw=black, tikzit fill=white, on layer=midlayer, line width=.2pt]
\tikzstyle{white}=[-, fill=white, draw=black, tikzit fill=white, tikzit draw=black]
\tikzstyle{bbox}=[-, fill=white]
\tikzstyle{stream}=[-, style=thick, tikzit draw=red, postaction=decorate, decoration={markings, mark=at position 0.5 with {\arrow[scale=1.2]{>}}}]
\tikzstyle{thin}=[-, line width=.2pt]

\tikzstyle{thick}=[-, style=very thick, tikzit draw=blue]

\newif\ifappendixproofs
\ifdefined\proofsinappendix
  \appendixproofstrue
\else
  \appendixproofsfalse
\fi

\ifappendixproofs
  \usepackage[createShortEnv]{proof-at-the-end}
  \pratendSetGlobal{end}
\else
  \newcommand{\pratendSetLocal}[1]{}
  \newcommand{\printProofs}[1][]{}

  \newenvironment{theoremE}[1][]{\begin{theorem}[#1]}{\end{theorem}}
  \newenvironment{lemmaE}[1][]{\begin{lemma}[#1]}{\end{lemma}}
  \newenvironment{propositionE}[1][]{\begin{proposition}[#1]}{\end{proposition}}
  \newenvironment{corollaryE}[1][]{\begin{corollary}[#1]}{\end{corollary}}
  \newenvironment{proofE}{\begin{proof}}{\end{proof}}
  \newenvironment{textAtEnd}{}{}
\fi

\newcommand{\proofof}[1]{%
  \phantomsection\label{proof:#1}%
  \begin{proof}[\mbox{\hyperref[#1]{\proofname}}]%
}

\newcommand{\seeproof}[1]{%
  See \hyperref[proof:#1]{proof} on page~\pageref{proof:#1}.%
}

\newcommand{\Z}{\mathbb{Z}}
\newcommand{\N}{\mathbb{N}}
\newcommand{\F}{\mathbb{F}}
\newcommand{\C}{\mathbb{C}}

\newcommand{\B}{\mathcal{B}}
\newcommand{\R}{\mathbb{R}}
\providecommand{\Tr}{\operatorname{Tr}}
\providecommand{\range}{\operatorname{range}}
\providecommand{\rank}{\operatorname{rank}}

\newenvironment{smallbmatrix}
  {\left[\begin{smallmatrix}}
  {\end{smallmatrix}\right]}

\ExplSyntaxOn
\seq_new:N \l__phasemat_vec_seq
\seq_new:N \l__phasemat_mat_seq
\seq_new:N \l__phasemat_col_seq
\tl_new:N  \l__phasemat_first_row_tl
\tl_new:N  \l__phasemat_colspec_tl
\int_new:N \l__phasemat_ncols_int

\NewDocumentCommand \phasemat { m m } {
  \ensuremath{ \__phasemat_build:nn {#1} {#2} }
}

\cs_new_protected:Nn \__phasemat_open_array:n { \begin{array}{#1} }
\cs_new_protected:Nn \__phasemat_separator: {
  \begingroup
  \transparent{0.5}
  \vrule width \arrayrulewidth
  \endgroup
}

\cs_new_protected:Nn \__phasemat_build:nn {
  \seq_set_split:Nnn \l__phasemat_vec_seq { \\ } {#1}
  \seq_set_split:Nnn \l__phasemat_mat_seq { \\ } {#2}
  \tl_set:Nx \l__phasemat_first_row_tl
    { \seq_item:Nn \l__phasemat_mat_seq { 1 } }
  \seq_set_split:NnV \l__phasemat_col_seq { & } \l__phasemat_first_row_tl
  \int_set:Nn \l__phasemat_ncols_int { \seq_count:N \l__phasemat_col_seq }
  \int_compare:nNnTF { \seq_count:N \l__phasemat_vec_seq } = { 1 }
    {
      \int_compare:nNnTF { \l__phasemat_ncols_int } = { 1 }
        { \tl_set:Nn \l__phasemat_colspec_tl { @{}c@{\hspace{1.5pt}\__phasemat_separator:\hspace{1.5pt}}c@{} } }
        {
          \tl_set:Nn \l__phasemat_colspec_tl
            { @{\hspace{2pt}}c@{\hspace{\arraycolsep}\__phasemat_separator:\hspace{\arraycolsep}} }
          \int_step_inline:nn { \l__phasemat_ncols_int } {
            \tl_put_right:Nn \l__phasemat_colspec_tl { c }
          }
          \tl_put_right:Nn \l__phasemat_colspec_tl { @{\hspace{2pt}} }
        }
    }
    {
      \tl_set:Nn \l__phasemat_colspec_tl
        { @{\hspace{2pt}}c@{\hspace{\arraycolsep}\__phasemat_separator:\hspace{\arraycolsep}} }
      \int_step_inline:nn { \l__phasemat_ncols_int } {
        \tl_put_right:Nn \l__phasemat_colspec_tl { c }
      }
      \tl_put_right:Nn \l__phasemat_colspec_tl { @{\hspace{2pt}} }
    }
  \exp_args:NV \__phasemat_open_array:n \l__phasemat_colspec_tl
    \int_step_inline:nn { \seq_count:N \l__phasemat_vec_seq } {
      \seq_item:Nn \l__phasemat_vec_seq {##1}
      \char_generate:nn { 38 } { 4 }
      \seq_item:Nn \l__phasemat_mat_seq {##1}
      \int_compare:nNnF {##1} = { \seq_count:N \l__phasemat_vec_seq } { \\ }
    }
  \end{array}
}
\ExplSyntaxOff

\renewcommand{\vb}[1]{\mathbf{#1}}
\renewcommand{\bar}[1]{\overline{#1}}

\newcommand{\discard}{\mathtt{disc}}
\newcommand{\codiscard}{\mathtt{codisc}}
\newcommand{\disc}{%
  {\mathrel{%
    \tikzsetnextfilename{sym-disc}%
    \tikz[x=1pt,y=1pt,baseline=-0.5ex]{
      \draw (3,0) -- (6,0);
      \draw[line cap=round] (6,4) -- (6,-4);
      \draw[line cap=round] (8,3) -- (8,-3);
      \draw[line cap=round] (10,2) -- (10,-2);
      \draw[line cap=round] (12,1) -- (12,-1);
    }%
  }}%
}
\ProvideDocumentCommand{\swap}{}{\mathrm{swap}}

\DeclareMathOperator{\Lim}{\mathtt{Lim}}
\DeclareMathOperator{\CPM}{CPM}
\DeclareMathOperator{\Proj}{Proj}
\DeclareMathOperator{\Obs}{Obs}
\DeclareMathOperator{\St}{St}
\DeclareMathOperator{\Gr}{Gr}
\DeclareMathOperator{\Unroll}{\mathtt{Unroll}}
\DeclareMathOperator{\Ext}{\mathtt{Ext}}
\DeclareMathOperator{\Ob}{ob}
\DeclareMathOperator{\Span}{span}

\newcommand{\AffRel}{\mathsf{AR}_{\infty}}
\newcommand{\FHilb}{\mathsf{FHilb}}
\newcommand{\Stab}{\mathsf{Stab}}
\newcommand{\StabRel}{\mathtt{StabGrp}}
\newcommand{\ZX}{\mathsf{ZX}}
\newcommand{\StabZX}{\ZX}
\newcommand{\ALR}{\mathsf{ALR}_{\mathsf{fd}}}
\newcommand{\sALR}{\mathsf{s}\ALR}
\newcommand{\ACR}{\mathsf{ACR}_{\mathsf{fd}}}
\newcommand{\dStabZX}{\delta\StabZX}

\NewDocumentCommand{\Pauli}{O{n}}{\mathcal{P}_p^{\otimes #1}}
\newcommand{\interp}[1]{\left\llbracket\ #1 \ \right\rrbracket}
\newcommand{\trans}{\top}
\newcommand{\preceqp}{\preceq_{\mathsf{P}}}
\newcommand{\preceql}{\preceq_{\mathsf{L}}}

\newtheorem{definition}{Definition}
\newtheorem{theorem}{Theorem}

\newtheorem{proposition}{Proposition}
\newtheorem{lemma}{Lemma}
\newtheorem{remark}{Remark}
\newtheorem{example}{Example}
\newtheorem{corollary}{Corollary}
\newtheorem{axiom}{Axiom}
\crefname{definition}{definition}{definitions}
\Crefname{definition}{Definition}{Definitions}
\crefname{theorem}{theorem}{theorems}
\Crefname{theorem}{Theorem}{Theorems}
\crefname{conjecture}{conjecture}{conjectures}
\Crefname{conjecture}{Conjecture}{Conjectures}
\crefname{proposition}{proposition}{propositions}
\Crefname{proposition}{Proposition}{Propositions}
\crefname{lemma}{lemma}{lemmas}
\Crefname{lemma}{Lemma}{Lemmas}
\crefname{remark}{remark}{remarks}
\Crefname{remark}{Remark}{Remarks}
\crefname{example}{example}{examples}
\Crefname{example}{Example}{Examples}
\crefname{corollary}{corollary}{corollaries}
\Crefname{corollary}{Corollary}{Corollaries}
\crefname{axiom}{axiom}{axioms}
\Crefname{axiom}{Axiom}{Axioms}

\providecommand{\Matrices}{\text{Mat}}
\RenewDocumentCommand{\Matrices}{o o o}{%
  \ensuremath{%
    \IfNoValueTF{#1}{%
      \F_p^{n\times n}%
    }{%
      \IfNoValueTF{#2}{%
        \IfNoValueTF{#3}{%
          \F_p(\delta)^{#1 \times #1}%
        }{%
          {#3}^{#1 \times #1}%
        }%
      }{%
        \IfNoValueTF{#3}{%
          \F_p(\delta)^{#1 \times #2}%
        }{%
          {#3}^{#1 \times #2}%
        }%
      }%
    }%
  }%
}

\providecommand{\Herm}{\text{Herm}}
\RenewDocumentCommand{\Herm}{o o}{%
  \ensuremath{%
    \IfNoValueTF{#1}{%
      \IfNoValueTF{#2}{%
        \mathsf{Herm}_n(\F_p(\delta))%
      }{%
        \mathsf{Herm}_n{#2}%
      }%
    }{%
      \IfNoValueTF{#2}{%
        \mathsf{Herm}_{#1}(\F_p(\delta))%
      }{%
        \mathsf{Herm}_{#1}{#2}%
      }%
    }%
  }%
}

\providecommand{\X}{\text{X}}
\providecommand{\Zp}{\mathbb{F}_p}
\renewcommand{\op}{\mathsf{op}}

\newcommand{\axref}[1]{\text{\hyperref[#1]{Ax.~\ref*{#1}}}}
\newcommand{\lemref}[1]{\text{\hyperref[#1]{Lem.~\ref*{#1}}}}
\newcommand{\themref}[1]{\text{\hyperref[#1]{Thm.~\ref*{#1}}}}
\newcommand{\propref}[1]{\text{\hyperref[#1]{Prop.~\ref*{#1}}}}
\newcommand{\defref}[1]{\text{\hyperref[#1]{Def.~\ref*{#1}}}}

\newcommand{\steq}[1]{\overset{\substack{#1}}{=}}

\newcommand{\indhyp}[1]{\text{\hyperlink{#1}{Ind.~Hyp.}}}

\newcommand{\axiomtag}[2][=]{%
  \overset{\text{\hypertarget{ax:#2}{}\axiomregister{#2}}\mathsf{#2}}{#1}}
\newcommand{\axiomtagp}[3][=]{%
  \overset{\text{\hypertarget{ax:#2}{}\hypertarget{ax:#3}{}\axiomregister{#2}\axiomregister{#3}}\mathsf{#3}}{#1}}

\makeatletter
\newcommand{\axiomregister}[1]{%
  \ifcsname axiomreg@#1\endcsname\else
    \expandafter\gdef\csname axiomreg@#1\endcsname{}%
    \if@filesw\immediate\write\@auxout{\string\@axiomdefined{#1}}\fi
  \fi}
\providecommand{\@axiomdefined}[1]{%
  \expandafter\gdef\csname axiomdef@#1\endcsname{}}
\makeatother
\newcommand{\axiomref}[1]{%
  \ifcsname axiomdef@#1\endcsname
    \hyperlink{ax:#1}{\ensuremath{\mathsf{#1}}}%
  \else
    \axiomrefundefined{#1}%
  \fi}
\newcommand{\axiomrefundefined}[1]{%
  \GenericWarning{}{Axiom reference `#1' is undefined: no \string\axiomtag{#1} in any table}%
  \textcolor{red}{\ensuremath{\mathsf{??#1??}}}}

\newcommand{\vdotss}{%
  \tikzsetnextfilename{sym-vdotss}%
  \tikz[baseline, every node/.style={inner sep=0}]{
    \node at (0,0) {.};
    \node at (0,4pt) {.};
    \node at (0,8pt) {.};
  }%
}

\title{The Delayed Stabilizer ZX-Calculus}
\newcommand{\titlerunning}{The Delayed Stabilizer ZX-Calculus}

\author{
  Cole Comfort\textsuperscript{1}\qquad\qquad Giovanni de Felice\textsuperscript{2}
  \institute{%
    \textsuperscript{1}Universit\'e Paris-Saclay, CNRS, ENS Paris-Saclay, Inria,
    CentraleSup\'elec,\\ Laboratoire M\'ethodes Formelles,
    91190 Gif-sur-Yvette, France\\
    \textsuperscript{2}Relational Intelligence Ltd.
  }
}
\newcommand{\authorrunning}{C. Comfort and G. de Felice}
\newcommand{\publicationstatus}{}

\usetikzlibrary{intersections}
\usetikzlibrary{fadings}
\usetikzlibrary{calc}

\usetikzlibrary{external}
\newif\ifcompileunrolledsurfacecode

\begin{document}

  \maketitle
 
  \begin{abstract}
  Many stabilizer quantum error-correcting codes are built from a finite pattern repeated across space or time, such as lattice codes, translation-invariant graph states, and quantum convolutional codes. Ordinary stabilizer ZX-diagrams capture only finite truncations of such systems, obscuring the repeated structure that defines them. We introduce the delayed stabilizer ZX-calculus, a finite graphical language for these infinite, translation-invariant processes. It extends the odd-prime-dimensional stabilizer ZX-calculus with a single new generator, the delay, which feeds data from one time step to the next.

  We equip the calculus with two semantics. In the first semantics, we interpret the behaviour of a delayed ZX-diagram as an equivalence class of sequences of  quantum channels; where two sequences are identified if they have the same information content. We show that the behaviour of a delayed ZX-diagram uniquely determines an infinite stabilizer group.
  In the second semantics, we interpret the delay as a formal variable, encoding the translation-invariant families of Pauli operators as generating functions. This allows us to represent a delayed ZX-diagram in terms of a tableau of generating functions, from which the infinite stabilizer group can be recovered.

  Finally, we give a complete axiomatization of the delayed stabilizer ZX-calculus, featuring generalised Euler decomposition and colour change rules. Using generalised forms of local complementation and pivoting, we reduce every diagram to a unique normal form. This establishes soundness, universality, and completeness for the generating tableau semantics.
  \end{abstract}
  
  \section{Introduction}\label{sec:intro}
  \pratendSetLocal{category=intro, end, restate}
  Stabilizer quantum mechanics is a cornerstone of quantum information theory and underlies much of quantum error correction and fault-tolerant quantum computation, describing quantum processes compactly through Pauli operators rather than exponentially large matrices~\cite{Nielsen2012,gottesman_stabilizer_1997}. The ZX-calculus, and more specifically the stabilizer ZX-calculus, allows such processes to be reasoned about using graph rewriting~\cite{Coecke2011,pqp,2012.13966}. The stabilizer ZX-calculus is sound, universal, and complete for stabilizer quantum mechanics in all prime dimensions~\cite{backens_complete_2014,quopitzx,poor}, meaning that reasoning about stabilizer quantum mechanics can be done exclusively using this graphical notation. This has been particularly useful for quantum error correction and measurement-based quantum computation, where diagrammatic rewrites correspond to transformations of graph states and stabilizer codes~\cite{KissingerWetering2024Book}.

The stabilizer ZX-calculus becomes inadequate, however, for processes defined by structure repeated across space or time. The quantum error-correction literature is full of examples: topological codes~\cite{Dennis2002,Kitaev2003}, in which a fixed local stabilizer pattern repeats across space, and convolutional and serial turbo codes~\cite{Poulin2009,jr_simple_2005,Ollivier2003,wildeentanglement,wildeconvolutional}, where the pattern repeats across time. Further examples include quantum channels with memory~\cite{kretschmann_quantum_2005}, resource states for measurement-based quantum computing~\cite{nest_universal_2006}, and quantum cellular automata~\cite{Haah2021}.

Such processes can be truncated to finite circuits, which themselves can be represented by ordinary ZX-diagrams. However, working with these truncations quickly becomes intractable, obscuring the recursive procedure that generates the circuit. For example, consider the ZX-calculus representation of the surface code from~\cite{grok}:
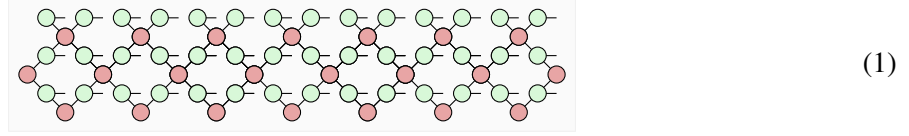
\begin{equation}\label{eq:surface_code}
  \begin{tikzpicture}
	\begin{pgfonlayer}{nodelayer}
		\node [style=wn] (2) at (1, -0.25) {};
		\node [style=wn] (5) at (2, -1.25) {};
		\node [style=wn] (6) at (2, 0.75) {};
		\node [style=bn] (9) at (1.5, 0.25) {};
		\node [style=bn] (12) at (1.5, -0.75) {};
		\node [style=none] (13) at (-1.5, 1.75) {};
		\node [style=none] (14) at (-1.5, -1.75) {};
		\node [style=none] (15) at (13.5, -1.75) {};
		\node [style=none] (16) at (13.5, 1.75) {};
		\node [style=wn] (17) at (3, -0.25) {};
		\node [style=bn] (24) at (2.5, 0.25) {};
		\node [style=bn] (25) at (2.5, -0.75) {};
		\node [style=none] (26) at (3, 0.25) {};
		\node [style=none] (27) at (2, 0.25) {};
		\node [style=none] (28) at (2, -0.75) {};
		\node [style=none] (29) at (3, -0.75) {};
		\node [style=wn] (30) at (-1, -0.25) {};
		\node [style=wn] (31) at (0, -1.25) {};
		\node [style=wn] (32) at (0, 0.75) {};
		\node [style=bn] (33) at (-0.5, 0.25) {};
		\node [style=bn] (34) at (-0.5, -0.75) {};
		\node [style=wn] (35) at (1, -0.25) {};
		\node [style=bn] (36) at (0.5, 0.25) {};
		\node [style=bn] (37) at (0.5, -0.75) {};
		\node [style=none] (38) at (1, 0.25) {};
		\node [style=none] (39) at (0, 0.25) {};
		\node [style=none] (40) at (0, -0.75) {};
		\node [style=none] (41) at (1, -0.75) {};
		\node [style=wn] (42) at (3, -0.25) {};
		\node [style=wn] (43) at (4, -1.25) {};
		\node [style=wn] (44) at (4, 0.75) {};
		\node [style=bn] (45) at (3.5, 0.25) {};
		\node [style=bn] (46) at (3.5, -0.75) {};
		\node [style=wn] (47) at (5, -0.25) {};
		\node [style=bn] (48) at (4.5, 0.25) {};
		\node [style=bn] (49) at (4.5, -0.75) {};
		\node [style=none] (50) at (5, 0.25) {};
		\node [style=none] (51) at (4, 0.25) {};
		\node [style=none] (52) at (4, -0.75) {};
		\node [style=none] (53) at (5, -0.75) {};
		\node [style=wn] (54) at (0, 0.75) {};
		\node [style=wn] (55) at (1, -0.25) {};
		\node [style=bn] (57) at (0.5, 1.25) {};
		\node [style=bn] (58) at (0.5, 0.25) {};
		\node [style=wn] (59) at (2, 0.75) {};
		\node [style=bn] (60) at (1.5, 1.25) {};
		\node [style=bn] (61) at (1.5, 0.25) {};
		\node [style=none] (62) at (2, 1.25) {};
		\node [style=none] (63) at (1, 1.25) {};
		\node [style=none] (64) at (1, 0.25) {};
		\node [style=none] (65) at (2, 0.25) {};
		\node [style=wn] (66) at (2, 0.75) {};
		\node [style=wn] (67) at (3, -0.25) {};
		\node [style=bn] (69) at (2.5, 1.25) {};
		\node [style=bn] (70) at (2.5, 0.25) {};
		\node [style=wn] (71) at (4, 0.75) {};
		\node [style=bn] (72) at (3.5, 1.25) {};
		\node [style=bn] (73) at (3.5, 0.25) {};
		\node [style=none] (74) at (4, 1.25) {};
		\node [style=none] (75) at (3, 1.25) {};
		\node [style=none] (76) at (3, 0.25) {};
		\node [style=none] (77) at (4, 0.25) {};
		\node [style=wn] (78) at (5, -0.25) {};
		\node [style=wn] (79) at (6, -1.25) {};
		\node [style=wn] (80) at (6, 0.75) {};
		\node [style=bn] (81) at (5.5, 0.25) {};
		\node [style=bn] (82) at (5.5, -0.75) {};
		\node [style=wn] (83) at (7, -0.25) {};
		\node [style=bn] (84) at (6.5, 0.25) {};
		\node [style=bn] (85) at (6.5, -0.75) {};
		\node [style=none] (86) at (7, 0.25) {};
		\node [style=none] (87) at (6, 0.25) {};
		\node [style=none] (88) at (6, -0.75) {};
		\node [style=none] (89) at (7, -0.75) {};
		\node [style=wn] (90) at (3, -0.25) {};
		\node [style=wn] (91) at (4, -1.25) {};
		\node [style=wn] (92) at (4, 0.75) {};
		\node [style=bn] (93) at (3.5, 0.25) {};
		\node [style=bn] (94) at (3.5, -0.75) {};
		\node [style=wn] (95) at (5, -0.25) {};
		\node [style=bn] (96) at (4.5, 0.25) {};
		\node [style=bn] (97) at (4.5, -0.75) {};
		\node [style=none] (98) at (5, 0.25) {};
		\node [style=none] (99) at (4, 0.25) {};
		\node [style=none] (100) at (4, -0.75) {};
		\node [style=none] (101) at (5, -0.75) {};
		\node [style=wn] (102) at (7, -0.25) {};
		\node [style=wn] (103) at (8, -1.25) {};
		\node [style=wn] (104) at (8, 0.75) {};
		\node [style=bn] (105) at (7.5, 0.25) {};
		\node [style=bn] (106) at (7.5, -0.75) {};
		\node [style=wn] (107) at (9, -0.25) {};
		\node [style=bn] (108) at (8.5, 0.25) {};
		\node [style=bn] (109) at (8.5, -0.75) {};
		\node [style=none] (110) at (9, 0.25) {};
		\node [style=none] (111) at (8, 0.25) {};
		\node [style=none] (112) at (8, -0.75) {};
		\node [style=none] (113) at (9, -0.75) {};
		\node [style=wn] (114) at (4, 0.75) {};
		\node [style=wn] (115) at (5, -0.25) {};
		\node [style=bn] (116) at (4.5, 1.25) {};
		\node [style=bn] (117) at (4.5, 0.25) {};
		\node [style=wn] (118) at (6, 0.75) {};
		\node [style=bn] (119) at (5.5, 1.25) {};
		\node [style=bn] (120) at (5.5, 0.25) {};
		\node [style=none] (121) at (6, 1.25) {};
		\node [style=none] (122) at (5, 1.25) {};
		\node [style=none] (123) at (5, 0.25) {};
		\node [style=none] (124) at (6, 0.25) {};
		\node [style=wn] (125) at (6, 0.75) {};
		\node [style=wn] (126) at (7, -0.25) {};
		\node [style=bn] (127) at (6.5, 1.25) {};
		\node [style=bn] (128) at (6.5, 0.25) {};
		\node [style=wn] (129) at (8, 0.75) {};
		\node [style=bn] (130) at (7.5, 1.25) {};
		\node [style=bn] (131) at (7.5, 0.25) {};
		\node [style=none] (132) at (8, 1.25) {};
		\node [style=none] (133) at (7, 1.25) {};
		\node [style=none] (134) at (7, 0.25) {};
		\node [style=none] (135) at (8, 0.25) {};
		\node [style=wn] (136) at (9, -0.25) {};
		\node [style=wn] (137) at (10, -1.25) {};
		\node [style=wn] (138) at (10, 0.75) {};
		\node [style=bn] (139) at (9.5, 0.25) {};
		\node [style=bn] (140) at (9.5, -0.75) {};
		\node [style=wn] (141) at (11, -0.25) {};
		\node [style=bn] (142) at (10.5, 0.25) {};
		\node [style=bn] (143) at (10.5, -0.75) {};
		\node [style=none] (144) at (11, 0.25) {};
		\node [style=none] (145) at (10, 0.25) {};
		\node [style=none] (146) at (10, -0.75) {};
		\node [style=none] (147) at (11, -0.75) {};
		\node [style=wn] (148) at (7, -0.25) {};
		\node [style=wn] (149) at (8, -1.25) {};
		\node [style=wn] (150) at (8, 0.75) {};
		\node [style=bn] (151) at (7.5, 0.25) {};
		\node [style=bn] (152) at (7.5, -0.75) {};
		\node [style=wn] (153) at (9, -0.25) {};
		\node [style=bn] (154) at (8.5, 0.25) {};
		\node [style=bn] (155) at (8.5, -0.75) {};
		\node [style=none] (156) at (9, 0.25) {};
		\node [style=none] (157) at (8, 0.25) {};
		\node [style=none] (158) at (8, -0.75) {};
		\node [style=none] (159) at (9, -0.75) {};
		\node [style=wn] (160) at (11, -0.25) {};
		\node [style=wn] (161) at (12, -1.25) {};
		\node [style=wn] (162) at (12, 0.75) {};
		\node [style=bn] (163) at (11.5, 0.25) {};
		\node [style=bn] (164) at (11.5, -0.75) {};
		\node [style=wn] (165) at (13, -0.25) {};
		\node [style=bn] (166) at (12.5, 0.25) {};
		\node [style=bn] (167) at (12.5, -0.75) {};
		\node [style=none] (168) at (13, 0.25) {};
		\node [style=none] (169) at (12, 0.25) {};
		\node [style=none] (170) at (12, -0.75) {};
		\node [style=none] (171) at (13, -0.75) {};
		\node [style=wn] (172) at (8, 0.75) {};
		\node [style=wn] (173) at (9, -0.25) {};
		\node [style=bn] (174) at (8.5, 1.25) {};
		\node [style=bn] (175) at (8.5, 0.25) {};
		\node [style=wn] (176) at (10, 0.75) {};
		\node [style=bn] (177) at (9.5, 1.25) {};
		\node [style=bn] (178) at (9.5, 0.25) {};
		\node [style=none] (179) at (10, 1.25) {};
		\node [style=none] (180) at (9, 1.25) {};
		\node [style=none] (181) at (9, 0.25) {};
		\node [style=none] (182) at (10, 0.25) {};
		\node [style=wn] (183) at (10, 0.75) {};
		\node [style=wn] (184) at (11, -0.25) {};
		\node [style=bn] (185) at (10.5, 1.25) {};
		\node [style=bn] (186) at (10.5, 0.25) {};
		\node [style=wn] (187) at (12, 0.75) {};
		\node [style=bn] (188) at (11.5, 1.25) {};
		\node [style=bn] (189) at (11.5, 0.25) {};
		\node [style=none] (190) at (12, 1.25) {};
		\node [style=none] (191) at (11, 1.25) {};
		\node [style=none] (192) at (11, 0.25) {};
		\node [style=none] (193) at (12, 0.25) {};
		\node [style=bn] (194) at (12.5, 1.25) {};
		\node [style=none] (195) at (13, 1.25) {};
		\node [style=bn] (196) at (-0.5, 1.25) {};
		\node [style=none] (197) at (0, 1.25) {};
	\end{pgfonlayer}
	\begin{pgfonlayer}{edgelayer}
		\draw [style=background] (16.center)
			 to (15.center)
			 to (14.center)
			 to (13.center)
			 to cycle;
		\draw (5) to (12);
		\draw (12) to (2);
		\draw (2) to (9);
		\draw (9) to (6);
		\draw (6) to (24);
		\draw (24) to (17);
		\draw (17) to (25);
		\draw (25) to (5);
		\draw (9) to (27.center);
		\draw (24) to (26.center);
		\draw (25) to (29.center);
		\draw (12) to (28.center);
		\draw (31) to (34);
		\draw (34) to (30);
		\draw (30) to (33);
		\draw (33) to (32);
		\draw (32) to (36);
		\draw (36) to (35);
		\draw (35) to (37);
		\draw (37) to (31);
		\draw (33) to (39.center);
		\draw (36) to (38.center);
		\draw (37) to (41.center);
		\draw (34) to (40.center);
		\draw (43) to (46);
		\draw (46) to (42);
		\draw (42) to (45);
		\draw (45) to (44);
		\draw (44) to (48);
		\draw (48) to (47);
		\draw (47) to (49);
		\draw (49) to (43);
		\draw (45) to (51.center);
		\draw (48) to (50.center);
		\draw (49) to (53.center);
		\draw (46) to (52.center);
		\draw (55) to (58);
		\draw (58) to (54);
		\draw (54) to (57);
		\draw (60) to (59);
		\draw (59) to (61);
		\draw (61) to (55);
		\draw (57) to (63.center);
		\draw (60) to (62.center);
		\draw (61) to (65.center);
		\draw (58) to (64.center);
		\draw (67) to (70);
		\draw (70) to (66);
		\draw (66) to (69);
		\draw (72) to (71);
		\draw (71) to (73);
		\draw (73) to (67);
		\draw (69) to (75.center);
		\draw (72) to (74.center);
		\draw (73) to (77.center);
		\draw (70) to (76.center);
		\draw (79) to (82);
		\draw (82) to (78);
		\draw (78) to (81);
		\draw (81) to (80);
		\draw (80) to (84);
		\draw (84) to (83);
		\draw (83) to (85);
		\draw (85) to (79);
		\draw (81) to (87.center);
		\draw (84) to (86.center);
		\draw (85) to (89.center);
		\draw (82) to (88.center);
		\draw (91) to (94);
		\draw (94) to (90);
		\draw (90) to (93);
		\draw (93) to (92);
		\draw (92) to (96);
		\draw (96) to (95);
		\draw (95) to (97);
		\draw (97) to (91);
		\draw (93) to (99.center);
		\draw (96) to (98.center);
		\draw (97) to (101.center);
		\draw (94) to (100.center);
		\draw (103) to (106);
		\draw (106) to (102);
		\draw (102) to (105);
		\draw (105) to (104);
		\draw (104) to (108);
		\draw (108) to (107);
		\draw (107) to (109);
		\draw (109) to (103);
		\draw (105) to (111.center);
		\draw (108) to (110.center);
		\draw (109) to (113.center);
		\draw (106) to (112.center);
		\draw (115) to (117);
		\draw (117) to (114);
		\draw (114) to (116);
		\draw (119) to (118);
		\draw (118) to (120);
		\draw (120) to (115);
		\draw (116) to (122.center);
		\draw (119) to (121.center);
		\draw (120) to (124.center);
		\draw (117) to (123.center);
		\draw (126) to (128);
		\draw (128) to (125);
		\draw (125) to (127);
		\draw (130) to (129);
		\draw (129) to (131);
		\draw (131) to (126);
		\draw (127) to (133.center);
		\draw (130) to (132.center);
		\draw (131) to (135.center);
		\draw (128) to (134.center);
		\draw (137) to (140);
		\draw (140) to (136);
		\draw (136) to (139);
		\draw (139) to (138);
		\draw (138) to (142);
		\draw (142) to (141);
		\draw (141) to (143);
		\draw (143) to (137);
		\draw (139) to (145.center);
		\draw (142) to (144.center);
		\draw (143) to (147.center);
		\draw (140) to (146.center);
		\draw (149) to (152);
		\draw (152) to (148);
		\draw (148) to (151);
		\draw (151) to (150);
		\draw (150) to (154);
		\draw (154) to (153);
		\draw (153) to (155);
		\draw (155) to (149);
		\draw (151) to (157.center);
		\draw (154) to (156.center);
		\draw (155) to (159.center);
		\draw (152) to (158.center);
		\draw (161) to (164);
		\draw (164) to (160);
		\draw (160) to (163);
		\draw (163) to (162);
		\draw (162) to (166);
		\draw (166) to (165);
		\draw (165) to (167);
		\draw (167) to (161);
		\draw (163) to (169.center);
		\draw (166) to (168.center);
		\draw (167) to (171.center);
		\draw (164) to (170.center);
		\draw (173) to (175);
		\draw (175) to (172);
		\draw (172) to (174);
		\draw (177) to (176);
		\draw (176) to (178);
		\draw (178) to (173);
		\draw (174) to (180.center);
		\draw (177) to (179.center);
		\draw (178) to (182.center);
		\draw (175) to (181.center);
		\draw (184) to (186);
		\draw (186) to (183);
		\draw (183) to (185);
		\draw (188) to (187);
		\draw (187) to (189);
		\draw (189) to (184);
		\draw (185) to (191.center);
		\draw (188) to (190.center);
		\draw (189) to (193.center);
		\draw (186) to (192.center);
		\draw (194) to (195.center);
		\draw (196) to (197.center);
		\draw (54) to (196);
		\draw (187) to (194);
	\end{pgfonlayer}
\end{tikzpicture}

\end{equation}

Properties of this diagram, representing a single truncation of the lattice, may not extend to larger truncations, forcing unwieldy use of ellipses.

This leaves a gap. We want a graphical language in which these families of codes are represented and manipulated as finite graphs that capture their translation-invariant symmetry, reducing computational problems about infinite stabilizer processes to problems in rewriting and graph theory.

We provide such a language by extending the odd-prime-dimensional stabilizer ZX-calculus with a single new generator \(\begin{tikzpicture}
	\begin{pgfonlayer}{nodelayer}
		\node [style=rmat] (1) at (0.75, 0) {\(\delta\)};
		\node [style=none] (2) at (0, 0.4) {};
		\node [style=none] (3) at (1.5, 0.4) {};
		\node [style=none] (4) at (1.5, -0.4) {};
		\node [style=none] (5) at (0, -0.4) {};
		\node [style=none] (6) at (0, 0) {};
		\node [style=none] (7) at (1.5, 0) {};
	\end{pgfonlayer}
	\begin{pgfonlayer}{edgelayer}
		\draw [style=background] (2.center)
			 to (3.center)
			 to (4.center)
			 to (5.center)
			 to cycle;
		\draw [style=background] (5.center) to (2.center);
		\draw [style=background] (2.center) to (3.center);
		\draw [style=background] (3.center) to (4.center);
		\draw [style=background] (4.center) to (5.center);
		\draw (7.center) to (6.center);
	\end{pgfonlayer}
\end{tikzpicture}
\), the delay, which passes data from one time step to the next. Delayed ZX-diagrams describe infinite processes obtained by iterating a finite pattern and connecting successive copies through the delays. For example, the horizontal translation-invariance of the surface code above is captured concisely by the delayed ZX-diagram
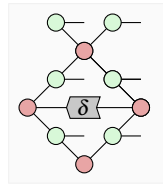
\[\begin{tikzpicture}
	\begin{pgfonlayer}{nodelayer}
		\node [style=wn] (0) at (1.5, -0.25) {};
		\node [style=none] (5) at (-2, 2.5) {};
		\node [style=none] (6) at (-2, -2.25) {};
		\node [style=none] (7) at (2.25, -2.25) {};
		\node [style=none] (8) at (2.25, 2.5) {};
		\node [style=wn] (16) at (-1.5, -0.25) {};
		\node [style=wn] (17) at (0, -1.75) {};
		\node [style=wn] (18) at (0, 1.25) {};
		\node [style=bn] (19) at (-0.75, 0.5) {};
		\node [style=bn] (20) at (-0.75, -1) {};
		\node [style=wn] (21) at (1.5, -0.25) {};
		\node [style=bn] (22) at (0.75, 0.5) {};
		\node [style=bn] (23) at (0.75, -1) {};
		\node [style=none] (24) at (1.5, 0.5) {};
		\node [style=none] (25) at (0, 0.5) {};
		\node [style=none] (26) at (0, -1) {};
		\node [style=none] (27) at (1.5, -1) {};
		\node [style=wn] (40) at (0, 1.25) {};
		\node [style=wn] (41) at (1.5, -0.25) {};
		\node [style=bn] (42) at (0.75, 2) {};
		\node [style=bn] (43) at (0.75, 0.5) {};
		\node [style=none] (48) at (1.5, 2) {};
		\node [style=none] (49) at (1.5, 0.5) {};
		\node [style=bn] (180) at (-0.75, 2) {};
		\node [style=none] (181) at (0, 2) {};
		\node [style=lmat] (182) at (-0.025, -0.25) {\(\delta\)};
	\end{pgfonlayer}
	\begin{pgfonlayer}{edgelayer}
		\draw [style=background] (8.center)
			 to (7.center)
			 to (6.center)
			 to (5.center)
			 to cycle;
		\draw (17) to (20);
		\draw (20) to (16);
		\draw (16) to (19);
		\draw (19) to (18);
		\draw (18) to (22);
		\draw (22) to (21);
		\draw (21) to (23);
		\draw (23) to (17);
		\draw (19) to (25.center);
		\draw (22) to (24.center);
		\draw (23) to (27.center);
		\draw (20) to (26.center);
		\draw (41) to (43);
		\draw (43) to (40);
		\draw (40) to (42);
		\draw (42) to (48.center);
		\draw (43) to (49.center);
		\draw (180) to (181.center);
		\draw (40) to (180);
		\draw (41) to (182);
		\draw (182) to (16);
	\end{pgfonlayer}
\end{tikzpicture}
\]
In this way, a single finite delayed ZX-diagram can represent systems such as spin chains, infinite graph states, lattice foliations, and quantum convolutional codes.

Adding a delay to the ZX-calculus is easy. However, it is not clear when two delayed ZX-diagrams should have the same interpretation.
Following Carette et al.~\cite{delayedtracememory}, we represent a delayed ZX-diagram as a sequence of finite approximations of an infinite process, obtained by iterating the repeating pattern and discarding the boundary:
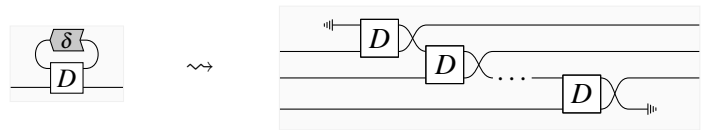
\[\begin{tikzpicture}
	\begin{pgfonlayer}{nodelayer}
		\node [style=none] (93) at (9.25, 0.7875) {\(D\)};
		\node [style=none] (94) at (17.725, 1.6375) {};
		\node [style=none] (95) at (6.6, 1.6375) {};
		\node [style=none] (96) at (17.725, -1.6375) {};
		\node [style=none] (97) at (6.6, -1.6375) {};
		\node [style=none] (98) at (8, 1.1375) {};
		\node [style=none] (99) at (8.75, 1.1375) {};
		\node [style=none] (100) at (9.75, 1.1375) {};
		\node [style=none] (101) at (8.75, 0.4375) {};
		\node [style=none] (102) at (9.75, 0.4375) {};
		\node [style=none] (103) at (10.5, 0.4375) {};
		\node [style=none] (104) at (10.5, 1.1375) {};
		\node [style=none] (105) at (11, 0.0875) {\(D\)};
		\node [style=none] (106) at (11.5, 0.4375) {};
		\node [style=none] (107) at (6.6, -0.2625) {};
		\node [style=none] (108) at (11.5, -0.2625) {};
		\node [style=none] (109) at (12.25, -0.2625) {};
		\node [style=none] (110) at (12.25, 0.4375) {};
		\node [style=none] (111) at (6.6, 0.4375) {};
		\node [style=none] (112) at (14.6, -0.6875) {\(D\)};
		\node [style=none] (113) at (15.1, -0.2625) {};
		\node [style=none] (114) at (6.6, -1.0875) {};
		\node [style=none] (115) at (15.1, -1.0875) {};
		\node [style=none] (116) at (15.85, -1.0875) {};
		\node [style=none] (117) at (15.85, -0.2625) {};
		\node [style=none] (118) at (15.85, 1.1375) {};
		\node [style=none] (119) at (15.85, 0.4375) {};
		\node [style=none] (120) at (12.75, -0.6375) {};
		\node [style=none] (121) at (13.275, -0.2625) {};
		\node [style=none] (122) at (12.75, -0.3125) {\(\cdots\)};
		\node [style=none] (123) at (16.35, -1.0875) {};
		\node [style=none] (124) at (17.725, -0.2625) {};
		\node [style=none] (125) at (17.725, 1.1375) {};
		\node [style=none] (126) at (17.725, 0.4375) {};
		\node [style=none] (127) at (8, 1.1375) {};
		\node [style=groundin] (128) at (8, 1.1375) {};
		\node [style=groundout] (129) at (16.35, -1.0875) {};
		\node [style=none] (130) at (8.75, 1.2875) {};
		\node [style=none] (131) at (8.75, 0.2875) {};
		\node [style=none] (132) at (9.75, 0.2875) {};
		\node [style=none] (133) at (9.75, 1.2875) {};
		\node [style=none] (134) at (10.475, 0.5875) {};
		\node [style=none] (135) at (10.475, -0.4125) {};
		\node [style=none] (136) at (11.475, -0.4125) {};
		\node [style=none] (137) at (11.475, 0.5875) {};
		\node [style=none] (138) at (14.1, -0.1875) {};
		\node [style=none] (139) at (14.1, -1.1875) {};
		\node [style=none] (140) at (15.1, -1.1875) {};
		\node [style=none] (141) at (15.1, -0.1875) {};
		\node [style=none] (142) at (1.475, -0.0125) {};
		\node [style=none] (143) at (1.475, 0.7375) {};
		\node [style=none] (144) at (0.475, -0.0125) {};
		\node [style=none] (145) at (0.475, 0.7375) {};
		\node [style=none] (146) at (-0.525, -0.5125) {};
		\node [style=none] (147) at (2.475, -0.5125) {};
		\node [style=box] (148) at (0.975, -0.2625) {\(D\)};
		\node [style=none] (149) at (-0.525, -0.8875) {};
		\node [style=none] (150) at (2.475, -0.8875) {};
		\node [style=none] (151) at (2.475, 1.1125) {};
		\node [style=none] (152) at (-0.525, 1.1125) {};
		\node [style=lmat] (153) at (0.975, 0.7375) {\(\delta\)};
		\node [style=none] (154) at (4.5, -0.0125) {\(\rightsquigarrow\)};
	\end{pgfonlayer}
	\begin{pgfonlayer}{edgelayer}
		\draw [style=background] (95.center)
			 to (94.center)
			 to (96.center)
			 to (97.center)
			 to cycle;
		\draw (98.center)
			 to (99.center)
			 to (100.center)
			 to [in=180, out=0] (103.center)
			 to (106.center)
			 to [in=180, out=0] (109.center);
		\draw (121.center)
			 to (113.center)
			 to [in=180, out=0] (116.center)
			 to (123.center);
		\draw (125.center)
			 to (118.center)
			 to (104.center)
			 to [in=0, out=180] (102.center)
			 to (101.center)
			 to (111.center);
		\draw (107.center)
			 to (108.center)
			 to [in=180, out=0] (110.center)
			 to (119.center)
			 to (126.center);
		\draw (114.center)
			 to (115.center)
			 to [in=180, out=0] (117.center)
			 to (124.center);
		\draw [style=bbox] (130.center)
			 to (133.center)
			 to (132.center)
			 to (131.center)
			 to cycle;
		\draw [style=bbox] (130.center) to (133.center);
		\draw [style=bbox] (133.center) to (132.center);
		\draw [style=bbox] (132.center) to (131.center);
		\draw [style=bbox] (131.center) to (130.center);
		\draw [style=bbox] (134.center)
			 to (137.center)
			 to (136.center)
			 to (135.center)
			 to cycle;
		\draw [style=bbox] (134.center) to (137.center);
		\draw [style=bbox] (137.center) to (136.center);
		\draw [style=bbox] (136.center) to (135.center);
		\draw [style=bbox] (135.center) to (134.center);
		\draw [style=bbox] (138.center)
			 to (141.center)
			 to (140.center)
			 to (139.center)
			 to cycle;
		\draw [style=bbox] (138.center) to (141.center);
		\draw [style=bbox] (141.center) to (140.center);
		\draw [style=bbox] (140.center) to (139.center);
		\draw [style=bbox] (139.center) to (138.center);
		\draw [style=background] (152.center)
			 to (151.center)
			 to (150.center)
			 to (149.center)
			 to cycle;
		\draw (146.center) to (147.center);
		\draw (145.center)
			 to [bend left=270, looseness=1.50] (144.center)
			 to (142.center)
			 to [bend right=90, looseness=1.50] (143.center)
			 to cycle;
	\end{pgfonlayer}
\end{tikzpicture}
\]
Each finite approximation captures only part of the infinite process, and different sequences of truncations may describe the same process. We therefore regard two sequences as equivalent when they approximate one another over time, following our recent work~\cite{obs}. We show that two sequences are equivalent in this sense if and only if they have the same infinite stabilizer group. These finite approximations are themselves quantum channels between finite-dimensional Hilbert spaces, so a delayed stabilizer diagram carries a genuine Hilbert-space semantics into which its stabilizer description embeds faithfully, not merely a combinatorial one.

This representation in terms of infinite stabilizer groups is difficult to work with directly. Because delayed diagrams are translation-invariant, we can instead borrow a finite presentation from the theory of quantum convolutional codes~\cite{wilde}. We treat the delay as a formal variable, so that infinite sequences of Pauli operators are represented by formal Laurent series in this variable. The stabilizers of a delayed diagram are then captured by a finite \emph{generating tableau} of fractions of polynomials, whose geometric series recover these infinite sequences of Pauli operators.

This more tractable semantics lets us give an equational theory for delayed stabilizer ZX-diagrams, the delayed stabilizer ZX-calculus \(\dStabZX\). Its spiders, multipliers, and H-boxes are labelled by fractions of polynomials that encode the translation-invariant structure, and it features generalised Euler decomposition and colour change rules.

We introduce a scalable notation that generalises Haah's polynomial representation of translation-invariant graph states~\cite{Haah2017,Haah2021,Haah2013} as the phases of scalable spiders. For example, the following delayed graph state reduces to a single scalable spider:
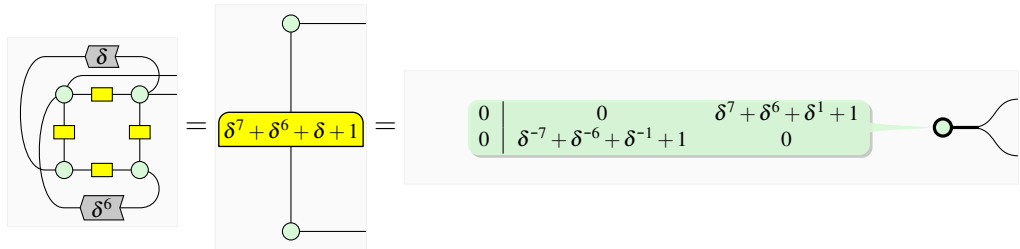
\[\begin{tikzpicture}
	\begin{pgfonlayer}{nodelayer}
		\node [style=none] (0) at (-0.5, 2.375) {};
		\node [style=none] (1) at (-0.5, -2.625) {};
		\node [style=none] (2) at (4, 2.375) {};
		\node [style=none] (3) at (4, -2.625) {};
		\node [style=bn] (4) at (1, 0.875) {};
		\node [style=bn] (5) at (1, -1.125) {};
		\node [style=bn] (6) at (3, -1.125) {};
		\node [style=none] (7) at (3, 0.875) {};
		\node [style=had] (8) at (2, 0.875) {};
		\node [style=had] (9) at (1, -0.125) {};
		\node [style=had] (10) at (2, -1.125) {};
		\node [style=had] (11) at (3, -0.125) {};
		\node [style=none] (12) at (4, 0.875) {};
		\node [style=lmat] (14) at (2, 1.875) {\(\delta\)};
		\node [style=none] (16) at (3, 1.875) {};
		\node [style=none] (17) at (0.5, 1.875) {};
		\node [style=none] (18) at (0.5, -1.125) {};
		\node [style=bn] (19) at (3, 0.875) {};
		\node [style=none] (20) at (4, 1.375) {};
		\node [style=none] (21) at (1.5, 1.375) {};
		\node [style=none] (141) at (9.5, 0) {\(=\)};
		\node [style=none] (142) at (10, 1.5) {};
		\node [style=none] (143) at (10, -1.5) {};
		\node [style=none] (144) at (26.25, 1.5) {};
		\node [style=none] (145) at (26.25, -1.5) {};
		\node [style=none] (147) at (26.25, -0.75) {};
		\node [style=none] (148) at (26.25, 0.75) {};
		\node [style=none] (149) at (25, 0) {};
		\node [style=gbn] (150) at (24.25, 0) {};
		\node [style=none] (151) at (25.2, 0) {};
		\node [style=bphase] (152) at (17, 0) {\(\phasemat{0 \\ 0}{0 & \delta^{7}+\delta^{6} + \delta^{1}  + 1 \\ \delta^{\minu 7}+\delta^{\minu 6} + \delta^{\minu 1}  + 1 & 0}\)};
		\node [style=none] (158) at (3, -1.125) {};
		\node [style=none] (160) at (1, 0.875) {};
		\node [style=lmat] (167) at (2, -2.125) {\(\delta^6\)};
		\node [style=none] (168) at (1, -2.125) {};
		\node [style=none] (169) at (3, -2.125) {};
		\node [style=none] (206) at (4.5, 0) {\(=\)};
		\node [style=none] (207) at (5, 3.25) {};
		\node [style=none] (208) at (5, -3.25) {};
		\node [style=none] (209) at (9, 3.25) {};
		\node [style=none] (210) at (9, -3.25) {};
		\node [style=none] (211) at (9, -2.75) {};
		\node [style=none] (212) at (9, 2.75) {};
		\node [style=bn] (213) at (7, 2.75) {};
		\node [style=bn] (214) at (7, -2.75) {};
		\node [style=uhad] (216) at (7, -0.05) {\(\delta^{7}+\delta^6+\delta +1\)};
	\end{pgfonlayer}
	\begin{pgfonlayer}{edgelayer}
		\draw [style=background] (1.center)
			 to (0.center)
			 to (2.center)
			 to (3.center)
			 to cycle;
		\draw (4) to (7.center);
		\draw (7.center) to (6);
		\draw (6) to (5);
		\draw (5) to (4);
		\draw (7.center)
			 to [in=0, out=0, looseness=1.75] (16.center)
			 to (17.center)
			 to [bend left=270, looseness=0.75] (18.center)
			 to (5);
		\draw (20.center)
			 to (21.center)
			 to [in=75, out=180] (4);
		\draw (19) to (12.center);
		\draw [style=background] (144.center)
			 to (145.center)
			 to (143.center)
			 to (142.center)
			 to cycle;
		\draw [in=-180, out=0, looseness=1.25] (149.center) to (148.center);
		\draw [in=180, out=0, looseness=1.25] (149.center) to (147.center);
		\draw [style=thick] (151.center) to (150);
		\draw [style=rbphase] (152) to (150);
		\draw (160.center)
			 to [in=-180, out=180, looseness=0.75] (168.center)
			 to (169.center)
			 to [in=0, out=0, looseness=1.75] (158.center);
		\draw [style=background] (210.center)
			 to (208.center)
			 to (207.center)
			 to (209.center)
			 to cycle;
		\draw (214) to (211.center);
		\draw (213) to (212.center);
		\draw (213) to (216);
		\draw (216) to (214);
	\end{pgfonlayer}
\end{tikzpicture}
\]
which unrolls into the following infinite graph state (where dotted lines denote wires connected by a Hadamard gate):

\par\noindent
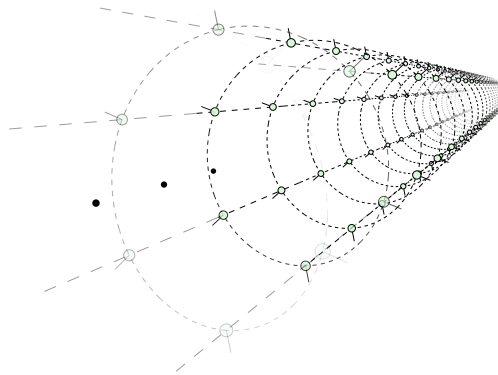
\makebox[\linewidth][c]{%
  \scalebox{0.45}{%
    
\begin{tikzpicture}[x=1cm,y=1cm,line cap=round]

  \pgfdeclarelayer{L0}
  \pgfdeclarelayer{L1}
  \pgfdeclarelayer{L2}
  \pgfdeclarelayer{L3}
  \pgfsetlayers{L0,L1,L2,L3,main}

  \def\R{2.4}
  \def\L{42}
  \def\turns{14}
  \def\nlines{6}
  \def\roll{12}

  \def\yaw{-69}
  \def\pitch{-6}
  \def\camdist{13}
  \def\pivotU{1.2}
  \def\focal{30}

  \def\nodescale{2}
  \def\tickfac{.8}

  \def\edgewidth{.4}
  \def\edgemin{.001}
  \def\dashpersegment{15}
  \def\dashfrac{.5}
  \def\linesegs{15}

  \def\knotsPerTurn{10}
  \def\frontExtTurns{1.9}
  \def\backExtTurns{6.5}

  \def\rearFadeStartFrac{0.5}
  \def\rearFadePower{.4}

  \def\dotrad{.1}
  \def\dotmin{.0000}
  \def\dotgap{1.2}
  \def\ndots{3}
  \def\frontDotStartU{-4}

  \def\bboxLeft{-21.2}
  \def\bboxRight{10.8}
  \def\bboxBottom{-14}
  \def\bboxTop{7.8}

  \tikzset{
    edge/.style={
      draw=hadamardedgecolor,
      line width=\lw pt,
      dash pattern=on \don cm off \doff cm
    },
    tick/.style={
      draw=hadamardedgecolor,
      line width=\lw pt
    }
  }

  \pgfmathsetmacro{\cyaw}{cos(\yaw)}
  \pgfmathsetmacro{\syaw}{sin(\yaw)}
  \pgfmathsetmacro{\cpit}{cos(\pitch)}
  \pgfmathsetmacro{\spit}{sin(\pitch)}
  \pgfmathsetmacro{\dphidu}{-360*\turns/\L}
  \pgfmathsetmacro{\kdeg}{pi/180}

  \pgfmathsetmacro{\frontExtU}{\frontExtTurns*\L/\turns}
  \pgfmathsetmacro{\backExtU}{\backExtTurns*\L/\turns}
  \pgfmathsetmacro{\uLo}{-\frontExtU}
  \pgfmathsetmacro{\uHi}{\L+\backExtU}

  \pgfmathsetmacro{\fadeStart}{\rearFadeStartFrac*\L}

  \newcommand{\proj}[2]{%
    \pgfmathsetmacro{\uu}{(#1)-\pivotU}%
    \pgfmathsetmacro{\Yw}{\R*cos(#2)}%
    \pgfmathsetmacro{\Zw}{\R*sin(#2)}%
    \pgfmathsetmacro{\Qx}{\cyaw*\uu + \syaw*\Zw}%
    \pgfmathsetmacro{\Qy}{\spit*\syaw*\uu + \cpit*\Yw - \spit*\cyaw*\Zw}%
    \pgfmathsetmacro{\Qz}{-\cpit*\syaw*\uu + \spit*\Yw + \cpit*\cyaw*\Zw + \camdist}%
    \pgfmathsetmacro{\px}{\focal*\Qx/\Qz}%
    \pgfmathsetmacro{\py}{\focal*\Qy/\Qz}%
    \pgfmathsetmacro{\pz}{\Qz}%
  }

  \newcommand{\projaxis}[1]{%
    \pgfmathsetmacro{\uu}{(#1)-\pivotU}%
    \pgfmathsetmacro{\Qzz}{-\cpit*\syaw*\uu + \camdist}%
    \pgfmathsetmacro{\axx}{\focal*(\cyaw*\uu)/\Qzz}%
    \pgfmathsetmacro{\axy}{\focal*(\spit*\syaw*\uu)/\Qzz}%
  }

  \newcommand{\setvis}[2]{%
    \proj{#1}{#2}%
    \pgfmathsetmacro{\Ncx}{\syaw*sin(#2)}%
    \pgfmathsetmacro{\Ncy}{\cpit*cos(#2) - \spit*\cyaw*sin(#2)}%
    \pgfmathsetmacro{\Ncz}{\spit*cos(#2) + \cpit*\cyaw*sin(#2)}%
    \pgfmathtruncatemacro{\vis}{(\Ncx*\Qx+\Ncy*\Qy+\Ncz*\Qz)<0 ? 1 : 0}%
  }

  \newcommand{\projT}[1]{%
    \pgfmathsetmacro{\phi}{\dphidu*(#1)+\roll}%
    \proj{#1}{\phi}%
    \pgfmathsetmacro{\Ypw}{-\R*sin(\phi)*\kdeg*\dphidu}%
    \pgfmathsetmacro{\Zpw}{ \R*cos(\phi)*\kdeg*\dphidu}%
    \pgfmathsetmacro{\Qxp}{\cyaw + \syaw*\Zpw}%
    \pgfmathsetmacro{\Qyp}{\spit*\syaw + \cpit*\Ypw - \spit*\cyaw*\Zpw}%
    \pgfmathsetmacro{\Qzp}{-\cpit*\syaw + \spit*\Ypw + \cpit*\cyaw*\Zpw}%
    \pgfmathsetmacro{\tx}{\focal*(\Qxp*\Qz - \Qx*\Qzp)/(\Qz*\Qz)}%
    \pgfmathsetmacro{\ty}{\focal*(\Qyp*\Qz - \Qy*\Qzp)/(\Qz*\Qz)}%
  }

  \newcommand{\setwidth}{%
    \pgfmathsetmacro{\lw}{max(\edgemin,\edgewidth*\camdist/\pz)}%
  }

  \newcommand{\setdash}[1]{%
    \pgfmathsetmacro{\period}{(#1)/\dashpersegment}%
    \pgfmathsetmacro{\don}{max(.015,\period*\dashfrac)}%
    \pgfmathsetmacro{\doff}{max(.015,\period*(1-\dashfrac))}%
  }

  \newcommand{\fadeop}[1]{%
    \pgfmathsetmacro{\fadeT}{max(0,min(1,((#1)-\fadeStart)/(\uHi-\fadeStart)))}%
    \pgfmathsetmacro{\op}{%
      ((#1)<0)
        ? max(0,((#1)+\frontExtU)/\frontExtU)
        : (((#1)<\fadeStart)
            ? 1
            : max(0,1 - pow(\fadeT,\rearFadePower)))%
    }%
  }

  \newcommand{\helixarc}[3]{%
    \pgfmathsetmacro{\hh}{(#2)-(#1)}%

    \projT{#1}%
    \pgfmathsetmacro{\Pax}{\px}%
    \pgfmathsetmacro{\Pay}{\py}%
    \pgfmathsetmacro{\Cax}{\px+\hh/3*\tx}%
    \pgfmathsetmacro{\Cay}{\py+\hh/3*\ty}%

    \projT{#2}%
    \pgfmathsetmacro{\Pbx}{\px}%
    \pgfmathsetmacro{\Pby}{\py}%
    \pgfmathsetmacro{\Cbx}{\px-\hh/3*\tx}%
    \pgfmathsetmacro{\Cby}{\py-\hh/3*\ty}%

    \pgfmathsetmacro{\umid}{((#1)+(#2))/2}%
    \setvis{\umid}{\dphidu*\umid+\roll}%
    \setwidth%

    \pgfmathsetmacro{\seglen}{%
      veclen(\Cax-\Pax,\Cay-\Pay)
      + veclen(\Cbx-\Cax,\Cby-\Cay)
      + veclen(\Pbx-\Cbx,\Pby-\Cby)
    }%
    \setdash{\seglen}%

    \ifnum\vis=1
      \def\lay{L2}%
    \else
      \def\lay{L0}%
    \fi

    \begin{pgfonlayer}{\lay}
      \draw[edge,opacity=#3]
        (\Pax,\Pay) .. controls (\Cax,\Cay) and (\Cbx,\Cby) .. (\Pbx,\Pby);
    \end{pgfonlayer}
  }

  \newcommand{\taperedline}[3]{%
    \pgfmathtruncatemacro{\NS}{\linesegs}%
    \foreach \q in {0,...,\numexpr\NS-1\relax}{%
      \pgfmathsetmacro{\ula}{(#1)+((#2)-(#1))*\q/\NS}%
      \pgfmathsetmacro{\ulb}{(#1)+((#2)-(#1))*(\q+1)/\NS}%

      \proj{\ula}{#3}%
      \pgfmathsetmacro{\Sax}{\px}%
      \pgfmathsetmacro{\Say}{\py}%

      \proj{\ulb}{#3}%
      \pgfmathsetmacro{\Sbx}{\px}%
      \pgfmathsetmacro{\Sby}{\py}%

      \pgfmathsetmacro{\umm}{(\ula+\ulb)/2}%
      \proj{\umm}{#3}%
      \setwidth%
      \fadeop{\umm}%

      \pgfmathsetmacro{\seglen}{veclen(\Sbx-\Sax,\Sby-\Say)}%
      \setdash{\seglen}%

      \draw[edge,opacity=\op] (\Sax,\Say)--(\Sbx,\Sby);%
    }%
  }

  \pgfmathsetmacro{\duK}{\L/\turns/\knotsPerTurn}
  \pgfmathtruncatemacro{\Ntot}{round((\uHi-\uLo)/\duK)}

  \foreach \s in {0,...,\numexpr\Ntot-1\relax}{%
    \pgfmathsetmacro{\ua}{\uLo+\s*\duK}%
    \pgfmathsetmacro{\um}{\ua+\duK/2}%
    \fadeop{\um}%
    \helixarc{\ua}{\ua+\duK}{\op}%
  }

  \foreach \i in {0,...,\numexpr\nlines-1\relax}{%
    \pgfmathsetmacro{\phl}{360*\i/\nlines+\roll}%
    \setvis{\L/2}{\phl}%

    \ifnum\vis=1
      \def\lay{L2}%
    \else
      \def\lay{L0}%
    \fi

    \begin{pgfonlayer}{\lay}
      \taperedline{\uLo}{\uHi}{\phl}
    \end{pgfonlayer}
  }

  \newcommand{\placebn}[2]{%
    \pgfmathsetmacro{\uN}{\L*((#2)-(#1)/\nlines)/\turns}%

    \ifdim\uN pt<\uLo pt\relax
    \else
      \ifdim\uN pt>\uHi pt\relax
      \else
        \pgfmathsetmacro{\phl}{360*(#1)/\nlines+\roll}%
        \setvis{\uN}{\phl}%
        \projaxis{\uN}%
        \setwidth%
        \fadeop{\uN}%

        \pgfmathsetmacro{\nscl}{\nodescale*\camdist/\pz}%
        \pgfmathsetmacro{\dxn}{\px-\axx}%
        \pgfmathsetmacro{\dyn}{\py-\axy}%
        \pgfmathsetmacro{\nrm}{sqrt(\dxn*\dxn+\dyn*\dyn)+.0001}%
        \pgfmathsetmacro{\tln}{\tickfac*\camdist/\pz}%

        \ifnum\vis=1
          \def\lay{L3}%
        \else
          \def\lay{L1}%
        \fi

        \begin{pgfonlayer}{\lay}
          \draw[tick,opacity=\op]
            (\px,\py)--($(\px,\py)+(\tln*\dxn/\nrm,\tln*\dyn/\nrm)$);
          \node[style=bn,transform shape,scale=\nscl,opacity=\op] at (\px,\py) {};
        \end{pgfonlayer}
      \fi
    \fi
  }

  \pgfmathtruncatemacro{\kLo}{floor(\turns*\uLo/\L)-1}
  \pgfmathtruncatemacro{\kHi}{ceil(\turns*\uHi/\L)+1}

  \foreach \i in {0,...,\numexpr\nlines-1\relax}{%
    \foreach \k in {\kLo,...,\kHi}{%
      \placebn{\i}{\k}%
    }%
  }

  \newcommand{\axisdot}[1]{%
    \projaxis{#1}%
    \pgfmathsetmacro{\rr}{max(\dotmin,\dotrad*\camdist/\Qzz)}%
    \fill (\axx,\axy) circle[radius=\rr];%
  }

  \foreach \j in {0,...,\numexpr\ndots-1\relax}{%
    \axisdot{\frontDotStartU-\j*\dotgap}%
  }

  \pgfresetboundingbox
  \path[use as bounding box]
    (\bboxLeft,\bboxBottom) rectangle (\bboxRight,\bboxTop);

\end{tikzpicture}%
  }%
}
\par

Within this notation we derive generalised forms of local complementation and pivoting:
\[\begin{tikzpicture}
	\begin{pgfonlayer}{nodelayer}
		\node [style=none] (0) at (-10.3, 1) {};
		\node [style=none] (1) at (-1, 1) {};
		\node [style=none] (2) at (-1, -1) {};
		\node [style=none] (3) at (-10.3, -1) {};
		\node [style=gbn] (4) at (-5.75, 0) {};
		\node [style=none] (5) at (-3.5, 0) {};
		\node [style=gbn] (6) at (-1.5, 0.5) {};
		\node [style=none] (7) at (-1, -0.5) {};
		\node [style=none] (8) at (-2, -0.5) {};
		\node [style=wirelable] (10) at (-4.625, 0.25) {$m+n$};
		\node [style=none] (11) at (-2, 0.5) {};
		\node [style=wirelable] (12) at (-2, -0.25) {$n$};
		\node [style=wirelable] (13) at (-2, 0.25) {$m$};
		\node [style=bphase] (9) at (-8.3, 0) {$\phasemat{\vb{x} \\ \vb{y}}{X & \bar E \\ E^\trans & Y}$};
	\end{pgfonlayer}
	\begin{pgfonlayer}{edgelayer}
		\draw [style=background] (2.center)
			 to (1.center)
			 to (0.center)
			 to (3.center)
			 to [in=180, out=0] cycle;
		\draw [style=thick] (7.center)
			 to (8.center)
			 to [in=0, out=-180] (5.center)
			 to (4);
		\draw [style=thick] (6)
			 to (11.center)
			 to [in=0, out=180] (5.center)
			 to (4);
		\draw [style=rbphase] (9) to (4);
	\end{pgfonlayer}
\end{tikzpicture}
 = \begin{tikzpicture}
	\begin{pgfonlayer}{nodelayer}
		\node [style=none] (17) at (1, 1) {};
		\node [style=none] (18) at (8.5, 1) {};
		\node [style=none] (19) at (8.5, -1) {};
		\node [style=none] (20) at (1, -1) {};
		\node [style=gbn] (15) at (4.5, -0.5) {};
		\node [style=none] (16) at (8.5, -0.5) {};
		\node [style=bphase] (21) at (4.5, 0.375) {$\phasemat{\vb{y} + \bar E X^{\minu 1}\vb{x}}{Y - \bar E X^{\minu 1}E^\trans}$};
	\end{pgfonlayer}
	\begin{pgfonlayer}{edgelayer}
		\draw [style=background] (19.center)
			 to (18.center)
			 to (17.center)
			 to (20.center)
			 to [in=180, out=0] cycle;
		\draw [style=thick] (15) to (16.center);
		\draw [style=dbphase] (21) to (15);
	\end{pgfonlayer}
\end{tikzpicture}
.\]
These rules reduce every delayed ZX-diagram to a unique normal form, proving that \(\dStabZX\) is sound, universal, and complete for the generating tableau semantics.

\paragraph*{Contributions.}
The main contributions of this paper are as follows.
\begin{enumerate}
  \item We introduce the directed mixed stabilizer ZX-calculus (\Cref{sec:mixed}), which extends the mixed ZX-calculus of Carette et al.~\cite{Carette2021} with directed rewrite rules that capture when one diagram approximates another (\Cref{def:mixed_stabzx} and \Cref{prop:stabzxdisc_alrdisc_equiv}).
  \item We give a semantics for the behaviour of stateful quantum processes as sequences of finite approximations by completely positive maps between finite-dimensional Hilbert spaces, where two sequences are identified when they contain the same information content (\Cref{sec:beh}).
  \item We show that translation-invariant stabilizer processes embed faithfully into these behaviours of completely positive maps (\Cref{prop:unroll_obs_functor,cor:obs_faithful}), and that each such behaviour is determined by a single infinite stabilizer group, capturing the infinite flow of Pauli operators (\Cref{thm:lim_faithful} and \Cref{cor:unroll}).
  \item For translation-invariant processes, we give a finite generating-tableau semantics whose geometric-series expansion recovers the infinite stabilizer group (\Cref{sec:salr}). Composition is preserved only up to approximation: the composition of infinite stabilizer groups approximates the composition of tableaux (\Cref{thm:gamma_oplax} and \Cref{prop:gamma_ext_approx}).
  \item We introduce the delayed stabilizer ZX-calculus \(\dStabZX\) (\Cref{sec:dzx}), which extends the odd-prime-dimensional ZX-calculus with the delay generator. It provides translation-invariant analogues of the generators (\Cref{subsec:dzx_derived_generators}) and equations (\Cref{subsec:dzx_axioms}) of the stabilizer ZX-calculus.
  \item Using generalised forms of local complementation and pivoting, we reduce every \(\dStabZX\)-diagram to a unique normal form (\Cref{subsec:dzx_scalable_normal_form}, \Cref{prop:dzx_schur_complementation}, \Cref{def:dzx_reduced_ap_form}, and \Cref{lem:AP_unique}). This proves that \(\dStabZX\) is sound, universal, and complete for the generating tableau semantics (\Cref{thm:dzx_completeness}).
\end{enumerate}

  \newpage
  \tableofcontents
  \newpage

  \section{Pure stabilizer quantum mechanics}\label{sec:stab}
  \pratendSetLocal{category=stab, end, restate}
  Stabilizer quantum mechanics is the fundamental structure which underlies the vast majority of finite-dimensional quantum error correction. See for example the seminal and highly influential work of Gottesman~\cite{gottesman_stabilizer_1997}.  The stabilizer formalism of quantum error correction exploits the symplectic representation of stabilizer quantum mechanics in order to redundantly store, manipulate and extract quantum information.

In this section, we first recall the pure stabilizer subtheory of quantum mechanics in terms of linear maps between Hilbert spaces, followed by its symplectic representation in terms of \(\F_p\) linear algebra, ending with the  corresponding graphical language. To formally relate all three pictures of stabilizer quantum mechanics, we will construct them as instances of the following structure

\begin{definition}
  A \textbf{strict \(\dag\)-compact closed category} (\textbf{\dag-CCC}) is a strict symmetric monoidal category \((\mathcal{C},\otimes,I)\), with symmetry \(\swap_{X,Y}:X\otimes Y\to Y\otimes X\) drawn as \(\begin{tikzpicture}
	\begin{pgfonlayer}{nodelayer}
		\node [style=none] (7) at (0.5, 0.15) {};
		\node [style=none] (9) at (1.25, 0.15) {};
		\node [style=none] (10) at (0.5, -0.15) {};
		\node [style=none] (11) at (1.25, -0.15) {};
		\node [style=none] (18) at (1.25, 0.325) {};
		\node [style=none] (19) at (1.25, -0.325) {};
		\node [style=none] (20) at (0.5, 0.325) {};
		\node [style=none] (21) at (0.5, -0.325) {};
	\end{pgfonlayer}
	\begin{pgfonlayer}{edgelayer}
		\draw [in=-180, out=0, looseness=1.50] (7.center) to (11.center);
		\draw [in=0, out=-180, looseness=1.50] (9.center) to (10.center);
		\draw [style=backgroundborderless] (18.center)
			 to (19.center)
			 to (21.center)
			 to (20.center)
			 to cycle;
		\draw [style=backgroundborderless] (20.center) to (18.center);
		\draw [style=backgroundborderless] (18.center) to (19.center);
		\draw [style=backgroundborderless] (19.center) to (21.center);
		\draw [style=backgroundborderless] (21.center) to (20.center);
	\end{pgfonlayer}
\end{tikzpicture}
\), equipped with:
  \begin{itemize}
    \item a \textbf{dagger}: an identity-on-objects involutive contravariant strict monoidal functor \((-)^\dag:\mathcal{C}^{\mathrm{op}}\to\mathcal{C}\) that is compatible with the symmetry, in that \(\swap_{X,Y}^\dag=\swap_{Y,X}\);
    \item for each object \(X\) a \textbf{dual} \(X^\ast\), with \(I^\ast=I\) and \((X\otimes Y)^\ast=Y^\ast\otimes X^\ast\), together with a \textbf{cup} \(\cup_X\coloneqq\begin{tikzpicture}
	\begin{pgfonlayer}{nodelayer}
		\node [style=none] (7) at (0.875, 0.175) {};
		\node [style=none] (18) at (0.5, 0.3) {};
		\node [style=none] (19) at (0.5, -0.3) {};
		\node [style=none] (20) at (1.25, 0.3) {};
		\node [style=none] (21) at (1.25, -0.3) {};
		\node [style=none] (22) at (0.875, -0.175) {};
		\node [style=none] (23) at (1.25, 0.175) {};
		\node [style=none] (24) at (1.25, -0.175) {};
	\end{pgfonlayer}
	\begin{pgfonlayer}{edgelayer}
		\draw (24.center)
			 to (22.center)
			 to [in=180, out=180, looseness=1.50] (7.center)
			 to (23.center);
		\draw [style=backgroundborderless] (19.center)
			 to (18.center)
			 to (20.center)
			 to (21.center)
			 to cycle;
		\draw [style=backgroundborderless] (18.center) to (20.center);
		\draw [style=backgroundborderless] (20.center) to (21.center);
		\draw [style=backgroundborderless] (21.center) to (19.center);
		\draw [style=backgroundborderless] (19.center) to (18.center);
	\end{pgfonlayer}
\end{tikzpicture}
:I\to X\otimes X^\ast\) and a \textbf{cap} \(\cap_X\coloneqq\begin{tikzpicture}
	\begin{pgfonlayer}{nodelayer}
		\node [style=none] (7) at (0.875, 0.175) {};
		\node [style=none] (18) at (1.25, 0.3) {};
		\node [style=none] (19) at (1.25, -0.3) {};
		\node [style=none] (20) at (0.5, 0.3) {};
		\node [style=none] (21) at (0.5, -0.3) {};
		\node [style=none] (22) at (0.875, -0.175) {};
		\node [style=none] (23) at (0.5, 0.175) {};
		\node [style=none] (24) at (0.5, -0.175) {};
	\end{pgfonlayer}
	\begin{pgfonlayer}{edgelayer}
		\draw (24.center)
			 to (22.center)
			 to [in=0, out=0, looseness=1.50] (7.center)
			 to (23.center);
		\draw [style=backgroundborderless] (20.center)
			 to (18.center)
			 to (19.center)
			 to (21.center)
			 to cycle;
		\draw [style=backgroundborderless] (20.center) to (18.center);
		\draw [style=backgroundborderless] (18.center) to (19.center);
		\draw [style=backgroundborderless] (19.center) to (21.center);
		\draw [style=backgroundborderless] (21.center) to (20.center);
	\end{pgfonlayer}
\end{tikzpicture}
:X^\ast\otimes X\to I\) satisfying \(\cap_X=(\cup_X)^\dag\circ\swap_{X^\ast,X}\).
  \end{itemize}
  Modulo the following equations allowing us to bend wires and slide maps around freely:
  \[\input{stab/strings/coherences.tikz}\]
\end{definition}

\subsection{The stabilizer formalism}
\label{section:stabilizer:hilbert}
We first recall the basic theory of quopit stabilizer quantum mechanics (see Gross~\cite{Gross} for reference).

Fix an odd prime \(p\), and write \(\F_p\) for the field of integers modulo \(p\). Denote the quopit Hilbert space by \(\mathcal{H}_p\coloneqq \Span\{\ket{x}\mid x\in\F_p\}\) and the character function by \(\chi(x)\coloneqq \exp(i2\pi x/p)\).

\begin{definition}
The odd prime  \textbf{quopit Pauli group}  \(\mathcal{P}_p\subseteq \mathcal{U}(\mathcal{H}_p)\) is  generated by the operators
 \[Z \ket{x}\coloneqq \chi(x) \ket{x}\qquad\text{and}\qquad X \ket{x} \coloneqq \ket{x+1}.\]
\end{definition}

The \emph{stabilizer formalism} of quantum error correction begins with the following observation:
\begin{lemma}
   Take a maximal Abelian subgroup \(S \subseteq \mathcal{P}_p^{\otimes n}\)
   such that \(\chi(a) 1_{\mathcal{H}_p}^{\otimes n} \in S\) if and only if
   \(a \equiv 0\mod p\). Up to global phase \(\exp(2\pi i \theta)\),  \(S\) uniquely determines a pure quantum
   state \(\ket{S} \in \mathcal{H}_p^{\otimes n}\) such that \(s \ket{S} =
   \ket{S}\) for all \(s \in S\), called a \textbf{stabilizer state} associated to the \textbf{stabilizer group} \(S\).
\end{lemma}

\begin{example} Consider the following stabilizer states and their stabilizer groups:
\begin{itemize}
  \item The computational basis state \(\ket{a}\) is the stabilizer state of \(\langle \chi(-a)Z\rangle\subseteq \mathcal{P}_p\).
  \item The Hadamard basis state \(\ket{+}\coloneqq \sfrac{1}{\sqrt{p}}\sum_{x\in\F_p}\ket{x}\) is the stabilizer state of \(\langle X\rangle\subseteq \mathcal{P}_p\).
  \item The Bell state \(\sfrac{1}{\sqrt{p}}\sum_{x\in\F_p}\ket{x}\ket{x}\) is the stabilizer state of \(\langle X\otimes X,\; Z\otimes Z^{\minu 1}\rangle\subseteq \mathcal{P}_p^{\otimes 2}\).
\end{itemize}
\end{example}

Stabilizer states and their adjoints assemble themselves into a \dag-CCC.

\begin{definition}
The \(\dag\)-CCC \(\Stab\) of quopit stabilizer maps is the \(\dag\)-CC subcategory of \(\FHilb\) generated by stabilizer states, Clifford operators, and the scalars \(\sfrac{1}{\sqrt{p}}\) and \(\sqrt{p}\), under tensor product, composition, and dagger.
\end{definition}

\subsection{The symplectic representation}
\label{section:stabilizer:symplectic}

The stabilizer formalism can be reformulated using
finite-dimensional symplectic linear algebra over \(\Zp\). We recall
the essential definitions and statements here. Throughout this section
we fix the standard symplectic vector space \((\Zp^{2n},\omega_n)\),
where \( \omega_n\bigl((\vb{x},\vb{z}),(\vb{x}',\vb{z}')\bigr) \coloneqq
\vb{x}^\trans\vb{z}'-\vb{z}^\trans\vb{x}' \).

\begin{definition}
A \textbf{symplectic vector space} \((V,\omega_V)\) is a finite-dimensional \(\F_p\)-vector space \(V\) equipped with a non-degenerate alternating bilinear form \(\omega_V:V\times V\to\F_p\).
\end{definition}

The Pauli group itself admits a symplectic description:

\begin{lemma}
Equipping \(\Zp\oplus\Zp^{2n}\) with the multiplication \(
(a,\vb{v})*(b,\vb{w})\coloneqq
(a+b-2^{\minu 1}\omega_n(\vb{v},\vb{w}),\,\vb{v}+\vb{w})\), there is a group isomorphism \(\nu:(\Zp\oplus\Zp^{2n}, *, \vb 0)\cong(\Pauli, \cdot, I)\) given by
\[(a,(\vb{x},\vb{z}))\mapsto
\chi\!(a+2^{\minu 1}\vb{z}^\trans\vb{x})\bigotimes_{k=1}^n X^{x_k}Z^{z_k}.\]
\end{lemma}

Two Pauli operators \(X^{x}Z^{z}\) and \(X^{x'}Z^{z'}\) commute precisely when \(\omega_n((\vb{x},\vb{z}),(\vb{x}',\vb{z}'))=0\), which allows us to characterise stabilizer groups in terms of symplectic linear algebra:

\begin{definition}
Given a linear subspace \(S\) of a symplectic vector space \((V, \omega)\), its \textbf{symplectic complement} is the space \(S^\omega \coloneqq \{\vb{v} \in V \mid \forall \vb{s} \in S,\, \omega(\vb{v},\vb{s}) = 0\}\). A linear subspace \(S\subseteq (V, \omega)\) is:
  \textbf{isotropic} if \(S \subseteq S^\omega\),
  \textbf{coisotropic} if \(S^\omega \subseteq S\),
  \textbf{Lagrangian} if \(S = S^\omega\).
An affine subspace is isotropic, coisotropic, Lagrangian if its linear component\footnote{The \textbf{linear component} of a non-empty affine subspace \(S \subseteq \Zp^n\) is the linear subspace \(\{\vb{x}-\vb{y} \mid \vb{x},\vb{y} \in S\}\subseteq \Zp^n\).} is.
\end{definition}

Under this representation, stabilizer states correspond to affine Lagrangian subspaces:

\begin{proposition}[{\cite[Lem.~8]{Gross}}]
\label{prop:stabilizer_group_symplectic}
Nonempty affine Lagrangian subspaces $L+\vb{a}\subseteq(\Zp^{2n},\omega_n)$
are in bijection with stabilizer subgroups of $\Pauli$ via
\(
L+\vb{a}\mapsto
\{\nu(\omega_n(\vb{a},\vb{b}),\vb{b})\mid \vb{b}\in L\}.
\)
\end{proposition}

Affine Lagrangian subspaces are better known in the quantum information literature by their tableaux.

\begin{definition}
A \textbf{stabilizer tableau} consists of a full-rank matrix \(H\in \F_p^{n\times 2n}\) and a vector \(\vb{a}\in \F_p^{2n}\) such that for any two rows \(\vb{h}_i,\vb{h}_j\) of \(H\), \(\omega_n(\vb{h}_i,\vb{h}_j)=0\). The coset \(\vb{a}+\ker(H)\subseteq(\F_p^{2n},\omega_n)\) is then an affine Lagrangian subspace, and so by \Cref{prop:stabilizer_group_symplectic} corresponds to a stabilizer state.
\end{definition}

\begin{example}Revisiting the previous example, we give the tableaux of the corresponding stabilizer states:

\begin{itemize}
  \item The computational basis state \(\ket{a}\), is stabilized by \(\chi(-a)Z\), therefore it has tableau \(H = \left[\begin{array}{c|c} 1 & 0 \end{array}\right]\) with \(\vb{a} = (a, 0)\), so that \(\vb{a} + \ker(H) = \{(a, z) \mid z\in \F_p\}\).
  \item The Hadamard basis state \(\ket{+}\coloneqq \sfrac{1}{\sqrt{p}} \sum_{x\in\F_p}\ket{x}\), is stabilized by \(X\), therefore it has a tableau \(H = \left[\begin{array}{c|c} 0 & 1 \end{array}\right]\) with \(\vb{a} = 0\), so that \(\vb{a} + \ker(H) = \{(x, 0) \mid x\in \F_p\}\).
  \item The Bell state \(\sfrac{1}{\sqrt{p}}\sum_{x\in\F_p}\ket{x}\ket{x}\) is stabilized by the commuting paulis \(X\otimes X\) and \(Z\otimes Z^{\minu 1}\), therefore it has a tableau
  \[
    H \;=\; \left[\begin{array}{cc|cc}
    1 & -1 & 0 & 0\\
    0 & 0  & 1 & 1
    \end{array}\right]
    \qquad\text{with}\qquad \vb{a} = 0,
  \]
  so that \(\vb{a} + \ker(H) = \{(x, x, z, -z) \mid x, z \in \F_p\}\).
\end{itemize}
\end{example}

Stabilizer maps are in turn represented by affine Lagrangian relations, composed via the usual relational composition following \cite{Weinstein1982}:

\begin{definition}
The category \(\ALR\) of affine Lagrangian relations has:
\begin{itemize}
  \item \textbf{objects}: finite-dimensional symplectic vector spaces \((V,\omega_V)\);
  \item \textbf{morphisms \((V,\omega_V)\to(W,\omega_W)\)}: affine Lagrangian subspaces of \((V,-\omega_V)\times (W,\omega_W)\);
  \item \textbf{composition}: given by relational composition, so that for \(R\subseteq V\times W\) and \(S\subseteq W\times U\),
  \[
    S\circ R \;\coloneqq\; \{(\vb{v}, \vb{u}) \mid \exists\, \vb{w}\in W:\, (\vb{v},\vb{w})\in R \text{ and } (\vb{w},\vb{u})\in S\};
  \]
  \item \textbf{symmetric monoidal structure}: given by the direct sum;
  \item \textbf{\(\dag\)-CC structure}: the dagger is given by the relational converse; whereas  the cups and caps are induced by the symplectic dual \((V,\omega_V)^*\coloneqq(V,-\omega_V)\).
\end{itemize}
\end{definition}

In order to state the equivalence between the Hilbert and symplectic pictures of stabilizer quantum mechanics, we need to quotient stabilizer maps by scalars:
\begin{definition}
Given a symmetric monoidal category \(\mathcal{C}\), let \(\Proj(\mathcal{C})\), denote the symmetric monoidal category of \textbf{projective maps} in \(\mathcal{C}\). This has the same objects as \(\mathcal{C}\), where the morphisms are given by equivalence classes of morphisms \([f]\) in \(\mathcal{C}\), so that \([f]=[g]\) if and only if there exists some invertible scalar \(\lambda:I\to I\) in \(\mathcal{C}\) such that \(\lambda f=g\).
\end{definition}

After applying this quotient,  the composition of stabilizer maps and their tableaux agree:
\begin{theorem}[{\cite{neretin_lectures_2011,comfort_graphical_2021}}]
  \label{theorem:stab_acr}
  There is a \dag-CC equivalence \(\Proj(\Stab)\simeq \ALR\).
\end{theorem}

\subsection{The stabilizer ZX-calculus}
Affine Lagrangian relations admit a sound, universal, and complete graphical language, called the quopit stabilizer ZX-calculus, which we denote by \(\ZX\). The stabilizer ZX-calculus was originally introduced by Booth and Carette~\cite{quopitzx}, after which its presentation was refined by Po\'or et al.~\cite{poor}. The interpretation into \(\ZX\) was later given by Booth et al.~\cite{gsa}.

The stabilizer ZX-calculus is the \dag-CCC generated by green and red \textbf{spiders}:
\begin{equation*}
  \begin{tikzpicture}
	\begin{pgfonlayer}{nodelayer}
		\node [style=none] (9) at (-1, 0.9375) {};
		\node [style=none] (10) at (2, 0.9375) {};
		\node [style=none] (11) at (2, -0.9375) {};
		\node [style=none] (12) at (-1, -0.9375) {};
		\node [style=bn] (14) at (0.5, -0.1875) {};
		\node [style=none] (15) at (1.5, 0.3125) {};
		\node [style=none] (16) at (1.5, -0.6875) {};
		\node [style=none] (17) at (-0.5, 0.3125) {};
		\node [style=none] (18) at (-0.5, -0.6875) {};
		\node [style=none] (19) at (-0.525, -0.1875) {$m$};
		\node [style=none] (20) at (2, 0.3125) {};
		\node [style=none] (21) at (2, -0.6875) {};
		\node [style=none] (26) at (0, 0.0125) {$\vdots$};
		\node [style=none] (27) at (1.375, -0.1875) {$n$};
		\node [style=none] (28) at (1, 0.0125) {$\vdots$};
		\node [style=none] (29) at (-1, 0.3125) {};
		\node [style=none] (30) at (-1, -0.6875) {};
		\node [style=bphase] (31) at (0.5, 0.625) {$\phasemat{a}{b}$};
	\end{pgfonlayer}
	\begin{pgfonlayer}{edgelayer}
		\draw [style=background] (11.center)
			 to (10.center)
			 to (9.center)
			 to (12.center)
			 to cycle;
		\draw (20.center)
			 to (15.center)
			 to [in=60, out=-180] (14);
		\draw (21.center)
			 to (16.center)
			 to [in=-60, out=180] (14);
		\draw (29.center)
			 to (17.center)
			 to [in=120, out=360] (14);
		\draw (14)
			 to [in=0, out=-120] (18.center)
			 to (30.center);
		\draw [style=dbphase] (31) to (14);
	\end{pgfonlayer}
\end{tikzpicture}
 \quad\text{and}\quad \begin{tikzpicture}
	\begin{pgfonlayer}{nodelayer}
		\node [style=none] (9) at (-1, 0.95) {};
		\node [style=none] (10) at (2, 0.95) {};
		\node [style=none] (11) at (2, -0.95) {};
		\node [style=none] (12) at (-1, -0.95) {};
		\node [style=wn] (14) at (0.5, -0.175) {};
		\node [style=none] (15) at (1.5, 0.325) {};
		\node [style=none] (16) at (1.5, -0.675) {};
		\node [style=none] (17) at (-0.5, 0.325) {};
		\node [style=none] (18) at (-0.5, -0.675) {};
		\node [style=none] (19) at (-0.525, -0.175) {$m$};
		\node [style=none] (20) at (2, 0.325) {};
		\node [style=none] (21) at (2, -0.675) {};
		\node [style=none] (26) at (0, 0.025) {$\vdots$};
		\node [style=none] (27) at (1.375, -0.175) {$n$};
		\node [style=none] (28) at (1, 0.025) {$\vdots$};
		\node [style=none] (29) at (-1, 0.325) {};
		\node [style=none] (30) at (-1, -0.675) {};
		\node [style=wphase] (31) at (0.5, 0.75) {$\phasemat{a}{b}$};
	\end{pgfonlayer}
	\begin{pgfonlayer}{edgelayer}
		\draw [style=background] (11.center)
			 to (10.center)
			 to (9.center)
			 to (12.center)
			 to cycle;
		\draw (20.center)
			 to (15.center)
			 to [in=60, out=-180] (14);
		\draw (21.center)
			 to (16.center)
			 to [in=-60, out=180] (14);
		\draw (29.center)
			 to (17.center)
			 to [in=120, out=360] (14);
		\draw (14)
			 to [in=0, out=-120] (18.center)
			 to (30.center);
		\draw [style=dwphase] (31) to (14);
	\end{pgfonlayer}
\end{tikzpicture}
, \qq{for all} a,b \in \F_p.
\end{equation*}
The parameter \(a\) is called the \textbf{affine phase}, and \(b\) the \textbf{symplectic phase} associated to the spider.

These spiders are interpreted in \(\FHilb\) as follows, up to a nonzero complex number:
\begin{align*}
  \interp{} &\coloneqq \sum_{x \in \F_p} \chi\!\left(ax + 2^{\minu 1}bx^2\right)\ket{x}^{\otimes n}\bra{x}^{\otimes m}\\
  \interp{} &\coloneqq F^{\otimes n} \sum_{x \in \F_p} \chi\!\left(ax + 2^{\minu 1}bx^2\right)\ket{x}^{\otimes n}\bra{x}^{\otimes m} (F^\dag)^{\otimes m},
\end{align*}
where \(F:\mathcal{H}_p\to\mathcal{H}_p\); \(\ket{x}\mapsto \sfrac{1}{\sqrt{p}}\sum_{y\in\F_p}\chi(xy)\ket{y}\), is the discrete Fourier transform.

Note that in particular, the Pauli operators are encoded in the affine phases, so that \(Z^z X^x \) is represented by the following diagram:

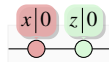
\[\begin{tikzpicture}
	\begin{pgfonlayer}{nodelayer}
		\node [style=none] (0) at (-1.25, 0.625) {};
		\node [style=none] (1) at (1.5, 0.625) {};
		\node [style=none] (2) at (1.5, -0.375) {};
		\node [style=none] (3) at (-1.25, -0.375) {};
		\node [style=none] (5) at (-1.25, 0) {};
		\node [style=none] (6) at (1.5, 0) {};
		\node [style=wn] (7) at (-0.5, 0) {};
		\node [style=bn] (8) at (0.75, 0) {};
		\node [style=wphase] (9) at (-0.5, 0.75) {\(\phasemat{x}{0}\)};
		\node [style=bphase] (10) at (0.75, 0.75) {\(\phasemat{z}{0}\)};
	\end{pgfonlayer}
	\begin{pgfonlayer}{edgelayer}
		\draw [style=background] (2.center)
			 to (1.center)
			 to (0.center)
			 to (3.center)
			 to cycle;
		\draw (5.center) to (6.center);
		\draw [style=dbphase] (10) to (8);
		\draw [style=dwphase] (9) to (7);
	\end{pgfonlayer}
\end{tikzpicture}
\]

From now on, however, we will only interpret stabilizer ZX-diagrams in \(\ALR\), thereby forgoing the need to keep track of nonzero scalars:
\[
  \interp{}
  \coloneqq
    \Bigl\{
      \bigl(
        \begin{smallbmatrix}
          \vb{x}\\ \vb{z}
        \end{smallbmatrix},
        \begin{smallbmatrix}
          \vb{x}'\\ \vb{z}'
        \end{smallbmatrix}
      \bigr)
      \in \F_p^{2m} \times \F_p^{2n}
    \ \Big| \
      z_i = z_j',
      \text{ and }
      \sum_{j=0}^{m-1} z_j+bx_0+a = \sum_{i=0}^{n-1}z_i'
    \Bigr\}
\]
\[
  \interp{}
  \coloneqq
    \Bigl\{
      \bigl(
        \begin{smallbmatrix}
          \vb{x}\\ \vb{z}
        \end{smallbmatrix},
        \begin{smallbmatrix}
          \vb{x}'\\ \vb{z}'
        \end{smallbmatrix}
      \bigr)
      \in \F_p^{2m} \times \F_p^{2n}
    \ \Big|\
      z_i = z_j',
      \text{ and }
      \sum_{j=0}^{m-1} x_j+bz_0+a = \sum_{i=0}^{n-1}x_i'
    \Bigr\}
\]

To state the equational theory in \Cref{fig:zx_axioms} succinctly, we define several notations.

\begin{figure}

\[\input{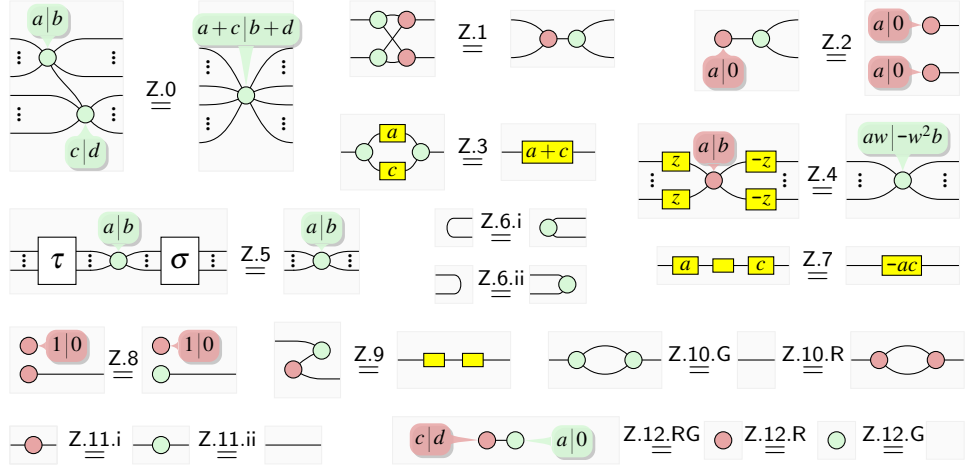}\]

\caption{Axioms of \(\StabZX\), for all \(a,b,c,d,z\in \F_p\), with \(z\neq 0\) and permutations \(\varsigma\) and \(\tau\).}
\label{fig:zx_axioms}
\end{figure}
First, for spiders with trivial phases:

\begin{equation*}
  \begin{tikzpicture}
	\begin{pgfonlayer}{nodelayer}
		\node [style=none] (9) at (-1, 0.75) {};
		\node [style=none] (10) at (2, 0.75) {};
		\node [style=none] (11) at (2, -0.75) {};
		\node [style=none] (12) at (-1, -0.75) {};
		\node [style=bn] (14) at (0.5, 0) {};
		\node [style=none] (15) at (1.5, 0.5) {};
		\node [style=none] (16) at (1.5, -0.5) {};
		\node [style=none] (17) at (-0.5, 0.5) {};
		\node [style=none] (18) at (-0.5, -0.5) {};
		\node [style=none] (19) at (-0.525, 0) {$m$};
		\node [style=none] (20) at (2, 0.5) {};
		\node [style=none] (21) at (2, -0.5) {};
		\node [style=none] (26) at (0, 0.2) {$\vdots$};
		\node [style=none] (27) at (1.375, 0) {$n$};
		\node [style=none] (28) at (1, 0.2) {$\vdots$};
		\node [style=none] (29) at (-1, 0.5) {};
		\node [style=none] (30) at (-1, -0.5) {};
	\end{pgfonlayer}
	\begin{pgfonlayer}{edgelayer}
		\draw [style=background] (11.center)
			 to (10.center)
			 to (9.center)
			 to (12.center)
			 to cycle;
		\draw (20.center)
			 to (15.center)
			 to [in=60, out=-180] (14);
		\draw (21.center)
			 to (16.center)
			 to [in=-60, out=180] (14);
		\draw (29.center)
			 to (17.center)
			 to [in=120, out=360] (14);
		\draw (14)
			 to [in=0, out=-120] (18.center)
			 to (30.center);
	\end{pgfonlayer}
\end{tikzpicture}
 \coloneqq \begin{tikzpicture}
	\begin{pgfonlayer}{nodelayer}
		\node [style=none] (32) at (3, 0.9375) {};
		\node [style=none] (33) at (6, 0.9375) {};
		\node [style=none] (34) at (6, -0.9375) {};
		\node [style=none] (35) at (3, -0.9375) {};
		\node [style=bn] (36) at (4.5, -0.1875) {};
		\node [style=none] (37) at (5.5, 0.3125) {};
		\node [style=none] (38) at (5.5, -0.6875) {};
		\node [style=none] (39) at (3.5, 0.3125) {};
		\node [style=none] (40) at (3.5, -0.6875) {};
		\node [style=none] (41) at (3.475, -0.1875) {$m$};
		\node [style=none] (42) at (6, 0.3125) {};
		\node [style=none] (43) at (6, -0.6875) {};
		\node [style=none] (44) at (4, 0.0125) {$\vdots$};
		\node [style=none] (45) at (5.375, -0.1875) {$n$};
		\node [style=none] (46) at (5, 0.0125) {$\vdots$};
		\node [style=none] (47) at (3, 0.3125) {};
		\node [style=none] (48) at (3, -0.6875) {};
		\node [style=bphase] (49) at (4.5, 0.625) {$\phasemat{0}{0}$};
	\end{pgfonlayer}
	\begin{pgfonlayer}{edgelayer}
		\draw [style=background] (34.center)
			 to (33.center)
			 to (32.center)
			 to (35.center)
			 to cycle;
		\draw (42.center)
			 to (37.center)
			 to [in=60, out=-180] (36);
		\draw (43.center)
			 to (38.center)
			 to [in=-60, out=180] (36);
		\draw (47.center)
			 to (39.center)
			 to [in=120, out=360] (36);
		\draw (36)
			 to [in=0, out=-120] (40.center)
			 to (48.center);
		\draw [style=dbphase] (49) to (36);
	\end{pgfonlayer}
\end{tikzpicture}

  \quad\text{and}\quad
  \begin{tikzpicture}
	\begin{pgfonlayer}{nodelayer}
		\node [style=none] (9) at (-1, 0.75) {};
		\node [style=none] (10) at (2, 0.75) {};
		\node [style=none] (11) at (2, -0.75) {};
		\node [style=none] (12) at (-1, -0.75) {};
		\node [style=wn] (14) at (0.5, 0) {};
		\node [style=none] (15) at (1.5, 0.5) {};
		\node [style=none] (16) at (1.5, -0.5) {};
		\node [style=none] (17) at (-0.5, 0.5) {};
		\node [style=none] (18) at (-0.5, -0.5) {};
		\node [style=none] (19) at (-0.525, 0) {$m$};
		\node [style=none] (20) at (2, 0.5) {};
		\node [style=none] (21) at (2, -0.5) {};
		\node [style=none] (26) at (0, 0.2) {$\vdots$};
		\node [style=none] (27) at (1.375, 0) {$n$};
		\node [style=none] (28) at (1, 0.2) {$\vdots$};
		\node [style=none] (29) at (-1, 0.5) {};
		\node [style=none] (30) at (-1, -0.5) {};
	\end{pgfonlayer}
	\begin{pgfonlayer}{edgelayer}
		\draw [style=background] (11.center)
			 to (10.center)
			 to (9.center)
			 to (12.center)
			 to cycle;
		\draw (20.center)
			 to (15.center)
			 to [in=60, out=-180] (14);
		\draw (21.center)
			 to (16.center)
			 to [in=-60, out=180] (14);
		\draw (29.center)
			 to (17.center)
			 to [in=120, out=360] (14);
		\draw (14)
			 to [in=0, out=-120] (18.center)
			 to (30.center);
	\end{pgfonlayer}
\end{tikzpicture}
 \coloneqq \begin{tikzpicture}
	\begin{pgfonlayer}{nodelayer}
		\node [style=none] (32) at (3, 0.9375) {};
		\node [style=none] (33) at (6, 0.9375) {};
		\node [style=none] (34) at (6, -0.9375) {};
		\node [style=none] (35) at (3, -0.9375) {};
		\node [style=wn] (36) at (4.5, -0.1875) {};
		\node [style=none] (37) at (5.5, 0.3125) {};
		\node [style=none] (38) at (5.5, -0.6875) {};
		\node [style=none] (39) at (3.5, 0.3125) {};
		\node [style=none] (40) at (3.5, -0.6875) {};
		\node [style=none] (41) at (3.475, -0.1875) {$m$};
		\node [style=none] (42) at (6, 0.3125) {};
		\node [style=none] (43) at (6, -0.6875) {};
		\node [style=none] (44) at (4, 0.0125) {$\vdots$};
		\node [style=none] (45) at (5.375, -0.1875) {$n$};
		\node [style=none] (46) at (5, 0.0125) {$\vdots$};
		\node [style=none] (47) at (3, 0.3125) {};
		\node [style=none] (48) at (3, -0.6875) {};
		\node [style=wphase] (49) at (4.5, 0.625) {$\phasemat{0}{0}$};
	\end{pgfonlayer}
	\begin{pgfonlayer}{edgelayer}
		\draw [style=background] (34.center)
			 to (33.center)
			 to (32.center)
			 to (35.center)
			 to cycle;
		\draw (42.center)
			 to (37.center)
			 to [in=60, out=-180] (36);
		\draw (43.center)
			 to (38.center)
			 to [in=-60, out=180] (36);
		\draw (47.center)
			 to (39.center)
			 to [in=120, out=360] (36);
		\draw (36)
			 to [in=0, out=-120] (40.center)
			 to (48.center);
		\draw [style=dwphase] (49) to (36);
	\end{pgfonlayer}
\end{tikzpicture}
.
\end{equation*}

We will also make frequent use of
\textbf{H-boxes}. Given \(z\in \F_p\) such that \(z\neq 0\):
\begin{equation*}
  \label{eq:hadamard_boxes}
  \input{stab/had/0.tikz} \coloneqq \input{stab/had/1.tikz}
  \quad\text{where}\quad
  \input{stab/had/2.tikz} \coloneqq \input{stab/had/3.tikz},
  \quad
  \input{stab/had/4.tikz} \coloneqq \input{stab/had/5.tikz},
  \quad\text{and}\quad
  \input{stab/had/6.tikz} \coloneqq \input{stab/had/7.tikz}.
\end{equation*}

We will use the following notation to denote relations given by linear transformation:
\begin{definition}\label{def:graph}
    Given a function \(f:X\to Y\), its \textbf{graph} (not to be confused with a graph state) is the relation \(\Gr(f)\coloneqq\{ (x, f(x) ) \ | \  x \in X \}\subseteq X\times Y\).
\end{definition}

When \(b\neq 0\), the \(H\)-box is the graph of the \emph{squeezed symplectic Fourier transform}:
\[ \interp{\input{stab/had/0.tikz}} \coloneqq
\Gr
\left[\begin{array}{c|c}
0    & -b^{\minu 1}\\ \hline
b & 0
\end{array}\right]
\]

We use the \(H\)-boxes to define \textbf{multipliers}:
\[
\begin{tikzpicture}
	\begin{pgfonlayer}{nodelayer}
		\node [style=none] (0) at (-2.5, 0) {};
		\node [style=none] (1) at (-0.5, 0) {};
		\node [style=rmat] (2) at (-1.5, 0) {$b$};
		\node [style=none] (3) at (-2.5, 0.375) {};
		\node [style=none] (4) at (-2.5, -0.375) {};
		\node [style=none] (5) at (-0.5, -0.375) {};
		\node [style=none] (6) at (-0.5, 0.375) {};
	\end{pgfonlayer}
	\begin{pgfonlayer}{edgelayer}
		\draw [style=background] (4.center)
			 to (3.center)
			 to (6.center)
			 to (5.center)
			 to cycle;
		\draw (0.center) to (1.center);
	\end{pgfonlayer}
\end{tikzpicture}
 \coloneqq \begin{tikzpicture}
	\begin{pgfonlayer}{nodelayer}
		\node [style=none] (8) at (0.5, 0) {};
		\node [style=none] (9) at (2.75, 0) {};
		\node [style=had] (10) at (1.25, 0) {$b$};
		\node [style=none] (11) at (0.5, 0.375) {};
		\node [style=none] (12) at (0.5, -0.375) {};
		\node [style=none] (13) at (2.75, -0.375) {};
		\node [style=none] (14) at (2.75, 0.375) {};
		\node [style=had] (15) at (2.125, 0) {\(\minu\)};
	\end{pgfonlayer}
	\begin{pgfonlayer}{edgelayer}
		\draw [style=background] (11.center)
			 to (14.center)
			 to (13.center)
			 to (12.center)
			 to cycle;
		\draw (8.center) to (9.center);
	\end{pgfonlayer}
\end{tikzpicture}
.
\]
When \(b\neq 0\), the multiplier is interpreted as the graph of the \textbf{squeezing map}:
\[ \interp{\begin{tikzpicture}
	\begin{pgfonlayer}{nodelayer}
		\node [style=none] (0) at (-2.5, 0) {};
		\node [style=none] (1) at (-0.5, 0) {};
		\node [style=rmat] (2) at (-1.5, 0) {$b$};
		\node [style=none] (3) at (-2.5, 0.375) {};
		\node [style=none] (4) at (-2.5, -0.375) {};
		\node [style=none] (5) at (-0.5, -0.375) {};
		\node [style=none] (6) at (-0.5, 0.375) {};
	\end{pgfonlayer}
	\begin{pgfonlayer}{edgelayer}
		\draw [style=background] (4.center)
			 to (3.center)
			 to (6.center)
			 to (5.center)
			 to cycle;
		\draw (0.center) to (1.center);
	\end{pgfonlayer}
\end{tikzpicture}
} \coloneqq
\Gr
\left[\begin{array}{c|c}
b & 0\\ \hline
0 & b^{\minu 1}
\end{array}\right]
\]

The \textbf{zero map} \(\begin{tikzpicture}
	\begin{pgfonlayer}{nodelayer}
		\node [style=none] (0) at (-0.75, 0.5) {};
		\node [style=none] (1) at (-0.75, -0.5) {};
		\node [style=none] (2) at (1.125, -0.5) {};
		\node [style=none] (3) at (1.125, 0.5) {};
		\node [style=wphase] (4) at (-0.25, 0) {\(\phasemat{1}{0}\)};
		\node [style=wn] (5) at (0.75, 0) {};
	\end{pgfonlayer}
	\begin{pgfonlayer}{edgelayer}
		\draw [style=background] (1.center)
			 to (0.center)
			 to (3.center)
			 to (2.center)
			 to cycle;
		\draw [style=background] (0.center) to (3.center);
		\draw [style=background] (3.center) to (2.center);
		\draw [style=background] (2.center) to (1.center);
		\draw [style=background] (1.center) to (0.center);
		\draw [style=rwphase] (4) to (5);
	\end{pgfonlayer}
\end{tikzpicture}
\) is interpreted as the empty subspace.
Finally, the dagger is given by:
\[
  \left(\right)^\dag
  \coloneqq
  \begin{tikzpicture}
	\begin{pgfonlayer}{nodelayer}
		\node [style=none] (9) at (-1, 0.9375) {};
		\node [style=none] (10) at (2, 0.9375) {};
		\node [style=none] (11) at (2, -0.9375) {};
		\node [style=none] (12) at (-1, -0.9375) {};
		\node [style=wn] (14) at (0.5, -0.1875) {};
		\node [style=none] (15) at (1.5, 0.3125) {};
		\node [style=none] (16) at (1.5, -0.6875) {};
		\node [style=none] (17) at (-0.5, 0.3125) {};
		\node [style=none] (18) at (-0.5, -0.6875) {};
		\node [style=none] (19) at (-0.525, -0.1875) {$m$};
		\node [style=none] (20) at (2, 0.3125) {};
		\node [style=none] (21) at (2, -0.6875) {};
		\node [style=none] (26) at (0, 0.0125) {$\vdots$};
		\node [style=none] (27) at (1.375, -0.1875) {$n$};
		\node [style=none] (28) at (1, 0.0125) {$\vdots$};
		\node [style=none] (29) at (-1, 0.3125) {};
		\node [style=none] (30) at (-1, -0.6875) {};
		\node [style=wphase] (31) at (0.5, 0.625) {$\phasemat{a}{\minu b}$};
	\end{pgfonlayer}
	\begin{pgfonlayer}{edgelayer}
		\draw [style=background] (11.center)
			 to (10.center)
			 to (9.center)
			 to (12.center)
			 to cycle;
		\draw (20.center)
			 to (15.center)
			 to [in=60, out=-180] (14);
		\draw (21.center)
			 to (16.center)
			 to [in=-60, out=180] (14);
		\draw (29.center)
			 to (17.center)
			 to [in=120, out=360] (14);
		\draw (14)
			 to [in=0, out=-120] (18.center)
			 to (30.center);
		\draw [style=dwphase] (31) to (14);
	\end{pgfonlayer}
\end{tikzpicture}

  \quad\text{and}\quad
  \left(\right)^\dag
  \coloneqq
  \begin{tikzpicture}
	\begin{pgfonlayer}{nodelayer}
		\node [style=none] (9) at (-1, 0.9375) {};
		\node [style=none] (10) at (2, 0.9375) {};
		\node [style=none] (11) at (2, -0.9375) {};
		\node [style=none] (12) at (-1, -0.9375) {};
		\node [style=bn] (14) at (0.5, -0.1875) {};
		\node [style=none] (15) at (1.5, 0.3125) {};
		\node [style=none] (16) at (1.5, -0.6875) {};
		\node [style=none] (17) at (-0.5, 0.3125) {};
		\node [style=none] (18) at (-0.5, -0.6875) {};
		\node [style=none] (19) at (-0.525, -0.1875) {$m$};
		\node [style=none] (20) at (2, 0.3125) {};
		\node [style=none] (21) at (2, -0.6875) {};
		\node [style=none] (26) at (0, 0.0125) {$\vdots$};
		\node [style=none] (27) at (1.375, -0.1875) {$n$};
		\node [style=none] (28) at (1, 0.0125) {$\vdots$};
		\node [style=none] (29) at (-1, 0.3125) {};
		\node [style=none] (30) at (-1, -0.6875) {};
		\node [style=bphase] (31) at (0.5, 0.625) {$\phasemat{\minu a}{\minu b}$};
	\end{pgfonlayer}
	\begin{pgfonlayer}{edgelayer}
		\draw [style=background] (11.center)
			 to (10.center)
			 to (9.center)
			 to (12.center)
			 to cycle;
		\draw (20.center)
			 to (15.center)
			 to [in=60, out=-180] (14);
		\draw (21.center)
			 to (16.center)
			 to [in=-60, out=180] (14);
		\draw (29.center)
			 to (17.center)
			 to [in=120, out=360] (14);
		\draw (14)
			 to [in=0, out=-120] (18.center)
			 to (30.center);
		\draw [style=dbphase] (31) to (14);
	\end{pgfonlayer}
\end{tikzpicture}

\]

\begin{theorem}[{\cite[Thm.~4.16]{comfort_graphical_2021}, \cite[Chap.~9]{neretin_lectures_2011},\cite[Cor.~46]{gsa}}]
  \label{theorem:stabzx_acr}
  There is a \dag-CC equivalence \(\ZX \cong \Proj(\Stab)\simeq \ALR\).
\end{theorem}

  \section{Noise and approximation in quantum mechanics}\label{sec:mixed}
  \pratendSetLocal{category=mixed, end, restate}
  
Pure quantum mechanics fails to capture noise, which is a fundamental feature of quantum mechanics. In this section, we review mixed stabilizer quantum mechanics which captures such notions as noise and measurement.
We frame mixed state quantum mechanics within categories with a notion of approximation and discarding:

\begin{definition}
  A symmetric monoidal category is \textbf{poset-enriched}, in case for all objects \(X\) and \(Y\), there is a poset \(\subseteq\) structure on the hom-set \(\mathcal{C}(X,Y)\) so that:
  \begin{itemize}
    \item
    given \(f\subseteq g\) and \(h\subseteq k\) of appropriate type, then \(f\otimes h \subseteq g\otimes h\);

    \item 
    and given \(f \subseteq g\) and \(h\subseteq k\) of appropriate type, then \(h\circ f \subseteq k\circ g\).
  \end{itemize}
\end{definition}
Given two processes \(f\) and \(g\) of the same type, we interpret \(f \subseteq g\) to denote that \(f\) is a less noisy version of \(g\), so that \(g\) \emph{approximates} \(f\). To model discarding of information, we use the following structure:
\begin{definition}[{\cite[Def.~2.5]{obs}}]\label{def:discard_bicategory}
  A \textbf{discard bicategory} is a poset-enriched symmetric monoidal category, such that for all objects \(X\), there exists a \textbf{discard map} \(\discard_X:X\to I\), denoted \(\begin{tikzpicture}
	\begin{pgfonlayer}{nodelayer}
		\node [style=none] (1) at (0, 0.3) {};
		\node [style=none] (2) at (0.75, 0.3) {};
		\node [style=none] (3) at (0.75, -0.3) {};
		\node [style=none] (4) at (0, -0.3) {};
		\node [style=none] (5) at (0, 0) {};
		\node [style=groundout] (6) at (0.375, 0) {};
	\end{pgfonlayer}
	\begin{pgfonlayer}{edgelayer}
		\draw [style=background] (3.center)
			 to (2.center)
			 to (1.center)
			 to (4.center)
			 to cycle;
		\draw (5.center) to (6);
	\end{pgfonlayer}
\end{tikzpicture}
\), where for all  $f: X \to Y$,
  \[\discard_Y \circ f \subseteq \discard_X, \quad \discard_I=1_I\quad\text{and}\quad\discard_{X\otimes Y}=\discard_X\otimes\discard_Y.\] 
  In a discard \dag-CCC, denote the dagger of the discard \(\codiscard \coloneqq \discard^\dag\) by \(\begin{tikzpicture}
	\begin{pgfonlayer}{nodelayer}
		\node [style=none] (71) at (11, 0.3) {};
		\node [style=none] (72) at (10.25, 0.3) {};
		\node [style=none] (73) at (10.25, -0.3) {};
		\node [style=none] (74) at (11, -0.3) {};
		\node [style=groundin] (80) at (10.625, 0) {};
		\node [style=none] (81) at (11, 0) {};
	\end{pgfonlayer}
	\begin{pgfonlayer}{edgelayer}
		\draw [style=background] (73.center)
			 to (72.center)
			 to (71.center)
			 to (74.center)
			 to cycle;
		\draw (81.center) to (80);
	\end{pgfonlayer}
\end{tikzpicture}
\).
\end{definition}
Interpreting the \(f\subseteq g\) to denote that \(g\) approximates \(f\), the discard map approximates all effects, so that it is maximally uninformative.
In the following subsections we upgrade the three pictures of stabilizer theory reviewed in the previous section from pure to mixed processes. We show that this naturally gives rise to discard bicategory structures.

\subsection{Completely positive maps}
In standard quantum theory, noisy processes are modeled by completely positive maps (See Nielsen and Chuang for reference \cite{Nielsen2012}).
\begin{definition}
Write \(\B(\mathcal{H}) := \mathcal{H}\otimes\mathcal{H}^*\) for the space of linear operators on a f.d. Hilbert space \(\mathcal{H}\). Recall that an operator \(\rho\in \B(\mathcal{H})\) is \textbf{positive} if \(\langle\psi,\rho\psi\rangle\geq 0\) for all \(\psi\in\mathcal{H}\). Moreover, recall that an operator \(\Phi:\B(\mathcal{H})\to \B(\mathcal{K})\) is \textbf{completely positive} if, for any f.d. Hilbert space \(\mathcal{E}\) and any positive operator \(\rho\in \B(\mathcal{E})\otimes \B(\mathcal{H})\), the operator \((1_{\B(\mathcal{E})}\otimes\Phi)(\rho)\) is positive.
\end{definition}
Completely positive maps form a poset-enriched category \cite[Lem.~6.5]{selinger2004towards}.
\begin{definition}
The poset enriched \dag-CCC, \(\CPM\), of \textbf{completely positive maps} has objects given by operator spaces \(\B(\mathcal{H})\) for all finite dimensional Hilbert spaces \(\mathcal{H}\) and morphisms given by \textbf{completely positive maps} \(\Phi:\B(\mathcal{H})\to \B(\mathcal{K})\). The \dag-CC structure is inherited from \(\FHilb\), and the poset enrichment is given by the \textbf{L\"owner order} so that \(\Phi\preceq\Psi\) if and only if \(\Psi-\Phi\) is completely positive.
\end{definition}

Any linear map \(T:\mathcal{H}\to\mathcal{K}\) induces a \textbf{pure} completely positive map \(\B(\mathcal{H})\to \B(\mathcal{K})\) given by \(\rho\mapsto T\rho T^\dag\); these are interpreted as lossless processes which do not interact with the environment.

On the other hand, given any orthonormal basis \(\{e_i\}\) of \(\mathcal{H}\), the \textbf{trace} \(\Tr_{\mathcal{H}}:\B(\mathcal{H})\to \C\) is the completely positive map \(\rho\mapsto \sum_i \langle e_i, \rho e_i\rangle\), which discards all information on \(\B(\mathcal{H})\) into the environment.

The Stinespring dilation theorem bridges these two extremes, characterising every completely positive map as a pure process on an enlarged space whose output is partially discarded:

\begin{lemma}
Given any completely positive map \(\Phi:\B(\mathcal{H})\to \B(\mathcal{K})\), there exists a Hilbert space \(\mathcal{E}\) and a linear map \(T:\mathcal{H}\to \mathcal{K}\otimes \mathcal{E}\),  such that \(\Phi(\rho) = (1_{\B(\mathcal{K})}\otimes\Tr_{\mathcal{E}})(T\rho T^\dag)\), called a \textbf{Stinespring dilation} of \(\Phi\).
\end{lemma}

Carette et al.~\cite[Def.~5]{Carette2021} also introduced a purification \emph{preorder} between completely positive maps: %

\begin{propositionE}
 \(\Proj(\CPM)\) is a discard bicategory with respect to the trace and the \textbf{purification order}. Given projective completely positive maps \(\Phi, \Psi: \B(\mathcal{H})\to \B(\mathcal{K})\)  in  \(\Proj(\CPM)\), \(\Phi \preceqp \Psi\) in case there exists some Hilbert space \(\mathcal{E}\) and  maps \(\Psi_0: \B(\mathcal{H})\to \B(\mathcal{K})\otimes \B(\mathcal{E})\) and \(\Phi_0: \B(\mathcal{E})\to \C\) in \(\Proj(\CPM)\) such that \(\Psi = (1_{\B(\mathcal{K})} \otimes \Tr_{\mathcal{E}}) \circ \Psi_0\) and \(\Phi = (1_{\B(\mathcal{K})} \otimes \Phi_0) \circ \Psi_0\).
 
 Moreover, \(\Proj(\CPM)\) is a discard bicategory with respect to the trace.
\end{propositionE}
\begin{proofE}
  By Carette et al. the purification order is a preorder enrichment on the full subcategory of  \(\CPM\) whose objects are of the form \(\mathcal{2}^{n}\), for all \(n\in \N\) \cite[Lem.~6]{Carette2021}. However, their argument holds for all finite dimensional Hilbert spaces so that the purification order is a preorder enrichment on \(\CPM\). The quotient by scalars makes the purification order a poset-enrichment on \(\Proj(\CPM)\).

  Next, we show that \(\Proj(\CPM)\) is a discard bicategory with respect to the trace.
  Take any effect \(\Xi: \B(\mathcal{H})\to \C\) in \(\Proj(\CPM)\). Let \(\Phi_0 = \Xi\) and \(\Psi_0 = 1_{\B(\mathcal{H})}\), then \((1_{\C} \otimes \Tr_{\mathcal{H}}) \circ \Psi_0 = \Tr_{\mathcal{H}}\) and \((1_{\C} \otimes \Phi_0) \circ \Psi_0 = \Xi\), and it follows that \(\Xi  \preceqp \Tr_{\mathcal{H}}\).
\end{proofE}

Carette et al. \cite{Carette2021} claim that the purification order and the L\"owner order are ``tightly linked''. We prove this claim formally:
\begin{theoremE}\label{thm:purification_lowner_equiv}
   \([\Phi] \preceqp [\Psi]\) in \(\Proj(\CPM)\) if and only if there is some  \(\lambda\in \R^{\geq 0}\) such that \(\Phi \preceql \lambda \Psi\) in \(\CPM\).
\end{theoremE}
\begin{proofE}
  Let \([\Psi]\) denote an equivalence class in \(\Proj(\CPM)\) with chosen representative \(\Psi\) in \(\CPM\). Suppose that \([\Phi] \preceqp [\Psi]\), then there exists \([\Phi_0]: \B(\mathcal{E}) \to \C\) and \([\Psi_0]: \B(\mathcal{H}) \to \B(\mathcal{K}) \otimes \B(\mathcal{E})\) such that \([\Phi] = (1 \otimes [\Phi_0]) \circ [\Psi_0]\) and \([\Psi] = (1 \otimes \Tr_{\mathcal{E}}) \circ [\Psi_0]\). Then there exists scalars \(a, b \in \R^{\geq 0}\) such that in \(\CPM\):
   \[ \Phi = a (1 \otimes \Phi_0) \circ \Psi_0 \quad  \Psi = b\,(1 \otimes \Tr_{\mathcal{E}}) \circ \Psi_0\]
   Now we normalize the effect \(\Phi_0\). Let \(c = \norm{\Phi_0}^{\minu 1}\) then \(c \Phi_0\) is a completely positive trace non-increasing map and we have \(\Phi = \frac{a}{c} (1 \otimes c \Phi_0) \circ \Psi_0\). Then let \(\lambda = \frac{a}{bc}\):
   \begin{align*}
      \lambda \Psi - \Phi &= \frac{b a}{b c} (1 \otimes \Tr_{\mathcal{E}}) \circ \Psi_0 - \frac{a}{c}(1 \otimes c \Phi_0) \circ \Psi_0 \\
      &= \frac{a}{c}(1 \otimes (\Tr_{\mathcal{E}} - c\Phi_0)) \circ \Psi_0
   \end{align*}
   Since \(c\Phi_0\) is trace non-increasing, \(f_0 \coloneqq (\Tr_{\mathcal{E}} - c\Phi_0)\) is a completely positive map. Then since \(\frac{a}{c} \geq 0\) and \(\Psi_0\) is completely positive, \(\lambda \Psi - \Phi\) is also completely positive. Therefore \(\Phi \preceql \lambda \Psi\).

   Now for the second implication, suppose \(\Phi \preceql \lambda \Psi\) in \(\CPM\). Then \(\Xi \coloneqq \lambda \Psi - \Phi\) is a completely positive map. Consider the qubit space \(\C^2\) with basis \(\ket{0}, \ket{1}\) and define \(\Psi_0: \B(\mathcal{H}) \to \B(\mathcal{K}) \otimes \B(\C^2)\) and \(\Phi_0: \B(\C^2) \to \C\) by:
    \[\Psi_0(\rho) \coloneqq \Phi(\rho) \otimes \dyad{0}{0} + \Xi(\rho) \otimes \dyad{1}{1} \qquad \Phi_0(\sigma) \coloneqq \bra{0} \sigma \ket{0}\]
    for any \(\rho\in \B(\mathcal{H})\) and qubit operator \(\sigma\in \B(\C^2)\). Then we have
    \( (1 \otimes \Tr_{\C^2}) \circ [\Psi_0] =  [(1 \otimes \Tr_{\C^2}) \circ \Psi_0] = [\Phi + \Xi] = [\lambda \Psi] = [\Psi] \) and \((1 \otimes [\Phi_0]) \circ [\Psi_0] = [(1 \otimes \Phi_0) \circ \Psi_0] = [\Phi]\) from which it follows immediately that \([\Phi] \preceqp [\Psi]\) in \(\Proj(\CPM)\).
\end{proofE}

\subsection{Completely positive stabilizer maps}
\label{sec:disc:symplectic}

By restricting to completely positive maps which admit a pure stabilizer Stinespring dilation, we obtain a category of mixed stabilizer processes:
\begin{definition}
  The \dag-CCC \(\Stab^\disc\) of \textbf{completely positive stabilizer maps} is the subcategory of \(\CPM\) whose objects are given by quopit Hilbert spaces \(\mathcal{H}_p^{\otimes n}\) and whose morphisms are given by completely positive maps which admit a pure stabilizer Stinespring dilation.
\end{definition}

What are the stabilizer groups of completely positive stabilizer maps?
Recall that a Pauli operator \(P\) stabilizes a pure state \(\ket{\psi}\) when \(P\ket{\psi} = \ket{\psi}\). For a completely positive map \(\Phi:\B(\mathcal{H}_p^{\otimes n})\to \B(\mathcal{K})\), we say that \(P\) \textbf{stabilizes} \(\Phi\) when \(\Phi\) is invariant under precomposition with the pure CP map induced by \(P\); equivalently, \(\Phi(P\rho P^\dag) = \Phi(\rho)\) for all \(\rho\in \B(\mathcal{H}_p^{\otimes n})\). The \textbf{stabilizer group} of \(\Phi\) is the subgroup of the Pauli group consisting of those operators which stabilize \(\Phi\).

How can the stabilizer group of a completely positive stabilizer map be captured symplectically, as for pure stabilizer maps?
By cyclicity of the trace, \(\Tr_{\mathcal{H}_p^{\otimes n}}(P\rho P^\dag) = \Tr_{\mathcal{H}_p^{\otimes n}}(\rho)\) for every Pauli \(P\), so the trace is stabilized by the entire Pauli group. Therefore, completely positive stabilizer maps should have larger stabilizer groups than pure stabilizer maps, motivating the following definition:
\begin{definition}
  Let \(\ALR^\disc\) denote the \dag-CCC with the same structure as \(\ALR\), where the morphisms are \textbf{affine coisotropic relations}. That is to say, these are affine relations which are either empty, or an affine relation whose linear component \(S\) satisfies \(S^\omega \subseteq S\).
\end{definition}
Booth and Comfort proved that \(\ALR^\disc\) admits a notion of Stinespring dilation \cite[Prop.~38]{denotational}:
\begin{lemma}\label{lem:symplectic_stinespring}
  Given an affine coisotropic relation \(R:  (V, \omega_V)\to \{\bullet\}\) there exists a symplectic vector space \((E,\omega_E)\) and an affine Lagrangian coisometry \(S: (V, \omega_V)\to (E,\omega_E)\) such that:
  \[R =  \discard_{(E,\omega_E)}\circ S\quad\text{where}\quad \discard_{(E,\omega_E)} \coloneqq \{(\vb v,\bullet)\mid \vb v \in E\} :(E,\omega_E) \to \{\bullet\}.\]
\end{lemma}
It follows immediately that:
\begin{proposition}
  \label{prop:acr_discard_bicategory}
  \(\ALR^\disc\) is a discard bicategory with respect to subspace inclusion and \(\discard\).
\end{proposition}

Taking the stabilizer groups of completely positive stabilizer maps induces the following equivalence:

\begin{theoremE} \label{thm:acr_stabdisc}
\(\Proj(\Stab^\disc)\) is a discard bicategory with respect to the purification order and the trace, and there is an equivalence of \dag-compact closed discard bicategories
\(\Proj(\Stab^\disc) \cong \ALR^\disc\).
\end{theoremE}
\begin{proofE}
  We prove that there is a functor \(\Proj(\Stab^\disc)\to \ALR^\disc\) which preserves the pseudo-purification and reflects the order given by subspace inclusion, sending the trace to the discard relation. From this it follows that \(\Proj(\Stab^\disc)\) is a discard bicategory and that the functor is an equivalence of discard bicategories.

  First, by Booth and Comfort \cite[Cor.~1]{denotational}, there is a \dag-CC equivalence  \(F:\Proj(\Stab^\disc) \simeq \ALR^\disc\) extending the \dag-CC equivalence \(\Proj(\Stab)\simeq \ALR\), sending the trace \(\Tr_{\mathcal{H}_p^{\otimes e}}:\B(\mathcal{H}_p^{\otimes e})\to\C\) in  \(\Proj(\Stab^\disc)\) to the discard relation \(\discard_{(\F_p^{2e}, \omega_{e})}\) in \(\ALR^\disc\).
  It is immediate that the discard is preserved; therefore, it remains to prove that the poset-enrichment is preserved and reflected.
  Because both categories are compact closed, it suffices to prove the claim on effects \(\B(\mathcal{H})\to \C\).

  Now we prove the preservation of the poset-enrichment.
  Take two effects \([\Phi],[\Psi]: \B(\mathcal{H}_p^{\otimes n})\to \C\) in \(\Proj(\Stab^\disc)\), such that \([\Phi] \preceqp [\Psi]\) is witnessed by \([\Psi_0]: \B(\mathcal{H}_p^{\otimes n})\to  \B(\mathcal{H}_p^{\otimes e})\) and \([\Phi_0]: \B(\mathcal{H}_p^{\otimes e})\to \C \). Then
  
  \begin{align*}
    F([\Phi])  = F([\Phi_0]) \circ F([\Psi_0])
    \subseteq \discard_{(\F_p^{2e}, \omega_e)} \circ F([\Psi_0])
     = F([\Tr_{\mathcal{H}_p^{\otimes e}}]) \circ F([\Psi_0])
     = F([\Tr_{\mathcal{H}_p^{\otimes e}}\circ \Psi_0])
     = F([\Psi])
  \end{align*}

  Finally, we show that the poset-enrichment is reflected. Take two effects \(R, S:(\F_p^{2n}, \omega_n)\to \{\bullet\}\) in \(\ALR^\disc\) such that \(R\subseteq S\). 
  By \Cref{lem:symplectic_stinespring} there exists some \(e\in \N\) and coisometry  \(V:(\F_p^{2n}, \omega_n)\to (\F_p^{2e}, \omega_e)\) such that  \(S = \discard_{(\F_p^{2e}, \omega_e)}\circ V\). 

  Let \(G\) denote the inverse of \(F\), when restricted to the objects of the form \((\F_p^{2n}, \omega_n)\) in \(\ALR^\dag\), for all \(n\in \N\). 
  Take \([\Phi_0] = G(R \circ V^\dag )\) and \([\Psi_0] = G(V)\). On the one hand,

   \[
    [\Tr_{\mathcal{H}_p^{\otimes e}}]\circ [\Psi_0]
    = G(\discard_{(\F_p^{2e}, \omega_e)}) \circ  G(V)
    = G(\discard_{(\F_p^{2e}, \omega_e)}\circ V)
    = G(S)
    \]

  On the other hand, because \(R\subseteq S\) we have
   \((V^\dag \circ V) \circ R =   \{(\vb v,\vb v) \ | \ \forall (\vb v, \bullet ) \in S \} \circ R =R\),
  therefore,
  \[
    [\Phi_0] \circ [\Psi_0]
    = G(  R\circ V^\dag )\circ G(V)
    = G(  R\circ V^\dag \circ V)
    = G(R)
  \]

\end{proofE}

\begin{corollaryE}\label{cor:cpm_stab_order_embedding}
  The canonical functor \(\Proj(\Stab^\disc)\to \Proj(\CPM)\) is faithful, preserves the \dag-compact-closed and discard bicategory structure, and reflects the poset enrichment.
\end{corollaryE}
\begin{proofE}
Let \(H:\Stab^\disc\to \CPM\) denote the canonical functor from mixed stabilizer maps to completely positive maps. The preservation of the \dag-compact closed structure, the discard bicategory structure and the faithfulness of \(H\) are immediate. It remains to show that \(H\) reflects the poset-enrichment.

Take effects \(\Phi,\Psi: \B(\mathcal{H}_p^{\otimes n})\to\C\) in \(\Stab^\disc\) with \([H(\Phi)] \preceqp[H(\Psi)]\). Take Stinespring dilations \(U:\mathcal{H}_p^{\otimes n}\to \mathcal{H}_p^{\otimes e}\) and \(V:\mathcal{H}_p^{\otimes n}\to \mathcal{H}_p^{\otimes e'}\) of \(H(\Phi)\) and \(H(\Psi)\) respectively.
By \Cref{thm:purification_lowner_equiv}, there exists some \(\lambda\in\R^{\geq 0}\) such that \(H(\Phi)\preceql \lambda H(\Psi)\). Therefore for any positive \(\rho\in \B(\mathcal{H}_p^{\otimes n})\), we have

\[ 0 \leq (\lambda H(\Psi) - H(\Phi))(\rho) = \lambda \Tr(V^\dag V \rho) - \Tr(U^\dag U \rho) = \Tr((\lambda V^\dag V - U^\dag U)\rho). \]

Recalling that a Hermitian operator \(A\) is positive if and only if \(\Tr(A\rho)\geq 0\) for all positive \(\rho\), it follows that \(\lambda V^\dag V - U^\dag U\) is positive. Invoking Douglas' range-inclusion lemma, it follows that \(\range(U^\dag)\subseteq \range(V^\dag)\).

Now, for any Pauli operator \(P\in \B(\mathcal{H}_p^{\otimes n})\), we have the following chain of logical equivalences:
\begin{align}
  & P \text{ stabilizes } H(\Phi) \notag\\
  \iff & H(\Phi)\circ(P{\cdot}P^\dag) = H(\Phi) \tag{definition} \\
  \iff & H(\Phi)(P\rho P^\dag) = H(\Phi)(\rho) \text{ for all } \rho\in \B(\mathcal{H}_p^{\otimes n}) \tag{evaluating} \\
  \iff & \Tr(U^\dag U  P\rho P^\dag) = \Tr(U^\dag U \rho) \text{ for all } \rho\in \B(\mathcal{H}_p^{\otimes n}) \tag{effect-operator formula} \\
  \iff & \Tr(P^\dag U^\dag U P \rho) = \Tr(U^\dag U \rho) \text{ for all } \rho\in \B(\mathcal{H}_p^{\otimes n}) \tag{cyclicity of the trace} \\
  \iff & P^\dag U^\dag U P = U^\dag U \tag{non-degeneracy of the trace} \\
  \iff & P U^\dag U P^\dag = U^\dag U \tag{Unitarity of \(P\)}
\end{align}
Because \(\range(U^\dag)\subseteq \range(V^\dag)\), it follows that \(P V^\dag V P^\dag = V^\dag V\), so that by symmetry, \(P\) also stabilizes \(H(\Psi)\). Therefore, the stabilizer group of \(\Phi\) is a subgroup of the stabilizer group of \(\Psi\). Recalling the functor \(F:\Proj(\Stab^\disc)\to \ALR^\disc\) from \Cref{thm:acr_stabdisc}, it follows that \(F([\Phi])\subseteq F([\Psi])\). Because \(F\) is an isomorphism of discard bicategories, it reflects the partial order enrichment, so that \([\Phi]\preceqp [\Psi]\) in \(\Proj(\Stab^\disc)\).
\end{proofE}

\subsection{The directed mixed stabilizer ZX-calculus}

Finally, building on the work of Carette et al.~\cite{Carette2021} and Booth et al.~\cite[Cor.~54]{gsa}, we can add discarding to the stabilizer ZX-calculus:

\begin{definition}\label{def:mixed_stabzx}
The \textbf{directed mixed stabilizer ZX-calculus}, \(\StabZX^\disc\), is the \dag-CC discard category presented by adding a generator
\(\begin{tikzpicture}
	\begin{pgfonlayer}{nodelayer}
		\node [style=none] (71) at (10.25, 0.3) {};
		\node [style=none] (72) at (11, 0.3) {};
		\node [style=none] (73) at (11, -0.3) {};
		\node [style=none] (74) at (10.25, -0.3) {};
		\node [style=groundout] (80) at (10.625, 0) {};
		\node [style=none] (81) at (10.25, 0) {};
	\end{pgfonlayer}
	\begin{pgfonlayer}{edgelayer}
		\draw [style=background] (73.center)
			 to (72.center)
			 to (71.center)
			 to (74.center)
			 to cycle;
		\draw (81.center) to (80);
	\end{pgfonlayer}
\end{tikzpicture}
:1\to 0\) to \(\StabZX\).
We impose that for all \(a,b\in\F_p\),
\[
\input{mixed/stab_zx/disc_isom/0.tikz}
=\input{mixed/stab_zx/disc_isom/1.tikz}
,\qquad
\input{mixed/stab_zx/disc_isom/2.tikz}
= \input{mixed/stab_zx/disc_isom/3.tikz}
,\qquad
 \input{mixed/stab_zx/disc_isom/4.tikz}
= \input{mixed/stab_zx/disc_isom/5.tikz}
,\qquad
\input{mixed/stab_zx/disc_isom/6.tikz}
\subseteq \input{mixed/stab_zx/disc_isom/7.tikz}
\qquad\text{and}\qquad
\input{mixed/stab_zx/disc_isom/8.tikz}
\subseteq \input{mixed/stab_zx/disc_isom/9.tikz}.
\]
\end{definition}

Putting everything together, we have:
\begin{theoremE}\label{prop:stabzxdisc_alrdisc_equiv}
There is an equivalence of \dag-CC discard bicategories
\(\StabZX^\disc\simeq\ALR^\disc \simeq \Proj(\Stab^\disc)\).
\end{theoremE}
\begin{proofE}
  The fact that this is a \dag-CC discard equivalence follows from \cite[Cor.~54]{gsa}.  To show that this equivalence preserves and reflects the preorder, first observe that the linear Lagrangian relations are poset enriched with respect to the discrete preorder, as all linear Lagrangian relations of the same type have the same dimension. For affine Lagrangian relations, we have almost the same story, except that all relations contain the empty set. Because affine coisotropic relations can be obtained by composing affine Lagrangian relations with the discard relation \(\disc_E\), the poset-enrichment with respect to subspace inclusion is determined by the poset \(\ACR(E, I)\), for which the discard relation is the top element.
\end{proofE}

  \section{Translation-invariant quantum processes}\label{sec:beh}
  \pratendSetLocal{category=beh, end, restate}
  In \Cref{sec:intro} we introduced the delay informally and claimed that
translation-invariant stabilizer processes can be captured by ``delaying'' finite ZX-diagrams. We now make this precise, beginning with the syntax. Following Katis et al.~\cite[\S~3.1, Prop.~3]{Katis1997}, we extend the
stabilizer ZX-calculus with a delay generator \(\delta\), subject only to the
equations governing its basic interaction with the existing generators:

\begin{definition}
  \label{def:statezx}
  The \(\dag\)-CCC \(\St(\StabZX)\), called the \textbf{stateful stabilizer ZX-calculus}, is given by adding a \textbf{delay} generator \(\begin{tikzpicture}
	\begin{pgfonlayer}{nodelayer}
		\node [style=rmat] (1) at (0.75, 0) {\(\delta\)};
		\node [style=none] (2) at (0, 0.4) {};
		\node [style=none] (3) at (1.5, 0.4) {};
		\node [style=none] (4) at (1.5, -0.4) {};
		\node [style=none] (5) at (0, -0.4) {};
		\node [style=none] (6) at (0, 0) {};
		\node [style=none] (7) at (1.5, 0) {};
	\end{pgfonlayer}
	\begin{pgfonlayer}{edgelayer}
		\draw [style=background] (2.center)
			 to (3.center)
			 to (4.center)
			 to (5.center)
			 to cycle;
		\draw [style=background] (5.center) to (2.center);
		\draw [style=background] (2.center) to (3.center);
		\draw [style=background] (3.center) to (4.center);
		\draw [style=background] (4.center) to (5.center);
		\draw (7.center) to (6.center);
	\end{pgfonlayer}
\end{tikzpicture}
\) to \(\StabZX\), modulo the sliding equations:
  \begin{equation}\label{eq:fdzx_sliding}
      \begin{tikzpicture}
	\begin{pgfonlayer}{nodelayer}
		\node [style=none] (75) at (3, -1.125) {};
		\node [style=none] (76) at (6.5, -1.125) {};
		\node [style=none] (77) at (6.5, 1.125) {};
		\node [style=none] (78) at (3, 1.125) {};
		\node [style=bn] (79) at (4.75, -0.125) {};
		\node [style=none] (80) at (6.5, 0.375) {};
		\node [style=none] (81) at (6.5, -0.625) {};
		\node [style=none] (82) at (3, 0.375) {};
		\node [style=none] (83) at (3, -0.625) {};
		\node [style=none] (84) at (5.75, 0.375) {};
		\node [style=none] (85) at (5.75, -0.625) {};
		\node [style=none] (86) at (3.75, 0.375) {};
		\node [style=none] (87) at (3.75, -0.625) {};
		\node [style=none] (88) at (6.25, -0.125) {$\vdotss$};
		\node [style=none] (89) at (3, 0.375) {};
		\node [style=none] (90) at (3.25, -0.125) {$\vdotss$};
		\node [style=bphase] (91) at (4.75, 0.75) {\(\phasemat{0}{b}\)};
		\node [style=rmat] (92) at (3.75, 0.375) {\(\delta\)};
		\node [style=rmat] (93) at (3.75, -0.625) {\(\delta\)};
	\end{pgfonlayer}
	\begin{pgfonlayer}{edgelayer}
		\draw [style=background] (77.center)
			 to (76.center)
			 to (75.center)
			 to (78.center)
			 to cycle;
		\draw (79)
			 to [in=360, out=135] (86.center)
			 to (82.center);
		\draw (79)
			 to [in=0, out=-135] (87.center)
			 to (83.center);
		\draw (79)
			 to [in=180, out=45] (84.center)
			 to [in=180, out=0] (80.center);
		\draw (79)
			 to [in=180, out=-45] (85.center)
			 to (81.center);
		\draw [style=dbphase] (91) to (79);
	\end{pgfonlayer}
\end{tikzpicture}

    =
      \begin{tikzpicture}
	\begin{pgfonlayer}{nodelayer}
		\node [style=none] (94) at (7.5, -1.125) {};
		\node [style=none] (95) at (11, -1.125) {};
		\node [style=none] (96) at (11, 1.125) {};
		\node [style=none] (97) at (7.5, 1.125) {};
		\node [style=bn] (98) at (9.25, -0.125) {};
		\node [style=none] (99) at (11, 0.375) {};
		\node [style=none] (100) at (11, -0.625) {};
		\node [style=none] (101) at (7.5, 0.375) {};
		\node [style=none] (102) at (7.5, -0.625) {};
		\node [style=none] (103) at (10.25, 0.375) {};
		\node [style=none] (104) at (10.25, -0.625) {};
		\node [style=none] (105) at (8.25, 0.375) {};
		\node [style=none] (106) at (8.25, -0.625) {};
		\node [style=none] (107) at (10.75, -0.125) {$\vdotss$};
		\node [style=none] (108) at (7.5, 0.375) {};
		\node [style=none] (109) at (7.75, -0.125) {$\vdotss$};
		\node [style=bphase] (110) at (9.25, 0.75) {\(\phasemat{0}{b}\)};
		\node [style=rmat] (111) at (10.25, 0.375) {\(\delta\)};
		\node [style=rmat] (112) at (10.25, -0.625) {\(\delta\)};
	\end{pgfonlayer}
	\begin{pgfonlayer}{edgelayer}
		\draw [style=background] (96.center)
			 to (95.center)
			 to (94.center)
			 to (97.center)
			 to cycle;
		\draw (98)
			 to [in=360, out=135] (105.center)
			 to (101.center);
		\draw (98)
			 to [in=0, out=-135] (106.center)
			 to (102.center);
		\draw (98)
			 to [in=180, out=45] (103.center)
			 to [in=180, out=0] (99.center);
		\draw (98)
			 to [in=180, out=-45] (104.center)
			 to (100.center);
		\draw [style=dbphase] (110) to (98);
	\end{pgfonlayer}
\end{tikzpicture}

    \qquad\text{and}\qquad
      \begin{tikzpicture}
	\begin{pgfonlayer}{nodelayer}
		\node [style=none] (114) at (14, -1.125) {};
		\node [style=none] (115) at (17.5, -1.125) {};
		\node [style=none] (116) at (17.5, 1.125) {};
		\node [style=none] (117) at (14, 1.125) {};
		\node [style=wn] (118) at (15.75, -0.125) {};
		\node [style=none] (119) at (17.5, 0.375) {};
		\node [style=none] (120) at (17.5, -0.625) {};
		\node [style=none] (121) at (14, 0.375) {};
		\node [style=none] (122) at (14, -0.625) {};
		\node [style=none] (123) at (16.75, 0.375) {};
		\node [style=none] (124) at (16.75, -0.625) {};
		\node [style=none] (125) at (14.75, 0.375) {};
		\node [style=none] (126) at (14.75, -0.625) {};
		\node [style=none] (127) at (17.25, -0.125) {$\vdotss$};
		\node [style=none] (128) at (14, 0.375) {};
		\node [style=none] (129) at (14.25, -0.125) {$\vdotss$};
		\node [style=wphase] (130) at (15.75, 0.75) {\(\phasemat{0}{b}\)};
		\node [style=rmat] (131) at (14.75, 0.375) {\(\delta\)};
		\node [style=rmat] (132) at (14.75, -0.625) {\(\delta\)};
	\end{pgfonlayer}
	\begin{pgfonlayer}{edgelayer}
		\draw [style=background] (116.center)
			 to (115.center)
			 to (114.center)
			 to (117.center)
			 to cycle;
		\draw (118)
			 to [in=360, out=135] (125.center)
			 to (121.center);
		\draw (118)
			 to [in=0, out=-135] (126.center)
			 to (122.center);
		\draw (118)
			 to [in=180, out=45] (123.center)
			 to [in=180, out=0] (119.center);
		\draw (118)
			 to [in=180, out=-45] (124.center)
			 to (120.center);
		\draw [style=dwphase] (130) to (118);
	\end{pgfonlayer}
\end{tikzpicture}

    =
      \begin{tikzpicture}
	\begin{pgfonlayer}{nodelayer}
		\node [style=none] (133) at (18.5, -1.125) {};
		\node [style=none] (134) at (22, -1.125) {};
		\node [style=none] (135) at (22, 1.125) {};
		\node [style=none] (136) at (18.5, 1.125) {};
		\node [style=wn] (137) at (20.25, -0.125) {};
		\node [style=none] (138) at (22, 0.375) {};
		\node [style=none] (139) at (22, -0.625) {};
		\node [style=none] (140) at (18.5, 0.375) {};
		\node [style=none] (141) at (18.5, -0.625) {};
		\node [style=none] (142) at (21.25, 0.375) {};
		\node [style=none] (143) at (21.25, -0.625) {};
		\node [style=none] (144) at (19.25, 0.375) {};
		\node [style=none] (145) at (19.25, -0.625) {};
		\node [style=none] (146) at (21.75, -0.125) {$\vdotss$};
		\node [style=none] (147) at (18.5, 0.375) {};
		\node [style=none] (148) at (18.75, -0.125) {$\vdotss$};
		\node [style=wphase] (149) at (20.25, 0.75) {\(\phasemat{0}{b}\)};
		\node [style=rmat] (150) at (21.25, 0.375) {\(\delta\)};
		\node [style=rmat] (151) at (21.25, -0.625) {\(\delta\)};
	\end{pgfonlayer}
	\begin{pgfonlayer}{edgelayer}
		\draw [style=background] (135.center)
			 to (134.center)
			 to (133.center)
			 to (136.center)
			 to cycle;
		\draw (137)
			 to [in=360, out=135] (144.center)
			 to (140.center);
		\draw (137)
			 to [in=0, out=-135] (145.center)
			 to (141.center);
		\draw (137)
			 to [in=180, out=45] (142.center)
			 to [in=180, out=0] (138.center);
		\draw (137)
			 to [in=180, out=-45] (143.center)
			 to (139.center);
		\draw [style=dwphase] (149) to (137);
	\end{pgfonlayer}
\end{tikzpicture}
.
  \end{equation}
\end{definition}

The delay is interpreted as the process which waits and then releases its input one time-step later. Apart from the Paulis, we interpret the stabilizer ZX-diagram as being translation invariant, so that they are repeated at each time step, as expressed by the commutation with the delay. 
However, the Paulis are different: they are interpreted as errors or signals localised at time \(t = 0\). Shifting such an error in time would move it to another step, so it does not commute with the delay.
In other words, the translation-invariant structure is the scaffolding; whereas, the Paulis propagate through this scaffolding.

Informally, such a diagram unfolds into an infinite, translation-invariant diagram as suggested in the introduction:

\begin{example}
  \label{ex:surface_flow}
Revisiting the surface code example given in \Cref{eq:surface_code} in the introduction, we see how the Pauli can flow through the infinite circuit:
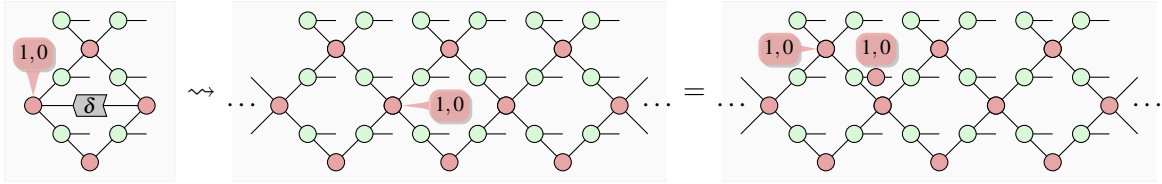
\[
    \begin{tikzpicture}
	\begin{pgfonlayer}{nodelayer}
		\node [style=none] (13) at (5.25, 2.375) {};
		\node [style=none] (14) at (5.25, -2.375) {};
		\node [style=none] (15) at (9.75, -2.375) {};
		\node [style=none] (16) at (9.75, 2.375) {};
		\node [style=wn] (198) at (7.5, -1.875) {};
		\node [style=bn] (199) at (6.75, -1.125) {};
		\node [style=wn] (200) at (9, -0.375) {};
		\node [style=bn] (201) at (8.25, 0.375) {};
		\node [style=bn] (202) at (8.25, -1.125) {};
		\node [style=none] (203) at (9, 0.375) {};
		\node [style=none] (204) at (7.5, -1.125) {};
		\node [style=none] (205) at (9, -1.125) {};
		\node [style=wn] (206) at (6, -0.375) {};
		\node [style=wn] (207) at (7.5, 1.125) {};
		\node [style=bn] (208) at (6.75, 0.375) {};
		\node [style=none] (209) at (7.5, 0.375) {};
		\node [style=lmat] (210) at (7.5, -0.375) {\(\delta\)};
		\node [style=wphase] (211) at (6, 1) {\(1,0\)};
		\node [style=wn] (212) at (7.5, 1.125) {};
		\node [style=bn] (213) at (6.75, 1.875) {};
		\node [style=bn] (214) at (8.25, 1.875) {};
		\node [style=none] (215) at (7.5, 1.875) {};
		\node [style=none] (216) at (9, 1.875) {};
	\end{pgfonlayer}
	\begin{pgfonlayer}{edgelayer}
		\draw [style=background] (16.center)
			 to (15.center)
			 to (14.center)
			 to (13.center)
			 to cycle;
		\draw (198) to (199);
		\draw (201) to (200);
		\draw (200) to (202);
		\draw (202) to (198);
		\draw (201) to (203.center);
		\draw (202) to (205.center);
		\draw (199) to (204.center);
		\draw (207) to (208);
		\draw (208) to (206);
		\draw (208) to (209.center);
		\draw (201) to (207);
		\draw (199) to (206);
		\draw (206) to (210);
		\draw (210) to (200);
		\draw [style=dwphase] (211) to (206);
		\draw (212) to (213);
		\draw (214) to (212);
		\draw (214) to (216.center);
		\draw (213) to (215.center);
	\end{pgfonlayer}
\end{tikzpicture}

  \rightsquigarrow
    \begin{tikzpicture}
	\begin{pgfonlayer}{nodelayer}
		\node [style=none] (13) at (1.75, 2.375) {};
		\node [style=none] (14) at (1.75, -2.375) {};
		\node [style=none] (15) at (13.25, -2.375) {};
		\node [style=none] (16) at (13.25, 2.375) {};
		\node [style=wn] (198) at (7.5, -1.875) {};
		\node [style=bn] (199) at (6.75, -1.125) {};
		\node [style=wn] (200) at (9, -0.375) {};
		\node [style=bn] (201) at (8.25, 0.375) {};
		\node [style=bn] (202) at (8.25, -1.125) {};
		\node [style=none] (203) at (9, 0.375) {};
		\node [style=none] (204) at (7.5, -1.125) {};
		\node [style=none] (205) at (9, -1.125) {};
		\node [style=wn] (206) at (6, -0.375) {};
		\node [style=wn] (207) at (7.5, 1.125) {};
		\node [style=bn] (208) at (6.75, 0.375) {};
		\node [style=none] (209) at (7.5, 0.375) {};
		\node [style=wn] (212) at (10.5, -1.875) {};
		\node [style=bn] (213) at (9.75, -1.125) {};
		\node [style=wn] (214) at (12, -0.375) {};
		\node [style=bn] (215) at (11.25, 0.375) {};
		\node [style=bn] (216) at (11.25, -1.125) {};
		\node [style=none] (217) at (12, 0.375) {};
		\node [style=none] (218) at (10.5, -1.125) {};
		\node [style=none] (219) at (12, -1.125) {};
		\node [style=wn] (220) at (9, -0.375) {};
		\node [style=wn] (221) at (10.5, 1.125) {};
		\node [style=bn] (222) at (9.75, 0.375) {};
		\node [style=none] (223) at (10.5, 0.375) {};
		\node [style=wn] (224) at (4.5, -1.875) {};
		\node [style=bn] (225) at (3.75, -1.125) {};
		\node [style=wn] (226) at (6, -0.375) {};
		\node [style=bn] (227) at (5.25, 0.375) {};
		\node [style=bn] (228) at (5.25, -1.125) {};
		\node [style=none] (229) at (6, 0.375) {};
		\node [style=none] (230) at (4.5, -1.125) {};
		\node [style=none] (231) at (6, -1.125) {};
		\node [style=wn] (232) at (3, -0.375) {};
		\node [style=wn] (233) at (4.5, 1.125) {};
		\node [style=bn] (234) at (3.75, 0.375) {};
		\node [style=none] (235) at (4.5, 0.375) {};
		\node [style=none] (236) at (2.25, 0.375) {};
		\node [style=none] (237) at (2.25, -1.125) {};
		\node [style=none] (238) at (12.75, -1.125) {};
		\node [style=none] (239) at (12.75, 0.375) {};
		\node [style=none] (240) at (2, -0.375) {\(\cdots\)};
		\node [style=none] (241) at (13, -0.375) {\(\cdots\)};
		\node [style=wphase] (244) at (7.5, -0.375) {\(1,0\)};
		\node [style=wn] (245) at (7.5, 1.125) {};
		\node [style=wn] (246) at (10.5, 1.125) {};
		\node [style=wn] (247) at (4.5, 1.125) {};
		\node [style=bn] (248) at (8.25, 1.875) {};
		\node [style=none] (249) at (9, 1.875) {};
		\node [style=bn] (250) at (6.75, 1.875) {};
		\node [style=none] (251) at (7.5, 1.875) {};
		\node [style=bn] (252) at (11.25, 1.875) {};
		\node [style=none] (253) at (12, 1.875) {};
		\node [style=bn] (254) at (9.75, 1.875) {};
		\node [style=none] (255) at (10.5, 1.875) {};
		\node [style=bn] (256) at (5.25, 1.875) {};
		\node [style=none] (257) at (6, 1.875) {};
		\node [style=bn] (258) at (3.75, 1.875) {};
		\node [style=none] (259) at (4.5, 1.875) {};
	\end{pgfonlayer}
	\begin{pgfonlayer}{edgelayer}
		\draw [style=backgroundborderless] (16.center)
			 to (15.center)
			 to (14.center)
			 to (13.center)
			 to cycle;
		\draw (198) to (199);
		\draw (201) to (200);
		\draw (200) to (202);
		\draw (202) to (198);
		\draw (201) to (203.center);
		\draw (202) to (205.center);
		\draw (199) to (204.center);
		\draw (207) to (208);
		\draw (208) to (206);
		\draw (208) to (209.center);
		\draw (201) to (207);
		\draw (199) to (206);
		\draw (212) to (213);
		\draw (215) to (214);
		\draw (214) to (216);
		\draw (216) to (212);
		\draw (215) to (217.center);
		\draw (216) to (219.center);
		\draw (213) to (218.center);
		\draw (221) to (222);
		\draw (222) to (220);
		\draw (222) to (223.center);
		\draw (215) to (221);
		\draw (213) to (220);
		\draw (224) to (225);
		\draw (227) to (226);
		\draw (226) to (228);
		\draw (228) to (224);
		\draw (227) to (229.center);
		\draw (228) to (231.center);
		\draw (225) to (230.center);
		\draw (233) to (234);
		\draw (234) to (232);
		\draw (234) to (235.center);
		\draw (227) to (233);
		\draw (225) to (232);
		\draw (238.center) to (214);
		\draw (214) to (239.center);
		\draw (232) to (236.center);
		\draw (237.center) to (232);
		\draw [style=backgroundborder] (13.center) to (16.center);
		\draw [style=backgroundborder] (15.center) to (14.center);
		\draw [style=lwphase] (244) to (226);
		\draw (248) to (249.center);
		\draw (250) to (251.center);
		\draw (252) to (253.center);
		\draw (254) to (255.center);
		\draw (256) to (257.center);
		\draw (258) to (259.center);
		\draw (247) to (258);
		\draw (247) to (256);
		\draw (245) to (250);
		\draw (245) to (248);
		\draw (246) to (254);
		\draw (246) to (252);
	\end{pgfonlayer}
\end{tikzpicture}

  =
    \begin{tikzpicture}
	\begin{pgfonlayer}{nodelayer}
		\node [style=none] (13) at (1.75, 2.375) {};
		\node [style=none] (14) at (1.75, -2.375) {};
		\node [style=none] (15) at (13.25, -2.375) {};
		\node [style=none] (16) at (13.25, 2.375) {};
		\node [style=wn] (198) at (7.5, -1.875) {};
		\node [style=bn] (199) at (6.75, -1.125) {};
		\node [style=wn] (200) at (9, -0.375) {};
		\node [style=bn] (201) at (8.25, 0.375) {};
		\node [style=bn] (202) at (8.25, -1.125) {};
		\node [style=none] (203) at (9, 0.375) {};
		\node [style=none] (204) at (7.5, -1.125) {};
		\node [style=none] (205) at (9, -1.125) {};
		\node [style=wn] (206) at (6, -0.375) {};
		\node [style=wn] (207) at (7.5, 1.125) {};
		\node [style=bn] (208) at (6.75, 0.375) {};
		\node [style=none] (209) at (7.5, 0.375) {};
		\node [style=wn] (212) at (10.5, -1.875) {};
		\node [style=bn] (213) at (9.75, -1.125) {};
		\node [style=wn] (214) at (12, -0.375) {};
		\node [style=bn] (215) at (11.25, 0.375) {};
		\node [style=bn] (216) at (11.25, -1.125) {};
		\node [style=none] (217) at (12, 0.375) {};
		\node [style=none] (218) at (10.5, -1.125) {};
		\node [style=none] (219) at (12, -1.125) {};
		\node [style=wn] (220) at (9, -0.375) {};
		\node [style=wn] (221) at (10.5, 1.125) {};
		\node [style=bn] (222) at (9.75, 0.375) {};
		\node [style=none] (223) at (10.5, 0.375) {};
		\node [style=wn] (224) at (4.5, -1.875) {};
		\node [style=bn] (225) at (3.75, -1.125) {};
		\node [style=wn] (226) at (6, -0.375) {};
		\node [style=bn] (227) at (5.25, 0.375) {};
		\node [style=bn] (228) at (5.25, -1.125) {};
		\node [style=none] (229) at (6.25, 0.375) {};
		\node [style=none] (230) at (4.5, -1.125) {};
		\node [style=none] (231) at (6, -1.125) {};
		\node [style=wn] (232) at (3, -0.375) {};
		\node [style=wn] (233) at (4.5, 1.125) {};
		\node [style=bn] (234) at (3.75, 0.375) {};
		\node [style=none] (235) at (4.5, 0.375) {};
		\node [style=none] (236) at (2.25, 0.375) {};
		\node [style=none] (237) at (2.25, -1.125) {};
		\node [style=none] (238) at (12.75, -1.125) {};
		\node [style=none] (239) at (12.75, 0.375) {};
		\node [style=none] (240) at (2, -0.375) {\(\cdots\)};
		\node [style=none] (241) at (13, -0.375) {\(\cdots\)};
		\node [style=wphase] (242) at (3.25, 1.125) {\(1,0\)};
		\node [style=wn] (243) at (5.825, 0.375) {};
		\node [style=wphase] (244) at (5.825, 1.125) {\(1,0\)};
		\node [style=wn] (245) at (7.5, 1.125) {};
		\node [style=wn] (246) at (10.5, 1.125) {};
		\node [style=wn] (247) at (4.5, 1.125) {};
		\node [style=bn] (248) at (8.25, 1.875) {};
		\node [style=none] (249) at (9, 1.875) {};
		\node [style=bn] (250) at (6.75, 1.875) {};
		\node [style=none] (251) at (7.5, 1.875) {};
		\node [style=bn] (252) at (11.25, 1.875) {};
		\node [style=none] (253) at (12, 1.875) {};
		\node [style=bn] (254) at (9.75, 1.875) {};
		\node [style=none] (255) at (10.5, 1.875) {};
		\node [style=bn] (256) at (5.25, 1.875) {};
		\node [style=none] (257) at (6, 1.875) {};
		\node [style=bn] (258) at (3.75, 1.875) {};
		\node [style=none] (259) at (4.5, 1.875) {};
	\end{pgfonlayer}
	\begin{pgfonlayer}{edgelayer}
		\draw [style=backgroundborderless] (16.center)
			 to (15.center)
			 to (14.center)
			 to (13.center)
			 to cycle;
		\draw (198) to (199);
		\draw (201) to (200);
		\draw (200) to (202);
		\draw (202) to (198);
		\draw (201) to (203.center);
		\draw (202) to (205.center);
		\draw (199) to (204.center);
		\draw (207) to (208);
		\draw (208) to (206);
		\draw (208) to (209.center);
		\draw (201) to (207);
		\draw (199) to (206);
		\draw (212) to (213);
		\draw (215) to (214);
		\draw (214) to (216);
		\draw (216) to (212);
		\draw (215) to (217.center);
		\draw (216) to (219.center);
		\draw (213) to (218.center);
		\draw (221) to (222);
		\draw (222) to (220);
		\draw (222) to (223.center);
		\draw (215) to (221);
		\draw (213) to (220);
		\draw (224) to (225);
		\draw (227) to (226);
		\draw (226) to (228);
		\draw (228) to (224);
		\draw (227) to (229.center);
		\draw (228) to (231.center);
		\draw (225) to (230.center);
		\draw (233) to (234);
		\draw (234) to (232);
		\draw (234) to (235.center);
		\draw (227) to (233);
		\draw (225) to (232);
		\draw (238.center) to (214);
		\draw (214) to (239.center);
		\draw (232) to (236.center);
		\draw (237.center) to (232);
		\draw [style=backgroundborder] (13.center) to (16.center);
		\draw [style=backgroundborder] (15.center) to (14.center);
		\draw [style=rwphase] (242) to (233);
		\draw [style=dwphase] (244) to (243);
		\draw (248) to (249.center);
		\draw (250) to (251.center);
		\draw (252) to (253.center);
		\draw (254) to (255.center);
		\draw (256) to (257.center);
		\draw (258) to (259.center);
		\draw (247) to (258);
		\draw (247) to (256);
		\draw (245) to (250);
		\draw (245) to (248);
		\draw (246) to (254);
		\draw (246) to (252);
	\end{pgfonlayer}
\end{tikzpicture}

\]
\end{example}

However, such infinite diagrams are only heuristic pictures of the intended translation-invariant process, rather than the semantics we work with directly. We formalise them instead through their \emph{unrollings}: finite truncations obtained by restricting the evolution of the process to finite intervals of time. To play nicely with the compact closed structure of \(\St(\StabZX)\), we fix no starting time, truncating over finite intervals of \(\Z\):

\begin{definition}
  \label{def:unrolling}
  Given a stateful stabilizer ZX-diagram \(D\), define its unrolling \(\Unroll(D)\) as a family of \(\ZX^\disc\) diagrams indexed by finite intervals of  \(\Z\):
  \begin{gather*}
    \mapsto  
      \left\{
        \bigotimes_{t=\ell}^r
        \begin{matrix}
          & \text{if }t=0\\[.5cm]
          \begin{tikzpicture}
	\begin{pgfonlayer}{nodelayer}
		\node [style=none] (9) at (-1, 0.95) {};
		\node [style=none] (10) at (2, 0.95) {};
		\node [style=none] (11) at (2, -0.95) {};
		\node [style=none] (12) at (-1, -0.95) {};
		\node [style=wn] (14) at (0.5, -0.175) {};
		\node [style=none] (15) at (1.5, 0.325) {};
		\node [style=none] (16) at (1.5, -0.675) {};
		\node [style=none] (17) at (-0.5, 0.325) {};
		\node [style=none] (18) at (-0.5, -0.675) {};
		\node [style=none] (19) at (-0.525, -0.175) {$m$};
		\node [style=none] (20) at (2, 0.325) {};
		\node [style=none] (21) at (2, -0.675) {};
		\node [style=none] (26) at (0, 0.025) {$\vdots$};
		\node [style=none] (27) at (1.375, -0.175) {$n$};
		\node [style=none] (28) at (1, 0.025) {$\vdots$};
		\node [style=none] (29) at (-1, 0.325) {};
		\node [style=none] (30) at (-1, -0.675) {};
		\node [style=wphase] (31) at (0.5, 0.75) {$\phasemat{0}{b}$};
	\end{pgfonlayer}
	\begin{pgfonlayer}{edgelayer}
		\draw [style=background] (11.center)
			 to (10.center)
			 to (9.center)
			 to (12.center)
			 to cycle;
		\draw (20.center)
			 to (15.center)
			 to [in=60, out=-180] (14);
		\draw (21.center)
			 to (16.center)
			 to [in=-60, out=180] (14);
		\draw (29.center)
			 to (17.center)
			 to [in=120, out=360] (14);
		\draw (14)
			 to [in=0, out=-120] (18.center)
			 to (30.center);
		\draw [style=dwphase] (31) to (14);
	\end{pgfonlayer}
\end{tikzpicture}
& \text{if } t\neq 0
        \end{matrix}
    \right\}_{\ell\leq r\in\Z}
    \quad
      \mapsto  
      \left\{
        \bigotimes_{t=\ell}^r
        \begin{matrix}
          & \text{if }t=0\\[.5cm]
          \begin{tikzpicture}
	\begin{pgfonlayer}{nodelayer}
		\node [style=none] (9) at (-1, 0.95) {};
		\node [style=none] (10) at (2, 0.95) {};
		\node [style=none] (11) at (2, -0.95) {};
		\node [style=none] (12) at (-1, -0.95) {};
		\node [style=bn] (14) at (0.5, -0.175) {};
		\node [style=none] (15) at (1.5, 0.325) {};
		\node [style=none] (16) at (1.5, -0.675) {};
		\node [style=none] (17) at (-0.5, 0.325) {};
		\node [style=none] (18) at (-0.5, -0.675) {};
		\node [style=none] (19) at (-0.525, -0.175) {$m$};
		\node [style=none] (20) at (2, 0.325) {};
		\node [style=none] (21) at (2, -0.675) {};
		\node [style=none] (26) at (0, 0.025) {$\vdots$};
		\node [style=none] (27) at (1.375, -0.175) {$n$};
		\node [style=none] (28) at (1, 0.025) {$\vdots$};
		\node [style=none] (29) at (-1, 0.325) {};
		\node [style=none] (30) at (-1, -0.675) {};
		\node [style=bphase] (31) at (0.5, 0.75) {$\phasemat{0}{b}$};
	\end{pgfonlayer}
	\begin{pgfonlayer}{edgelayer}
		\draw [style=background] (11.center)
			 to (10.center)
			 to (9.center)
			 to (12.center)
			 to cycle;
		\draw (20.center)
			 to (15.center)
			 to [in=60, out=-180] (14);
		\draw (21.center)
			 to (16.center)
			 to [in=-60, out=180] (14);
		\draw (29.center)
			 to (17.center)
			 to [in=120, out=360] (14);
		\draw (14)
			 to [in=0, out=-120] (18.center)
			 to (30.center);
		\draw [style=dbphase] (31) to (14);
	\end{pgfonlayer}
\end{tikzpicture}
& \text{if }t\neq 0
        \end{matrix}
      \right\}_{\ell\leq r\in\Z}\\
  \begin{tikzpicture}
	\begin{pgfonlayer}{nodelayer}
		\node [style=rmat] (1) at (0.75, 0) {\(\delta\)};
		\node [style=none] (2) at (0, 0.4) {};
		\node [style=none] (3) at (1.5, 0.4) {};
		\node [style=none] (4) at (1.5, -0.4) {};
		\node [style=none] (5) at (0, -0.4) {};
		\node [style=none] (6) at (0, 0) {};
		\node [style=none] (7) at (1.5, 0) {};
	\end{pgfonlayer}
	\begin{pgfonlayer}{edgelayer}
		\draw [style=background] (2.center)
			 to (3.center)
			 to (4.center)
			 to (5.center)
			 to cycle;
		\draw [style=background] (5.center) to (2.center);
		\draw [style=background] (2.center) to (3.center);
		\draw [style=background] (3.center) to (4.center);
		\draw [style=background] (4.center) to (5.center);
		\draw (7.center) to (6.center);
	\end{pgfonlayer}
\end{tikzpicture}
\mapsto 
    \left\{
      \begin{tikzpicture}
	\begin{pgfonlayer}{nodelayer}
		\node [style=none] (12) at (6, 1.0875) {};
		\node [style=none] (13) at (0.325, 1.0875) {};
		\node [style=none] (14) at (6, -1.0875) {};
		\node [style=none] (15) at (0.325, -1.0875) {};
		\node [style=none] (48) at (0.325, -0.5375) {};
		\node [style=none] (49) at (4.5, -0.5375) {};
		\node [style=none] (51) at (5.5, 0.0875) {};
		\node [style=none] (61) at (6, 0.0875) {};
		\node [style=none] (68) at (0.325, -0.0375) {};
		\node [style=none] (69) at (0.825, -0.0375) {};
		\node [style=none] (71) at (1.825, 0.5875) {};
		\node [style=none] (80) at (6, 0.5875) {};
		\node [style=none] (83) at (3.45, -0.0125) {\scalebox{.8}{\(r-\ell-1\)}};
		\node [style=none] (84) at (2, -0.0125) {\(\vdotss\)};
		\node [style=none] (87) at (1, 0.6125) {};
		\node [style=none] (88) at (0.325, 0.6125) {};
		\node [style=none] (89) at (5.5, -0.5375) {};
		\node [style=none] (90) at (6, -0.5375) {};
		\node [style=groundout] (91) at (1, 0.6125) {};
		\node [style=groundin] (92) at (5.5, -0.5375) {};
	\end{pgfonlayer}
	\begin{pgfonlayer}{edgelayer}
		\draw [style=background] (13.center)
			 to (12.center)
			 to (14.center)
			 to (15.center)
			 to cycle;
		\draw (48.center)
			 to (49.center)
			 to [in=180, out=0] (51.center)
			 to (61.center);
		\draw (80.center)
			 to (71.center)
			 to [in=0, out=180] (69.center)
			 to (68.center);
		\draw (88.center) to (87.center);
		\draw (89.center) to (90.center);
	\end{pgfonlayer}
\end{tikzpicture}

    \right\}_{\ell< r\in\Z}
    \bigsqcup\qquad
    \left\{
      \begin{tikzpicture}
	\begin{pgfonlayer}{nodelayer}
		\node [style=none] (12) at (6, 1.0875) {};
		\node [style=none] (13) at (0.325, 1.0875) {};
		\node [style=none] (14) at (6, -1.0875) {};
		\node [style=none] (15) at (0.325, -1.0875) {};
	\end{pgfonlayer}
	\begin{pgfonlayer}{edgelayer}
		\draw [style=background] (13.center)
			 to (12.center)
			 to (14.center)
			 to (15.center)
			 to cycle;
	\end{pgfonlayer}
\end{tikzpicture}

    \right\}_{\ell=r \in \Z}
  \end{gather*}
  where the tensor product and composition of unrollings are defined pointwise.
\end{definition}

Revisiting the infinite ZX-diagram above:
\begin{example} Truncating over the interval \([-1,1]\):
  \label{ex:trunc}
  \[
      \Unroll\left(\right)_{-1\leq 1}
    =\qquad
      \begin{tikzpicture}
	\begin{pgfonlayer}{nodelayer}
		\node [style=none] (13) at (1.75, 2.375) {};
		\node [style=none] (14) at (1.75, -2.375) {};
		\node [style=none] (15) at (13.25, -2.375) {};
		\node [style=none] (16) at (13.25, 2.375) {};
		\node [style=wn] (198) at (7.5, -1.875) {};
		\node [style=bn] (199) at (6.75, -1.125) {};
		\node [style=wn] (200) at (9, -0.375) {};
		\node [style=bn] (201) at (8.25, 0.375) {};
		\node [style=bn] (202) at (8.25, -1.125) {};
		\node [style=none] (203) at (9, 0.375) {};
		\node [style=none] (204) at (7.5, -1.125) {};
		\node [style=none] (205) at (9, -1.125) {};
		\node [style=wn] (206) at (6, -0.375) {};
		\node [style=wn] (207) at (7.5, 1.125) {};
		\node [style=bn] (208) at (6.75, 0.375) {};
		\node [style=none] (209) at (7.5, 0.375) {};
		\node [style=wn] (212) at (10.5, -1.875) {};
		\node [style=bn] (213) at (9.75, -1.125) {};
		\node [style=wn] (214) at (12, -0.375) {};
		\node [style=bn] (215) at (11.25, 0.375) {};
		\node [style=bn] (216) at (11.25, -1.125) {};
		\node [style=none] (217) at (12, 0.375) {};
		\node [style=none] (218) at (10.5, -1.125) {};
		\node [style=none] (219) at (12, -1.125) {};
		\node [style=wn] (220) at (9, -0.375) {};
		\node [style=wn] (221) at (10.5, 1.125) {};
		\node [style=bn] (222) at (9.75, 0.375) {};
		\node [style=none] (223) at (10.5, 0.375) {};
		\node [style=wn] (224) at (4.5, -1.875) {};
		\node [style=bn] (225) at (3.75, -1.125) {};
		\node [style=wn] (226) at (6, -0.375) {};
		\node [style=bn] (227) at (5.25, 0.375) {};
		\node [style=bn] (228) at (5.25, -1.125) {};
		\node [style=none] (229) at (6, 0.375) {};
		\node [style=none] (230) at (4.5, -1.125) {};
		\node [style=none] (231) at (6, -1.125) {};
		\node [style=wn] (232) at (3, -0.375) {};
		\node [style=wn] (233) at (4.5, 1.125) {};
		\node [style=bn] (234) at (3.75, 0.375) {};
		\node [style=none] (235) at (4.5, 0.375) {};
		\node [style=groundin] (237) at (2.25, -0.375) {};
		\node [style=groundout] (238) at (12.75, -0.375) {};
		\node [style=wphase] (244) at (7.5, -0.375) {\(1,0\)};
		\node [style=bn] (247) at (8.25, 1.875) {};
		\node [style=none] (248) at (9, 1.875) {};
		\node [style=bn] (249) at (6.75, 1.875) {};
		\node [style=none] (250) at (7.5, 1.875) {};
		\node [style=bn] (251) at (11.25, 1.875) {};
		\node [style=none] (252) at (12, 1.875) {};
		\node [style=bn] (253) at (9.75, 1.875) {};
		\node [style=none] (254) at (10.5, 1.875) {};
		\node [style=bn] (255) at (5.25, 1.875) {};
		\node [style=none] (256) at (6, 1.875) {};
		\node [style=bn] (257) at (3.75, 1.875) {};
		\node [style=none] (258) at (4.5, 1.875) {};
	\end{pgfonlayer}
	\begin{pgfonlayer}{edgelayer}
		\draw [style=background] (16.center)
			 to (15.center)
			 to (14.center)
			 to (13.center)
			 to cycle;
		\draw (198) to (199);
		\draw (201) to (200);
		\draw (200) to (202);
		\draw (202) to (198);
		\draw (201) to (203.center);
		\draw (202) to (205.center);
		\draw (199) to (204.center);
		\draw (207) to (208);
		\draw (208) to (206);
		\draw (208) to (209.center);
		\draw (201) to (207);
		\draw (199) to (206);
		\draw (212) to (213);
		\draw (215) to (214);
		\draw (214) to (216);
		\draw (216) to (212);
		\draw (215) to (217.center);
		\draw (216) to (219.center);
		\draw (213) to (218.center);
		\draw (221) to (222);
		\draw (222) to (220);
		\draw (222) to (223.center);
		\draw (215) to (221);
		\draw (213) to (220);
		\draw (224) to (225);
		\draw (227) to (226);
		\draw (226) to (228);
		\draw (228) to (224);
		\draw (227) to (229.center);
		\draw (228) to (231.center);
		\draw (225) to (230.center);
		\draw (233) to (234);
		\draw (234) to (232);
		\draw (234) to (235.center);
		\draw (227) to (233);
		\draw (225) to (232);
		\draw (238) to (214);
		\draw (237) to (232);
		\draw [style=lwphase] (244) to (226);
		\draw (247) to (248.center);
		\draw (249) to (250.center);
		\draw (251) to (252.center);
		\draw (253) to (254.center);
		\draw (255) to (256.center);
		\draw (257) to (258.center);
		\draw (233) to (257);
		\draw (233) to (255);
		\draw (207) to (249);
		\draw (207) to (247);
		\draw (221) to (253);
		\draw (221) to (251);
	\end{pgfonlayer}
\end{tikzpicture}

  \]
\end{example}

However, as stands, unrolling is not a functor, so the interpretation above is not well-defined. Two diagrams that are equal in \(\St(\StabZX)\) can unroll to distinct sequences, depending on the order in which the delay is applied.
For example, an error introduced outside a given interval can still propagate back into that interval, despite the fact that it can not be observed in the smaller truncation. Therefore two distinct unrollings may eventually agree only once the time interval is widened far enough to reveal the error.
In \Cref{sec:behaviours} we recover functoriality by identifying two unrollings whenever they exhibit the same behaviour in this sense.

We show that these behaviours can be interpreted as equivalence classes of compatible finite observations, given by sequences of completely positive maps between finite-dimensional Hilbert spaces. Thus the observation is not a single fixed truncation, but a coherent family of finite observations, and no infinite-dimensional channel semantics has to be chosen.

Finally, in \Cref{sec:infinite}
we show that the family of stabiliser groups given by a family for truncations determines a unique infinite stabilizer group; thus capturing the infinite flows of Pauli operators through the translation invariant stabilizer quantum process.

\subsection{Observational equivalence}
\label{sec:behaviours}
To compare unrollings, we first capture the structure of an unrolling abstractly, in any discard bicategory, so that we can reuse it for completely positive maps and for stabilizer groups:
\begin{definition}
  \label{def:monotone}
Let \(\mathcal C\) be a discard bicategory, and take two bi-infinite sequence of objects in  \(\mathcal{C}\), 
\(\vb X,\vb Y\in (\Ob(\mathcal C))^{\Z}\).
A \textbf{monotone sequence} \(\vb U:\vb X\to\vb Y\) is a morphisms in \(\mathcal{C}\) indexed by finite intervals of  \(\Z\)
\[
\left\{\begin{tikzpicture}
	\begin{pgfonlayer}{nodelayer}
		\node [style=none] (0) at (-1.5, 1.875) {};
		\node [style=none] (1) at (-1.5, -1.875) {};
		\node [style=none] (2) at (5.5, -1.875) {};
		\node [style=none] (3) at (5.5, 1.875) {};
		\node [style=none] (4) at (1, 1.375) {};
		\node [style=none] (5) at (1, -1.375) {};
		\node [style=none] (6) at (3, -1.375) {};
		\node [style=none] (7) at (3, 1.375) {};
		\node [style=none] (8) at (2, 0) {\(U_{\ell,r}\)};
		\node [style=none] (9) at (1, 1.125) {};
		\node [style=none] (10) at (1, 0.5) {};
		\node [style=none] (11) at (1, -0.5) {};
		\node [style=none] (12) at (1, -1.125) {};
		\node [style=none] (13) at (-1.5, 1.125) {};
		\node [style=none] (14) at (-1.5, 0.5) {};
		\node [style=none] (15) at (-1.5, -0.5) {};
		\node [style=none] (16) at (-1.5, -1.125) {};
		\node [style=wirelable] (17) at (-0.5, 1.375) {\(X_{\ell}\)};
		\node [style=wirelable] (18) at (-0.5, 0.75) {\(X_{\ell+1}\)};
		\node [style=wirelable] (19) at (-0.5, -0.25) {\(X_{r-1}\)};
		\node [style=wirelable] (20) at (-0.5, -0.875) {\(X_{r}\)};
		\node [style=none] (21) at (0.5, 0) {\(\vdotss\)};
		\node [style=none] (22) at (3, 1.125) {};
		\node [style=none] (23) at (3, 0.5) {};
		\node [style=none] (24) at (3, -0.5) {};
		\node [style=none] (25) at (3, -1.125) {};
		\node [style=none] (26) at (5.5, 1.125) {};
		\node [style=none] (27) at (5.5, 0.5) {};
		\node [style=none] (28) at (5.5, -0.5) {};
		\node [style=none] (29) at (5.5, -1.125) {};
		\node [style=wirelable] (30) at (4.5, 1.375) {\(Y_{\ell}\)};
		\node [style=wirelable] (31) at (4.5, 0.75) {\(Y_{\ell+1}\)};
		\node [style=wirelable] (32) at (4.5, -0.25) {\(Y_{r-1}\)};
		\node [style=wirelable] (33) at (4.5, -0.875) {\(Y_{r}\)};
		\node [style=none] (34) at (3.5, 0) {\(\vdotss\)};
	\end{pgfonlayer}
	\begin{pgfonlayer}{edgelayer}
		\draw [style=background] (3.center)
			 to (2.center)
			 to (1.center)
			 to (0.center)
			 to cycle;
		\draw [style=background] (0.center) to (3.center);
		\draw [style=background] (3.center) to (2.center);
		\draw [style=background] (2.center) to (1.center);
		\draw [style=background] (1.center) to (0.center);
		\draw [style=bbox] (5.center)
			 to (4.center)
			 to (7.center)
			 to (6.center)
			 to cycle;
		\draw (13.center) to (9.center);
		\draw (10.center) to (14.center);
		\draw (11.center) to (15.center);
		\draw (16.center) to (12.center);
		\draw (26.center) to (22.center);
		\draw (23.center) to (27.center);
		\draw (24.center) to (28.center);
		\draw (29.center) to (25.center);
	\end{pgfonlayer}
\end{tikzpicture}
\right\}_{\ell\leq r \in \Z}
\]
such that for all \(\ell\leq r \in \Z\):
\[
\begin{tikzpicture}
	\begin{pgfonlayer}{nodelayer}
		\node [style=none] (0) at (-2.625, 1.875) {};
		\node [style=none] (1) at (-2.625, -1.875) {};
		\node [style=none] (2) at (4.875, -1.875) {};
		\node [style=none] (3) at (4.875, 1.875) {};
		\node [style=none] (4) at (-0.375, 1.375) {};
		\node [style=none] (5) at (-0.375, -0.875) {};
		\node [style=none] (6) at (2.375, -0.875) {};
		\node [style=none] (7) at (2.375, 1.375) {};
		\node [style=none] (8) at (1, 0.25) {\(U_{\ell-1,r}\)};
		\node [style=none] (9) at (-0.375, 1.125) {};
		\node [style=none] (10) at (-0.375, 0.5) {};
		\node [style=none] (11) at (-0.375, -0.5) {};
		\node [style=none] (13) at (-2.625, 1.125) {};
		\node [style=none] (14) at (-2.625, 0.5) {};
		\node [style=none] (15) at (-2.625, -0.5) {};
		\node [style=wirelable] (17) at (-1.625, 1.375) {\(X_{\ell-1}\)};
		\node [style=wirelable] (18) at (-1.625, 0.75) {\(X_{\ell}\)};
		\node [style=wirelable] (19) at (-1.625, -0.25) {\(X_{r}\)};
		\node [style=none] (21) at (-0.875, 0) {\(\vdotss\)};
		\node [style=none] (22) at (2.375, 1.125) {};
		\node [style=none] (23) at (2.375, 0.5) {};
		\node [style=none] (24) at (2.375, -0.5) {};
		\node [style=none] (26) at (4.375, 1.125) {};
		\node [style=none] (27) at (4.875, 0.5) {};
		\node [style=none] (28) at (4.875, -0.5) {};
		\node [style=wirelable] (30) at (3.625, 1.375) {\(Y_{\ell-1}\)};
		\node [style=wirelable] (31) at (3.625, 0.75) {\(Y_{\ell}\)};
		\node [style=wirelable] (32) at (3.625, -0.25) {\(Y_{r}\)};
		\node [style=none] (34) at (2.875, 0) {\(\vdotss\)};
		\node [style=groundout] (70) at (4.375, 1.125) {};
	\end{pgfonlayer}
	\begin{pgfonlayer}{edgelayer}
		\draw [style=background] (3.center)
			 to (2.center)
			 to (1.center)
			 to (0.center)
			 to cycle;
		\draw [style=background] (0.center) to (3.center);
		\draw [style=background] (3.center) to (2.center);
		\draw [style=background] (2.center) to (1.center);
		\draw [style=background] (1.center) to (0.center);
		\draw [style=bbox] (5.center)
			 to (4.center)
			 to (7.center)
			 to (6.center)
			 to cycle;
		\draw (13.center) to (9.center);
		\draw (10.center) to (14.center);
		\draw (11.center) to (15.center);
		\draw (26.center) to (22.center);
		\draw (23.center) to (27.center);
		\draw (24.center) to (28.center);
	\end{pgfonlayer}
\end{tikzpicture}

\subseteq
\begin{tikzpicture}
	\begin{pgfonlayer}{nodelayer}
		\node [style=none] (35) at (5.875, 1.875) {};
		\node [style=none] (36) at (5.875, -1.875) {};
		\node [style=none] (37) at (12.125, -1.875) {};
		\node [style=none] (38) at (12.125, 1.875) {};
		\node [style=none] (39) at (8.375, 0.875) {};
		\node [style=none] (40) at (8.375, -0.875) {};
		\node [style=none] (41) at (10.375, -0.875) {};
		\node [style=none] (42) at (10.375, 0.875) {};
		\node [style=none] (43) at (9.375, 0) {\(U_{\ell,r}\)};
		\node [style=none] (44) at (7.875, 1.125) {};
		\node [style=none] (45) at (8.375, 0.5) {};
		\node [style=none] (46) at (8.375, -0.5) {};
		\node [style=none] (48) at (5.875, 1.125) {};
		\node [style=none] (49) at (5.875, 0.5) {};
		\node [style=none] (50) at (5.875, -0.5) {};
		\node [style=wirelable] (52) at (6.875, 1.375) {\(X_{\ell-1}\)};
		\node [style=wirelable] (53) at (6.875, 0.75) {\(X_{\ell}\)};
		\node [style=wirelable] (54) at (6.875, -0.25) {\(X_{r}\)};
		\node [style=none] (56) at (7.875, 0) {\(\vdotss\)};
		\node [style=none] (58) at (10.375, 0.5) {};
		\node [style=none] (59) at (10.375, -0.5) {};
		\node [style=none] (62) at (12.125, 0.5) {};
		\node [style=none] (63) at (12.125, -0.5) {};
		\node [style=wirelable] (66) at (11.375, 0.75) {\(Y_{\ell}\)};
		\node [style=wirelable] (67) at (11.375, -0.25) {\(Y_{r}\)};
		\node [style=none] (69) at (10.875, 0) {\(\vdotss\)};
		\node [style=groundout] (73) at (7.875, 1.125) {};
	\end{pgfonlayer}
	\begin{pgfonlayer}{edgelayer}
		\draw [style=background] (38.center)
			 to (37.center)
			 to (36.center)
			 to (35.center)
			 to cycle;
		\draw [style=background] (35.center) to (38.center);
		\draw [style=background] (38.center) to (37.center);
		\draw [style=background] (37.center) to (36.center);
		\draw [style=background] (36.center) to (35.center);
		\draw [style=bbox] (40.center)
			 to (39.center)
			 to (42.center)
			 to (41.center)
			 to cycle;
		\draw (48.center) to (44.center);
		\draw (45.center) to (49.center);
		\draw (46.center) to (50.center);
		\draw (58.center) to (62.center);
		\draw (59.center) to (63.center);
	\end{pgfonlayer}
\end{tikzpicture}

\quad\text{and}\quad
\begin{tikzpicture}
	\begin{pgfonlayer}{nodelayer}
		\node [style=none] (75) at (14.375, 1.875) {};
		\node [style=none] (76) at (14.375, -1.875) {};
		\node [style=none] (77) at (21.875, -1.875) {};
		\node [style=none] (78) at (21.875, 1.875) {};
		\node [style=none] (79) at (16.625, 0.875) {};
		\node [style=none] (80) at (16.625, -1.375) {};
		\node [style=none] (81) at (19.375, -1.375) {};
		\node [style=none] (82) at (19.375, 0.875) {};
		\node [style=none] (83) at (18, -0.25) {\(U_{\ell,r+1}\)};
		\node [style=none] (85) at (16.625, 0.5) {};
		\node [style=none] (86) at (16.625, -0.5) {};
		\node [style=none] (87) at (16.625, -1.125) {};
		\node [style=none] (89) at (14.375, 0.5) {};
		\node [style=none] (90) at (14.375, -0.5) {};
		\node [style=none] (91) at (14.375, -1.125) {};
		\node [style=wirelable] (93) at (15.375, 0.75) {\(X_{\ell}\)};
		\node [style=wirelable] (94) at (15.375, -0.25) {\(X_{r}\)};
		\node [style=wirelable] (95) at (15.375, -0.875) {\(X_{r+1}\)};
		\node [style=none] (96) at (16.125, 0) {\(\vdotss\)};
		\node [style=none] (98) at (19.375, 0.5) {};
		\node [style=none] (99) at (19.375, -0.5) {};
		\node [style=none] (100) at (19.375, -1.125) {};
		\node [style=none] (102) at (21.875, 0.5) {};
		\node [style=none] (103) at (21.875, -0.5) {};
		\node [style=none] (104) at (21.375, -1.125) {};
		\node [style=wirelable] (106) at (20.625, 0.75) {\(Y_{\ell}\)};
		\node [style=wirelable] (107) at (20.625, -0.25) {\(Y_{r}\)};
		\node [style=wirelable] (108) at (20.625, -0.875) {\(Y_{r+1}\)};
		\node [style=none] (109) at (19.875, 0) {\(\vdotss\)};
		\node [style=groundout] (140) at (21.375, -1.125) {};
	\end{pgfonlayer}
	\begin{pgfonlayer}{edgelayer}
		\draw [style=background] (78.center)
			 to (77.center)
			 to (76.center)
			 to (75.center)
			 to cycle;
		\draw [style=background] (75.center) to (78.center);
		\draw [style=background] (78.center) to (77.center);
		\draw [style=background] (77.center) to (76.center);
		\draw [style=background] (76.center) to (75.center);
		\draw [style=bbox] (80.center)
			 to (79.center)
			 to (82.center)
			 to (81.center)
			 to cycle;
		\draw (85.center) to (89.center);
		\draw (86.center) to (90.center);
		\draw (91.center) to (87.center);
		\draw (98.center) to (102.center);
		\draw (99.center) to (103.center);
		\draw (104.center) to (100.center);
	\end{pgfonlayer}
\end{tikzpicture}

\subseteq
\begin{tikzpicture}
	\begin{pgfonlayer}{nodelayer}
		\node [style=none] (110) at (22.875, 1.875) {};
		\node [style=none] (111) at (22.875, -1.875) {};
		\node [style=none] (112) at (28.875, -1.875) {};
		\node [style=none] (113) at (28.875, 1.875) {};
		\node [style=none] (114) at (25.125, 0.875) {};
		\node [style=none] (115) at (25.125, -0.875) {};
		\node [style=none] (116) at (27.125, -0.875) {};
		\node [style=none] (117) at (27.125, 0.875) {};
		\node [style=none] (118) at (26.125, 0) {\(U_{\ell,r}\)};
		\node [style=none] (120) at (25.125, 0.5) {};
		\node [style=none] (121) at (25.125, -0.5) {};
		\node [style=none] (122) at (24.875, -1.125) {};
		\node [style=none] (124) at (22.875, 0.5) {};
		\node [style=none] (125) at (22.875, -0.5) {};
		\node [style=none] (126) at (22.875, -1.125) {};
		\node [style=wirelable] (128) at (23.875, 0.75) {\(X_{\ell}\)};
		\node [style=wirelable] (129) at (23.875, -0.25) {\(X_{r}\)};
		\node [style=wirelable] (130) at (23.875, -0.875) {\(X_{r+1}\)};
		\node [style=none] (131) at (24.625, 0) {\(\vdotss\)};
		\node [style=none] (132) at (27.125, 0.5) {};
		\node [style=none] (133) at (27.125, -0.5) {};
		\node [style=none] (134) at (28.875, 0.5) {};
		\node [style=none] (135) at (28.875, -0.5) {};
		\node [style=wirelable] (136) at (28.125, 0.75) {\(Y_{\ell}\)};
		\node [style=wirelable] (137) at (28.125, -0.25) {\(Y_{r}\)};
		\node [style=none] (138) at (27.625, 0) {\(\vdotss\)};
		\node [style=groundout] (143) at (24.625, -1.125) {};
	\end{pgfonlayer}
	\begin{pgfonlayer}{edgelayer}
		\draw [style=background] (113.center)
			 to (112.center)
			 to (111.center)
			 to (110.center)
			 to cycle;
		\draw [style=background] (110.center) to (113.center);
		\draw [style=background] (113.center) to (112.center);
		\draw [style=background] (112.center) to (111.center);
		\draw [style=background] (111.center) to (110.center);
		\draw [style=bbox] (115.center)
			 to (114.center)
			 to (117.center)
			 to (116.center)
			 to cycle;
		\draw (120.center) to (124.center);
		\draw (121.center) to (125.center);
		\draw (126.center) to (122.center);
		\draw (132.center) to (134.center);
		\draw (133.center) to (135.center);
	\end{pgfonlayer}
\end{tikzpicture}
.
\]
\end{definition}

This notion of monotone sequences was originally defined for natural number indexed sequences of completely positive maps by Carette et al.~\cite[\S~D]{delayedtracememory}. However, we use the more general version given by Comfort and de Felice~\cite[Def.~4.1]{obs}.

A monotone sequence over the interval \(\ell \le r\) records the observations
available between timesteps \(\ell\) and \(r\). Monotonicity says that observing
for longer can only refine the previous observations.
As discussed previously, comparing monotone sequences by equality is too rigid to properly capture their behaviour.
To remedy this problem, we introduce an order on sequences capturing when the observations of one sequence of truncations will eventually be shown to approximate those of the other:

\begin{definition}{\cite[Def.~4.2]{obs}}
  \label{def:approximation}
Given monotone sequences \(\vb U,\vb V:\vb X\to\vb Y\), say that \(\vb V\)
\textbf{approximates} \(\vb U\), written \(\vb U\subseteq\vb V\), if for every
\(\ell\leq r\in\Z\), there exists some \textbf{lookahead} \(R\in \N\) and 
\textbf{lookbehind} \(L\in \N\) such that
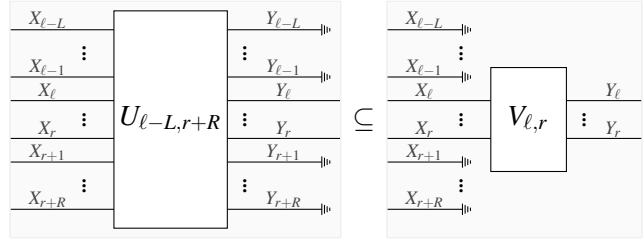
\[
\begin{tikzpicture}
	\begin{pgfonlayer}{nodelayer}
		\node [style=none] (85) at (-3.25, 3.125) {};
		\node [style=none] (86) at (-3.25, -3.125) {};
		\node [style=none] (87) at (5.5, -3.125) {};
		\node [style=none] (88) at (5.5, 3.125) {};
		\node [style=none] (89) at (-0.5, 2.875) {};
		\node [style=none] (90) at (-0.5, -2.875) {};
		\node [style=none] (91) at (2.5, -2.875) {};
		\node [style=none] (92) at (2.5, 2.875) {};
		\node [style=none] (93) at (1, 0) {\(U_{\ell-L,r+R}\)};
		\node [style=none] (94) at (-0.5, 1.125) {};
		\node [style=none] (95) at (-0.5, 0.5) {};
		\node [style=none] (96) at (-0.5, -0.5) {};
		\node [style=none] (97) at (-0.5, -1.125) {};
		\node [style=none] (98) at (-3.25, 1.125) {};
		\node [style=none] (99) at (-3.25, 0.5) {};
		\node [style=none] (100) at (-3.25, -0.5) {};
		\node [style=none] (101) at (-3.25, -1.125) {};
		\node [style=wirelable] (102) at (-2.25, 1.375) {\(X_{\ell-1}\)};
		\node [style=wirelable] (103) at (-2.25, 0.75) {\(X_{\ell}\)};
		\node [style=wirelable] (104) at (-2.25, -0.25) {\(X_{r}\)};
		\node [style=wirelable] (105) at (-2.25, -0.875) {\(X_{r+1}\)};
		\node [style=none] (106) at (-1.25, 0) {\(\vdotss\)};
		\node [style=none] (107) at (2.5, 0.5) {};
		\node [style=none] (108) at (2.5, -0.5) {};
		\node [style=none] (109) at (5.5, 0.5) {};
		\node [style=none] (110) at (5.5, -0.5) {};
		\node [style=wirelable] (111) at (4, 0.75) {\(Y_{\ell}\)};
		\node [style=wirelable] (112) at (4, -0.25) {\(Y_{r}\)};
		\node [style=none] (113) at (3, 0) {\(\vdotss\)};
		\node [style=groundout] (114) at (5, 1.125) {};
		\node [style=groundout] (115) at (5, -1.125) {};
		\node [style=none] (116) at (-0.5, 2.375) {};
		\node [style=none] (117) at (-3.25, 2.375) {};
		\node [style=wirelable] (118) at (-2.25, 2.625) {\(X_{\ell-L}\)};
		\node [style=groundout] (119) at (5, 2.375) {};
		\node [style=none] (120) at (-1.25, 1.75) {\(\vdotss\)};
		\node [style=none] (121) at (-0.5, -2.375) {};
		\node [style=none] (122) at (-3.25, -2.375) {};
		\node [style=wirelable] (123) at (-2.25, -2.125) {\(X_{r+R}\)};
		\node [style=groundout] (124) at (5, -2.375) {};
		\node [style=none] (125) at (-1.25, -1.75) {\(\vdotss\)};
		\node [style=none] (126) at (2.5, 1.125) {};
		\node [style=none] (127) at (2.5, -1.125) {};
		\node [style=none] (128) at (2.5, 2.375) {};
		\node [style=none] (129) at (2.5, -2.375) {};
		\node [style=none] (130) at (3, -1.75) {\(\vdotss\)};
		\node [style=none] (131) at (3, 1.75) {\(\vdotss\)};
		\node [style=wirelable] (132) at (4, 1.375) {\(Y_{\ell-1}\)};
		\node [style=wirelable] (133) at (4, 2.625) {\(Y_{\ell-L}\)};
		\node [style=wirelable] (134) at (4, -0.875) {\(Y_{r+1}\)};
		\node [style=wirelable] (135) at (4, -2.125) {\(Y_{r+R}\)};
	\end{pgfonlayer}
	\begin{pgfonlayer}{edgelayer}
		\draw [style=background] (88.center)
			 to (87.center)
			 to (86.center)
			 to (85.center)
			 to cycle;
		\draw [style=background] (85.center) to (88.center);
		\draw [style=background] (88.center) to (87.center);
		\draw [style=background] (87.center) to (86.center);
		\draw [style=background] (86.center) to (85.center);
		\draw [style=bbox] (90.center)
			 to (89.center)
			 to (92.center)
			 to (91.center)
			 to cycle;
		\draw (98.center) to (94.center);
		\draw (95.center) to (99.center);
		\draw (96.center) to (100.center);
		\draw (101.center) to (97.center);
		\draw (107.center) to (109.center);
		\draw (108.center) to (110.center);
		\draw (117.center) to (116.center);
		\draw (122.center) to (121.center);
		\draw (128.center) to (119);
		\draw (126.center) to (114);
		\draw (127.center) to (115);
		\draw (129.center) to (124);
	\end{pgfonlayer}
\end{tikzpicture}

\subseteq
\begin{tikzpicture}
	\begin{pgfonlayer}{nodelayer}
		\node [style=none] (35) at (6.5, 3.125) {};
		\node [style=none] (36) at (6.5, -3.125) {};
		\node [style=none] (37) at (13.25, -3.125) {};
		\node [style=none] (38) at (13.25, 3.125) {};
		\node [style=none] (39) at (9.25, 1.375) {};
		\node [style=none] (40) at (9.25, -1.375) {};
		\node [style=none] (41) at (11.25, -1.375) {};
		\node [style=none] (42) at (11.25, 1.375) {};
		\node [style=none] (43) at (10.25, 0) {\(V_{\ell,r}\)};
		\node [style=none] (44) at (8.5, 1.125) {};
		\node [style=none] (45) at (9.25, 0.5) {};
		\node [style=none] (46) at (9.25, -0.5) {};
		\node [style=none] (47) at (8.5, -1.125) {};
		\node [style=none] (48) at (6.5, 1.125) {};
		\node [style=none] (49) at (6.5, 0.5) {};
		\node [style=none] (50) at (6.5, -0.5) {};
		\node [style=none] (51) at (6.5, -1.125) {};
		\node [style=wirelable] (52) at (7.5, 1.375) {\(X_{\ell-1}\)};
		\node [style=wirelable] (53) at (7.5, 0.75) {\(X_{\ell}\)};
		\node [style=wirelable] (54) at (7.5, -0.25) {\(X_{r}\)};
		\node [style=wirelable] (55) at (7.5, -0.875) {\(X_{r+1}\)};
		\node [style=none] (56) at (8.5, 0) {\(\vdotss\)};
		\node [style=none] (58) at (11.25, 0.5) {};
		\node [style=none] (59) at (11.25, -0.5) {};
		\node [style=none] (62) at (13.25, 0.5) {};
		\node [style=none] (63) at (13.25, -0.5) {};
		\node [style=wirelable] (66) at (12.5, 0.75) {\(Y_{\ell}\)};
		\node [style=wirelable] (67) at (12.5, -0.25) {\(Y_{r}\)};
		\node [style=none] (69) at (11.75, 0) {\(\vdotss\)};
		\node [style=groundout] (73) at (8.5, 1.125) {};
		\node [style=groundout] (74) at (8.5, -1.125) {};
		\node [style=none] (75) at (8.5, 2.375) {};
		\node [style=none] (76) at (6.5, 2.375) {};
		\node [style=wirelable] (77) at (7.5, 2.625) {\(X_{\ell-L}\)};
		\node [style=groundout] (78) at (8.5, 2.375) {};
		\node [style=none] (79) at (8.5, 1.75) {\(\vdotss\)};
		\node [style=none] (80) at (8.5, -2.375) {};
		\node [style=none] (81) at (6.5, -2.375) {};
		\node [style=wirelable] (82) at (7.5, -2.125) {\(X_{r+R}\)};
		\node [style=groundout] (83) at (8.5, -2.375) {};
		\node [style=none] (84) at (8.5, -1.75) {\(\vdotss\)};
	\end{pgfonlayer}
	\begin{pgfonlayer}{edgelayer}
		\draw [style=background] (38.center)
			 to (37.center)
			 to (36.center)
			 to (35.center)
			 to cycle;
		\draw [style=background] (35.center) to (38.center);
		\draw [style=background] (38.center) to (37.center);
		\draw [style=background] (37.center) to (36.center);
		\draw [style=background] (36.center) to (35.center);
		\draw [style=bbox] (40.center)
			 to (39.center)
			 to (42.center)
			 to (41.center)
			 to cycle;
		\draw (48.center) to (44.center);
		\draw (45.center) to (49.center);
		\draw (46.center) to (50.center);
		\draw (51.center) to (47.center);
		\draw (58.center) to (62.center);
		\draw (59.center) to (63.center);
		\draw (76.center) to (75.center);
		\draw (81.center) to (80.center);
	\end{pgfonlayer}
\end{tikzpicture}

\]
\end{definition}

Intuitively, \(\vb U\subseteq\vb V\) in case for every observation made in \(\vb U\) from time \(\ell\) to \(r\) can be verified by observing \(\vb V\), possibly over a larger interval of time.

This allows us to define an equivalence relation which equates monotone sequences with the same behaviour:
\begin{definition}{\cite[Def.~4.6]{obs}}
Two monotone sequences \(\vb U,\vb V:\vb X\to\vb Y\) are
\textbf{observationally equivalent} in case they approximate each other so that 
\(\vb U\subseteq\vb V\) and \(\vb V\subseteq\vb U\).
\end{definition}

Observational equivalence classes of monotone sequences over a discard bicategory themselves form a discard bicategory:

\begin{definition}\label{def:obs}{\cite[Def.~4.8]{obs}}
Given a discard bicategory \(\mathcal C\), the discard bicategory \textbf{observational behaviours},
\(\Obs(\mathcal C)\), has:
\begin{itemize}
    \item objects: \(\Z\)-indexed sets of objects in \(\mathcal C\): \(\vb X=\{X_k\}_{k\in\Z}\);
    \item morphisms: observational equivalence classes of monotone sequences \([\vb U]:\vb X\to\vb Y\);
    \item poset enrichment: given by approximation so that \([\vb V]\subseteq[\vb U]\) if and only if \(\vb V\subseteq\vb U\);
    \item discard bicategory structure: defined pointwise.
\end{itemize}
\end{definition}

\begin{theorem}
The data in \Cref{def:obs} makes \(\Obs(\mathcal C)\) a
discard bicategory.
\end{theorem}

It follows from essentially the same argument as that of Comfort and de Felice~\cite[Lem.~6.7]{obs}, that this yields functorial semantics for stateful ZX-diagrams:
\begin{proposition}\label{prop:unroll_obs_functor}
The unrollings of stateful ZX-diagrams induce a \dag-compact closed functor
 \[\Unroll:\St(\StabZX)\to \Obs(\StabZX^\disc).\]
\end{proposition}

The quotient by observational equivalence is essential: the two sides of the sliding equation \eqref{eq:fdzx_sliding} unroll to distinct but observationally equivalent sequences. 

The construction of observational behaviours itself is functorial \cite[Prop.~4.10]{obs}; therefore, we can interpret behaviours of \(\ZX^\disc\) diagrams in terms of behaviours of completely positive maps:

\begin{proposition}\label{cor:obs_faithful}
There is a faithful \dag-compact closed discard bicategory functor
\[\Obs(\StabZX^\disc)\to\Obs(\Proj(\CPM))\]
given pointwise by \(\StabZX^\disc \cong \Proj(\Stab^\disc)\to\Proj(\CPM)\).
\end{proposition}

Observational equivalence classes of quantum processes may still induce \emph{infinite-dimensional} ones under additional analytic hypotheses. However, arbitrary completely positive maps are too broad for this purpose: one would expect to impose that the maps are trace-preserving or trace-nonincreasing in the Schr\"odinger picture; equivalently unital or subunital in the Heisenberg picture. Since our construction does not rely on infinite dimensional quantum channels, we interpret processes through Smolin's ``Church of the Larger Hilbert Space,'' as finite processes which can be continued into larger finite-dimensional Hilbert spaces indefinitely.

\subsection{Infinite stabilizer groups}
\label{sec:infinite}
Independent of the convergence of behaviours to infinite quantum channels, we can still represent them \emph{concretely} in terms of their infinite stabilizer groups:
\begin{definition}
  Let \(\AffRel\) denote the \dag-compact closed category whose objects are (possibly infinite-dimensional) vector spaces over \(\F_p\), whose morphisms \(R:X\to Y\) are affine subspaces \(R\subseteq X\times Y\). The discard and poset enrichment are given by the discard relation and subspace inclusion, as before.
\end{definition}

It follows immediately from \cite[Cor.~5.6]{obs} that observational equivalence classes of \emph{finite} stabilizer groups can be coherently glued together into a unique \emph{infinite} stabilizer group consisting of the infinite Pauli operators which stabilize each finite truncation:

\begin{theorem}\label{thm:lim_faithful}
  There is a faithful \dag-compact closed discard functor
  \(
  \Lim:\Obs(\ALR^\disc)\rightarrowtail \AffRel
  \)
  sending
  \[
    [\vb R]:\vb X\to\vb Y
    \quad \longmapsto \quad
    \left\{ (\vb x, \vb y) \ \middle| \ \forall \ell\leq r\in\Z:\; \bigl((x_\ell, \ldots, x_r),(y_\ell, \ldots, y_r)\bigr) \in R_{\ell,r}\right\} \subseteq \bigl(\, \prod_{j\in \Z} X_j\, \bigr) \times \bigl(\, \prod_{j\in \Z} Y_j\, \bigr).
  \]
\end{theorem}

The quotient by observational equivalence is crucial for the faithfulness of \(\Lim\). Here, faithfulness means that monotone sequences have the same infinite stabilizer group if and only if they are observationally equivalent.

By precomposing this limit functor with the unrolling, each stateful \(\ZX\)-diagram determines a unique infinite stabilizer group:

\begin{corollary}\label{cor:unroll}
  There is a \dag-compact closed functor 
  \[\StabRel:\St(\StabZX)\xrightarrow{\Unroll} \Obs(\StabZX^\disc)\simeq \Obs(\ALR^\disc) \overset{\Lim}{\rightarrowtail} \AffRel\]
\end{corollary}

For example, consider the infinite stabilizer group induced by the delay generator:

\begin{example}\label{ex:beh_delay}
The delay is sent to the relation which reindexes the Paulis:
\[
\StabRel\bigl(\begin{tikzpicture}
	\begin{pgfonlayer}{nodelayer}
		\node [style=rmat] (1) at (0.75, 0) {\(\delta\)};
		\node [style=none] (2) at (0, 0.4) {};
		\node [style=none] (3) at (1.5, 0.4) {};
		\node [style=none] (4) at (1.5, -0.4) {};
		\node [style=none] (5) at (0, -0.4) {};
		\node [style=none] (6) at (0, 0) {};
		\node [style=none] (7) at (1.5, 0) {};
	\end{pgfonlayer}
	\begin{pgfonlayer}{edgelayer}
		\draw [style=background] (2.center)
			 to (3.center)
			 to (4.center)
			 to (5.center)
			 to cycle;
		\draw [style=background] (5.center) to (2.center);
		\draw [style=background] (2.center) to (3.center);
		\draw [style=background] (3.center) to (4.center);
		\draw [style=background] (4.center) to (5.center);
		\draw (7.center) to (6.center);
	\end{pgfonlayer}
\end{tikzpicture}
\bigr)
=
\left\{
\left(
\begin{smallbmatrix}
\vb x\\
\vb z
\end{smallbmatrix},
\begin{smallbmatrix}
\vb x'\\
\vb z'
\end{smallbmatrix}
\right)
\;\middle|\;
\forall t\in\Z,\;
x'_t=x_{t-1},\;
z'_t=z_{t-1}
\right\}  \subseteq
  (\F_p^\Z)^2 \times (\F_p^\Z)^2
.\]
\end{example}

  \section{Generating tableaux for time-delayed circuits}\label{sec:salr}
  \pratendSetLocal{category=salr, end, restate}
  
Stabilizer ZX-diagrams can be represented by affine Lagrangian relations, which themselves are finitely presented by tableaux.
Therefore, it is a natural question to ask if the infinite stabilizer groups of stateful ZX-diagrams also possess such a finite representation. In this section, we show how stateful ZX-diagrams can be represented by a generalised stabilizer tableau which we call a \emph{generating tableau}.

To this end, we represent an infinite sequence of Pauli operators \((\vb p_j \in \F_p^2)_{j\in\N}\) as a formal Laurent series, called its \textbf{generating function}:
\[
  P(\delta) \coloneqq \sum_{j\ge 0}^\infty \vb p_j \delta^j \in \F_{p}((\delta))^{2n}.
\]
Similarly, we represent the delay by the relation which multiplies by the indeterminate \(\delta\):
\begin{equation}\label{eq:delay_interpretation}
    \interp {\begin{tikzpicture}
	\begin{pgfonlayer}{nodelayer}
		\node [style=rmat] (1) at (0.75, 0) {\(\delta\)};
		\node [style=none] (2) at (0, 0.4) {};
		\node [style=none] (3) at (1.5, 0.4) {};
		\node [style=none] (4) at (1.5, -0.4) {};
		\node [style=none] (5) at (0, -0.4) {};
		\node [style=none] (6) at (0, 0) {};
		\node [style=none] (7) at (1.5, 0) {};
	\end{pgfonlayer}
	\begin{pgfonlayer}{edgelayer}
		\draw [style=background] (2.center)
			 to (3.center)
			 to (4.center)
			 to (5.center)
			 to cycle;
		\draw [style=background] (5.center) to (2.center);
		\draw [style=background] (2.center) to (3.center);
		\draw [style=background] (3.center) to (4.center);
		\draw [style=background] (4.center) to (5.center);
		\draw (7.center) to (6.center);
	\end{pgfonlayer}
\end{tikzpicture}
}
  \coloneqq
    \bigl\{
      \bigl(
        \begin{smallbmatrix}
          \vb x\\
          \vb z
        \end{smallbmatrix}
        ,
        \delta\cdot 
        \begin{smallbmatrix}
          \vb x\\
          \vb z
        \end{smallbmatrix}
      \bigr)
    \; \big|\;
     \vb x,\vb z \in \F_p((\delta))
    \bigr\}
  \subseteq
    \F_p((\delta))^2 \times \F_p((\delta))^2.
\end{equation}

Formal Laurent series still take infinite data to encode. However, some arise as the expansion of a \emph{rational function}, i.e.\ as fractions \(\tfrac{f(\delta)}{g(\delta)}\) of polynomials \(f,g\in\F_p[\delta]\). Rational functions in turn form a \emph{field}, denoted \(\F_p(\delta)\). A rational function is thus a finite description of an infinite sequence, whose generating function is recovered from it by the \textbf{geometric series expansion}, computed by long division of polynomials.
For instance, consider the geometric series of the following fraction of polynomials
\[\tfrac{\delta}{(1-\delta)^2} = \sum_{t\ge 0}^\infty t \delta^t.\]

Interpreting the delay as multiplication by \(\delta\) as in \Cref{eq:delay_interpretation} and interpreting the undelayed \(\ZX\) generators as before; in this section, we find that a stateful ZX-diagram can be described by an affine subspace of \(\F_p(\delta)^{2n}\). 
We identify which affine subspaces over \(\F_p(\delta)\) arise this way, the \emph{shifted affine Lagrangian} subspaces, and show that each admits a finite \emph{generating tableau} generalising the tableau of a stabilizer group.

\subsection{Generating tableaux and the shifted symplectic form}
Interpreting stateful ZX-diagrams as \(\F_p(\delta)\)-affine relations according to \Cref{eq:delay_interpretation} produces a special class of affine relations, which we call \emph{shifted affine Lagrangian relations}. These are analogous to affine Lagrangian relations over \(\F_p\), except where the symplectic form is twisted by the involution\footnote{Not to be confused with the conjugation of complex numbers.} \(
\bar{(-)}:\F_p(\delta)\to\F_p(\delta)\) given by \(\bar f(\delta)\coloneqq f(\delta^{\minu 1})\). Following Wilde and Brun~\cite{wildeentanglement} and Wilde et al.~\cite{wildeconvolutional}:

\begin{definition}
  Given some \(n\in \N\), the standard \textbf{shifted symplectic form} \(\varpi_n\colon \F_p(\delta)^{2n} \times \F_p(\delta)^{2n} \to \F_p(\delta)\) is the sesquilinear form given by extending scalars along \(\F_p(\delta)/\F_p\), and then twisting the symplectic form with the conjugation \(\bar{(-)}\):
  \begin{equation}
    \varpi_n(f(\delta), g(\delta)) \coloneqq 
     \omega_n(f(\delta), g(\delta^{\minu 1})) 
  \end{equation}
  More generally, given any \(\F_p\)-symplectic vector space, one can obtain a shifted symplectic vector space by extending scalars along \(\F_p(\delta)/\F_p\) and twisting the symplectic form.
\end{definition}

Just as the symplectic form captures the commutation of Paulis, the shifted symplectic form captures the translation-invariant commutation of sequences of Paulis:

\begin{remark}
  To see this, remark that two finitely-supported sequences of Pauli operators encoded by \(f\) and \(g\) commute precisely when \(\sum_t\omega_n(f_t,g_t)=0\).
  Now, the \(k\)th coefficient in \(\varpi_n(f,g)\) is given by \(\sum_t\omega_n(f_t,g_{t-k})\). This is zero precisely when the sequence of Pauli operators encoded by \(f\) commutes with the shifted sequence of Pauli operators encoded by \(\delta^k g\).
  Therefore, \(\varpi_n(f,g)=0\) precisely when the Pauli operators encoded by \(f\) commute with all possible shifted sequences of Pauli operators in \(g\).
\end{remark}

Recalling that stabilizer states can be represented by affine coisotropic subspaces with respect to the symplectic form, we represent translation-invariant stabilizer states in the same way, except now with respect to the \emph{shifted} symplectic form:

\begin{definition}
  Given a \(\F_p(\delta)\)-linear subspace \(S+\vb{a}(\delta)\subseteq (V, \varpi )\) of a shifted symplectic vector space, the \textbf{shifted symplectic complement} is the linear subspace:
  \[S^{\varpi} = \{ \vb{f}(\delta) \in V \ | \ \forall \vb{g}(\delta) \in S,\, \varpi(\vb{f}(\delta), \vb{g}(\delta)) = 0\} \subseteq (V, \varpi ).\]

  A \textbf{shifted affine Lagrangian subspace} \(S+\vb{a}(\delta)\subseteq (V, \varpi )\) is an affine subspace with \(S^{\varpi} = S\).
\end{definition}

It follows immediately that this generalises the unshifted notion:
\begin{lemma}
  Given an \(\F_p\)-affine Lagrangian subspace of \((\F_p^{2n},\omega_n)\), its field extension along \(\F_p(\delta)/\F_p\) is a shifted affine Lagrangian subspace of \((\F_p(\delta)^{2n}, \varpi_{n})\).
\end{lemma}
Similarly, the delay given in \Cref{eq:delay_interpretation} is a shifted Lagrangian subspace:
\begin{lemmaE}
  The following subspace is shifted Lagrangian:
  \[
    \left\{
      \Bigl(
        \begin{smallbmatrix}
          \vb x\\
          \vb z
        \end{smallbmatrix}
        ,
        \delta\cdot 
        \begin{smallbmatrix}
          \vb x\\
          \vb z
        \end{smallbmatrix}
      \Bigr)
      \; \middle|\;
      \vb x, \vb z\in \F_p(\delta)
    \right\}
    \subseteq
    (\F_p(\delta)^{2n}, -\varpi_n)
    \times
    (\F_p(\delta)^{2m}, \varpi_m).
  \]
\end{lemmaE}
\begin{proofE}
  Writing \(L\) for the subspace above, any two elements \((\vb v,\delta\vb v),(\vb w,\delta\vb w)\in L\) commute since scaling by \(\delta\) preserves the shifted form, as \(\delta\bar\delta=\delta\,\delta^{\minu 1}=1\):
  \[
    -\varpi_n(\vb v,\vb w)+\varpi_m(\delta\vb v,\delta\vb w)
    =-\varpi_n(\vb v,\vb w)+\delta\bar\delta\,\varpi_n(\vb v,\vb w)
    =0.
  \]
  Hence \(L\) is isotropic, and as \(\vb v\mapsto(\vb v,\delta\vb v)\) is injective, \(\dim L=2\) is half the ambient dimension, so \(L\) is Lagrangian.
\end{proofE}

Defining the dagger
\footnote{The dagger here is not to be confused with the Hermitian adjoint of a bounded linear map between Hilbert spaces.} 
to be conjugate transpose
\((-)^\dag \coloneqq \bar{(-)}^\top\), we obtain a generalisation of the notion of a stabilizer tableau:
\begin{definition}
A \textbf{generating tableau} for a shifted affine Lagrangian subspace \(S \subseteq (\F_p(\delta)^{2n}, \varpi_n)\) consists of a full-rank matrix
\(
    H
  =
  \left[\begin{array}{c|c}
    X & Z
  \end{array}\right]
  \in \Matrices[n][2n][\F_p(\delta)]
\)
and a vector
\(\vb a\in \F_p(\delta)^{2n}\)
such that
\(\vb a+\ker(H) = S\), 
where moreover
\(X Z^\dag=Z X^\dag\), 
so that all rows 
\(\vb h_i,\vb h_j\) 
of \(H\) commute:   
\(\varpi_n(\vb h_i,\vb h_j)=0\). 

In case \(X=0\), say that \(H\) is the tableau associated to a \textbf{shifted graph-state}.
\end{definition}

Denoting the set of \(n\times n\) Hermitian matrices by
\(
    \Herm[n]
  \coloneqq
    \{
        \Sigma \in \Matrices[n][n][\F_p(\delta)]
      \ \mid\
        \Sigma^\dag=\Sigma
    \}
\),
we find that generating tableaux admit a unique normal form, generalizing that of Gottesman and Cleve~\cite{Cleve1997}:
\begin{propositionE}\label{prop:generating_tableau}
  Given shifted affine Lagrangian \(S\subseteq (\F_p(\delta)^{2n}, \varpi_n)\), there exists some unique
  \(0\leq m\leq n \in \N\),
  \(L\in \Matrices[m][(n-m)]\),
  \(\Sigma\in \Herm[m]\),
  \(\vb a\in \F_p(\delta)^{2n}\),
  and permutation \(\varsigma \in \Matrices[n][n]\) such that:
  \[
      S
    =
      \ker\left(
          \left[\begin{array}{cc|cc}
            I_m &  L & \Sigma    & 0\\
            0     & 0 & -L^\dag & I_{n-m}
          \end{array}\right]
        \left[\begin{array}{c|c}
          \varsigma & 0\\ \hline
          0         & \varsigma
        \end{array}\right]
      \right)
      +
      \vb a.
  \]
  Call this \textbf{canonical form} for the generating tableau associated to \(R\).
\end{propositionE}
\begin{proofE}
  Write \(R=S+\vb a\), where \(S\subseteq \F_p(\delta)^{2n}\) is the linear part of \(R\).
  Because \(R\) is shifted affine Lagrangian, \(S\) is shifted Lagrangian.
  Let \(\widetilde{S}\coloneqq \{\bar s\mid s\in S\}\subseteq \F_p(\delta)^{2n}\).

  Since \(\varpi_n(x,s)=\omega_n(x,\bar s)\), we have
  \(S^{\varpi_n}=\widetilde{S}^{\omega_n}\).
  The involution \(s\mapsto \bar s\) is bijective, therefore
  \(\dim_{\F_p(\delta)}(\widetilde{S})=\dim_{\F_p(\delta)}(S)\). Because \(\omega_n\) is non-degenerate as an \(\F_p(\delta)\)-bilinear form,
  \(\dim_{\F_p(\delta)}(\widetilde{S}^{\omega_n})=2n-\dim_{\F_p(\delta)}(\widetilde{S})\).
  Because \(S\) is shifted Lagrangian, \(S=S^{\varpi_n}=\widetilde{S}^{\omega_n}\), hence
  \(\dim_{\F_p(\delta)}(S)=2n-\dim_{\F_p(\delta)}(S)\), therefore \(\dim_{\F_p(\delta)}(S)=n\).

  Therefore there is a full-rank matrix
  \(H=\left[\begin{array}{c|c} X & Z \end{array}\right]\in \Matrices[n][2n][\F_p(\delta)]\)
  with \(S=\ker(H)\). Let \(m\coloneqq \rank(X)\). Applying elementary row-operations, we can reduce this matrix to the following form, while preserving the kernel, reordering the columns with a shifted symplectic permutation matrix \(\varsigma\):
  \[
    H\sim
    \left[\begin{array}{cc|cc}
      I_m & A & B & 0\\
      0   & 0 & C & I_{n-m}
    \end{array}\right]
    \left[\begin{array}{c|c}
      \varsigma & 0\\ \hline
      0                 & \varsigma
    \end{array}\right].
  \]
  The shifted Lagrangian condition on \(S\) forces \(C=-A^\dag\) and \(B^\dag= B\).
  Setting \(L\coloneqq A\) and \(\Sigma\coloneqq B\):
  \[
      S
    =
      \ker\left(
        \left[\begin{array}{cc|cc}
          I_m &  L & \Sigma    & 0\\
          0     & 0 & -L^\dag & I_{n-m}
        \end{array}\right]
        \circ
        \left[\begin{array}{c|c}
          \varsigma & 0\\ \hline
          0                 & \varsigma
        \end{array}\right]
      \right),
  \]
  To prove uniqueness, suppose two canonical tableaux define the same linear subspace \(S\), with data
  \((m,L,\Sigma,\varsigma)\) and \((m',L',\Sigma',\varsigma')\), and let the corresponding full-rank matrices be
  \(H=[X\mid Z]\) and \(H'=[X'\mid Z']\).
  Since \(\ker(H)=\ker(H')=S\) and both have rank \(n\), there is some \(U\in \mathrm{GL}_n(\F_p(\delta))\) such that \(H'=UH\).
  Hence \(X'=UX\), so that 
  \(m'=\rank(X')=\rank(UX)=\rank(X)=m\),  thus \(m\) is unique.

  Since \(H'=UH\) with \(U\) invertible, \(H\) and \(H'\) have the same row space, and hence the same reduced row-echelon form, whose pivot columns are therefore determined by \(S\). The \(m=\rank(X)\) pivots among the \(X\)-columns determine \(\varsigma\), and the reduced form then determines \(L\) and \(\Sigma\). Finally, \(\vb a\) is determined modulo \(S\) by the coset \(\vb a+S=R\).
\end{proofE}

The generating tableaux of shifted Lagrangian subspaces were originally used to represent the stabilizer groups of quantum convolutional codes: the infinite-time, translation-invariant analogue to stabilizer codes (see Wilde's thesis for reference \cite{wilde}). There is a breadth of literature on quantum convolutional codes in the qubit setting~\cite{Ollivier2003,jr_simple_2005,wildeentanglement,wildeconvolutional}; however, the compositional structure of quantum convolutional coding theory remains unexplored.  By working in odd-prime dimensions, we show that translation-invariant stabilizer quantum processes can be composed as relations, in the same way as their finite-dimensional analogues:

\begin{definition}\label{def:salr}
The \dag-CC feedback category 
\(\sALR\) of shifted \(\F_p(\delta)\)-affine 
Lagrangian relations between finite-dimensional 
\(\F_p(\delta)\)-shifted symplectic vector spaces has objects given by shifted symplectic vector spaces, whose morphisms \((V, \varpi_V) \to (W, \varpi_W)\) are either shifted affine Lagrangian subspaces of \((V, -\varpi_V )\times (W, \varpi_W )\) or the empty set. The rest of the \dag-CC structure is the same as for \(\ALR\).
\end{definition}

Interpreting the delay as in \Cref{eq:delay_interpretation}, we have that:
\begin{proposition}
  \label{prop:compact_closed_functor}
  There is a \dag-compact closed functor
  \(\Ext:\St(\StabZX)\to \sALR\)
  given by interpreting the delay
  \(\begin{tikzpicture}
	\begin{pgfonlayer}{nodelayer}
		\node [style=rmat] (1) at (0.75, 0) {\(\delta\)};
		\node [style=none] (2) at (0, 0.4) {};
		\node [style=none] (3) at (1.5, 0.4) {};
		\node [style=none] (4) at (1.5, -0.4) {};
		\node [style=none] (5) at (0, -0.4) {};
		\node [style=none] (6) at (0, 0) {};
		\node [style=none] (7) at (1.5, 0) {};
	\end{pgfonlayer}
	\begin{pgfonlayer}{edgelayer}
		\draw [style=background] (2.center)
			 to (3.center)
			 to (4.center)
			 to (5.center)
			 to cycle;
		\draw [style=background] (5.center) to (2.center);
		\draw [style=background] (2.center) to (3.center);
		\draw [style=background] (3.center) to (4.center);
		\draw [style=background] (4.center) to (5.center);
		\draw (7.center) to (6.center);
	\end{pgfonlayer}
\end{tikzpicture}
\)
  as multiplication by \(\delta\); and by interpreting the \(\ZX\) generators in \(\ALR \subseteq \sALR\), given by extending scalars along \(\F_p(\delta)/\F_p\).
\end{proposition}
\begin{proof}
  Functoriality is immediate, as given any \(\F_p(\delta)\)-linear vector space \(S\), by linearity, we have \(\delta \cdot S = S\).
\end{proof}

\begin{example}\label{ex:delayed_graph_state_ext}
  By delaying graph states, we obtain \emph{shifted} graph-states:
  \[
    \Ext\left(\begin{tikzpicture}
	\begin{pgfonlayer}{nodelayer}
		\node [style=none] (0) at (-0.5, 2) {};
		\node [style=none] (1) at (-0.5, -2) {};
		\node [style=none] (2) at (4, 2) {};
		\node [style=none] (3) at (4, -2) {};
		\node [style=bn] (4) at (1, 0.5) {};
		\node [style=bn] (5) at (1, -1.5) {};
		\node [style=bn] (6) at (3, -1.5) {};
		\node [style=none] (7) at (3, 0.5) {};
		\node [style=had] (8) at (2, 0.5) {\(d\)};
		\node [style=had] (9) at (1, -0.5) {\(a\)};
		\node [style=had] (10) at (2, -1.5) {\(b\)};
		\node [style=had] (11) at (3, -0.5) {\(c\)};
		\node [style=none] (12) at (4, 0.5) {};
		\node [style=none] (13) at (4, -1.5) {};
		\node [style=lmat] (14) at (2, 1.5) {\(\delta\)};
		\node [style=none] (16) at (3, 1.5) {};
		\node [style=none] (17) at (0.5, 1.5) {};
		\node [style=none] (18) at (0.5, -1.5) {};
		\node [style=bn] (19) at (3, 0.5) {};
		\node [style=none] (20) at (4, 1) {};
		\node [style=none] (21) at (1.5, 1) {};
	\end{pgfonlayer}
	\begin{pgfonlayer}{edgelayer}
		\draw [style=background] (1.center)
			 to (0.center)
			 to (2.center)
			 to (3.center)
			 to cycle;
		\draw (4) to (7.center);
		\draw (7.center) to (6);
		\draw (6) to (5);
		\draw (5) to (4);
		\draw (7.center)
			 to [in=0, out=0, looseness=1.75] (16.center)
			 to (17.center)
			 to [bend left=270, looseness=0.75] (18.center)
			 to (5);
		\draw (20.center)
			 to (21.center)
			 to [in=75, out=180] (4);
		\draw (19) to (12.center);
		\draw (6) to (13.center);
	\end{pgfonlayer}
\end{tikzpicture}
\right)
    =
    \ker
      \left[\begin{array}{ccc|ccc}
        1 & 0 & 0 & 0              & \delta a+ d & 0 \\
        0 & 1 & 0 & \delta^{\minu 1} a+ d & 0           & \delta^{\minu 1} b+ c\\
        0 & 0 & 1 & 0                & \delta b+ c & 0
      \end{array}\right]
    \subseteq (\F_p^{6}(\delta),\varpi_3).
  \]
\end{example}

\begin{example}\label{ex:delayed_cx_ext}
  The delayed controlled \(X\) gate has the following interpretation:
  \[
    \Ext\left(\begin{tikzpicture}
	\begin{pgfonlayer}{nodelayer}
		\node [style=none] (0) at (13.125, 1.25) {};
		\node [style=none] (1) at (13.125, -1.25) {};
		\node [style=none] (2) at (17.125, -1.25) {};
		\node [style=none] (3) at (17.125, 1.25) {};
		\node [style=none] (41) at (16.5, 0) {};
		\node [style=none] (42) at (13.75, 0) {};
		\node [style=none] (43) at (16.5, 0.75) {};
		\node [style=none] (44) at (13.75, 0.75) {};
		\node [style=bn] (45) at (14.25, 0) {};
		\node [style=wn] (46) at (15.25, -0.75) {};
		\node [style=none] (47) at (13.125, -0.75) {};
		\node [style=none] (48) at (15.25, 0) {};
		\node [style=none] (49) at (16.5, -0.75) {};
		\node [style=none] (50) at (17.125, -0.75) {};
		\node [style=lmat] (51) at (15.25, 0.75) {\(\delta\)};
	\end{pgfonlayer}
	\begin{pgfonlayer}{edgelayer}
		\draw [style=background] (3.center)
			 to (2.center)
			 to (1.center)
			 to (0.center)
			 to cycle;
		\draw [style=background] (0.center) to (3.center);
		\draw [style=background] (3.center) to (2.center);
		\draw [style=background] (2.center) to (1.center);
		\draw [style=background] (1.center) to (0.center);
		\draw (45)
			 to (42.center)
			 to [bend left=90, looseness=0.75] (44.center)
			 to (43.center)
			 to [bend left=90, looseness=0.75] (41.center)
			 to [in=0, out=-180] (46);
		\draw (45)
			 to (48.center)
			 to [in=180, out=0] (49.center)
			 to (50.center);
		\draw (45) to (46);
		\draw (47.center) to (46);
	\end{pgfonlayer}
\end{tikzpicture}
\right)
    =
    \Gr
    \left[\begin{array}{c|c}
      (1+\delta)^{\minu 1} & 0\\ \hline
      0                   & 1+\delta^{\minu 1}
    \end{array}\right]
    \subseteq (\F_p^{2}(\delta),-\varpi_1)\times(\F_p^{2}(\delta),\varpi_1).
  \]
\end{example}

These shifted graph states have been extensively studied by Haah in the context of lattice codes, topological quantum computing, and quantum cellular automata \cite{Haah2017, Haah2021, Haah2013}. In his work, the entries of the graph are restricted to the ring \(\F_p[\delta,\delta^{\minu 1}]\subset \F_p(\delta)\) of Laurent polynomials as well as related group algebras and modules.

\subsection{From generating tableaux to behaviours}
Having interpreted stateful ZX-diagrams in \(\sALR\) in the previous subsection, we now relate this generating-tableau semantics to the unrolling semantics of \Cref{cor:unroll}. We show that the geometric series expansion recovers the 
\emph{rational} stabilizers of the unrollings of a stateful ZX-diagram.

First, it is an immediate consequence of~\cite[Thm.~3.6]{obs} that:

\begin{theoremE}\label{thm:gamma_oplax}
  There is a \dag-CC oplax normal functor
  \(\Gamma:\sALR\to \AffRel\)
  factoring through \(\Obs(\ACR)\), which sends morphisms
  \(
      R
    \subseteq
      (
          \F_p^{2n}(\delta), -\varpi_n)
        \times
          (\F_p^{2m}(\delta), \varpi_m
      )
  \)
  to the \(\F_p\)-affine subspace of \((\F_p^\Z)^{2n}\times (\F_p^\Z)^{2m}\) given pointwise by interpreting geometric series as infinite sequences in \(\F_p\).
\end{theoremE}

Being oplax normal, means that \(\Gamma\) preserves identities and satisfies \(\Gamma(S\circ R)\subseteq \Gamma(S)\circ \Gamma(R)\). When this inclusion is strict, composition in \(\AffRel\) is interpreted to \emph{approximate} composition in \(\sALR\).

For example, consider the symplectic analogue of the observation made by Comfort and de Felice~\cite{obs}:
\begin{example}
  \label{ex:not_invertible}
  Consider the shifted Lagrangian relation:
  \begin{equation*}
      M_{\delta+1}
    \coloneqq
      \Gr
      \left[\begin{array}{c|c}
         \delta +1 & 0\\ \hline
         0 & (\delta^{\minu 1}+1)^{\minu 1}
      \end{array}\right]
    =
      \Ext\Bigl(\begin{tikzpicture}
	\begin{pgfonlayer}{nodelayer}
		\node [style=none] (63) at (-0.1, -0.75) {};
		\node [style=none] (64) at (2.875, -0.75) {};
		\node [style=none] (65) at (2.875, 0.75) {};
		\node [style=none] (66) at (-0.1, 0.75) {};
		\node [style=none] (67) at (-0.1, -0.125) {};
		\node [style=none] (68) at (2.875, -0.125) {};
		\node [style=bn] (87) at (0.525, -0.125) {};
		\node [style=wn] (88) at (2.25, -0.125) {};
		\node [style=rmat] (89) at (1.375, 0.25) {\(\delta\)};
		\node [style=none] (90) at (1.375, -0.5) {};
	\end{pgfonlayer}
	\begin{pgfonlayer}{edgelayer}
		\draw [style=background] (64.center)
			 to (63.center)
			 to (66.center)
			 to (65.center)
			 to cycle;
		\draw [style=background] (65.center) to (64.center);
		\draw [style=background] (64.center) to (63.center);
		\draw [style=background] (63.center) to (66.center);
		\draw [style=background] (66.center) to (65.center);
		\draw (67.center) to (87);
		\draw (88) to (68.center);
		\draw (88)
			 to [in=0, out=135] (89.center)
			 to [in=45, out=-180] (87);
		\draw (87)
			 to [in=180, out=-45] (90.center)
			 to [in=-135, out=0] (88);
	\end{pgfonlayer}
\end{tikzpicture}
\Bigr)
  \end{equation*}
  This is only partially invertible in \(\Gamma(\sALR)\subset \AffRel\):
  \[
    \Gamma(M_{\delta+1}) \circ \Gamma(M_{\delta+1}^{\minu 1})=1,
    \quad\text{however}\quad
    \Gamma(M_{\delta+1}^{\minu 1}) \circ \Gamma(M_{\delta+1})\supsetneq 1.
  \]
  
\end{example}

Because observational equivalence is defined by finite approximations, exhibiting the left inverse would require expanding the entire infinite geometric series, that is, it would require infinite lookahead. The right inverse, by contrast, can be computed lazily with finite lookahead, without fully expanding the series at each stage of the unrolling.
The oplaxness of \(\Gamma\) is the categorical shadow of a \emph{catastrophic} encoder in the theory of quantum convolutional codes. The feedback transfer function \((1+\delta)^{\minu 1}\) of \(M_{\delta+1}\) is the textbook example: inverting it propagates a single error indefinitely, so it admits no finite-depth, non-catastrophic inverse~\cite{Ollivier2003, wilde}. %

As an immediate consequence of \cite[Prop.~6.10]{obs}, we can also relate the unrolling semantics to the generating-tableau semantics,

\begin{proposition}\label{prop:gamma_ext_approx}
For all delayed ZX-diagrams \(D\in \St(\StabZX)\),  \(\Gamma(\Ext(D))\subseteq \StabRel(D)\).
\end{proposition}

This is merely an inclusion because the geometric series expansion produces only the \emph{rational} sequences of Pauli operators which stabilize its unrollings.
This justifies our initial claim at the beginning of the section that the generating-tableau is a finitary approximation of the infinite space of stabilizers.

\begin{example}
  \label{ex:delayed_graph_state_gamma}
  Continuing \Cref{ex:delayed_graph_state_ext}, let
  \[
    \Ext\left(\right)
    =
    \ker\left[\begin{array}{ccc|ccc}
        1 & 0 & 0 & 0              & \delta a+ d & 0 \\
      0 & 1 & 0 & \delta^{\minu 1} a+ d & 0           & \delta^{\minu 1} b+ c\\
      0 & 0 & 1 & 0                & \delta b+ c & 0
    \end{array}\right]
    \eqqcolon R.
  \]
  A rational sequence of Paulis encoded by \((\vb x,\vb z)\in ((\F_p^3)^\Z)^2\) is in \(\Gamma(R)\) if and only if:
  \[
      x_{0,t}=-az_{1,t-1}-dz_{1,t},
    \quad
      x_{1,t}=-az_{0,t+1}-dz_{0,t}-bz_{2,t+1}-cz_{2,t},
    \quad\text{and}\quad
      x_{2,t}=-bz_{1,t-1}-cz_{1,t},
    \quad\forall t\in\Z.
  \]
  Therefore, the finitely supported stabilizer subgroup encoded by 
  \(
    \Gamma(R)\cap
    (
      \substack{\Z\\ \oplus}\F_p^3
    )^2
  \)
  is generated by:
  \[
      K_{0,t}:=X_{1,t-1}^{\minu a}X_{1,t}^{\minu d}Z_{0,t},
    \quad
      K_{1,t}:=X_{0,t+1}^{\minu a}X_{0,t}^{\minu d}X_{2,t+1}^{\minu b}X_{2,t}^{\minu c}Z_{1,t},
    \quad\text{and}\quad
      K_{2,t}:=X_{1,t-1}^{\minu b}X_{1,t}^{\minu c}Z_{2,t},
    \quad\forall t\in\Z.
  \]
  where \(X_{n,t}\) (respectively \(Z_{n,t}\)) denotes the Pauli \(X\)
  (respectively Pauli \(Z\)) on wire \(n\in\{0,1,2\}\) at time \(t\).
\end{example}

\begin{example}\label{ex:delayed_cx_gamma}
  Continuing \Cref{ex:delayed_cx_ext},
  \[
    \Ext\left(\begin{tikzpicture}
	\begin{pgfonlayer}{nodelayer}
		\node [style=none] (0) at (13.125, 1.25) {};
		\node [style=none] (1) at (13.125, -1.25) {};
		\node [style=none] (2) at (17.125, -1.25) {};
		\node [style=none] (3) at (17.125, 1.25) {};
		\node [style=none] (41) at (16.5, 0) {};
		\node [style=none] (42) at (13.75, 0) {};
		\node [style=none] (43) at (16.5, 0.75) {};
		\node [style=none] (44) at (13.75, 0.75) {};
		\node [style=bn] (45) at (14.25, 0) {};
		\node [style=wn] (46) at (15.25, -0.75) {};
		\node [style=none] (47) at (13.125, -0.75) {};
		\node [style=none] (48) at (15.25, 0) {};
		\node [style=none] (49) at (16.5, -0.75) {};
		\node [style=none] (50) at (17.125, -0.75) {};
		\node [style=lmat] (51) at (15.25, 0.75) {\(\delta\)};
	\end{pgfonlayer}
	\begin{pgfonlayer}{edgelayer}
		\draw [style=background] (3.center)
			 to (2.center)
			 to (1.center)
			 to (0.center)
			 to cycle;
		\draw [style=background] (0.center) to (3.center);
		\draw [style=background] (3.center) to (2.center);
		\draw [style=background] (2.center) to (1.center);
		\draw [style=background] (1.center) to (0.center);
		\draw (45)
			 to (42.center)
			 to [bend left=90, looseness=0.75] (44.center)
			 to (43.center)
			 to [bend left=90, looseness=0.75] (41.center)
			 to [in=0, out=-180] (46);
		\draw (45)
			 to (48.center)
			 to [in=180, out=0] (49.center)
			 to (50.center);
		\draw (45) to (46);
		\draw (47.center) to (46);
	\end{pgfonlayer}
\end{tikzpicture}
\right)
    =
    \Gr
    \left[\begin{array}{c|c}
      (1+\delta)^{\minu 1} & 0\\ \hline
      0                   & 1+\delta^{\minu 1}
    \end{array}\right]
  \]
  Applying the geometric series expansion, two rational sequences of Paulis encoded by 
  \(\bigl((\vb x,\vb z),(\vb x',\vb z')\bigr)\in ((\F_p^\Z)^2)^2\)
  are in this relation if and only if:
  \[
      x_t=x'_t+x'_{t-1}
    \qquad\text{and}\qquad
      z'_t=z_t+z_{t+1},
    \qquad\forall t\in\Z.
  \]
  Therefore, the finitely supported stabilizer subgroup  
  is generated by the following local Pauli operators:
  \[
      L^X_t:=X_t^{\mathrm{in}}X_{t+1}^{\mathrm{in}}X_t^{\mathrm{out}}
    \qquad\text{and}\qquad
      L^Z_t:=Z_t^{\mathrm{in}}Z_t^{\mathrm{out}}Z_{t-1}^{\mathrm{out}},
    \qquad\forall t\in\Z.
  \]
\end{example}

  \section{The delayed ZX-calculus\texorpdfstring{: \(\dStabZX\)}{}}\label{sec:dzx}
  \pratendSetLocal{category=dzx, end, restate}
  In this section, we add enough equations to \(\St(\StabZX)\) to obtain a complete axiomatisation of \(\sALR\). We call the resulting language the \textbf{delayed stabilizer ZX-calculus}, denoted \(\dStabZX\). The delayed ZX-calculus can be regarded as a quantum/symplectic analogue of the graphical calculus for signal flow graphs
\cite{baez,Bonchi2021, Bonchi2015, Bonchi2017}.

\subsection{Derived notation}\label{subsec:dzx_derived_generators}
In this subsection, we define generalised spiders, H-boxes and multipliers in \(\St(\StabZX)\), giving their interpretations in \(\sALR\). This will allow us to state the equational theory of \(\dStabZX\) in the following subsection. We prove that these notations are well-defined in \Cref{subsec:dzx_axiom_lemmas}.

\begin{definition}\label{def:dzx_poly_multiplier}
   Polynomial multipliers are defined by induction on their degree:
   \begin{equation*}
         \begin{tikzpicture}
	\begin{pgfonlayer}{nodelayer}
		\node [style=none] (51) at (-4.375, -1.25) {};
		\node [style=none] (52) at (0, -1.25) {};
		\node [style=none] (53) at (0, 1.25) {};
		\node [style=none] (54) at (-4.375, 1.25) {};
		\node [style=none] (55) at (-4.375, 0) {};
		\node [style=none] (56) at (0, 0) {};
		\node [style=rmat] (61) at (-2.25, 0) {\(f(\delta)+b\delta^{n+1}\)};
	\end{pgfonlayer}
	\begin{pgfonlayer}{edgelayer}
		\draw [style=background] (53.center)
			 to (52.center)
			 to (51.center)
			 to (54.center)
			 to cycle;
		\draw (55.center) to (61);
		\draw (61) to (56.center);
	\end{pgfonlayer}
\end{tikzpicture}

      \coloneqq
         \begin{tikzpicture}
	\begin{pgfonlayer}{nodelayer}
		\node [style=none] (8) at (1, -1.25) {};
		\node [style=none] (9) at (6.5, -1.25) {};
		\node [style=none] (10) at (6.5, 1.25) {};
		\node [style=none] (11) at (1, 1.25) {};
		\node [style=none] (41) at (1, 0) {};
		\node [style=bn] (43) at (1.75, 0) {};
		\node [style=wn] (44) at (5.75, 0) {};
		\node [style=none] (45) at (6.5, 0) {};
		\node [style=none] (46) at (2.75, 0.5) {};
		\node [style=rmat] (47) at (3.75, -0.5) {\(f(\delta)\)};
		\node [style=none] (48) at (4.75, 0.5) {};
		\node [style=none] (49) at (2.75, -0.5) {};
		\node [style=none] (50) at (4.75, -0.5) {};
		\node [style=rmat] (63) at (3, 0.5) {\(b\delta^n\)};
		\node [style=rmat] (64) at (4.5, 0.5) {\(\delta\)};
	\end{pgfonlayer}
	\begin{pgfonlayer}{edgelayer}
		\draw [style=background] (10.center)
			 to (9.center)
			 to (8.center)
			 to (11.center)
			 to cycle;
		\draw (41.center) to (43);
		\draw (44) to (45.center);
		\draw (43)
			 to [in=-180, out=60] (46.center)
			 to (48.center)
			 to [in=120, out=0] (44);
		\draw (44)
			 to [in=0, out=-120] (50.center)
			 to (47.center)
			 to (49.center)
			 to [in=-60, out=180] (43);
	\end{pgfonlayer}
\end{tikzpicture}

   \end{equation*}
\end{definition}

These restrict to \(\StabZX\) multipliers when \(f(\delta)=f(0) \in \F_p\), and the delay when  \(f(\delta) = \delta\):
\begin{equation*}
      \Ext\left(\begin{tikzpicture}
	\begin{pgfonlayer}{nodelayer}
		\node [style=none] (51) at (-4.375, -0.75) {};
		\node [style=none] (52) at (-1, -0.75) {};
		\node [style=none] (53) at (-1, 0.75) {};
		\node [style=none] (54) at (-4.375, 0.75) {};
		\node [style=none] (55) at (-4.375, 0) {};
		\node [style=none] (56) at (-1, 0) {};
		\node [style=rmat] (61) at (-2.75, 0) {\(f(\delta)\)};
	\end{pgfonlayer}
	\begin{pgfonlayer}{edgelayer}
		\draw [style=background] (53.center)
			 to (52.center)
			 to (51.center)
			 to (54.center)
			 to cycle;
		\draw (55.center) to (61);
		\draw (61) to (56.center);
	\end{pgfonlayer}
\end{tikzpicture}
\right)
   = 
      \Gr
      \left[\begin{array}{c|c}
         f(\delta) & 0\\ \hline
         0 & f(\delta^{\minu 1})^{\minu 1}
      \end{array}\right]
\end{equation*}

Using two polynomial multipliers, we define rational multipliers:
\begin{definition}\label{def:dzx_rational_multiplier}
   Given \(f(\delta),g(\delta) \in \F_p[\delta]\), define 
   \[
         \begin{tikzpicture}
	\begin{pgfonlayer}{nodelayer}
		\node [style=none] (72) at (7.125, -0.5) {};
		\node [style=none] (73) at (9.875, -0.5) {};
		\node [style=none] (74) at (9.875, 0.5) {};
		\node [style=none] (75) at (7.125, 0.5) {};
		\node [style=none] (76) at (7.125, 0) {};
		\node [style=none] (77) at (9.875, 0) {};
		\node [style=lmat] (79) at (8.5, 0) {\(g(\delta)\)};
	\end{pgfonlayer}
	\begin{pgfonlayer}{edgelayer}
		\draw [style=background] (72.center)
			 to (75.center)
			 to (74.center)
			 to (73.center)
			 to cycle;
		\draw [style=background] (74.center) to (73.center);
		\draw [style=background] (73.center) to (72.center);
		\draw [style=background] (72.center) to (75.center);
		\draw [style=background] (75.center) to (74.center);
		\draw (76.center) to (79);
		\draw (79) to (77.center);
	\end{pgfonlayer}
\end{tikzpicture}

      \coloneqq
         \begin{tikzpicture}
	\begin{pgfonlayer}{nodelayer}
		\node [style=none] (81) at (11.375, -0.75) {};
		\node [style=none] (82) at (15.625, -0.75) {};
		\node [style=none] (83) at (15.625, 0.75) {};
		\node [style=none] (84) at (11.375, 0.75) {};
		\node [style=none] (85) at (11.375, 0) {};
		\node [style=none] (86) at (15.625, 0) {};
		\node [style=rmat] (87) at (13.5, 0) {\(g(\delta)\)};
		\node [style=none] (88) at (14.75, 0) {};
		\node [style=none] (89) at (12.25, 0) {};
		\node [style=none] (90) at (14.75, 0.5) {};
		\node [style=none] (91) at (12.25, 0.5) {};
		\node [style=none] (92) at (12.25, -0.5) {};
		\node [style=none] (93) at (14.75, -0.5) {};
	\end{pgfonlayer}
	\begin{pgfonlayer}{edgelayer}
		\draw [style=background] (83.center)
			 to (82.center)
			 to (81.center)
			 to (84.center)
			 to cycle;
		\draw [style=background] (83.center) to (82.center);
		\draw [style=background] (82.center) to (81.center);
		\draw [style=background] (81.center) to (84.center);
		\draw [style=background] (84.center) to (83.center);
		\draw (85.center)
			 to [in=180, out=0, looseness=0.75] (91.center)
			 to (90.center)
			 to [bend left=90, looseness=1.25] (88.center)
			 to (87.center)
			 to (89.center)
			 to [bend right=90, looseness=1.25] (92.center)
			 to (93.center)
			 to [in=180, out=0, looseness=0.75] (86.center);
	\end{pgfonlayer}
\end{tikzpicture}

      \quad\text{and}\quad
         \begin{tikzpicture}
	\begin{pgfonlayer}{nodelayer}
		\node [style=none] (51) at (-5.625, -0.5) {};
		\node [style=none] (52) at (-1.75, -0.5) {};
		\node [style=none] (53) at (-1.75, 0.5) {};
		\node [style=none] (54) at (-5.625, 0.5) {};
		\node [style=none] (55) at (-5.625, 0) {};
		\node [style=none] (56) at (-1.75, 0) {};
		\node [style=rmat] (61) at (-3.75, 0) {\(f(\delta)/g(\delta)\)};
	\end{pgfonlayer}
	\begin{pgfonlayer}{edgelayer}
		\draw [style=background] (53.center)
			 to (52.center)
			 to (51.center)
			 to (54.center)
			 to cycle;
		\draw (55.center) to (61);
		\draw (61) to (56.center);
	\end{pgfonlayer}
\end{tikzpicture}

      \coloneqq
         \begin{tikzpicture}
	\begin{pgfonlayer}{nodelayer}
		\node [style=none] (63) at (-0.5, -0.5) {};
		\node [style=none] (64) at (4.125, -0.5) {};
		\node [style=none] (65) at (4.125, 0.5) {};
		\node [style=none] (66) at (-0.5, 0.5) {};
		\node [style=none] (67) at (-0.5, 0) {};
		\node [style=none] (68) at (4.125, 0) {};
		\node [style=rmat] (69) at (0.875, 0) {\(f(\delta)\)};
		\node [style=lmat] (70) at (2.625, 0) {\(g(\delta)\)};
	\end{pgfonlayer}
	\begin{pgfonlayer}{edgelayer}
		\draw [style=background] (65.center)
			 to (64.center)
			 to (63.center)
			 to (66.center)
			 to cycle;
		\draw [style=background] (65.center) to (64.center);
		\draw [style=background] (64.center) to (63.center);
		\draw [style=background] (63.center) to (66.center);
		\draw [style=background] (66.center) to (65.center);
		\draw (67.center) to (69);
		\draw (69) to (68.center);
	\end{pgfonlayer}
\end{tikzpicture}
.
   \]
\end{definition}

We define generalized, \emph{directed} \(H\)-boxes labelled by rational functions:

\begin{definition}\label{def:dzx_directed_fourier_box}
   Given \(h(\delta)\in \F_p(\delta)\), define
   \[
      \begin{tikzpicture}
	\begin{pgfonlayer}{nodelayer}
		\node [style=none] (51) at (-7.125, -0.5) {};
		\node [style=none] (52) at (-4.75, -0.5) {};
		\node [style=none] (53) at (-4.75, 0.5) {};
		\node [style=none] (54) at (-7.125, 0.5) {};
		\node [style=none] (55) at (-7.125, 0) {};
		\node [style=none] (56) at (-4.75, 0) {};
		\node [style=rhad] (61) at (-6, 0) {\(h(\delta)\)};
	\end{pgfonlayer}
	\begin{pgfonlayer}{edgelayer}
		\draw [style=background] (53.center) to (52.center) to (51.center) to (54.center) to cycle;
		\draw (55.center) to (61);
		\draw (61) to (56.center);
	\end{pgfonlayer}
\end{tikzpicture}

   \coloneqq
      \begin{tikzpicture}
	\begin{pgfonlayer}{nodelayer}
		\node [style=none] (63) at (-2.75, -0.5) {};
		\node [style=none] (64) at (1.25, -0.5) {};
		\node [style=none] (65) at (1.25, 0.5) {};
		\node [style=none] (66) at (-2.75, 0.5) {};
		\node [style=none] (67) at (-2.75, 0) {};
		\node [style=none] (68) at (1.25, 0) {};
		\node [style=rmat] (69) at (-1.375, 0) {\(h(\delta)\)};
		\node [style=had] (70) at (0.25, 0) {};
	\end{pgfonlayer}
	\begin{pgfonlayer}{edgelayer}
		\draw [style=background] (65.center) to (64.center) to (63.center) to (66.center) to cycle;
		\draw [style=background] (65.center) to (64.center);
		\draw [style=background] (64.center) to (63.center);
		\draw [style=background] (63.center) to (66.center);
		\draw [style=background] (66.center) to (65.center);
		\draw (67.center) to (69);
		\draw (69) to (68.center);
	\end{pgfonlayer}
\end{tikzpicture}

   \quad\text{and}\quad
      \begin{tikzpicture}
	\begin{pgfonlayer}{nodelayer}
		\node [style=none] (80) at (9.125, -0.5) {};
		\node [style=none] (81) at (5.25, -0.5) {};
		\node [style=none] (82) at (5.25, 0.5) {};
		\node [style=none] (83) at (9.125, 0.5) {};
		\node [style=none] (84) at (9.125, 0) {};
		\node [style=none] (85) at (5.25, 0) {};
		\node [style=lmat] (86) at (7.75, 0) {\(h(\delta)\)};
		\node [style=had] (87) at (6.25, 0) {};
	\end{pgfonlayer}
	\begin{pgfonlayer}{edgelayer}
		\draw [style=background] (80.center) to (83.center) to (82.center) to (81.center) to cycle;
		\draw [style=background] (82.center) to (81.center);
		\draw [style=background] (81.center) to (80.center);
		\draw [style=background] (80.center) to (83.center);
		\draw [style=background] (83.center) to (82.center);
		\draw (84.center) to (86);
		\draw (86) to (85.center);
	\end{pgfonlayer}
\end{tikzpicture}

   \eqqcolon
      \begin{tikzpicture}
	\begin{pgfonlayer}{nodelayer}
		\node [style=none] (71) at (15, -0.5) {};
		\node [style=none] (72) at (12.625, -0.5) {};
		\node [style=none] (73) at (12.625, 0.5) {};
		\node [style=none] (74) at (15, 0.5) {};
		\node [style=none] (75) at (15, 0) {};
		\node [style=none] (76) at (12.625, 0) {};
		\node [style=lhad] (77) at (13.875, 0) {\(h(\delta)\)};
	\end{pgfonlayer}
	\begin{pgfonlayer}{edgelayer}
		\draw [style=background] (73.center) to (72.center) to (71.center) to (74.center) to cycle;
		\draw (75.center) to (77);
		\draw (77) to (76.center);
	\end{pgfonlayer}
\end{tikzpicture}
.
   \]
\end{definition}

\begin{textAtEnd}
\begin{lemma}
   \label{lem:self-conj-laurent}
   For \(f(\delta)\in \F_p[\delta, \delta^{\minu 1}]\), we have \(f(\delta)=f(\delta^{\minu 1})\) if and only if \(f(\delta) \in \F_p[\delta+\delta^{\minu 1}]\).
\end{lemma}
\begin{proof}
   First, we prove that if \(h(\delta)\in \F_p[\delta+\delta^{\minu 1}]\), then \( h(\delta) = h(\delta^{\minu 1})  \in \F_p[\delta,\delta^{\minu 1}]\).  Since, \(\bar{\delta+\delta^{\minu 1}} = \delta^{\minu 1}+\delta = \delta+\delta^{\minu 1}\)
   every element of \(\F_p[\delta+\delta^{\minu 1}]\) is self-symmetric.

   For the converse direction, we prove that \((\delta^j+\delta^{\minu j}) \in \F_p[\delta+\delta^{\minu 1}]\) for all \(j\in \N\),  which implies that every self-conjugate Laurent polynomial \(h(\delta) = h_0 + \sum_{j=1}^{n} h_j (\delta+\delta^{\minu j}) \in \F_p[\delta+\delta^{\minu 1}]\).

   We prove that  \((\delta^j+\delta^{\minu j}) \in \F_p[\delta+\delta^{\minu 1}]\) by induction on \(j \in \N\).
   For the base case of \(j=0\):
   \[\delta^0 + \delta^{\minu 0} = 2(\delta+\delta^{\minu 1})^0 \in \F_p[\delta+\delta^{\minu 1}]\]
   
   For the inductive hypothesis, take some \(0<j\in \N\) such that for every \(m\in \N\),  such that \(m<j\), \((\delta^m+\delta^{\minu m}) \in \F_p[\delta+\delta^{\minu 1}]\).
   From the binomial theorem, we have:
   \begin{align}
      (\delta+\delta^{\minu 1})^{j}
   &=
         \sum_{k=0}^{j}\binom{j}{k}\,\delta^{j-2k} 
      =
         \binom{j}{0} \delta^{j}
         + \sum_{k=1}^{\lfloor j/2\rfloor}\binom{j}{k} \delta^{j-2k}
         + \sum_{k=\lfloor j/2\rfloor+1}^{j-1}\binom{j}{k}\,\delta^{j-2k}
         + \binom{j}{j} \delta^{\minu j}                                           \\
   &=
         \delta^{j}+\delta^{\minu j}
      + \sum_{k=1}^{\lfloor j/2\rfloor}
            \binom{j}{k}\left(\delta^{j-2k}+\delta^{\minu (j-2k)}\right)
   \end{align}

   Therefore, rearranging terms:
   \begin{align}
      \delta^{j}+\delta^{\minu j}
      =
      (\delta+\delta^{\minu 1})^{j}
      - \sum_{k=1}^{\lfloor j/2\rfloor}
      \binom{j}{k}\left(\delta^{j-2k}+\delta^{\minu (j-2k)}\right)
   \end{align}
   It is clear that  \((\delta+\delta^{\minu 1})^{j}= (\delta+\delta^{\minu 1})^{j} \in \F_p[\delta+\delta^{\minu 1}]\).  Moreover all of the other summands are also in \(\F_p[\delta+\delta^{\minu 1}]\) by the inductive hypothesis.  Therefore, because \(\F_p[\delta+\delta^{\minu 1}]\) is a ring, \({\delta^{j}+\delta^{\minu j} \in \F_p[\delta+\delta^{\minu 1}]}\), as desired.
\end{proof}

\begin{lemma}\label{lem:self-conj-rational}
   For \(f(\delta)\in \F_p(\delta)\), we have \(f(\delta)=f(\delta^{\minu 1})\) if and only if \(f(\delta) \in \F_p(\delta+\delta^{\minu 1})\).
\end{lemma}
\begin{proof}
   First, remark that \(0\) is both a  self-conjugate and in \(\F_p(\delta+\delta^{\minu 1})\). 
   Take a non-zero self-conjugate rational function 
   \[\frac{f(\delta)}{g(\delta)}= \bar{\left(\frac{f(\delta)}{g(\delta)}\right)}=\frac{f(\delta^{\minu 1})}{g(\delta^{\minu 1})}  \in \F_p(\delta)\]
   for \(f(\delta), g(\delta) \in \F_p[\delta]\). Then:
   \[\frac{f(\delta)}{g(\delta)} = \bar{\left(\frac{f(\delta)}{g(\delta)}\right)}= \frac{f(\delta^{\minu 1})}{g(\delta^{\minu 1})} = \frac{f(\delta^{\minu 1})g(\delta)}{g(\delta^{\minu 1}) g(\delta)}\]
   Remark that \(g(\delta^{\minu 1})g(\delta)\) and  \(f(\delta^{\minu 1})g(\delta)\) are Laurent polynomials. 
   Moreover,
   \[\bar{g(\delta^{\minu 1}) g(\delta) }= g(\delta)g(\delta^{\minu 1}) =g(\delta^{\minu 1}) g(\delta)  \in \F_p[\delta+\delta^{\minu 1}]\]
   by \Cref{lem:self-conj-laurent}. Similarly, since \(f(\delta) = g(\delta)f(\delta^{\minu 1})/g(\delta^{\minu 1})\), another application of \Cref{lem:self-conj-laurent} gives:
   \[\bar{f(\delta^{\minu 1})g(\delta)} = f(\delta)g(\delta^{\minu 1}) =   \frac{f(\delta^{\minu 1})g(\delta)g(\delta^{\minu 1})}{g(\delta^{\minu 1})}  = f(\delta^{\minu 1})g(\delta) \in \F_p[\delta+\delta^{\minu 1}] \]
   Therefore, every self-conjugate rational function is a fraction of self-conjugate Laurent polynomials in \(\F_p[\delta+\delta^{\minu 1}]\) as desired.
\end{proof}
\end{textAtEnd}

To define undirected generalised H-boxes and generalised spiders, it is helpful to understand the field \(\{f(\delta)\ | \ f(\delta) = f(\delta^{\minu 1})\} \subset \F_p(\delta)\) of \emph{self conjugate} rational functions:

\begin{propositionE}\label{prop:self-conj-rational}
   The self-conjugate rational functions form the field of fractions \(\F_p(\delta+\delta^{\minu 1})\subseteq \F_p(\delta)\) of the ring \(\F_p[\delta+\delta^{\minu 1}]\) of self-conjugate Laurent polynomials.
   
   Concretely, every self-conjugate rational function can be written as
   \[
      \frac{f(\delta)+f(\delta^{\minu 1})-f(0)}{g(\delta)+g(\delta^{\minu 1})-g(0)}
   \]
   for some polynomials \(f(\delta),g(\delta)\in \F_p[\delta]\).
\end{propositionE}
\begin{proofE}
   This follows immediately from \Cref{lem:self-conj-laurent} and \Cref{lem:self-conj-rational}.
\end{proofE}

The left and right-facing \(H\)-boxes labelled by these self-conjugate rational functions coincide:
\begin{definition}\label{def:dzx_undirected_fourier_box}
   Given \(h(\delta) \in \F_p(\delta+\delta^{\minu 1})\), define
   \[
         \begin{tikzpicture}
	\begin{pgfonlayer}{nodelayer}
		\node [style=none] (51) at (-5.125, -0.5) {};
		\node [style=none] (52) at (-2.75, -0.5) {};
		\node [style=none] (53) at (-2.75, 0.5) {};
		\node [style=none] (54) at (-5.125, 0.5) {};
		\node [style=none] (55) at (-5.125, 0) {};
		\node [style=none] (56) at (-2.75, 0) {};
		\node [style=rhad] (61) at (-4, 0) {\(h(\delta)\)};
	\end{pgfonlayer}
	\begin{pgfonlayer}{edgelayer}
		\draw [style=background] (53.center) to (52.center) to (51.center) to (54.center) to cycle;
		\draw (55.center) to (61);
		\draw (61) to (56.center);
	\end{pgfonlayer}
\end{tikzpicture}

      \eqqcolon
         \begin{tikzpicture}
	\begin{pgfonlayer}{nodelayer}
		\node [style=none] (79) at (2.375, -0.5) {};
		\node [style=none] (80) at (0, -0.5) {};
		\node [style=none] (81) at (0, 0.5) {};
		\node [style=none] (82) at (2.375, 0.5) {};
		\node [style=none] (83) at (2.375, 0) {};
		\node [style=none] (84) at (0, 0) {};
		\node [style=had] (85) at (1.25, 0) {\(h(\delta)\)};
	\end{pgfonlayer}
	\begin{pgfonlayer}{edgelayer}
		\draw [style=background] (79.center) to (82.center) to (81.center) to (80.center) to cycle;
		\draw [style=background] (81.center) to (80.center);
		\draw [style=background] (80.center) to (79.center);
		\draw [style=background] (79.center) to (82.center);
		\draw [style=background] (82.center) to (81.center);
		\draw (83.center) to (85);
		\draw (85) to (84.center);
	\end{pgfonlayer}
\end{tikzpicture}

      \coloneqq
         \begin{tikzpicture}
	\begin{pgfonlayer}{nodelayer}
		\node [style=none] (71) at (7.375, -0.5) {};
		\node [style=none] (72) at (5, -0.5) {};
		\node [style=none] (73) at (5, 0.5) {};
		\node [style=none] (74) at (7.375, 0.5) {};
		\node [style=none] (75) at (7.375, 0) {};
		\node [style=none] (76) at (5, 0) {};
		\node [style=lhad] (77) at (6.25, 0) {\(h(\delta)\)};
	\end{pgfonlayer}
	\begin{pgfonlayer}{edgelayer}
		\draw [style=background] (73.center) to (72.center) to (71.center) to (74.center) to cycle;
		\draw (75.center) to (77);
		\draw (77) to (76.center);
	\end{pgfonlayer}
\end{tikzpicture}
.
   \]
\end{definition}

Finally, we define generalised spiders with rational affine phases and self-conjugate rational symplectic phases:
\begin{definition}\label{def:dzx_shifted_spider}
   Take some \(a(\delta)\in \F_p(\delta)\) and  \(b(\delta)\in \F_p(\delta+\delta^{\minu 1})\).

   Choose \(h(\delta), k(\delta)\in \F_p[\delta+\delta^{\minu 1}]\) and \(f(\delta), g(\delta)\in \F_p[\delta]\) such that:
   
   \[
         b(\delta) = \frac{h(\delta)}{k(\delta)},
      \quad
         h(\delta) =  f(\delta)+f(\delta^{\minu 1})-f(0),
      \quad\text{and}\quad
      k(\delta) = g(\delta)+g(\delta^{\minu 1}) - g(0)
   \]

   We define \((a(\delta),b(\delta))\)-labelled green spiders by case distinction:
   \[


               & \text{otherwise.}\\
         \end{cases}
   \]
   Where red spiders are defined as follows:
   \[
         %

   \]

\end{definition}

\subsection{Axioms}\label{subsec:dzx_axioms}
Using the notation from the previous subsection, we impose the additional equations:

\begin{axiom}[Spider-fusion]\label{ax:dzx_fusion}
   Given \( (a(\delta), b(\delta)), (c(\delta), d(\delta)) \in   \F_p(\delta) \times \F_p(\delta+\delta^{\minu 1})\):
   \begin{equation*}
         %
.
   \]
\end{axiom}

All of the axioms can be easily verified to be sound.

We reframe the presentation of \(\dStabZX\) by taking the rational spiders and \(H\)-boxes as primitives, providing a self-contained axiomatisation in \Cref{fig:dzx_axioms}. We close this subsection by recording the well-definedness of the derived notation and the equivalence of the two presentations.

\begin{restatable}{proposition}{propDZXWellDefined}\label{prop:dzx_well_defined}
   The derived notation of \(\dStabZX\) is well-defined.
\end{restatable}
\seeproof{prop:dzx_well_defined}

A series of menial calculations show that:
\begin{proposition}\label{prop:dzx_equivalent_presentations}
The presentation in \Cref{subsec:dzx_axioms} is equivalent to that of \Cref{fig:dzx_axioms}, where we identify
\[\begin{tikzpicture}
	\begin{pgfonlayer}{nodelayer}
		\node [style=none] (51) at (-4.375, -0.5) {};
		\node [style=none] (52) at (-1.5, -0.5) {};
		\node [style=none] (53) at (-1.5, 0.5) {};
		\node [style=none] (54) at (-4.375, 0.5) {};
		\node [style=none] (55) at (-4.375, 0) {};
		\node [style=none] (56) at (-1.5, 0) {};
		\node [style=rmat] (61) at (-3, 0) {\(h(\delta)\)};
	\end{pgfonlayer}
	\begin{pgfonlayer}{edgelayer}
		\draw [style=background] (53.center)
			 to (52.center)
			 to (51.center)
			 to (54.center)
			 to cycle;
		\draw (55.center) to (61);
		\draw (61) to (56.center);
	\end{pgfonlayer}
\end{tikzpicture}
\quad \text{with} \quad \begin{tikzpicture}
	\begin{pgfonlayer}{nodelayer}
		\node [style=none] (51) at (-4.375, -0.5) {};
		\node [style=none] (52) at (-0.75, -0.5) {};
		\node [style=none] (53) at (-0.75, 0.5) {};
		\node [style=none] (54) at (-4.375, 0.5) {};
		\node [style=none] (55) at (-4.375, 0) {};
		\node [style=none] (56) at (-0.75, 0) {};
		\node [style=rhad] (61) at (-3, 0) {\(h(\delta)\)};
		\node [style=had] (62) at (-1.75, 0) {\(\minu\)};
	\end{pgfonlayer}
	\begin{pgfonlayer}{edgelayer}
		\draw [style=background] (53.center)
			 to (52.center)
			 to (51.center)
			 to (54.center)
			 to cycle;
		\draw (55.center) to (61);
		\draw (61) to (56.center);
	\end{pgfonlayer}
\end{tikzpicture}
; \quad\text{conversely, identifying} \quad 
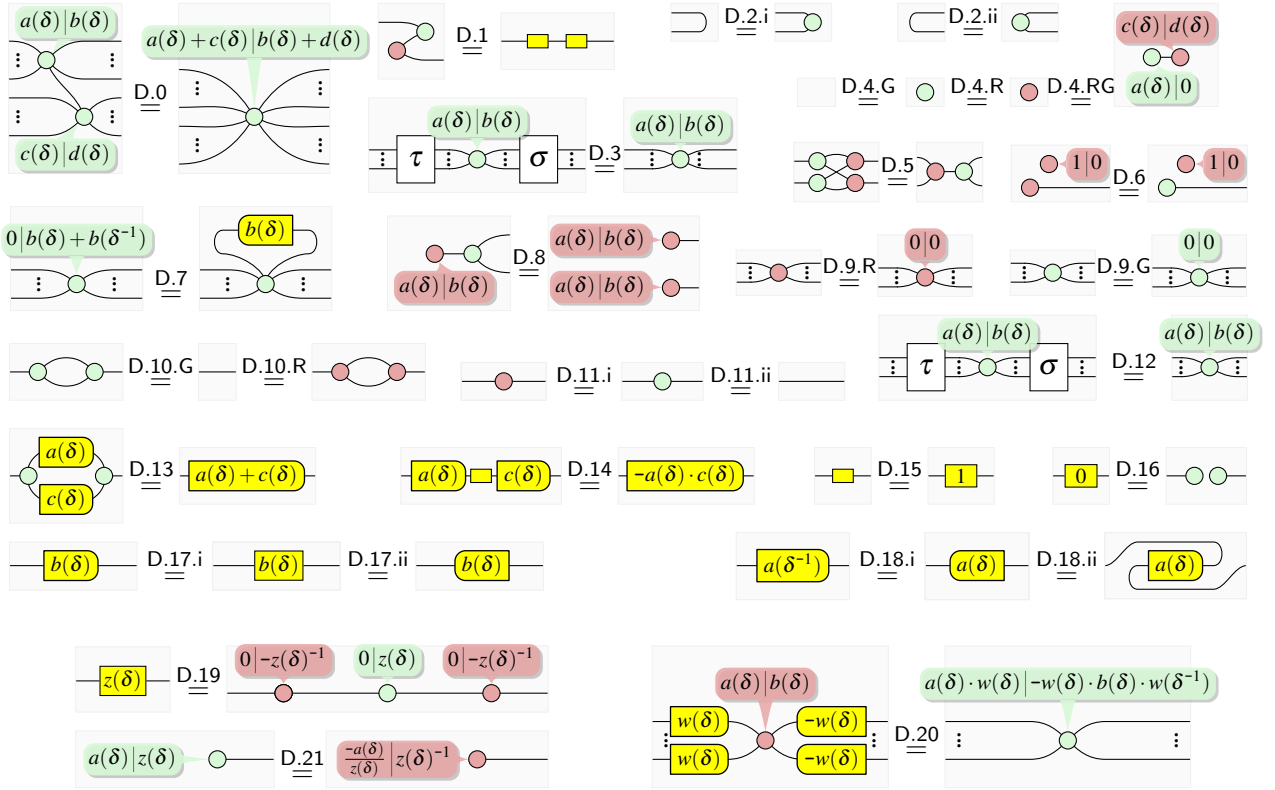
\begin{tikzpicture}
	\begin{pgfonlayer}{nodelayer}
		\node [style=none] (51) at (-4.375, -0.5) {};
		\node [style=none] (52) at (-1.5, -0.5) {};
		\node [style=none] (53) at (-1.5, 0.5) {};
		\node [style=none] (54) at (-4.375, 0.5) {};
		\node [style=none] (55) at (-4.375, 0) {};
		\node [style=none] (56) at (-1.5, 0) {};
		\node [style=rhad] (61) at (-3, 0) {\(h(\delta)\)};
	\end{pgfonlayer}
	\begin{pgfonlayer}{edgelayer}
		\draw [style=background] (53.center)
			 to (52.center)
			 to (51.center)
			 to (54.center)
			 to cycle;
		\draw (55.center) to (61);
		\draw (61) to (56.center);
	\end{pgfonlayer}
\end{tikzpicture}
\quad\text{with}\quad \begin{tikzpicture}
	\begin{pgfonlayer}{nodelayer}
		\node [style=none] (51) at (-4.375, -0.5) {};
		\node [style=none] (52) at (-0.75, -0.5) {};
		\node [style=none] (53) at (-0.75, 0.5) {};
		\node [style=none] (54) at (-4.375, 0.5) {};
		\node [style=none] (55) at (-4.375, 0) {};
		\node [style=none] (56) at (-0.75, 0) {};
		\node [style=rmat] (61) at (-3, 0) {\(h(\delta)\)};
		\node [style=had] (62) at (-1.75, 0) {};
	\end{pgfonlayer}
	\begin{pgfonlayer}{edgelayer}
		\draw [style=background] (53.center)
			 to (52.center)
			 to (51.center)
			 to (54.center)
			 to cycle;
		\draw (55.center) to (61);
		\draw (61) to (56.center);
	\end{pgfonlayer}
\end{tikzpicture}
.\]
\end{proposition}

\clearpage
\begin{figure}[p]
\vspace*{\fill}
\centering
\[\input{dzx/axioms/eqs/axioms.tikz}\]

\small
\renewcommand{\arraystretch}{1.15}
\begin{tabular}{@{}p{0.3\linewidth}p{0.32\linewidth}p{0.30\linewidth}@{}}
\textbf{Parameter class} & \textbf{Parameters} & \textbf{Conditions} \\
Permutations &
\(\sigma,\tau\) &
--- \\
Rational functions &
\(a(\delta), c(\delta), w(\delta)\in \F_p(\delta)\) &
\(w(\delta)\neq 0\) \\
Self-conjugate rational functions &
\(b(\delta), d(\delta), z(\delta)\in \F_p(\delta+\delta^{\minu 1})\) &
\(z(\delta)\neq 0\) 
\end{tabular}

\caption{Axioms of the delayed stabilizer ZX-calculus.}
\label{fig:dzx_axioms}
\vspace*{\fill}
\end{figure}
\clearpage

\begin{textAtEnd}
\subsection{Proof of well-definedness}\label{subsec:dzx_axiom_lemmas}
\phantomsection\label{proof:prop:dzx_well_defined}
In this section, we prove some useful lemmas for \(\dStabZX\), showing that our derived notation in the previous subsection is well-defined.

First, we show that the rational function multipliers are well-defined.

By \Cref{ax:dzx_colour}, \Cref{ax:h_flip} and by definition of the \(H\)-box it follows that:
\begin{lemma}\label{lem:polynomial_inverse}
   For any nonzero \(f(\delta) \in \F_p[\delta]\):
   \[
      \begin{tikzpicture}
\begin{pgfonlayer}{nodelayer}
		\node [style=none] (0) at (2.5, 0.5) {};
		\node [style=none] (1) at (2.5, -0.5) {};
		\node [style=none] (2) at (-1.25, -0.5) {};
		\node [style=none] (3) at (-1.25, 0.5) {};
		\node [style=rmat] (4) at (-0.25, 0) {\(f(\delta)\)};
		\node [style=lmat] (5) at (1.5, 0) {\(f(\delta)\)};
		\node [style=none] (6) at (2.5, 0) {};
		\node [style=none] (7) at (-1.25, 0) {};
\end{pgfonlayer}
\begin{pgfonlayer}{edgelayer}
		\draw [style=background] (0.center)
			 to (1.center)
			 to (2.center)
			 to (3.center)
			 to cycle;
		\draw [style=background] (3.center) to (0.center);
		\draw [style=background] (0.center) to (1.center);
		\draw [style=background] (1.center) to (2.center);
		\draw [style=background] (2.center) to (3.center);
		\draw (6.center) to (7.center);
\end{pgfonlayer}
\end{tikzpicture}

      =
      \begin{tikzpicture}
\begin{pgfonlayer}{nodelayer}
		\node [style=none] (48) at (25, 0.5) {};
		\node [style=none] (49) at (25, -0.5) {};
		\node [style=none] (50) at (24, -0.5) {};
		\node [style=none] (51) at (24, 0.5) {};
		\node [style=none] (52) at (25, 0) {};
		\node [style=none] (53) at (24, 0) {};
\end{pgfonlayer}
\begin{pgfonlayer}{edgelayer}
		\draw [style=background] (48.center)
			 to (49.center)
			 to (50.center)
			 to (51.center)
			 to cycle;
		\draw [style=background] (51.center) to (48.center);
		\draw [style=background] (48.center) to (49.center);
		\draw [style=background] (49.center) to (50.center);
		\draw [style=background] (50.center) to (51.center);
		\draw (52.center) to (53.center);
\end{pgfonlayer}
\end{tikzpicture}

      =
      \begin{tikzpicture}
\begin{pgfonlayer}{nodelayer}
		\node [style=none] (102) at (-6.875, 0.5) {};
		\node [style=none] (103) at (-6.875, -0.5) {};
		\node [style=none] (104) at (-10.625, -0.5) {};
		\node [style=none] (105) at (-10.625, 0.5) {};
		\node [style=lmat] (106) at (-9.625, 0) {\(f(\delta)\)};
		\node [style=rmat] (107) at (-7.875, 0) {\(f(\delta)\)};
		\node [style=none] (108) at (-6.875, 0) {};
		\node [style=none] (109) at (-10.625, 0) {};
\end{pgfonlayer}
\begin{pgfonlayer}{edgelayer}
		\draw [style=background] (102.center)
			 to (103.center)
			 to (104.center)
			 to (105.center)
			 to cycle;
		\draw [style=background] (105.center) to (102.center);
		\draw [style=background] (102.center) to (103.center);
		\draw [style=background] (103.center) to (104.center);
		\draw [style=background] (104.center) to (105.center);
		\draw (108.center) to (109.center);
\end{pgfonlayer}
\end{tikzpicture}

   \] 
\end{lemma}
\begin{proof}
      \[
   
   \steq{\lemref{lem:box_inverse}}
   \begin{tikzpicture}
\begin{pgfonlayer}{nodelayer}
		\node [style=none] (8) at (9.25, 0.5) {};
		\node [style=none] (9) at (9.25, -0.5) {};
		\node [style=none] (10) at (3.5, -0.5) {};
		\node [style=none] (11) at (3.5, 0.5) {};
		\node [style=rmat] (12) at (4.5, 0) {\(f(\delta)\)};
		\node [style=lmat] (13) at (8.25, 0) {\(f(\delta)\)};
		\node [style=none] (14) at (9.25, 0) {};
		\node [style=none] (15) at (3.5, 0) {};
		\node [style=had] (17) at (5.75, 0) {};
		\node [style=had] (18) at (6.75, 0) {\(\minu\)};
\end{pgfonlayer}
\begin{pgfonlayer}{edgelayer}
		\draw [style=background] (8.center)
			 to (9.center)
			 to (10.center)
			 to (11.center)
			 to cycle;
		\draw [style=background] (11.center) to (8.center);
		\draw [style=background] (8.center) to (9.center);
		\draw [style=background] (9.center) to (10.center);
		\draw [style=background] (10.center) to (11.center);
		\draw (14.center) to (15.center);
\end{pgfonlayer}
\end{tikzpicture}

   \steq{\defref{def:dzx_directed_fourier_box}}
   \begin{tikzpicture}
\begin{pgfonlayer}{nodelayer}
		\node [style=none] (19) at (14.5, 0.5) {};
		\node [style=none] (20) at (14.5, -0.5) {};
		\node [style=none] (21) at (10.25, -0.5) {};
		\node [style=none] (22) at (10.25, 0.5) {};
		\node [style=rhad] (23) at (11.25, 0) {\(f(\delta)\)};
		\node [style=lhad] (24) at (13.25, 0) {\(\minu f(\delta)\)};
		\node [style=none] (25) at (14.5, 0) {};
		\node [style=none] (26) at (10.25, 0) {};
\end{pgfonlayer}
\begin{pgfonlayer}{edgelayer}
		\draw [style=background] (19.center)
			 to (20.center)
			 to (21.center)
			 to (22.center)
			 to cycle;
		\draw [style=background] (22.center) to (19.center);
		\draw [style=background] (19.center) to (20.center);
		\draw [style=background] (20.center) to (21.center);
		\draw [style=background] (21.center) to (22.center);
		\draw (25.center) to (26.center);
\end{pgfonlayer}
\end{tikzpicture}

   \steq{\axiomref{Z.11}}
   \begin{tikzpicture}
\begin{pgfonlayer}{nodelayer}
		\node [style=none] (28) at (20.5, 0.5) {};
		\node [style=none] (29) at (20.5, -0.5) {};
		\node [style=none] (30) at (15.5, -0.5) {};
		\node [style=none] (31) at (15.5, 0.5) {};
		\node [style=rhad] (32) at (16.5, 0) {\(f(\delta)\)};
		\node [style=lhad] (33) at (19.25, 0) {\(\minu f(\delta)\)};
		\node [style=none] (34) at (20.5, 0) {};
		\node [style=none] (35) at (15.5, 0) {};
		\node [style=wn] (37) at (17.75, 0) {};
\end{pgfonlayer}
\begin{pgfonlayer}{edgelayer}
		\draw [style=background] (28.center)
			 to (29.center)
			 to (30.center)
			 to (31.center)
			 to cycle;
		\draw [style=background] (31.center) to (28.center);
		\draw [style=background] (28.center) to (29.center);
		\draw [style=background] (29.center) to (30.center);
		\draw [style=background] (30.center) to (31.center);
		\draw (34.center) to (35.center);
\end{pgfonlayer}
\end{tikzpicture}

   \steq{\axref{ax:dzx_colour}}
   \begin{tikzpicture}
\begin{pgfonlayer}{nodelayer}
		\node [style=none] (38) at (23, 0.5) {};
		\node [style=none] (39) at (23, -0.5) {};
		\node [style=none] (40) at (21.5, -0.5) {};
		\node [style=none] (41) at (21.5, 0.5) {};
		\node [style=none] (44) at (23, 0) {};
		\node [style=none] (45) at (21.5, 0) {};
		\node [style=bn] (47) at (22.25, 0) {};
\end{pgfonlayer}
\begin{pgfonlayer}{edgelayer}
		\draw [style=background] (38.center)
			 to (39.center)
			 to (40.center)
			 to (41.center)
			 to cycle;
		\draw [style=background] (41.center) to (38.center);
		\draw [style=background] (38.center) to (39.center);
		\draw [style=background] (39.center) to (40.center);
		\draw [style=background] (40.center) to (41.center);
		\draw (44.center) to (45.center);
\end{pgfonlayer}
\end{tikzpicture}

   \steq{\axiomref{Z.11}}
   
   \]
   \begin{align*}
   
   & \steq{\lemref{lem:box_inverse}}
   \input{dzx/axiom_lemmas/polynomial_inverse/eq2/0.tikz}
   \steq{\lemref{lem:polynomial_inverse}}
   \input{dzx/axiom_lemmas/polynomial_inverse/eq2/1.tikz}\\
   &\steq{\defref{def:dzx_directed_fourier_box}}
   \input{dzx/axiom_lemmas/polynomial_inverse/eq2/2.tikz}
   \steq{\axiomref{Z.11}}
   \input{dzx/axiom_lemmas/polynomial_inverse/eq2/3.tikz}
   \steq{\axref{ax:dzx_colour}}
   \input{dzx/axiom_lemmas/polynomial_inverse/eq2/4.tikz}
   \steq{\axiomref{Z.11}}
   \input{dzx/axiom_lemmas/polynomial_inverse/eq2/5.tikz}
   \end{align*}
\end{proof}
Moreover, it follows from several simple induction arguments that the previous lemma implies that rational arithmetic can be derived by polynomial multipliers and their converses:
\begin{lemma} \label{lem:mutliplier_arithmetic}
For all \(f(\delta), g(\delta), h(\delta), k(\delta), s(\delta), t(\delta)\in \F_p[\delta]\) with \(g(\delta), k(\delta), t(\delta)\) nonzero:
\[
\begin{tikzpicture}
\begin{pgfonlayer}{nodelayer}
		\node [style=none] (51) at (-7, -1) {};
		\node [style=none] (52) at (-0.25, -1) {};
		\node [style=none] (53) at (-0.25, 1) {};
		\node [style=none] (54) at (-7, 1) {};
		\node [style=none] (61) at (-5.25, 0.5) {};
		\node [style=none] (63) at (-7, 0) {};
		\node [style=none] (64) at (-0.25, 0) {};
		\node [style=none] (65) at (-5.25, -0.5) {};
		\node [style=bn] (66) at (-6, 0) {};
		\node [style=wn] (67) at (-1.25, 0) {};
		\node [style=none] (94) at (-2, 0.5) {};
		\node [style=none] (95) at (-2, -0.5) {};
		\node [style=rmat] (96) at (-4.5, 0.5) {\(f(\delta)\)};
		\node [style=rmat] (97) at (-4.5, -0.5) {\(h(\delta)\)};
		\node [style=lmat] (98) at (-2.75, 0.5) {\(g(\delta)\)};
		\node [style=lmat] (99) at (-2.75, -0.5) {\(k(\delta)\)};
\end{pgfonlayer}
\begin{pgfonlayer}{edgelayer}
		\draw [style=background] (53.center)
			 to (52.center)
			 to (51.center)
			 to (54.center)
			 to cycle;
		\draw (64.center) to (67);
		\draw (66) to (63.center);
		\draw (66)
			 to [in=180, out=45] (61.center)
			 to (94.center)
			 to [in=135, out=0] (67);
		\draw (67)
			 to [in=0, out=-135] (95.center)
			 to (65.center)
			 to [in=-45, out=180] (66);
\end{pgfonlayer}
\end{tikzpicture}

=
\begin{tikzpicture}
\begin{pgfonlayer}{nodelayer}
		\node [style=none] (100) at (0.75, -1) {};
		\node [style=none] (101) at (10.5, -1) {};
		\node [style=none] (102) at (10.5, 1) {};
		\node [style=none] (103) at (0.75, 1) {};
		\node [style=none] (104) at (2.5, 0.5) {};
		\node [style=none] (106) at (0.75, 0) {};
		\node [style=none] (107) at (10.5, 0) {};
		\node [style=none] (108) at (2.5, -0.5) {};
		\node [style=bn] (109) at (1.75, 0) {};
		\node [style=wn] (110) at (6.5, 0) {};
		\node [style=none] (118) at (5.75, 0.5) {};
		\node [style=none] (119) at (5.75, -0.5) {};
		\node [style=rmat] (120) at (3.25, 0.5) {\(f(\delta)\)};
		\node [style=rmat] (121) at (3.25, -0.5) {\(h(\delta)\)};
		\node [style=rmat] (122) at (5, 0.5) {\(k(\delta)\)};
		\node [style=rmat] (123) at (5, -0.5) {\(g(\delta)\)};
		\node [style=lmat] (124) at (7.75, 0) {\(g(\delta)\)};
		\node [style=lmat] (125) at (9.5, 0) {\(k(\delta)\)};
\end{pgfonlayer}
\begin{pgfonlayer}{edgelayer}
		\draw [style=background] (102.center)
			 to (101.center)
			 to (100.center)
			 to (103.center)
			 to cycle;
		\draw (107.center) to (110);
		\draw (109) to (106.center);
		\draw (109)
			 to [in=180, out=45] (104.center)
			 to (118.center)
			 to [in=135, out=0] (110);
		\draw (110)
			 to [in=0, out=-135] (119.center)
			 to (108.center)
			 to [in=-45, out=180] (109);
\end{pgfonlayer}
\end{tikzpicture}

\]
\[
\begin{tikzpicture}
\begin{pgfonlayer}{nodelayer}
		\node [style=none] (143) at (-10, -1) {};
		\node [style=none] (144) at (-0.25, -1) {};
		\node [style=none] (145) at (-0.25, 1) {};
		\node [style=none] (146) at (-10, 1) {};
		\node [style=none] (147) at (-5.25, 0.5) {};
		\node [style=none] (148) at (-10, 0) {};
		\node [style=none] (149) at (-0.25, 0) {};
		\node [style=none] (150) at (-5.25, -0.5) {};
		\node [style=bn] (151) at (-6, 0) {};
		\node [style=wn] (152) at (-1.25, 0) {};
		\node [style=none] (153) at (-2, 0.5) {};
		\node [style=none] (154) at (-2, -0.5) {};
		\node [style=rmat] (155) at (-4.5, 0.5) {\(h(\delta)\)};
		\node [style=rmat] (156) at (-4.5, -0.5) {\(s(\delta)\)};
		\node [style=lmat] (157) at (-2.75, 0.5) {\(k(\delta)\)};
		\node [style=lmat] (158) at (-2.75, -0.5) {\(t(\delta)\)};
		\node [style=rmat] (159) at (-9, 0) {\(f(\delta)\)};
		\node [style=lmat] (160) at (-7.25, 0) {\(g(\delta)\)};
\end{pgfonlayer}
\begin{pgfonlayer}{edgelayer}
		\draw [style=background] (145.center)
			 to (144.center)
			 to (143.center)
			 to (146.center)
			 to cycle;
		\draw (149.center) to (152);
		\draw (151) to (148.center);
		\draw (151)
			 to [in=180, out=45] (147.center)
			 to (153.center)
			 to [in=135, out=0] (152);
		\draw (152)
			 to [in=0, out=-135] (154.center)
			 to (150.center)
			 to [in=-45, out=180] (151);
\end{pgfonlayer}
\end{tikzpicture}

=
\begin{tikzpicture}
\begin{pgfonlayer}{nodelayer}
		\node [style=none] (126) at (0.75, -1) {};
		\node [style=none] (127) at (9.75, -1) {};
		\node [style=none] (128) at (9.75, 1) {};
		\node [style=none] (129) at (0.75, 1) {};
		\node [style=none] (130) at (2.5, 0.5) {};
		\node [style=none] (132) at (0.75, 0) {};
		\node [style=none] (133) at (9.75, 0) {};
		\node [style=none] (134) at (2.5, -0.5) {};
		\node [style=bn] (135) at (1.75, 0) {};
		\node [style=wn] (136) at (8.75, 0) {};
		\node [style=none] (137) at (8, 0.5) {};
		\node [style=none] (138) at (8, -0.5) {};
		\node [style=rmat] (139) at (3.75, 0.5) {\(f(\delta)h(\delta)\)};
		\node [style=rmat] (140) at (3.75, -0.5) {\(f(\delta)s(\delta)\)};
		\node [style=lmat] (141) at (6.75, 0.5) {\(g(\delta)k(\delta)\)};
		\node [style=lmat] (142) at (6.75, -0.5) {\(g(\delta)t(\delta)\)};
\end{pgfonlayer}
\begin{pgfonlayer}{edgelayer}
		\draw [style=background] (128.center)
			 to (127.center)
			 to (126.center)
			 to (129.center)
			 to cycle;
		\draw (133.center) to (136);
		\draw (135) to (132.center);
		\draw (135)
			 to [in=180, out=45] (130.center)
			 to (137.center)
			 to [in=135, out=0] (136);
		\draw (136)
			 to [in=0, out=-135] (138.center)
			 to (134.center)
			 to [in=-45, out=180] (135);
\end{pgfonlayer}
\end{tikzpicture}

\]
\[
\begin{tikzpicture}
\begin{pgfonlayer}{nodelayer}
		\node [style=none] (178) at (-11.5, -0.5) {};
		\node [style=none] (179) at (-5, -0.5) {};
		\node [style=none] (180) at (-5, 0.5) {};
		\node [style=none] (181) at (-11.5, 0.5) {};
		\node [style=none] (183) at (-11.5, 0) {};
		\node [style=none] (184) at (-5, 0) {};
		\node [style=rmat] (194) at (-9.75, 0) {\(h(\delta)f(\delta)\)};
		\node [style=lmat] (196) at (-7, 0) {\(h(\delta)g(\delta)\)};
\end{pgfonlayer}
\begin{pgfonlayer}{edgelayer}
		\draw [style=background] (180.center)
			 to (179.center)
			 to (178.center)
			 to (181.center)
			 to cycle;
		\draw (183.center) to (184.center);
\end{pgfonlayer}
\end{tikzpicture}

=
\begin{tikzpicture}
\begin{pgfonlayer}{nodelayer}
		\node [style=none] (199) at (-4, -0.5) {};
		\node [style=none] (200) at (-0.25, -0.5) {};
		\node [style=none] (201) at (-0.25, 0.5) {};
		\node [style=none] (202) at (-4, 0.5) {};
		\node [style=none] (203) at (-4, 0) {};
		\node [style=none] (204) at (-0.25, 0) {};
		\node [style=rmat] (206) at (-3, 0) {\(f(\delta)\)};
		\node [style=lmat] (208) at (-1.25, 0) {\(g(\delta)\)};
\end{pgfonlayer}
\begin{pgfonlayer}{edgelayer}
		\draw [style=background] (201.center)
			 to (200.center)
			 to (199.center)
			 to (202.center)
			 to cycle;
		\draw (203.center) to (204.center);
\end{pgfonlayer}
\end{tikzpicture}

=
\begin{tikzpicture}
\begin{pgfonlayer}{nodelayer}
		\node [style=none] (209) at (0.75, -0.5) {};
		\node [style=none] (210) at (8, -0.5) {};
		\node [style=none] (211) at (8, 0.5) {};
		\node [style=none] (212) at (0.75, 0.5) {};
		\node [style=none] (213) at (0.75, 0) {};
		\node [style=none] (214) at (8, 0) {};
		\node [style=rmat] (215) at (1.75, 0) {\(f(\delta)\)};
		\node [style=lmat] (216) at (3.5, 0) {\(g(\delta)\)};
		\node [style=rmat] (217) at (5.25, 0) {\(h(\delta)\)};
		\node [style=lmat] (218) at (7, 0) {\(k(\delta)\)};
\end{pgfonlayer}
\begin{pgfonlayer}{edgelayer}
		\draw [style=background] (211.center)
			 to (210.center)
			 to (209.center)
			 to (212.center)
			 to cycle;
		\draw (213.center) to (214.center);
\end{pgfonlayer}
\end{tikzpicture}

=
\begin{tikzpicture}
\begin{pgfonlayer}{nodelayer}
		\node [style=none] (220) at (9, -0.5) {};
		\node [style=none] (221) at (15.75, -0.5) {};
		\node [style=none] (222) at (15.75, 0.5) {};
		\node [style=none] (223) at (9, 0.5) {};
		\node [style=none] (224) at (9, 0) {};
		\node [style=none] (225) at (15.75, 0) {};
		\node [style=rmat] (226) at (11, 0) {\(f(\delta)h(\delta\)};
		\node [style=lmat] (227) at (13.75, 0) {\(g(\delta)k(\delta)\)};
\end{pgfonlayer}
\begin{pgfonlayer}{edgelayer}
		\draw [style=background] (222.center)
			 to (221.center)
			 to (220.center)
			 to (223.center)
			 to cycle;
		\draw (224.center) to (225.center);
\end{pgfonlayer}
\end{tikzpicture}

\]
\end{lemma}
Therefore rational function multipliers are well-defined, and thus, so are directed \(H\)-boxes.  By \Cref{ax:h_flip} it follows that undirected \(H\)-boxes are also well-defined.

Finally:

\begin{lemma}
   Generalised spiders are well-defined.
\end{lemma}
\begin{proof}
   Because generalised spiders are defined by case distinction, we need to show that both ways of forming spiders are compatible with each other.

   First, given \(f(\delta),g(\delta), a(\delta), b(\delta)\in \F_p(\delta)\) and \(x,y\in \F_p\):
   \begin{align*}
      &\input{dzx/axiom_lemmas/spiders/0/00.tikz}
      \steq{\defref{def:dzx_shifted_spider}} \input{dzx/axiom_lemmas/spiders/0/0.tikz}\\
      &\steq{\axiomref{Z.0}}
      \input{dzx/axiom_lemmas/spiders/0/2.tikz}
      \steq{\defref{def:dzx_directed_fourier_box}\\ \lemref{lem:box_inverse}\\ \axiomref{Z.2}} \input{dzx/axiom_lemmas/spiders/0/3.tikz}\\
      &\steq{\axiomref{Z.4}} \input{dzx/axiom_lemmas/spiders/0/4.tikz}
      \steq{\defref{def:dzx_poly_multiplier}} \input{dzx/axiom_lemmas/spiders/0/5.tikz}\\
      &\steq{\defref{def:dzx_directed_fourier_box}} \input{dzx/axiom_lemmas/spiders/0/55.tikz}
      \steq{\defref{def:dzx_shifted_spider}} \input{dzx/axiom_lemmas/spiders/0/6.tikz}
   \end{align*}

   Using the definition of generalised red spiders and \lemref{lem:box_inverse}, this implies that:
   \[


   \]

   Now, suppose that \(f(\delta)+f(\delta^{\minu 1}) + x\) and  \(g(\delta)+g(\delta^{\minu 1}) + y\) are both nonzero. We have that:

   \begin{align*}
      &%

   \end{align*}

   As before, this implies the colour swapped equation:

   \[
      %

   \]
   
   Together, these imply that generalised spiders are well-defined.
\end{proof}

Therefore, it follows that:

\propDZXWellDefined*

\end{textAtEnd}

\subsection{Completeness}
\label{subsec:dzx_scalable_normal_form}
In this subsection, we prove that \(\dStabZX\) is sound, universal and complete with respect to its interpretation into \(\sALR\), generalising the proof technique used in \cite{gsa}. To prove completeness, we introduce scalable notation which allows us to bundle multiple wires together:

\begin{definition}
  \label{eq:strict}
  Given any \(n \in \N\), we introduce a ``thick'' wire
  \[


  \]
  which represents \(n\) wires bundled together; alongside generators which allow for wires to be split apart and merged back together, so that for all \(n,m\in \N\):\footnote{Formally, these are the components of the (op)laxators from \(\dStabZX\) into a suitable choice of strictification.}
  \begin{equation*}
      %
  \end{equation*}

\end{definition}

We define scalable spiders whose affine phases are given by vectors and whose symplectic phases are given by Hermitian matrices:

\begin{definition}
  \label{def:scalable_shifted_graph}
   Given a Hermitian matrix \(G\in \Herm\) and a vector \(\vb{a}\in \F_p(\delta)^n\),  the \textbf{scalable shifted graph state} diagram
   \(%
\)
   associated to \(G\) and \(\vb a\) is constructed by:
   \begin{enumerate}
   \item  first, assembling the \(\dStabZX\) diagram \(\bigotimes_{j=0}^{n-1} %
\);
   \item second, for each \( 0\leq j<k < n\) connecting the \(j\)th and \(k\)th wires with the following \(\dStabZX\) diagrams:
   \[
         %
;
   \]
   \item 
      and then finally bundling all of the wires together.
   \end{enumerate}
   Scalable spiders with different numbers of input and output wires are defined in the obvious way (see \Cref{def:scalable_spider}).
\end{definition}
The scalable notation allows for the connectivity information of graph states to be recovered in terms of their adjacency matrices, making the proof of completeness more tractable, and less combinatorial.
\begin{example}
  \label{ex:scalable_normal}
   \[
   %

   \]
\end{example}

Not all \(\dStabZX\)-diagrams are equivalent to graph states, however, they can be rewritten to the following form, consisting of a graph state where not all edges are necessarily connected to the output:

\begin{definition}
  \label{def:graph_like_form}
A \(\dStabZX\)-diagram is in \textbf{graph-like form} when
it can be decomposed into a graph-state with affine shift where some of
 the output wires are capped off by some number of green effects \(%
\]
\end{definition}

\begin{lemma}\label{lem:dzx_graph_form}
All \(\dStabZX\)-diagrams can be rewritten to graph-like form.
\end{lemma}
\begin{proof}
All of the non-delayed components can be put into graph-like form by \cite[Prop.~34]{gsa}; moreover, the delay can be rewritten to a directed Hadamard edge as follows:
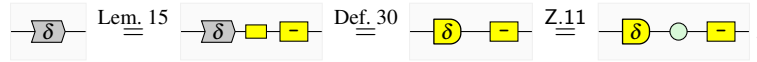
\[
%
.
\]
\end{proof}

We can refine graph-like diagrams so that the nodes which are not connected to the output are not allowed to be connected to each other:

\begin{definition}\label{def:dzx_ap_form}
  A graph-like \(\dStabZX\)-diagram is in \textbf{AP-form} when it takes the form
  \begin{equation}
    \label{eq:dzx-AP-form}
    %

  \end{equation}
  for \(m,n \in \N\), \(\vb{x} \in \F_p(\delta)^m\), \(\vb{y} \in \F_p(\delta)^n\), \(E
  \in \Matrices[m][n][\F_p(\delta)]\), and \(Y \in \Herm[n]\).
\end{definition}

To reduce \(\dStabZX\)-diagrams to AP-form, we shall use the following rule, which generalises local complementation and pivoting to the translation-invariant setting:
\begin{restatable}[Schur complementation]{proposition}{propDZXSchurComplementation}\label{prop:dzx_schur_complementation}
Given \(Y \in \Herm[n]\), invertible \(X \in \Herm[m]\), \(E \in \Matrices\),
  \(\vb{x} \in \F^m\) and \(\vb{y} \in \F^n\),
\[\begin{tikzpicture}
	\begin{pgfonlayer}{nodelayer}
		\node [style=none] (0) at (-10.3, 1) {};
		\node [style=none] (1) at (-1, 1) {};
		\node [style=none] (2) at (-1, -1) {};
		\node [style=none] (3) at (-10.3, -1) {};
		\node [style=gbn] (4) at (-5.75, 0) {};
		\node [style=none] (5) at (-3.5, 0) {};
		\node [style=gbn] (6) at (-1.5, 0.5) {};
		\node [style=none] (7) at (-1, -0.5) {};
		\node [style=none] (8) at (-2, -0.5) {};
		\node [style=wirelable] (10) at (-4.625, 0.25) {$m+n$};
		\node [style=none] (11) at (-2, 0.5) {};
		\node [style=wirelable] (12) at (-2, -0.25) {$n$};
		\node [style=wirelable] (13) at (-2, 0.25) {$m$};
		\node [style=bphase] (9) at (-8.3, 0) {$\phasemat{\vb{x} \\ \vb{y}}{X &  E \\ E^\dag & Y}$};
	\end{pgfonlayer}
	\begin{pgfonlayer}{edgelayer}
		\draw [style=background] (2.center)
			 to (1.center)
			 to (0.center)
			 to (3.center)
			 to [in=180, out=0] cycle;
		\draw [style=thick] (7.center)
			 to (8.center)
			 to [in=0, out=-180] (5.center)
			 to (4);
		\draw [style=thick] (6)
			 to (11.center)
			 to [in=0, out=180] (5.center)
			 to (4);
		\draw [style=rbphase] (9) to (4);
	\end{pgfonlayer}
\end{tikzpicture}
=\begin{tikzpicture}
	\begin{pgfonlayer}{nodelayer}
		\node [style=none] (17) at (1, 1) {};
		\node [style=none] (18) at (8.5, 1) {};
		\node [style=none] (19) at (8.5, -1) {};
		\node [style=none] (20) at (1, -1) {};
		\node [style=gbn] (15) at (4.5, -0.5) {};
		\node [style=none] (16) at (8.5, -0.5) {};
		\node [style=bphase] (21) at (4.5, 0.375) {$\phasemat{\vb{y} +  E X^{\minu 1}\vb{x}}{Y - E X^{\minu 1}E^\dag}$};
	\end{pgfonlayer}
	\begin{pgfonlayer}{edgelayer}
		\draw [style=background] (19.center)
			 to (18.center)
			 to (17.center)
			 to (20.center)
			 to [in=180, out=0] cycle;
		\draw [style=thick] (15) to (16.center);
		\draw [style=dbphase] (21) to (15);
	\end{pgfonlayer}
\end{tikzpicture}
\]
\end{restatable}
\seeproof{prop:dzx_schur_complementation}

\begin{restatable}{proposition}{propDZXAPForm}\label{prop:dzx_ap_form}
Every graph-like \(\dStabZX\)-diagram can be put into AP-form.
\end{restatable}
\seeproof{prop:dzx_ap_form}

The AP-form can be refined to a unique normal form using Gaussian elimination:

\begin{definition}\label{def:dzx_reduced_ap_form}
   A graph-like \(\dStabZX\)-diagram is in \textbf{reduced AP-form} either when it takes the form
   \(\begin{tikzpicture}
	\begin{pgfonlayer}{nodelayer}
		\node [style=none] (0) at (-5.8, 0.375) {};
		\node [style=none] (1) at (-2.375, 0.375) {};
		\node [style=none] (2) at (-2.375, -0.375) {};
		\node [style=none] (3) at (-5.8, -0.375) {};
		\node [style=ggn] (8) at (-3.375, 0) {};
		\node [style=bn] (20) at (-4.125, 0) {};
		\node [style=bn] (20) at (-4.125, 0) {};
		\node [style=wirelable] (12) at (-2.75, 0.225) {$n$};
		\node [style=none] (15) at (-2.375, 0) {};
		\node [style=gphase] (21) at (-5.125, 0) {\(\phasemat{1}{0}\)};
	\end{pgfonlayer}
	\begin{pgfonlayer}{edgelayer}
		\draw [style=background] (2.center)
			 to (1.center)
			 to (0.center)
			 to (3.center)
			 to [in=180, out=0] cycle;
		\draw [style=thick] (8) to (15.center);
		\draw [style=rbphase] (21) to (20);
	\end{pgfonlayer}
\end{tikzpicture}
\)
   or when it takes the following form
   
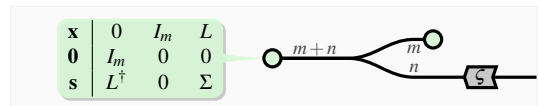
\[\begin{tikzpicture}
	\begin{pgfonlayer}{nodelayer}
		\node [style=none] (0) at (-4.925, 1.375) {};
		\node [style=none] (1) at (9, 1.375) {};
		\node [style=none] (2) at (9, -1.375) {};
		\node [style=none] (3) at (-4.925, -1.375) {};
		\node [style=gbn] (4) at (2, 0) {};
		\node [style=none] (5) at (4, 0) {};
		\node [style=gbn] (6) at (6.25, 0.5) {};
		\node [style=none] (8) at (5.75, -0.5) {};
		\node [style=bphase] (9) at (-1.5, 0) {$\phasemat{\vb{x} \\ \vb{0} \\ \vb{s}}{0 & I_m & L \\ I_m & 0 & 0 \\ L^\dag & 0 & \Sigma}$};
		\node [style=wirelable] (10) at (3.125, 0.25) {$m+n$};
		\node [style=none] (11) at (5.75, 0.5) {};
		\node [style=wirelable] (12) at (5.75, -0.25) {$n$};
		\node [style=wirelable] (13) at (5.75, 0.25) {$m$};
		\node [style=glmat] (14) at (7.5, -0.5) {\(\varsigma\)};
		\node [style=none] (15) at (9, -0.5) {};
	\end{pgfonlayer}
	\begin{pgfonlayer}{edgelayer}
		\draw [style=background] (2.center)
			 to (1.center)
			 to (0.center)
			 to (3.center)
			 to [in=180, out=0] cycle;
		\draw [style=thick] (4)
			 to (5.center)
			 to [in=180, out=0] (8.center)
			 to (14);
		\draw [style=thick] (4)
			 to (5.center)
			 to [in=-180, out=0] (11.center)
			 to (6);
		\draw [style=thick] (15.center) to (14);
		\draw [style=rbphase] (9) to (4);
	\end{pgfonlayer}
\end{tikzpicture}
\]
   for some
   \(0\leq m\leq n\),
   \(\vb{x}\in\F_p(\delta)^m\),
   \(\vb{s}\in\F_p(\delta)^{n-m}\),
   \(L\in \Matrices[m][(n-m)][\F_p(\delta)]\),
   \(\Sigma\in \Herm[n-m]\),
   and permutation
   \(\varsigma\in \Matrices[n][n][\F_p]\).
\end{definition}

Because the data of the reduced AP-form associated to a \(\dStabZX\)-diagram \(D\) is precisely that of its associated generating tableau given in \Cref{prop:generating_tableau}, it follows that:
\begin{lemma}\label{lem:AP_unique}
   The reduced AP-form of a \(\dStabZX\)-diagram is unique. Moreover, every shifted affine Lagrangian relation can be represented as a \(\dStabZX\)-diagram.
\end{lemma}

Moreover, every AP-form diagram can be rewritten to reduced AP-form:

\begin{restatable}{proposition}{propDZXReducedAPForm}\label{prop:dzx_reduced_ap_form}
Every  \(\dStabZX\)-diagram in AP-form can be rewritten to reduced AP-form.
\end{restatable}
\seeproof{prop:dzx_reduced_ap_form}

This implies immediately that:

\begin{theorem}\label{thm:dzx_completeness}
   The interpretation \(\interp{-}:\dStabZX\to \sALR\) is a \dag-CC equivalence, so that the delayed stabilizer ZX-calculus is sound, universal, and complete for shifted affine Lagrangian relations.
\end{theorem}

\subsection{Examples}
We end with examples demonstrating  how the delayed stabilizer ZX-calculus can be used to reason about quantum convolutional codes and lattice codes.

First, we recall the delayed controlled-X gate from
\Cref{ex:delayed_cx_ext,ex:delayed_cx_gamma}.
\begin{example}
  \label{ex:push_pauli_gate}
   If we want to push the Pauli-\(X\) gate past the delayed controlled-\(X\) gate, we need infinite time, as it propagates infinitely far into the future:
      \begin{align*}
      \input{dzx/examples/push_pauli_gate/CX_gun/0.tikz}
      &= \input{dzx/examples/push_pauli_gate/CX_gun/1.tikz}
      = \input{dzx/examples/push_pauli_gate/CX_gun/2.tikz}
      = \input{dzx/examples/push_pauli_gate/CX_gun/3.tikz}
      = \input{dzx/examples/push_pauli_gate/CX_gun/4.tikz} \\
      &= \input{dzx/examples/push_pauli_gate/CX_gun/5.tikz}
      = \overset{\infty}{\cdots} = \input{dzx/examples/push_pauli_gate/CX_gun/6.tikz}
      \end{align*}

   However, using the axioms of \(\dStabZX\), the delayed controlled-X gate can be simplified into a rational multiplier:
      \begin{align*}
      \input{dzx/examples/push_pauli_gate/CX_calculation/0.tikz}
      &= \input{dzx/examples/push_pauli_gate/CX_calculation/1.tikz}
      = \input{dzx/examples/push_pauli_gate/CX_calculation/2.tikz}
      = \input{dzx/examples/push_pauli_gate/CX_calculation/3.tikz}
      = \input{dzx/examples/push_pauli_gate/CX_calculation/4.tikz} \\
      &= \input{dzx/examples/push_pauli_gate/CX_calculation/5.tikz}
      = \input{dzx/examples/push_pauli_gate/CX_calculation/6.tikz}
      = \input{dzx/examples/push_pauli_gate/CX_calculation/7.tikz}
      = \input{dzx/examples/push_pauli_gate/CX_calculation/8.tikz}
      = \input{dzx/examples/push_pauli_gate/CX_calculation/9.tikz}
      \end{align*}
   Which allows us to easily compute its stabilizers using only finite reasoning:
      \begin{align*}
      \input{dzx/examples/push_pauli_gate/CX_commutation/0.tikz}
      &= \input{dzx/examples/push_pauli_gate/CX_commutation/1.tikz}
      = \input{dzx/examples/push_pauli_gate/CX_commutation/2.tikz} \\
      &= \input{dzx/examples/push_pauli_gate/CX_commutation/3.tikz}
      = \input{dzx/examples/push_pauli_gate/CX_commutation/4.tikz}
      = \input{dzx/examples/push_pauli_gate/CX_commutation/5.tikz}
      \end{align*}
\end{example}

Next, we return to the delayed graph state from
\Cref{ex:delayed_graph_state_ext,ex:delayed_graph_state_gamma}:
\begin{example} The tableaux of a delayed graph-state becomes explicit after reduction to the normal form:
  
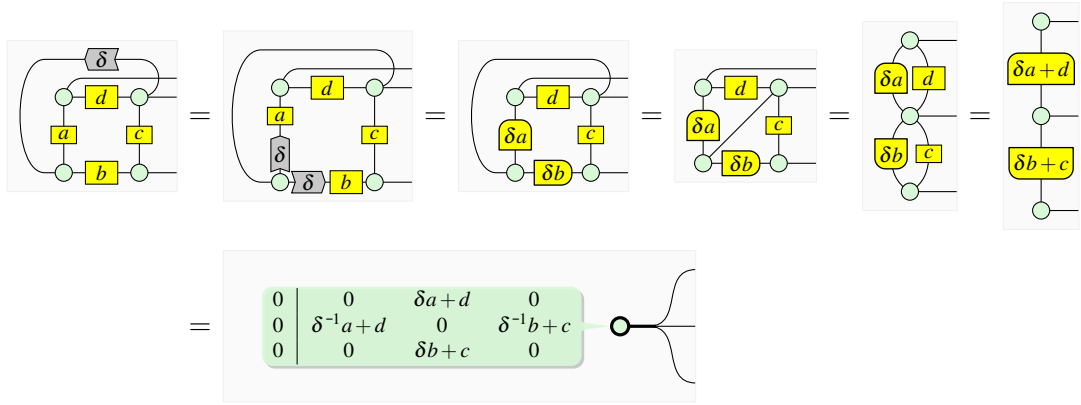
\label{ex:tableaux_delayed_graph}
   \begin{align*}
   \input{dzx/examples/tableaux_delayed_graph/d_graph_calculation/0.tikz}
   &= \input{dzx/examples/tableaux_delayed_graph/d_graph_calculation/1.tikz}
   = \input{dzx/examples/tableaux_delayed_graph/d_graph_calculation/2.tikz}
   = \input{dzx/examples/tableaux_delayed_graph/d_graph_calculation/3.tikz}
   = \input{dzx/examples/tableaux_delayed_graph/d_graph_calculation/4.tikz}
   = \input{dzx/examples/tableaux_delayed_graph/d_graph_calculation/5.tikz} \\
   &= \input{dzx/examples/tableaux_delayed_graph/d_graph_calculation/6.tikz}
   \end{align*}
    This unrolls to an infinite graph state built from tiling square plaquettes in a line:
    
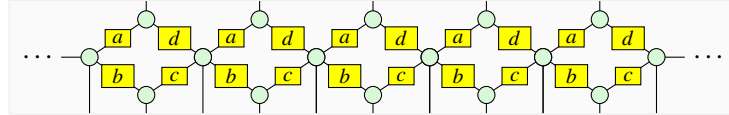
\[\begin{tikzpicture}
	\begin{pgfonlayer}{nodelayer}
		\node [style=none] (0) at (11.375, 1.5) {};
		\node [style=none] (1) at (11.375, -1.5) {};
		\node [style=none] (2) at (30.625, -1.5) {};
		\node [style=none] (3) at (30.625, 1.5) {};
		\node [style=bn] (4) at (15, 1) {};
		\node [style=bn] (5) at (15, -1) {};
		\node [style=bn] (6) at (16.5, 0) {};
		\node [style=bn] (7) at (13.5, 0) {};
		\node [style=had] (8) at (14.25, 0.5) {\(a\)};
		\node [style=had] (9) at (15.75, 0.5) {\(d\)};
		\node [style=had] (10) at (15.75, -0.5) {\(c\)};
		\node [style=had] (11) at (14.25, -0.5) {\(b\)};
		\node [style=none] (12) at (13.5, -1.5) {};
		\node [style=none] (13) at (15, 1.5) {};
		\node [style=none] (14) at (16.5, -1.5) {};
		\node [style=none] (15) at (15, -1.5) {};
		\node [style=bn] (16) at (18, 1) {};
		\node [style=bn] (17) at (18, -1) {};
		\node [style=bn] (18) at (19.5, 0) {};
		\node [style=bn] (19) at (16.5, 0) {};
		\node [style=had] (20) at (17.25, 0.5) {\(a\)};
		\node [style=had] (21) at (18.75, 0.5) {\(d\)};
		\node [style=had] (22) at (18.75, -0.5) {\(c\)};
		\node [style=had] (23) at (17.25, -0.5) {\(b\)};
		\node [style=none] (25) at (18, 1.5) {};
		\node [style=none] (26) at (19.5, -1.5) {};
		\node [style=none] (27) at (18, -1.5) {};
		\node [style=bn] (30) at (27, 1) {};
		\node [style=bn] (31) at (27, -1) {};
		\node [style=bn] (32) at (28.5, 0) {};
		\node [style=bn] (33) at (25.5, 0) {};
		\node [style=had] (34) at (26.25, 0.5) {\(a\)};
		\node [style=had] (35) at (27.75, 0.5) {\(d\)};
		\node [style=had] (36) at (27.75, -0.5) {\(c\)};
		\node [style=had] (37) at (26.25, -0.5) {\(b\)};
		\node [style=none] (38) at (27, 1.5) {};
		\node [style=none] (39) at (28.5, -1.5) {};
		\node [style=none] (40) at (27, -1.5) {};
		\node [style=none] (41) at (29.25, 0) {};
		\node [style=none] (42) at (12.75, 0) {};
		\node [style=none] (43) at (29.25, 0) {};
		\node [style=none] (44) at (12.75, 0) {};
		\node [style=none] (45) at (25.5, -1.5) {};
		\node [style=none] (50) at (29.875, 0) {\(\cdots\)};
		\node [style=none] (52) at (12.125, 0) {\(\cdots\)};
		\node [style=bn] (53) at (21, 1) {};
		\node [style=bn] (54) at (21, -1) {};
		\node [style=bn] (55) at (22.5, 0) {};
		\node [style=bn] (56) at (19.5, 0) {};
		\node [style=had] (57) at (20.25, 0.5) {\(a\)};
		\node [style=had] (58) at (21.75, 0.5) {\(d\)};
		\node [style=had] (59) at (21.75, -0.5) {\(c\)};
		\node [style=had] (60) at (20.25, -0.5) {\(b\)};
		\node [style=none] (61) at (19.5, -1.5) {};
		\node [style=none] (62) at (21, 1.5) {};
		\node [style=none] (63) at (22.5, -1.5) {};
		\node [style=none] (64) at (21, -1.5) {};
		\node [style=bn] (65) at (22.5, 0) {};
		\node [style=bn] (66) at (24, 1) {};
		\node [style=bn] (67) at (24, -1) {};
		\node [style=bn] (68) at (25.5, 0) {};
		\node [style=bn] (69) at (22.5, 0) {};
		\node [style=had] (70) at (23.25, 0.5) {\(a\)};
		\node [style=had] (71) at (24.75, 0.5) {\(d\)};
		\node [style=had] (72) at (24.75, -0.5) {\(c\)};
		\node [style=had] (73) at (23.25, -0.5) {\(b\)};
		\node [style=none] (74) at (22.5, -1.5) {};
		\node [style=none] (75) at (24, 1.5) {};
		\node [style=none] (76) at (25.5, -1.5) {};
		\node [style=none] (77) at (24, -1.5) {};
		\node [style=bn] (78) at (25.5, 0) {};
	\end{pgfonlayer}
	\begin{pgfonlayer}{edgelayer}
		\draw [style=background] (3.center)
			 to (2.center)
			 to (1.center)
			 to (0.center)
			 to cycle;
		\draw [style=background] (0.center) to (3.center);
		\draw [style=background] (3.center) to (2.center);
		\draw [style=background] (2.center) to (1.center);
		\draw [style=background] (1.center) to (0.center);
		\draw (7) to (4);
		\draw (4) to (6);
		\draw (6) to (5);
		\draw (5) to (7);
		\draw (12.center) to (7);
		\draw (15.center) to (5);
		\draw (14.center) to (6);
		\draw (13.center) to (4);
		\draw (19) to (16);
		\draw (16) to (18);
		\draw (18) to (17);
		\draw (17) to (19);
		\draw (27.center) to (17);
		\draw (26.center) to (18);
		\draw (25.center) to (16);
		\draw (33) to (30);
		\draw (30) to (32);
		\draw (32) to (31);
		\draw (31) to (33);
		\draw (40.center) to (31);
		\draw (39.center) to (32);
		\draw (38.center) to (30);
		\draw (42.center) to (7);
		\draw (32) to (41.center);
		\draw (45.center) to (33);
		\draw (56) to (53);
		\draw (53) to (55);
		\draw (55) to (54);
		\draw (54) to (56);
		\draw (61.center) to (56);
		\draw (64.center) to (54);
		\draw (63.center) to (55);
		\draw (62.center) to (53);
		\draw (69) to (66);
		\draw (66) to (68);
		\draw (68) to (67);
		\draw (67) to (69);
		\draw (74.center) to (69);
		\draw (77.center) to (67);
		\draw (76.center) to (68);
		\draw (75.center) to (66);
	\end{pgfonlayer}
\end{tikzpicture}
\]
\end{example}

By taking multiple delayed traces of a graph state, we obtain a surface code. For example, we can add another delayed edge to the previous example:
\begin{example}\label{ex:delayed_surface_code}
   Consider the following delayed graph state for \(r\in \N\):
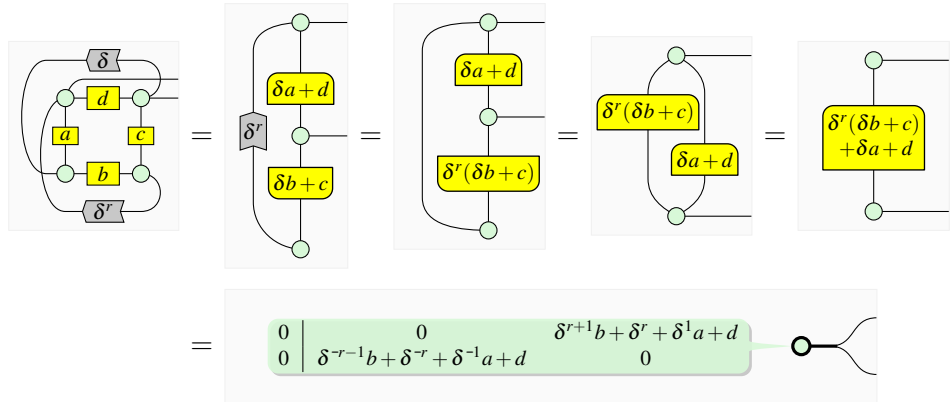
   \begin{align*}
   \input{dzx/examples/delayed_surface_code/0.tikz}
   &= \input{dzx/examples/delayed_surface_code/1.tikz}
   = \input{dzx/examples/delayed_surface_code/2.tikz}
   = \input{dzx/examples/delayed_surface_code/3.tikz}
   = \input{dzx/examples/delayed_surface_code/4.tikz} \\
   &= \input{dzx/examples/delayed_surface_code/5.tikz}
   \end{align*}

   This unrolls to an infinite surface code with square plaquettes wrapped around a cylinder of ``radius'' \(r\): the diamond-shaped finite delayed diagram determines the shape of each plaquette, while the delays determine how successive plaquettes are glued together around the cylinder. For \(r=6\) and \(a=b=c=d=1\), this is depicted in \Cref{fig:unrolled_surface_code}.

   \begin{figure}[h]
      \centering
      \input{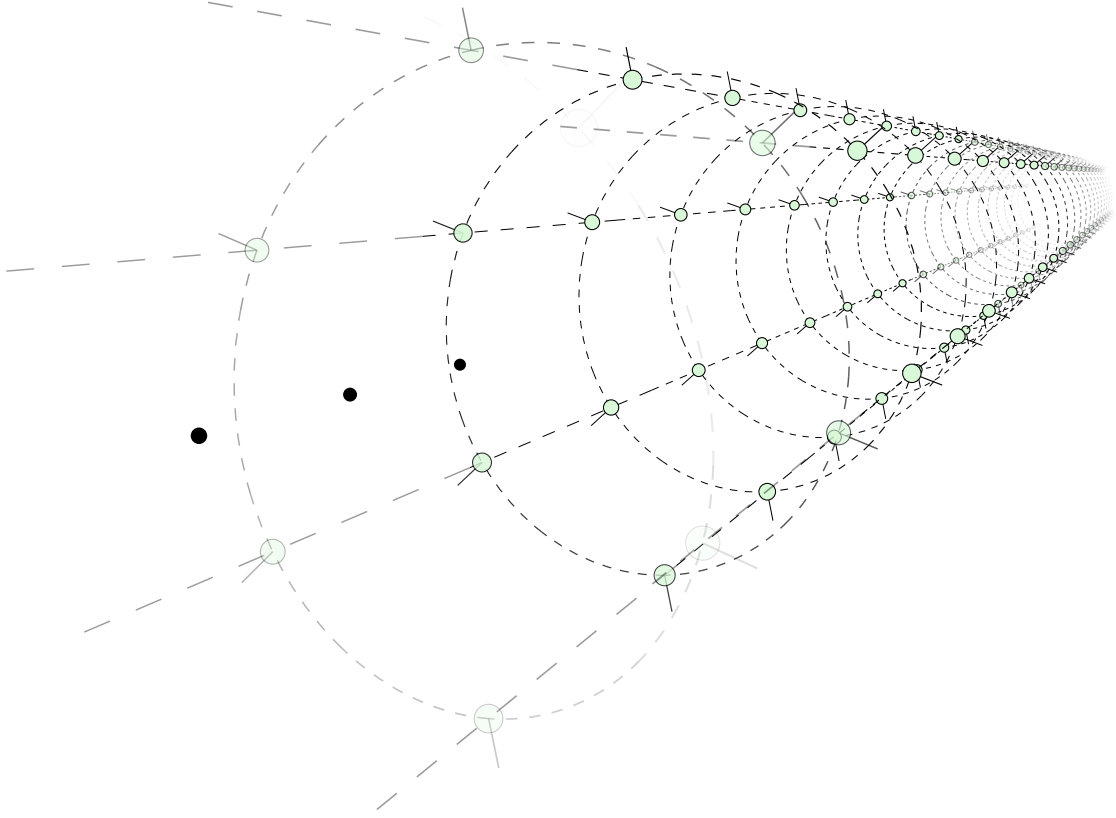}
      \caption{Infinitely unrolled surface code of \Cref{ex:delayed_surface_code}.\\ \noindent Dotted lines denote wires connected by a Hadamard gate.}
      \label{fig:unrolled_surface_code}
   \end{figure}

\end{example}

  \section*{Acknowledgements}
  This work has been partially funded by the European Union through the MSCA SE
  project QCOMICAL and within the framework of ``Plan France 2030'', under
  the research projects EPIQ ANR-22-PETQ-0007 and HQI-R\&D ANR-22-PNCQ-0002.

  \newpage
  \bibliographystyle{eptcs}
  \bibliography{style/bibliography}
  
  \newpage 
  \appendix

  \ifappendixproofs
  \section{Proofs of \Cref{sec:mixed}}
  \IfFileExists{\jobname-pratendmixed.tex}{%
    \printProofs[mixed]%
  }{}

  \section{Proofs of \Cref{sec:salr}}
  \IfFileExists{\jobname-pratendsalr.tex}{%
    \printProofs[salr]%
  }{}

  \section{Proofs of \Cref{sec:dzx}}
  \IfFileExists{\jobname-pratenddzx.tex}{%
    \printProofs[dzx]%
  }{}
  \fi

  \begingroup
  \BeforeBeginEnvironment{definition}{\par}
  \BeforeBeginEnvironment{theorem}{\par}
  \BeforeBeginEnvironment{lemma}{\par}
  \BeforeBeginEnvironment{proposition}{\par}
  \BeforeBeginEnvironment{corollary}{\par}
  
\subsection{Graphical affine algebra}

\begin{figure}[h]
\centering
\[\input{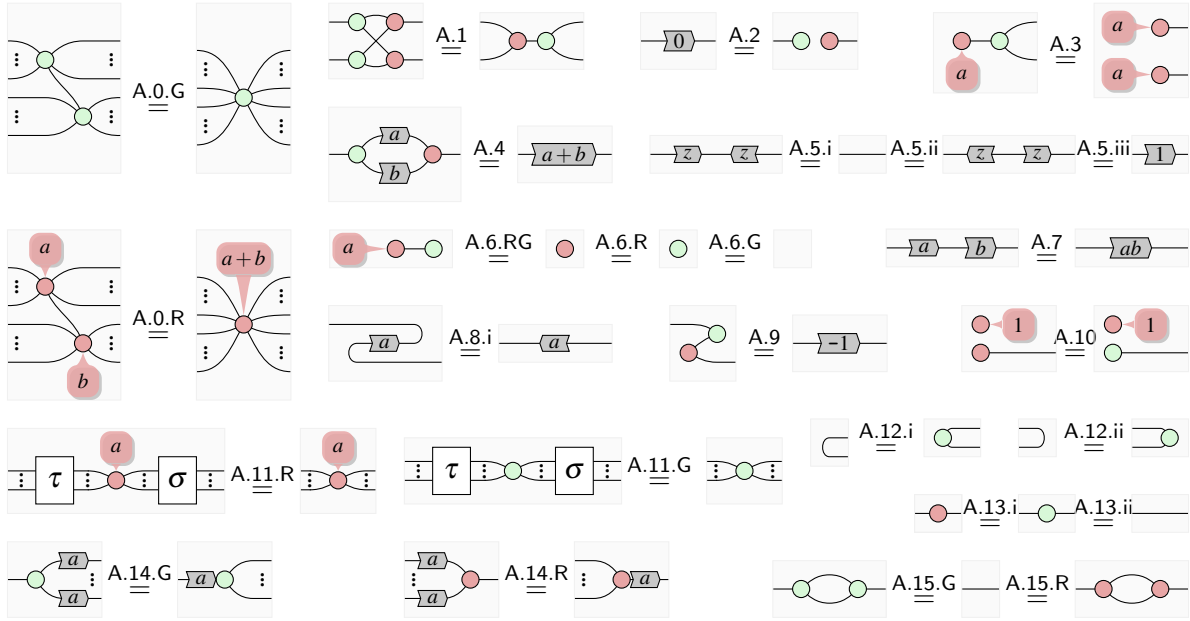}\]
\caption{Axioms of the Graphical Affine Algebra~\cite{gsa} (GAA), where \(a, b,z\in \mathbb{K}\) with \(z\neq 0\).}
\label{fig:gaa_axioms}%
\end{figure}

It is easy to see that:

\begin{lemma}
  \label{lem:gaa_functor}%
  The interpretation of the scalable generators defines a strict symmetric
  monoidal functor \(\mathsf{GAA} \to \StabZX\).
\end{lemma}

Therefore, we can work with rational function labelled multipliers as we would expect.

\subsection{Miscellaneous non-scalable lemmas}
In this appendix, we prove various non-scalable lemmas in \(\dStabZX\) to help us prove completeness.

\begin{lemma}
  \label{lem:box_inverse}
  It follows immediately from \axiomref{D.20} and \axiomref{D.11} that for any invertible \(z\in \F_p(\delta)\):
  \begin{align*}
      \begin{tikzpicture}
	\begin{pgfonlayer}{nodelayer}
		\node [style=none] (0) at (1.25, 0) {};
		\node [style=rhad] (2) at (-0.75, 0) {$\minu z$};
		\node [style=none] (3) at (-1.5, 0) {};
		\node [style=lhad] (4) at (0.5, 0) {$z$};
		\node [style=none] (bg0) at (-1.5, 0.5) {};
		\node [style=none] (bg1) at (1.25, 0.5) {};
		\node [style=none] (bg2) at (1.25, -0.5) {};
		\node [style=none] (bg3) at (-1.5, -0.5) {};
	\end{pgfonlayer}
	\begin{pgfonlayer}{edgelayer}
		\draw [style=background] (bg0.center)
			 to (bg1.center)
			 to (bg2.center)
			 to (bg3.center)
			 to cycle;
		\draw [in=180, out=0] (3.center) to (2);
		\draw [in=180, out=0] (2) to (4);
		\draw [in=180, out=0] (4) to (0.center);
	\end{pgfonlayer}
\end{tikzpicture}

    =
      \begin{tikzpicture}
	\begin{pgfonlayer}{nodelayer}
		\node [style=none] (7) at (1, 0) {};
		\node [style=none] (8) at (-1, 0) {};
		\node [style=none] (bg0) at (-1, 0.5) {};
		\node [style=none] (bg1) at (1, 0.5) {};
		\node [style=none] (bg2) at (1, -0.5) {};
		\node [style=none] (bg3) at (-1, -0.5) {};
	\end{pgfonlayer}
	\begin{pgfonlayer}{edgelayer}
		\draw [style=background] (bg0.center)
			 to (bg1.center)
			 to (bg2.center)
			 to (bg3.center)
			 to cycle;
		\draw (8.center) to (7.center);
	\end{pgfonlayer}
\end{tikzpicture}

  \end{align*}
\end{lemma}

\begin{lemma}
  \label{lem:box_product}
  It follows immediately from \axiomref{D.20} that for all \(z\in \F_p(\delta)\):
  \begin{align*}
      \begin{tikzpicture}
	\begin{pgfonlayer}{nodelayer}
		\node [style=none] (0) at (1.125, 0) {};
		\node [style=lhad] (1) at (0.375, 0) {$z$};
		\node [style=rhad] (2) at (-0.625, 0) {$z$};
		\node [style=none] (3) at (-1.375, 0) {};
		\node [style=none] (bg0) at (-1.375, 0.5) {};
		\node [style=none] (bg1) at (1.125, 0.5) {};
		\node [style=none] (bg2) at (1.125, -0.5) {};
		\node [style=none] (bg3) at (-1.375, -0.5) {};
	\end{pgfonlayer}
	\begin{pgfonlayer}{edgelayer}
		\draw [style=background] (bg0.center)
			 to (bg1.center)
			 to (bg2.center)
			 to (bg3.center)
			 to cycle;
		\draw [in=180, out=0] (3.center) to (2);
		\draw [in=180, out=0] (2) to (1);
		\draw [in=180, out=0] (1) to (0.center);
	\end{pgfonlayer}
\end{tikzpicture}

    =
      \begin{tikzpicture}
	\begin{pgfonlayer}{nodelayer}
		\node [style=none] (5) at (-1, 0) {};
		\node [style=none] (6) at (1, 0) {};
		\node [style=rmat] (7) at (0, 0) {\(\minu 1\)};
		\node [style=none] (bg0) at (-1, 0.5) {};
		\node [style=none] (bg1) at (1, 0.5) {};
		\node [style=none] (bg2) at (1, -0.5) {};
		\node [style=none] (bg3) at (-1, -0.5) {};
	\end{pgfonlayer}
	\begin{pgfonlayer}{edgelayer}
		\draw [style=background] (bg0.center)
			 to (bg1.center)
			 to (bg2.center)
			 to (bg3.center)
			 to cycle;
		\draw [in=180, out=0] (5.center) to (7);
		\draw [in=180, out=0] (7) to (6.center);
	\end{pgfonlayer}
\end{tikzpicture}

  \end{align*}
\end{lemma}

Similarly, it follows easily that:
\begin{lemma}
  \label{lem:compact_antipode}
    \begin{align*}
      \input{gsa_completeness/miscellaneous_non_scalable/lemmas/compact_antipode/statement/0.tikz}
    =
      \input{gsa_completeness/miscellaneous_non_scalable/lemmas/compact_antipode/statement/1.tikz},
      \quad
      \input{gsa_completeness/miscellaneous_non_scalable/lemmas/compact_antipode/statement/2.tikz}
    =
      \input{gsa_completeness/miscellaneous_non_scalable/lemmas/compact_antipode/statement/3.tikz},
      \quad
      \input{gsa_completeness/miscellaneous_non_scalable/lemmas/compact_antipode/statement/4.tikz}
    =
      \input{gsa_completeness/miscellaneous_non_scalable/lemmas/compact_antipode/statement/5.tikz},
      \quad\text{and}\quad
      \input{gsa_completeness/miscellaneous_non_scalable/lemmas/compact_antipode/statement/6.tikz}
    =
      \input{gsa_completeness/miscellaneous_non_scalable/lemmas/compact_antipode/statement/7.tikz}.
  \end{align*}
\end{lemma}

\begin{lemma}
  \label{lem:box_identity}
  For any invertible \(z \in \F_p(\delta)\)
  \begin{align*}


  \end{align*}
\end{proof}

It follows from a simple induction argument that:
\begin{lemma}
  \label{lem:antipode_spider_symplectic}
  For any \(a,b \in \F_p(\delta)\):
    \begin{align*}
      %

  \end{align*}
\end{lemma}

\begin{lemma}
  \label{lem:colour_inverted}
  It follows from \axiomref{D.20} and \Cref{lem:box_inverse} that for any invertible
  \(z \in \F_p(\delta)\),
  \begin{align*}
      %
 
  \end{align*}
\end{lemma}

\begin{lemma}
  \label{lem:copy_swapped}
  A colour-swapped version of \axiomref{Z.7} is derivable: for any \(a \in \F_p(\delta)\),
  \begin{align*}
      %

  \end{align*}
\end{proof}

It follows from a simple induction proof that:
\begin{proposition}
  \label{gsa:prop:zero_normal_forms}
  For any \(m,n \in \N\) and \(D \in \ZX(m,n)\) and invertible \(a \in \F_p(\delta)\):
  \begin{align*}
      %

  \end{align*}
\end{proposition}

\begin{lemma}
  \label{lem:box_antipode}
  Boxes and antipodes commute: for any invertible \(z \in \F_p(\delta)\),
  \begin{align*}
      %

  \end{align*}
\end{proof}

\begin{lemma}
  \label{lem:box_opposite_inverse}
  The opposite-inverse box can be written as follows: for any invertible \(z \in \F_p(\delta)\),
  \begin{align*}
      %

  \end{align*}
\end{lemma}

\begin{proof}
  \begin{align*}
    &
      \input{gsa_completeness/miscellaneous_non_scalable/lemmas/box_opposite_inverse/proof/0.tikz}
    \steq{\axiomref{D.14}}
      \input{gsa_completeness/miscellaneous_non_scalable/lemmas/box_opposite_inverse/proof/1.tikz}
    \steq{\lemref{lem:box_product}}
      \input{gsa_completeness/miscellaneous_non_scalable/lemmas/box_opposite_inverse/proof/2.tikz} 
    \steq{\lemref{lem:box_antipode}}
      \input{gsa_completeness/miscellaneous_non_scalable/lemmas/box_opposite_inverse/proof/3.tikz}
    \\
    &\steq{\axiomref{D.1}}
      \input{gsa_completeness/miscellaneous_non_scalable/lemmas/box_opposite_inverse/proof/4.tikz}
    \steq{\axiomref{A.7} \\ \axiomref{A.13}}
      \input{gsa_completeness/miscellaneous_non_scalable/lemmas/box_opposite_inverse/proof/5.tikz}
  \end{align*}
\end{proof}

\begin{lemma}
  \label{lem:box_swapped}
  For any invertible \(z \in \F_p(\delta+\delta^{\minu 1})\), 
  \begin{align*}


  \end{align*}
\end{lemma}

\begin{proof}
  The first equation  is exactly \(\axiomref{D.19}\). For the second equation, observe:
    \begin{align*}
      &
        %

    \end{align*}
\end{proof}

\begin{lemma} 
  \label{lem:push_pauli_state}
  Strictly-affine red states absorb arbitrary phases and vice-versa:
  \begin{align*}
      %

  \end{align*}
\end{proof}

\begin{lemma}
  \label{lem:symplectic_states}
  States with non-zero symplectic part can all be represented using both
  green and red spiders: for any \(a \in \F_p(\delta)\) and invertible \(z \in \F_p(\delta+\delta^{\minu 1})\),
  \begin{align*}
      %

  \end{align*}
\end{lemma}

\begin{proof}
  It suffices to prove the first claim, as this immediately implies the other:
  \begin{align*}
    &
      \input{gsa_completeness/miscellaneous_non_scalable/lemmas/symplectic_states/proof/0.tikz}
    \steq{\axiomref{D.20}} 
      \input{gsa_completeness/miscellaneous_non_scalable/lemmas/symplectic_states/proof/1.tikz} 
    \steq{\lemref{lem:box_swapped}}
      \input{gsa_completeness/miscellaneous_non_scalable/lemmas/symplectic_states/proof/2.tikz} 
    \steq{\axref{ax:dzx_fusion}}
      \input{gsa_completeness/miscellaneous_non_scalable/lemmas/symplectic_states/proof/3.tikz}
    \\
    &\steq{\lemref{lem:push_pauli_state}}
      \input{gsa_completeness/miscellaneous_non_scalable/lemmas/symplectic_states/proof/4.tikz} 
    \steq{\axref{ax:dzx_fusion}}
      \input{gsa_completeness/miscellaneous_non_scalable/lemmas/symplectic_states/proof/5.tikz}
  \end{align*}
  The second equation follows from an analogous argument.
\end{proof}

\begin{lemma}
  \label{lem:pauli_push}
  Red spiders with arbitrary phases copy affine green spiders and vice-versa:
  \[


  \]
\end{lemma}

\begin{proof}
  First of all,
  \begin{align*}
      %

  \end{align*}
  Then, we separate the equation into two cases based on whether the green spider
  is affine or not.
  In case $d = 0$, the green spider is affine and therefore:
  \begin{align*}
      %

  \end{align*}
  Note that if \(d = 0\), then \(ad + c = c\) and so the lemma holds.
  Otherwise, \(d \neq 0\) and therefore \(d^{\minu 1}\) exists, so we can apply
  the state-change lemma:
  \begin{align*}
    &
      %

  \end{align*}
  We can prove the second equation of the lemma using Hadamard-boxes as
  follows:
  \begin{align*}
    &
      \input{gsa_completeness/miscellaneous_non_scalable/lemmas/pauli_push/proof/eq3/0.tikz}
    \steq{\axref{ax:dzx_colour}}
      \input{gsa_completeness/miscellaneous_non_scalable/lemmas/pauli_push/proof/eq3/1.tikz}
    \steq{\lemref{lem:box_inverse}\\ \axiomref{D.20}}
      \input{gsa_completeness/miscellaneous_non_scalable/lemmas/pauli_push/proof/eq3/2.tikz} 
    \steq{\lemref{lem:pauli_push}}
      \input{gsa_completeness/miscellaneous_non_scalable/lemmas/pauli_push/proof/eq3/3.tikz}
    \\
    &\steq{\lemref{lem:box_inverse}}
      \input{gsa_completeness/miscellaneous_non_scalable/lemmas/pauli_push/proof/eq3/4.tikz} 
    \steq{\axref{ax:dzx_colour}}
      \input{gsa_completeness/miscellaneous_non_scalable/lemmas/pauli_push/proof/eq3/5.tikz}
  \end{align*}
\end{proof}

\subsection{Scalable generators}\label{sec:scalable_generators}

In this appendix, we introduce scalable notation for \(\dStabZX\) to help us prove completeness.

\begin{definition}
  \label{def:matrix_arrow}

  \textbf{Scalable matrix arrows} are defined by induction:

  
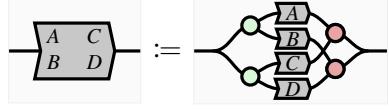
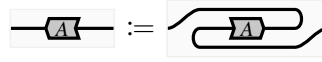
\begin{equation*}
      \begin{tikzpicture}
	\begin{pgfonlayer}{nodelayer}
		\node [style=none] (0) at (0.5, 1.375) {};
		\node [style=none] (1) at (4, 1.375) {};
		\node [style=none] (2) at (4, -1.375) {};
		\node [style=none] (3) at (0.5, -1.375) {};
		\node [style=grmat] (9) at (2.25, 0) {$\begin{matrix}A & C \\ B & D\end{matrix}$};
		\node [style=none] (10) at (4, 0) {};
		\node [style=none] (11) at (0.5, 0) {};
	\end{pgfonlayer}
	\begin{pgfonlayer}{edgelayer}
		\draw [style=background] (2.center)
			 to (1.center)
			 to (0.center)
			 to (3.center)
			 to cycle;
		\draw [style=thick] (11.center) to (9);
		\draw [style=thick] (9) to (10.center);
	\end{pgfonlayer}
\end{tikzpicture}

    \coloneqq
      \begin{tikzpicture}
	\begin{pgfonlayer}{nodelayer}
		\node [style=none] (5) at (6, 1.4) {};
		\node [style=none] (6) at (11.125, 1.4) {};
		\node [style=none] (7) at (11.125, -1.4) {};
		\node [style=none] (8) at (6, -1.4) {};
		\node [style=none] (12) at (6, 0) {};
		\node [style=none] (13) at (6.5, 0) {};
		\node [style=ggn] (14) at (7.5, 0.675) {};
		\node [style=ggn] (15) at (7.5, -0.675) {};
		\node [style=grn] (16) at (9.75, 0.475) {};
		\node [style=grn] (17) at (9.75, -0.475) {};
		\node [style=grmat] (18) at (8.625, 1) {$A$};
		\node [style=grmat] (19) at (8.625, 0.325) {$B$};
		\node [style=grmat] (20) at (8.625, -0.325) {$C$};
		\node [style=grmat] (21) at (8.625, -1) {$D$};
		\node [style=none] (22) at (11.125, 0) {};
		\node [style=none] (23) at (10.625, 0) {};
		\node [style=grn] (24) at (9.75, 0.475) {};
		\node [style=grn] (25) at (9.75, -0.475) {};
		\node [style=none] (26) at (8.25, 1) {};
		\node [style=none] (27) at (9, 1) {};
		\node [style=none] (28) at (8.25, 0.325) {};
		\node [style=none] (29) at (9, 0.325) {};
		\node [style=none] (30) at (8.25, -0.325) {};
		\node [style=none] (31) at (9, -0.325) {};
		\node [style=none] (32) at (8.25, -1) {};
		\node [style=none] (33) at (9, -1) {};
	\end{pgfonlayer}
	\begin{pgfonlayer}{edgelayer}
		\draw [style=background] (6.center)
			 to (5.center)
			 to (8.center)
			 to (7.center)
			 to cycle;
		\draw [style=thick] (12.center)
			 to (13.center)
			 to [in=180, out=0] (15);
		\draw [style=thick, in=-180, out=0] (13.center) to (14);
		\draw [style=thick] (22.center)
			 to (23.center)
			 to [in=0, out=180] (25);
		\draw [style=thick, in=0, out=180] (23.center) to (24);
		\draw [style=very thick] (14)
			 to [in=-180, out=30, looseness=0.75] (26.center)
			 to [in=-180, out=30, looseness=0.75] (18.center)
			 to [in=-180, out=30, looseness=0.75] (27.center)
			 to [in=150, out=0, looseness=0.75] (16);
		\draw [style=very thick] (17)
			 to [in=0, out=135, looseness=0.75] (29.center)
			 to [in=0, out=135, looseness=0.75] (19.center)
			 to [in=0, out=135, looseness=0.75] (28.center)
			 to [in=-30, out=180, looseness=0.75] (14);
		\draw [style=very thick] (15)
			 to [in=-180, out=30, looseness=0.75] (30.center)
			 to [in=-180, out=30, looseness=0.75] (20.center)
			 to [in=-180, out=30, looseness=0.75] (31.center)
			 to [in=-135, out=0, looseness=0.75] (16);
		\draw [style=very thick] (17)
			 to [in=0, out=-150, looseness=0.75] (33.center)
			 to [in=0, out=-150, looseness=0.75] (21.center)
			 to [in=0, out=-150, looseness=0.75] (32.center)
			 to [in=-30, out=180, looseness=0.75] (15);
	\end{pgfonlayer}
\end{tikzpicture}

  \end{equation*}
  Where
  \begin{equation*}
      \begin{tikzpicture}
	\begin{pgfonlayer}{nodelayer}
		\node [style=none] (0) at (7.125, -0.5) {};
		\node [style=none] (1) at (9.875, -0.5) {};
		\node [style=none] (2) at (9.875, 0.5) {};
		\node [style=none] (3) at (7.125, 0.5) {};
		\node [style=none] (4) at (7.125, 0) {};
		\node [style=none] (5) at (9.875, 0) {};
		\node [style=glmat] (6) at (8.5, 0) {\(A\)};
	\end{pgfonlayer}
	\begin{pgfonlayer}{edgelayer}
		\draw [style=background] (0.center)
			 to (3.center)
			 to (2.center)
			 to (1.center)
			 to cycle;
		\draw [style=background] (2.center) to (1.center);
		\draw [style=background] (1.center) to (0.center);
		\draw [style=background] (0.center) to (3.center);
		\draw [style=background] (3.center) to (2.center);
		\draw [style=thick] (4.center) to (6);
		\draw [style=thick] (6) to (5.center);
	\end{pgfonlayer}
\end{tikzpicture}

    \coloneqq
      \begin{tikzpicture}
	\begin{pgfonlayer}{nodelayer}
		\node [style=none] (0) at (11.375, -0.75) {};
		\node [style=none] (1) at (15.625, -0.75) {};
		\node [style=none] (2) at (15.625, 0.75) {};
		\node [style=none] (3) at (11.375, 0.75) {};
		\node [style=none] (4) at (11.375, 0) {};
		\node [style=none] (5) at (15.625, 0) {};
		\node [style=rmatt] (6) at (13.5, 0) {\(A\)};
		\node [style=none] (7) at (14.75, 0) {};
		\node [style=none] (8) at (12.25, 0) {};
		\node [style=none] (9) at (14.75, 0.5) {};
		\node [style=none] (10) at (12.25, 0.5) {};
		\node [style=none] (11) at (12.25, -0.5) {};
		\node [style=none] (12) at (14.75, -0.5) {};
	\end{pgfonlayer}
	\begin{pgfonlayer}{edgelayer}
		\draw [style=background] (2.center)
			 to (1.center)
			 to (0.center)
			 to (3.center)
			 to cycle;
		\draw [style=background] (2.center) to (1.center);
		\draw [style=background] (1.center) to (0.center);
		\draw [style=background] (0.center) to (3.center);
		\draw [style=background] (3.center) to (2.center);
		\draw [style=thick] (4.center)
			 to [in=180, out=0, looseness=0.75] (10.center)
			 to (9.center)
			 to [bend left=90, looseness=1.25] (7.center)
			 to (6.center)
			 to (8.center)
			 to [bend right=90, looseness=1.25] (11.center)
			 to (12.center)
			 to [in=180, out=0, looseness=0.75] (5.center);
	\end{pgfonlayer}
\end{tikzpicture}

  \end{equation*}
\end{definition}
\begin{definition}
  \label{def:unlabelled_scalable}
  The unlabelled scalable H-box is defined by induction:

  
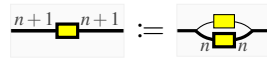
\begin{equation*}
      \begin{tikzpicture}
	\begin{pgfonlayer}{nodelayer}
		\node [style=none] (0) at (-4, 0) {};
		\node [style=none] (1) at (-1, 0) {};
		\node [style=ghad] (2) at (-2.5, 0) {};
		\node [style=none] (3) at (-4, 0.7) {};
		\node [style=none] (4) at (-1, 0.7) {};
		\node [style=none] (5) at (-1, -0.7) {};
		\node [style=none] (6) at (-4, -0.7) {};
		\node [style=wirelable] (23) at (-3.375, 0.3) {$n+1$};
		\node [style=wirelable] (24) at (-1.625, 0.3) {$n+1$};
	\end{pgfonlayer}
	\begin{pgfonlayer}{edgelayer}
		\draw [style=background] (5.center)
			 to (4.center)
			 to (3.center)
			 to (6.center)
			 to cycle;
		\draw [style=thick] (1.center)
			 to (2.center)
			 to (0.center);
		\draw [style=thick] (2) to (0.center);
		\draw [style=thick] (2) to (1.center);
	\end{pgfonlayer}
\end{tikzpicture}
    \coloneqq
      \begin{tikzpicture}
	\begin{pgfonlayer}{nodelayer}
		\node [style=none] (10) at (0, 0) {};
		\node [style=none] (11) at (2.5, 0) {};
		\node [style=none] (13) at (0, 0.7) {};
		\node [style=none] (14) at (2.5, 0.7) {};
		\node [style=none] (15) at (2.5, -0.675) {};
		\node [style=none] (16) at (0, -0.675) {};
		\node [style=none] (19) at (0.5, 0) {};
		\node [style=none] (20) at (2, 0) {};
		\node [style=ghad] (22) at (1.25, -0.25) {};
		\node [style=wirelable] (25) at (0.75, -0.425) {$n$};
		\node [style=wirelable] (26) at (1.75, -0.425) {$n$};
		\node [style=none] (27) at (0, 0) {};
		\node [style=none] (28) at (2.5, 0) {};
		\node [style=none] (29) at (0.5, 0) {};
		\node [style=none] (30) at (2, 0) {};
		\node [style=had] (31) at (1.25, 0.25) {};
	\end{pgfonlayer}
	\begin{pgfonlayer}{edgelayer}
		\draw [style=background] (15.center)
			 to (14.center)
			 to (13.center)
			 to (16.center)
			 to cycle;
		\draw [style=thick] (11.center)
			 to (20.center)
			 to [in=0, out=-135] (22.center)
			 to [in=-45, out=180] (19.center)
			 to (10.center);
		\draw (28.center) to (30.center);
		\draw (29.center) to (27.center);
		\draw (30.center)
			 to [in=0, out=135] (31.center)
			 to [in=45, out=180] (29.center);
	\end{pgfonlayer}
\end{tikzpicture}

  \end{equation*}

  Given a matrix \(A\), the \(A\)-labelled scalable \(H\)-boxes are defined as follows:

  
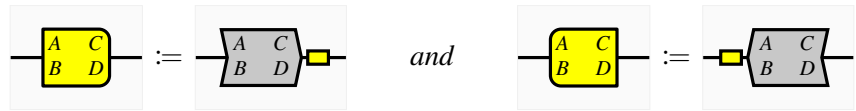
\begin{equation*}
      \begin{tikzpicture}
	\begin{pgfonlayer}{nodelayer}
		\node [style=none] (0) at (0.5, 1.375) {};
		\node [style=none] (1) at (4, 1.375) {};
		\node [style=none] (2) at (4, -1.375) {};
		\node [style=none] (3) at (0.5, -1.375) {};
		\node [style=grhad] (9) at (2.25, 0) {$\begin{matrix}A & C \\ B & D\end{matrix}$};
		\node [style=none] (10) at (4, 0) {};
		\node [style=none] (11) at (0.5, 0) {};
	\end{pgfonlayer}
	\begin{pgfonlayer}{edgelayer}
		\draw [style=background] (2.center)
			 to (1.center)
			 to (0.center)
			 to (3.center)
			 to cycle;
		\draw [style=thick] (11.center) to (9);
		\draw [style=thick] (9) to (10.center);
	\end{pgfonlayer}
\end{tikzpicture}

    \coloneqq
      \begin{tikzpicture}
	\begin{pgfonlayer}{nodelayer}
		\node [style=none] (0) at (0.5, 1.375) {};
		\node [style=none] (1) at (4.5, 1.375) {};
		\node [style=none] (2) at (4.5, -1.375) {};
		\node [style=none] (3) at (0.5, -1.375) {};
		\node [style=grmat] (9) at (2.25, 0) {$\begin{matrix}A & C \\ B & D\end{matrix}$};
		\node [style=none] (10) at (4.5, 0) {};
		\node [style=none] (11) at (0.5, 0) {};
		\node [style=ghad] (12) at (3.75, 0) {};
	\end{pgfonlayer}
	\begin{pgfonlayer}{edgelayer}
		\draw [style=background] (2.center)
			 to (1.center)
			 to (0.center)
			 to (3.center)
			 to cycle;
		\draw [style=thick] (11.center) to (9);
		\draw [style=thick] (9) to (10.center);
	\end{pgfonlayer}
\end{tikzpicture}

    \qquad\text{and}\qquad
      \begin{tikzpicture}
	\begin{pgfonlayer}{nodelayer}
		\node [style=none] (0) at (0.5, 1.375) {};
		\node [style=none] (1) at (4, 1.375) {};
		\node [style=none] (2) at (4, -1.375) {};
		\node [style=none] (3) at (0.5, -1.375) {};
		\node [style=glhad] (9) at (2.25, 0) {$\begin{matrix}A & C \\ B & D\end{matrix}$};
		\node [style=none] (10) at (4, 0) {};
		\node [style=none] (11) at (0.5, 0) {};
	\end{pgfonlayer}
	\begin{pgfonlayer}{edgelayer}
		\draw [style=background] (2.center)
			 to (1.center)
			 to (0.center)
			 to (3.center)
			 to cycle;
		\draw [style=thick] (11.center) to (9);
		\draw [style=thick] (9) to (10.center);
	\end{pgfonlayer}
\end{tikzpicture}

    \coloneqq
      \begin{tikzpicture}
	\begin{pgfonlayer}{nodelayer}
		\node [style=none] (0) at (4.5, 1.375) {};
		\node [style=none] (1) at (0.5, 1.375) {};
		\node [style=none] (2) at (0.5, -1.375) {};
		\node [style=none] (3) at (4.5, -1.375) {};
		\node [style=glmat] (9) at (2.75, 0) {$\begin{matrix}A & C \\ B & D\end{matrix}$};
		\node [style=none] (10) at (0.5, 0) {};
		\node [style=none] (11) at (4.5, 0) {};
		\node [style=ghad] (12) at (1.25, 0) {};
	\end{pgfonlayer}
	\begin{pgfonlayer}{edgelayer}
		\draw [style=background] (2.center)
			 to (1.center)
			 to (0.center)
			 to (3.center)
			 to cycle;
		\draw [style=thick] (11.center) to (9);
		\draw [style=thick] (9) to (10.center);
	\end{pgfonlayer}
\end{tikzpicture}

  \end{equation*}
\end{definition}

\begin{definition}
  \label{def:scalable_spider}
  The scalable spiders are defined by induction on the dimension:
  
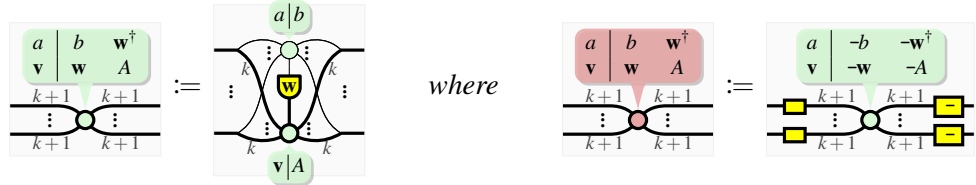
\begin{equation*}
      \begin{tikzpicture}
	\begin{pgfonlayer}{nodelayer}
		\node [style=wirelable] (139) at (-3.375, -1.5) {$k+1$};
		\node [style=wirelable] (140) at (-1.55, -1.5) {$k+1$};
		\node [style=wirelable] (141) at (-3.375, -0.25) {$k+1$};
		\node [style=wirelable] (142) at (-1.55, -0.25) {$k+1$};
		\node [style=none] (93) at (-4.5, 1.7) {};
		\node [style=none] (94) at (-0.5, 1.7) {};
		\node [style=none] (95) at (-0.5, -1.825) {};
		\node [style=none] (96) at (-4.5, -1.825) {};
		\node [style=gbn] (71) at (-2.5, -0.875) {};
		\node [style=none] (72) at (-1.75, -0.5) {};
		\node [style=none] (73) at (-1.75, -1.25) {};
		\node [style=none] (74) at (-3.25, -0.5) {};
		\node [style=none] (75) at (-3.25, -1.25) {};
		\node [style=none] (83) at (-3.4, -0.875) {$\vdotss$};
		\node [style=none] (85) at (-1.675, -0.875) {$\vdotss$};
		\node [style=none] (86) at (-4.5, -0.5) {};
		\node [style=none] (87) at (-4.5, -1.25) {};
		\node [style=none] (128) at (-0.5, -0.5) {};
		\node [style=none] (129) at (-0.5, -1.25) {};
		\node [style=bphase] (64) at (-2.5, 0.825) {$\phasemat{a \\ \vb{v}}{b & {\vb{w}}^\dag \\ {\vb{w}} & A}$};
	\end{pgfonlayer}
	\begin{pgfonlayer}{edgelayer}
		\draw [style=background] (95.center)
			 to (94.center)
			 to (93.center)
			 to (96.center)
			 to cycle;
		\draw [style=thick] (128.center)
			 to (72.center)
			 to [in=60, out=180] (71);
		\draw [style=thick] (129.center)
			 to (73.center)
			 to [in=-60, out=180] (71);
		\draw [style=thick] (86.center)
			 to (74.center)
			 to [in=120, out=360] (71);
		\draw [style=thick] (87.center)
			 to (75.center)
			 to [in=-120, out=0] (71);
		\draw [style=dbphase] (64) to (71);
	\end{pgfonlayer}
\end{tikzpicture} 
    \coloneqq
      \begin{tikzpicture}
	\begin{pgfonlayer}{nodelayer}
		\node [style=wirelable] (143) at (1.25, -1.6) {$k$};
		\node [style=wirelable] (144) at (3.25, -1.6) {$k$};
		\node [style=wirelable] (145) at (3.4, 0.475) {$k$};
		\node [style=wirelable] (146) at (1.1, 0.475) {$k$};
		\node [style=none] (105) at (0.25, 2.175) {};
		\node [style=none] (106) at (4.25, 2.175) {};
		\node [style=none] (107) at (4.25, -2.35) {};
		\node [style=none] (108) at (0.25, -2.35) {};
		\node [style=bn] (58) at (2.25, 0.975) {};
		\node [style=gbn] (59) at (2.25, -1.225) {};
		\node [style=none] (62) at (2.25, 0.55) {};
		\node [style=gdhad] (63) at (2.25, 0) {$\vb{w}$};
		\node [style=none] (97) at (0.25, 0.975) {};
		\node [style=none] (98) at (0.25, -1.225) {};
		\node [style=none] (99) at (0.875, 0.975) {};
		\node [style=none] (100) at (0.875, -1.225) {};
		\node [style=none] (122) at (0.75, -0.125) {$\vdotss$};
		\node [style=none] (125) at (1.75, 0.9) {$\vdotss$};
		\node [style=none] (127) at (1.7, -1.1) {$\vdotss$};
		\node [style=none] (132) at (4.25, 0.975) {};
		\node [style=none] (133) at (4.25, -1.225) {};
		\node [style=none] (134) at (3.625, 0.975) {};
		\node [style=none] (135) at (3.625, -1.225) {};
		\node [style=none] (136) at (3.75, -0.125) {$\vdotss$};
		\node [style=none] (137) at (2.75, 0.9) {$\vdotss$};
		\node [style=none] (138) at (2.8, -1.1) {$\vdotss$};
		\node [style=bphase] (65) at (2.25, 1.875) {$\phasemat{a}{b}$};
		\node [style=bphase] (66) at (2.25, -2.125) {$\phasemat{\vb{v}}{A}$};
		\node [style=bn] (147) at (2.25, 0.975) {};
		\node [style=none] (148) at (0.875, 0.975) {};
		\node [style=none] (149) at (0.25, 0.975) {};
		\node [style=bn] (150) at (2.25, 0.975) {};
		\node [style=none] (151) at (3.625, 0.975) {};
		\node [style=none] (152) at (4.25, 0.975) {};
		\node [style=bn] (153) at (2.25, 0.975) {};
		\node [style=none] (154) at (0.875, -1.225) {};
		\node [style=none] (155) at (3.625, -1.225) {};
		\node [style=none] (156) at (0.25, -1.225) {};
		\node [style=none] (157) at (4.25, -1.225) {};
	\end{pgfonlayer}
	\begin{pgfonlayer}{edgelayer}
		\draw [style=background] (107.center)
			 to (106.center)
			 to (105.center)
			 to (108.center)
			 to cycle;
		\draw [style=thick] (98.center)
			 to (100.center)
			 to [in=-150, out=-30] (59);
		\draw [style=thick] (97.center)
			 to (99.center)
			 to [in=165, out=-30] (59);
		\draw [in=135, out=45] (99.center) to (58);
		\draw [in=-165, out=30] (100.center) to (58);
		\draw [style=thick] (59)
			 to [in=-150, out=-30] (135.center)
			 to (133.center);
		\draw [style=thick] (132.center)
			 to (134.center)
			 to [in=15, out=-150] (59);
		\draw [in=-15, out=150] (135.center) to (58);
		\draw [in=45, out=135] (134.center) to (58);
		\draw [style=thick] (59) to (63);
		\draw (63) to (62.center);
		\draw (62.center) to (58);
		\draw [style=dbphase] (65) to (58);
		\draw [style=ubphase] (66) to (59);
		\draw (149.center)
			 to (148.center)
			 to [in=135, out=45] (147);
		\draw (152.center)
			 to (151.center)
			 to [in=45, out=135] (150);
		\draw (156.center)
			 to (154.center)
			 to [in=-165, out=30] (153);
		\draw (153)
			 to [in=150, out=-15] (155.center)
			 to (157.center);
	\end{pgfonlayer}
\end{tikzpicture}

    \qquad \text{where} \qquad
      \begin{tikzpicture}
	\begin{pgfonlayer}{nodelayer}
		\node [style=wirelable] (139) at (-3.375, -1.5) {$k+1$};
		\node [style=wirelable] (140) at (-1.55, -1.5) {$k+1$};
		\node [style=wirelable] (141) at (-3.375, -0.25) {$k+1$};
		\node [style=wirelable] (142) at (-1.55, -0.25) {$k+1$};
		\node [style=none] (93) at (-4.5, 1.7) {};
		\node [style=none] (94) at (-0.5, 1.7) {};
		\node [style=none] (95) at (-0.5, -1.825) {};
		\node [style=none] (96) at (-4.5, -1.825) {};
		\node [style=gwn] (71) at (-2.5, -0.875) {};
		\node [style=none] (72) at (-1.75, -0.5) {};
		\node [style=none] (73) at (-1.75, -1.25) {};
		\node [style=none] (74) at (-3.25, -0.5) {};
		\node [style=none] (75) at (-3.25, -1.25) {};
		\node [style=none] (83) at (-3.4, -0.875) {$\vdotss$};
		\node [style=none] (85) at (-1.675, -0.875) {$\vdotss$};
		\node [style=none] (86) at (-4.5, -0.5) {};
		\node [style=none] (87) at (-4.5, -1.25) {};
		\node [style=none] (128) at (-0.5, -0.5) {};
		\node [style=none] (129) at (-0.5, -1.25) {};
		\node [style=wphase] (64) at (-2.5, 0.825) {$\phasemat{a \\ \vb{v}}{b & {\vb{w}}^\dag \\ {\vb{w}} & A}$};
	\end{pgfonlayer}
	\begin{pgfonlayer}{edgelayer}
		\draw [style=background] (95.center)
			 to (94.center)
			 to (93.center)
			 to (96.center)
			 to cycle;
		\draw [style=thick] (128.center)
			 to (72.center)
			 to [in=60, out=180] (71);
		\draw [style=thick] (129.center)
			 to (73.center)
			 to [in=-60, out=180] (71);
		\draw [style=thick] (86.center)
			 to (74.center)
			 to [in=120, out=360] (71);
		\draw [style=thick] (87.center)
			 to (75.center)
			 to [in=-120, out=0] (71);
		\draw [style=dwphase] (64) to (71);
	\end{pgfonlayer}
\end{tikzpicture} 
    \coloneqq
      \begin{tikzpicture}
	\begin{pgfonlayer}{nodelayer}
		\node [style=wirelable] (143) at (2.375, -1.5) {$k+1$};
		\node [style=wirelable] (144) at (4.2, -1.5) {$k+1$};
		\node [style=wirelable] (145) at (2.375, -0.25) {$k+1$};
		\node [style=wirelable] (146) at (4.2, -0.25) {$k+1$};
		\node [style=none] (147) at (0.5, 1.7) {};
		\node [style=none] (148) at (6, 1.7) {};
		\node [style=none] (149) at (6, -1.825) {};
		\node [style=none] (150) at (0.5, -1.825) {};
		\node [style=gbn] (151) at (3.25, -0.875) {};
		\node [style=none] (152) at (4, -0.5) {};
		\node [style=none] (153) at (4, -1.25) {};
		\node [style=none] (154) at (2.5, -0.5) {};
		\node [style=none] (155) at (2.5, -1.25) {};
		\node [style=none] (156) at (2.35, -0.875) {$\vdotss$};
		\node [style=none] (157) at (4.075, -0.875) {$\vdotss$};
		\node [style=none] (158) at (0.5, -0.5) {};
		\node [style=none] (159) at (0.5, -1.25) {};
		\node [style=none] (160) at (6, -0.5) {};
		\node [style=none] (161) at (6, -1.25) {};
		\node [style=bphase] (162) at (3.25, 0.825) {$\phasemat{a \\ \vb{v}}{\minu b & \minu {\vb{w}}^\dag \\ \minu {\vb{w}} & \minu A}$};
		\node [style=ghad] (163) at (1.25, -0.525) {};
		\node [style=ghad] (164) at (1.25, -1.275) {};
		\node [style=ghad] (165) at (5.325, -1.275) {\(\minu\)};
		\node [style=ghad] (166) at (5.325, -0.525) {\(\minu\)};
	\end{pgfonlayer}
	\begin{pgfonlayer}{edgelayer}
		\draw [style=background] (149.center)
			 to (148.center)
			 to (147.center)
			 to (150.center)
			 to cycle;
		\draw [style=thick] (160.center)
			 to (152.center)
			 to [in=60, out=180] (151);
		\draw [style=thick] (161.center)
			 to (153.center)
			 to [in=-60, out=180] (151);
		\draw [style=thick] (158.center)
			 to (154.center)
			 to [in=120, out=360] (151);
		\draw [style=thick] (159.center)
			 to (155.center)
			 to [in=-120, out=0] (151);
		\draw [style=dbphase] (162) to (151);
	\end{pgfonlayer}
\end{tikzpicture}
  \end{equation*}
\end{definition}

\subsection{Scalable lemmas}\label{sec:scalable_results}

In this appendix, we prove various scalable lemmas in \(\dStabZX\) to help prove completeness. 

It is a consequence of the work of Booth et al.~\cite{gsa} that all lemmas here hold for \(\StabZX\) when the labels are restricted from \(\F_p(\delta)\) and \(\F_p(\delta+\delta^{\minu 1})\) to \(\F_p\). Therefore, there are a few forward references to these lemmas when rewriting diagrams with non-delayed labels.

\begin{lemma}
    \label{eq:arrow_inverse}
For any invertible matrix \(Z\in\Matrices\): 
\[
    \begin{tikzpicture}
	\begin{pgfonlayer}{nodelayer}
		\node [style=none] (0) at (-0.5, 0) {};
		\node [style=none] (1) at (-3, 0) {};
		\node [style=grmat] (2) at (-1.75, 0) {$Z^{\minu 1}$};
		\node [style=none] (3) at (-0.5, 0.375) {};
		\node [style=none] (4) at (-0.5, -0.375) {};
		\node [style=none] (5) at (-3, -0.375) {};
		\node [style=none] (6) at (-3, 0.375) {};
		\node [style=wirelable] (7) at (-0.75, 0.175) {$m$};
		\node [style=wirelable] (8) at (-2.75, 0.175) {$m$};
	\end{pgfonlayer}
	\begin{pgfonlayer}{edgelayer}
		\draw [style=background] (4.center)
			 to (3.center)
			 to (6.center)
			 to (5.center)
			 to cycle;
		\draw [style=thick] (0.center) to (1.center);
	\end{pgfonlayer}
\end{tikzpicture}
  =
    \begin{tikzpicture}
	\begin{pgfonlayer}{nodelayer}
		\node [style=none] (10) at (0.5, 0) {};
		\node [style=none] (11) at (2.5, 0) {};
		\node [style=glmat] (12) at (1.5, 0) {$Z$};
		\node [style=none] (13) at (0.5, 0.375) {};
		\node [style=none] (14) at (0.5, -0.375) {};
		\node [style=none] (15) at (2.5, -0.375) {};
		\node [style=none] (16) at (2.5, 0.375) {};
		\node [style=wirelable] (17) at (0.75, 0.175) {$m$};
		\node [style=wirelable] (18) at (2.25, 0.175) {$m$};
	\end{pgfonlayer}
	\begin{pgfonlayer}{edgelayer}
		\draw [style=background] (14.center)
			 to (13.center)
			 to (16.center)
			 to (15.center)
			 to cycle;
		\draw [style=thick] (10.center) to (11.center);
	\end{pgfonlayer}
\end{tikzpicture}.
\]
\end{lemma}

\begin{lemma}
  \label{lem:scalable_bigebra}%
  For any scalable dimension \(n \in \N\):
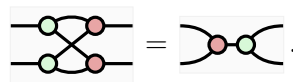
  \[
      \begin{tikzpicture}
	\begin{pgfonlayer}{nodelayer}
		\node [style=none] (22) at (4.5, 1) {};
		\node [style=none] (23) at (7.75, 1) {};
		\node [style=none] (24) at (7.75, -1) {};
		\node [style=none] (25) at (4.5, -1) {};
		\node [style=none] (27) at (4.5, 0.5) {};
		\node [style=none] (28) at (4.5, -0.5) {};
		\node [style=ggn] (29) at (5.5, 0.5) {};
		\node [style=ggn] (30) at (5.5, -0.5) {};
		\node [style=grn] (31) at (6.75, 0.5) {};
		\node [style=grn] (32) at (6.75, -0.5) {};
		\node [style=none] (33) at (7.75, 0.5) {};
		\node [style=none] (34) at (7.75, -0.5) {};
	\end{pgfonlayer}
	\begin{pgfonlayer}{edgelayer}
		\draw [style=background] (24.center)
			 to (23.center)
			 to (22.center)
			 to (25.center)
			 to cycle;
		\draw [style=thick] (27.center) to (29);
		\draw [style=thick] (28.center) to (30);
		\draw [style=thick] (31) to (33.center);
		\draw [style=thick] (32) to (34.center);
		\draw [style=thick] (29) to (32);
		\draw [style=thick] (30) to (31);
		\draw [style=thick, bend left=45] (29) to (31);
		\draw [style=thick, bend right=45] (30) to (32);
	\end{pgfonlayer}
\end{tikzpicture}
    =
      \begin{tikzpicture}
	\begin{pgfonlayer}{nodelayer}
		\node [style=none] (18) at (9.75, 0.75) {};
		\node [style=none] (19) at (12.5, 0.75) {};
		\node [style=none] (20) at (12.5, -0.75) {};
		\node [style=none] (21) at (9.75, -0.75) {};
		\node [style=none] (35) at (9.75, 0.5) {};
		\node [style=none] (36) at (9.75, -0.5) {};
		\node [style=grn] (37) at (10.75, 0) {};
		\node [style=ggn] (38) at (11.5, 0) {};
		\node [style=none] (39) at (12.5, 0.5) {};
		\node [style=none] (40) at (12.5, -0.5) {};
	\end{pgfonlayer}
	\begin{pgfonlayer}{edgelayer}
		\draw [style=background] (20.center)
			 to (19.center)
			 to (18.center)
			 to (21.center)
			 to cycle;
		\draw [style=thick] (37) to (38);
		\draw [style=thick, in=0, out=-120] (37) to (36.center);
		\draw [style=thick, in=360, out=120] (37) to (35.center);
		\draw [style=thick, in=180, out=60] (38) to (39.center);
		\draw [style=thick, in=180, out=-60] (38) to (40.center);
	\end{pgfonlayer}
\end{tikzpicture}.
  \]
\end{lemma}

\begin{proof}
  This follows from a straightforward induction on the type \(n \in \N\) of the scalable wires.
\end{proof}

\begin{lemma}
  \label{lem:scalable_copy}%
  For any \(n \in \N\) and \(\vb{x} \in \F_p(\delta)^n\):
  \[
      \begin{tikzpicture}
	\begin{pgfonlayer}{nodelayer}
		\node [style=none] (4) at (11.5, 1) {};
		\node [style=none] (5) at (14.25, 1) {};
		\node [style=none] (6) at (14.25, -1) {};
		\node [style=none] (7) at (11.5, -1) {};
		\node [style=grn] (8) at (12.25, 0) {};
		\node [style=gbn] (9) at (13.25, 0) {};
		\node [style=none] (15) at (14.25, 0.5) {};
		\node [style=none] (16) at (14.25, -0.5) {};
		\node [style=wphase] (ph0) at (12.25, 0.85) {$\phasemat{\vb{x}}{0}$};
	\end{pgfonlayer}
	\begin{pgfonlayer}{edgelayer}
		\draw [style=background] (6.center)
			 to (5.center)
			 to (4.center)
			 to (7.center)
			 to cycle;
		\draw [style=thick] (9) to (8);
		\draw [style=thick, bend left] (9) to (15.center);
		\draw [style=thick, bend right] (9) to (16.center);
		\draw [style=dwphase] (ph0) to (8);
	\end{pgfonlayer}
\end{tikzpicture}
    =
      \begin{tikzpicture}
	\begin{pgfonlayer}{nodelayer}
		\node [style=none] (0) at (16.25, 1) {};
		\node [style=none] (1) at (19.3, 1) {};
		\node [style=none] (2) at (19.3, -1) {};
		\node [style=none] (3) at (16.25, -1) {};
		\node [style=grn] (11) at (18.5, 0.5) {};
		\node [style=grn] (12) at (18.5, -0.5) {};
		\node [style=none] (13) at (19.25, 0.5) {};
		\node [style=none] (14) at (19.25, -0.5) {};
		\node [style=wphase] (ph1) at (17.5, 0.5) {$\phasemat{\vb{x}}{0}$};
		\node [style=wphase] (ph2) at (17.5, -0.5) {$\phasemat{\vb{x}}{0}$};
	\end{pgfonlayer}
	\begin{pgfonlayer}{edgelayer}
		\draw [style=background] (2.center)
			 to (1.center)
			 to (0.center)
			 to (3.center)
			 to cycle;
		\draw [style=thick] (12) to (14.center);
		\draw [style=thick] (11) to (13.center);
		\draw [style=rwphase] (ph1) to (11);
		\draw [style=rwphase] (ph2) to (12);
	\end{pgfonlayer}
\end{tikzpicture}.
  \]
\end{lemma}

\begin{proof}
  This follows from a straightforward induction on the type \(n \in \N\) of the scalable wires.
\end{proof}

\begin{lemma}
  \label{lem:scalable_mat}%
  For any \(m,n,k \in \N\) and \(A \in \Matrices[k][n]\), \(B \in \Matrices[n][m]\),
  the block decomposition of a matrix arrow and the composition of matrix
  arrows hold diagrammatically; see~\Cref{def:matrix_arrow}.
\end{lemma}

\begin{proof}
  This follows from a straightforward induction on the type \(n \in \N\) of the scalable wires.
\end{proof}

\begin{lemma}
  \label{lem:scalable_plus}%
  For any \(n \in \N\) and matrices \(A, B \in \Matrices[n][m][\F_p(\delta)]\),
  \begin{gather*}
      \input{gsa_completeness/scalable_results/lemmas/scalable_plus/0.tikz}
    =
      \input{gsa_completeness/scalable_results/lemmas/scalable_plus/1.tikz}
    ,\qquad
      \input{gsa_completeness/scalable_results/lemmas/scalable_plus/2.tikz}
    =
      \input{gsa_completeness/scalable_results/lemmas/scalable_plus/3.tikz}
      ,\qquad
      \input{gsa_completeness/scalable_results/lemmas/scalable_plus/4.tikz}
    =
      \input{gsa_completeness/scalable_results/lemmas/scalable_plus/5.tikz}
    ,\qquad
      \input{gsa_completeness/scalable_results/lemmas/scalable_plus/6.tikz}
    =
      \input{gsa_completeness/scalable_results/lemmas/scalable_plus/7.tikz}
    \\
      \input{gsa_completeness/scalable_results/lemmas/scalable_plus/8.tikz}
    =
      \input{gsa_completeness/scalable_results/lemmas/scalable_plus/9.tikz}
    =
      \input{gsa_completeness/scalable_results/lemmas/scalable_plus/10.tikz}
    .
  \end{gather*}
\end{lemma}

\begin{proof}
  This follows from a straightforward induction on  \(n\) and \(m\).
\end{proof}

\begin{lemma}
  \label{lem:scalable_id}%
  For any \(n \in \N\): 
  \[
      \begin{tikzpicture}
	\begin{pgfonlayer}{nodelayer}
		\node [style=none] (3) at (2.75, 0.5) {};
		\node [style=none] (4) at (4.5, 0.5) {};
		\node [style=none] (5) at (4.5, -0.5) {};
		\node [style=none] (6) at (2.75, -0.5) {};
		\node [style=none] (20) at (2.75, 0) {};
		\node [style=none] (21) at (4.5, 0) {};
		\node [style=grn] (22) at (3.625, 0) {};
		\node [style=none] (23) at (4.125, 0) {};
	\end{pgfonlayer}
	\begin{pgfonlayer}{edgelayer}
		\draw [style=background] (5.center)
			 to (4.center)
			 to (3.center)
			 to (6.center)
			 to cycle;
		\draw [style=thick] (20.center) to (21.center);
	\end{pgfonlayer}
\end{tikzpicture}

    =
      \begin{tikzpicture}
	\begin{pgfonlayer}{nodelayer}
		\node [style=none] (7) at (6, 0.5) {};
		\node [style=none] (8) at (7.5, 0.5) {};
		\node [style=none] (9) at (7.5, -0.5) {};
		\node [style=none] (10) at (6, -0.5) {};
		\node [style=none] (15) at (6, 0) {};
		\node [style=none] (16) at (7.5, 0) {};
		\node [style=gbn] (17) at (6.75, 0) {};
	\end{pgfonlayer}
	\begin{pgfonlayer}{edgelayer}
		\draw [style=background] (9.center)
			 to (8.center)
			 to (7.center)
			 to (10.center)
			 to cycle;
		\draw [style=thick] (15.center) to (16.center);
	\end{pgfonlayer}
\end{tikzpicture}

    =
      \begin{tikzpicture}
	\begin{pgfonlayer}{nodelayer}
		\node [style=none] (24) at (9, 0.5) {};
		\node [style=none] (25) at (10.5, 0.5) {};
		\node [style=none] (26) at (10.5, -0.5) {};
		\node [style=none] (27) at (9, -0.5) {};
		\node [style=rmatt] (28) at (9.725, 0) {\(1\)};
		\node [style=none] (29) at (10.5, 0) {};
		\node [style=none] (30) at (9, 0) {};
	\end{pgfonlayer}
	\begin{pgfonlayer}{edgelayer}
		\draw [style=background] (26.center)
			 to (25.center)
			 to (24.center)
			 to (27.center)
			 to cycle;
		\draw [style=thick, in=360, out=180] (29.center) to (28);
		\draw [style=thick] (30.center) to (28);
	\end{pgfonlayer}
\end{tikzpicture}

    =
     \begin{tikzpicture}
	\begin{pgfonlayer}{nodelayer}
		\node [style=none] (11) at (12, 0.5) {};
		\node [style=none] (12) at (13.5, 0.5) {};
		\node [style=none] (13) at (13.5, -0.5) {};
		\node [style=none] (14) at (12, -0.5) {};
		\node [style=none] (18) at (12, 0) {};
		\node [style=none] (19) at (13.5, 0) {};
	\end{pgfonlayer}
	\begin{pgfonlayer}{edgelayer}
		\draw [style=background] (13.center)
			 to (12.center)
			 to (11.center)
			 to (14.center)
			 to cycle;
		\draw [style=thick] (18.center) to (19.center);
	\end{pgfonlayer}
\end{tikzpicture}
.
  \]
\end{lemma}

\begin{proof}
  This follows from a straightforward induction on the type \(n \in \N\) of the scalable wires.
\end{proof}

\begin{lemma}
  \label{gsa:lem:scalable_symplectic_states}
  For any \(n \in \N\), \(\vb x \in \F_p(\delta)^n\) and invertible \(X \in \Herm\),
  \[
      \begin{tikzpicture}
	\begin{pgfonlayer}{nodelayer}
		\node [style=gbn] (20) at (-0.25, -0.25) {};
		\node [style=none] (21) at (1, -0.25) {};
		\node [style=bphase] (31) at (-0.25, 0.65) {$\phasemat{\vb x}{X}$};
		\node [style=none] (bg0) at (-1, 0.75) {};
		\node [style=none] (bg1) at (1, 0.75) {};
		\node [style=none] (bg2) at (1, -0.75) {};
		\node [style=none] (bg3) at (-1, -0.75) {};
	\end{pgfonlayer}
	\begin{pgfonlayer}{edgelayer}
		\draw [style=background] (bg0.center)
			 to (bg1.center)
			 to (bg2.center)
			 to (bg3.center)
			 to cycle;
		\draw [style=thick] (21.center) to (20);
		\draw [style=dbphase] (31) to (20);
	\end{pgfonlayer}
\end{tikzpicture}

    =
      \begin{tikzpicture}
	\begin{pgfonlayer}{nodelayer}
		\node [style=none] (34) at (1.75, -0.375) {};
		\node [style=gwn] (35) at (-0.75, -0.375) {};
		\node [style=whphase] (36) at (0, 0.525) {$\phasemat{\minu X^{\minu 1}\vb x}{X^{\minu 1}}$};
		\node [style=none] (bg0) at (-1.75, 0.875) {};
		\node [style=none] (bg1) at (1.75, 0.875) {};
		\node [style=none] (bg2) at (1.75, -0.875) {};
		\node [style=none] (bg3) at (-1.75, -0.875) {};
	\end{pgfonlayer}
	\begin{pgfonlayer}{edgelayer}
		\draw [style=background] (bg0.center)
			 to (bg1.center)
			 to (bg2.center)
			 to (bg3.center)
			 to cycle;
		\draw [style=thick] (35) to (34.center);
		\draw [style=dwphase] (36) to (35);
	\end{pgfonlayer}
\end{tikzpicture}

    \qquad\text{and}\qquad
      \begin{tikzpicture}
	\begin{pgfonlayer}{nodelayer}
		\node [style=gwn] (28) at (-0.25, -0.25) {};
		\node [style=none] (29) at (1, -0.25) {};
		\node [style=whphase] (32) at (-0.25, 0.65) {$\phasemat{\vb x}{X}$};
		\node [style=none] (bg0) at (-1, 0.75) {};
		\node [style=none] (bg1) at (1, 0.75) {};
		\node [style=none] (bg2) at (1, -0.75) {};
		\node [style=none] (bg3) at (-1, -0.75) {};
	\end{pgfonlayer}
	\begin{pgfonlayer}{edgelayer}
		\draw [style=background] (bg0.center)
			 to (bg1.center)
			 to (bg2.center)
			 to (bg3.center)
			 to cycle;
		\draw [style=thick] (29.center) to (28);
		\draw [style=dwphase] (32) to (28);
	\end{pgfonlayer}
\end{tikzpicture}

    =
      \begin{tikzpicture}
	\begin{pgfonlayer}{nodelayer}
		\node [style=gbn] (25) at (-0.5, -0.375) {};
		\node [style=none] (26) at (1.75, -0.375) {};
		\node [style=bphase] (33) at (0, 0.525) {$\phasemat{\minu X^{\minu 1}\vb x}{X^{\minu 1}}$};
		\node [style=none] (bg0) at (-1.75, 0.875) {};
		\node [style=none] (bg1) at (1.75, 0.875) {};
		\node [style=none] (bg2) at (1.75, -0.875) {};
		\node [style=none] (bg3) at (-1.75, -0.875) {};
	\end{pgfonlayer}
	\begin{pgfonlayer}{edgelayer}
		\draw [style=background] (bg0.center)
			 to (bg1.center)
			 to (bg2.center)
			 to (bg3.center)
			 to cycle;
		\draw [style=thick] (25) to (26.center);
		\draw [style=dbphase] (33) to (25);
	\end{pgfonlayer}
\end{tikzpicture}

  \]
\end{lemma}

\begin{proof}
  The proof is analogous to the proof of \Cref{lem:symplectic_states}.
\end{proof}

\begin{lemma}
  \label{lem:scalable_fusion}
  
  For any \(\vb{x},\vb{y} \in \F_p(\delta)^n\) and \(X,Y \in \Herm\),
  \[
      \begin{tikzpicture}
	\begin{pgfonlayer}{nodelayer}
		\node [style=gbn] (17) at (-0.5, 1.5) {};
		\node [style=none] (18) at (-2, 2.25) {};
		\node [style=none] (19) at (-2, 0.75) {};
		\node [style=none] (20) at (1.5, 2.25) {};
		\node [style=none] (21) at (1.5, 0.75) {};
		\node [style=none] (22) at (1.75, 1.75) {$\vdots$};
		\node [style=none] (23) at (-1.75, 1.75) {$\vdots$};
		\node [style=gbn] (24) at (0.5, -1.5) {};
		\node [style=none] (25) at (-1.5, -0.75) {};
		\node [style=none] (26) at (-1.5, -2.25) {};
		\node [style=none] (27) at (2, -0.75) {};
		\node [style=none] (28) at (2, -2.25) {};
		\node [style=none] (29) at (1.75, -1.25) {$\vdots$};
		\node [style=none] (30) at (-1.75, -1.25) {$\vdots$};
		\node [style=none] (44) at (-2, -0.75) {};
		\node [style=none] (45) at (-2, -2.25) {};
		\node [style=none] (46) at (2, 2.25) {};
		\node [style=none] (47) at (2, 0.75) {};
		\node [style=bphase] (lblphase1) at (-0.5, 2.4) {$\phasemat{\vb{x}}{X}$};
		\node [style=bphase] (lblphase2) at (0.5, -2.4) {$\phasemat{\vb{y}}{Y}$};
		\node [style=none] (bg0) at (-2, 2.5) {};
		\node [style=none] (bg1) at (2, 2.5) {};
		\node [style=none] (bg2) at (2, -2.5) {};
		\node [style=none] (bg3) at (-2, -2.5) {};
	\end{pgfonlayer}
	\begin{pgfonlayer}{edgelayer}
		\draw [style=background] (bg0.center)
			 to (bg1.center)
			 to (bg2.center)
			 to (bg3.center)
			 to cycle;
		\draw [style=thick] (17)
					 to [in=180, out=30] (20.center)
					 to (46.center);
		\draw [style=thick] (17)
					 to [in=180, out=-30] (21.center)
					 to (47.center);
		\draw [style=thick, bend right] (17) to (18.center);
		\draw [style=thick, bend left] (17) to (19.center);
		\draw [style=thick, bend left] (24) to (27.center);
		\draw [style=thick, bend right] (24) to (28.center);
		\draw [style=thick] (44.center)
					 to (25.center)
					 to [in=150, out=360] (24);
		\draw [style=thick] (45.center)
					 to (26.center)
					 to [in=210, out=0] (24);
		\draw [style=thick, in=135, out=-45] (17) to (24);
		\draw [style=dbphase] (lblphase1) to (17);
		\draw [style=ubphase] (lblphase2) to (24);
	\end{pgfonlayer}
\end{tikzpicture}

    =
      \begin{tikzpicture}
	\begin{pgfonlayer}{nodelayer}
		\node [style=gbn] (31) at (0, 0) {};
		\node [style=none] (32) at (2.25, 2.25) {};
		\node [style=none] (33) at (2.25, 0.75) {};
		\node [style=none] (34) at (2.25, -0.75) {};
		\node [style=none] (35) at (2.25, -2.25) {};
		\node [style=none] (36) at (-2.25, 2.25) {};
		\node [style=none] (37) at (-2.25, 0.75) {};
		\node [style=none] (38) at (-2.25, -0.75) {};
		\node [style=none] (39) at (-2.25, -2.25) {};
		\node [style=none] (40) at (2, 1.75) {$\vdots$};
		\node [style=none] (41) at (2, -1.25) {$\vdots$};
		\node [style=none] (42) at (-2, 1.75) {$\vdots$};
		\node [style=none] (43) at (-2, -1.25) {$\vdots$};
		\node [style=bphase] (lblphase3) at (0, 1.65) {$\phasemat{\vb{x}+\vb{y}}{X+Y}$};
		\node [style=none] (bg0) at (-2.25, 2.5) {};
		\node [style=none] (bg1) at (2.25, 2.5) {};
		\node [style=none] (bg2) at (2.25, -2.5) {};
		\node [style=none] (bg3) at (-2.25, -2.5) {};
	\end{pgfonlayer}
	\begin{pgfonlayer}{edgelayer}
		\draw [style=background] (bg0.center)
			 to (bg1.center)
			 to (bg2.center)
			 to (bg3.center)
			 to cycle;
		\draw [style=thick, in=180, out=30, looseness=1.25] (31) to (33.center);
		\draw [style=thick, in=180, out=60] (31) to (32.center);
		\draw [style=thick, in=180, out=-30] (31) to (34.center);
		\draw [style=thick, in=180, out=-60] (31) to (35.center);
		\draw [style=thick, in=0, out=120] (31) to (36.center);
		\draw [style=thick, in=0, out=150] (31) to (37.center);
		\draw [style=thick, in=0, out=-150] (31) to (38.center);
		\draw [style=thick, in=0, out=-120] (31) to (39.center);
		\draw [style=dbphase] (lblphase3) to (31);
	\end{pgfonlayer}
\end{tikzpicture}

    \qquad\text{and}\qquad
      \begin{tikzpicture}
	\begin{pgfonlayer}{nodelayer}
		\node [style=gwn] (49) at (-0.5, 1.5) {};
		\node [style=none] (50) at (-2, 2.25) {};
		\node [style=none] (51) at (-2, 0.75) {};
		\node [style=none] (52) at (1.5, 2.25) {};
		\node [style=none] (53) at (1.5, 0.75) {};
		\node [style=none] (54) at (1.75, 1.75) {$\vdots$};
		\node [style=none] (55) at (-1.75, 1.75) {$\vdots$};
		\node [style=gwn] (56) at (0.5, -1.5) {};
		\node [style=none] (57) at (-1.5, -0.75) {};
		\node [style=none] (58) at (-1.5, -2.25) {};
		\node [style=none] (59) at (2, -0.75) {};
		\node [style=none] (60) at (2, -2.25) {};
		\node [style=none] (61) at (1.75, -1.25) {$\vdots$};
		\node [style=none] (62) at (-1.75, -1.25) {$\vdots$};
		\node [style=none] (76) at (-2, -0.75) {};
		\node [style=none] (77) at (-2, -2.25) {};
		\node [style=none] (78) at (2, 2.25) {};
		\node [style=none] (79) at (2, 0.75) {};
		\node [style=none] (80) at (-0.025, 0.125) {};
		\node [style=whphase] (lblphase4) at (-0.5, 2.4) {$\phasemat{\vb{x}}{X}$};
		\node [style=whphase] (lblphase5) at (0.5, -2.4) {$\phasemat{\vb{y}}{Y}$};
		\node [style=none] (bg0) at (-2, 2.5) {};
		\node [style=none] (bg1) at (2, 2.5) {};
		\node [style=none] (bg2) at (2, -2.5) {};
		\node [style=none] (bg3) at (-2, -2.5) {};
	\end{pgfonlayer}
	\begin{pgfonlayer}{edgelayer}
		\draw [style=background] (bg0.center)
			 to (bg1.center)
			 to (bg2.center)
			 to (bg3.center)
			 to cycle;
		\draw [style=thick] (49)
			 to [in=180, out=30] (52.center)
			 to (78.center);
		\draw [style=thick] (49)
			 to [in=180, out=-30] (53.center)
			 to (79.center);
		\draw [style=thick, bend right] (49) to (50.center);
		\draw [style=thick, bend left] (49) to (51.center);
		\draw [style=thick, bend left] (56) to (59.center);
		\draw [style=thick, bend right] (56) to (60.center);
		\draw [style=thick] (56)
			 to [in=360, out=150] (57.center)
			 to (76.center);
		\draw [style=thick] (56)
			 to [in=0, out=-150] (58.center)
			 to (77.center);
		\draw [style=thick] (49)
			 to [in=94, out=-45] (80.center)
			 to [in=135, out=-86] (56);
		\draw [style=dwphase] (lblphase4) to (49);
		\draw [style=uwphase] (lblphase5) to (56);
	\end{pgfonlayer}
\end{tikzpicture}

    =
      \begin{tikzpicture}
	\begin{pgfonlayer}{nodelayer}
		\node [style=gwn] (63) at (0, 0) {};
		\node [style=none] (64) at (2.25, 2.25) {};
		\node [style=none] (65) at (2.25, 0.75) {};
		\node [style=none] (66) at (2.25, -0.75) {};
		\node [style=none] (67) at (2.25, -2.25) {};
		\node [style=none] (68) at (-2.25, 2.25) {};
		\node [style=none] (69) at (-2.25, 0.75) {};
		\node [style=none] (70) at (-2.25, -0.75) {};
		\node [style=none] (71) at (-2.25, -2.25) {};
		\node [style=none] (72) at (2, 1.75) {$\vdots$};
		\node [style=none] (73) at (2, -1.25) {$\vdots$};
		\node [style=none] (74) at (-2, 1.75) {$\vdots$};
		\node [style=none] (75) at (-2, -1.25) {$\vdots$};
		\node [style=whphase] (lblphase6) at (0, 1.65) {$\phasemat{\vb{x}+\vb{y}}{X+Y}$};
		\node [style=none] (bg0) at (-2.25, 2.5) {};
		\node [style=none] (bg1) at (2.25, 2.5) {};
		\node [style=none] (bg2) at (2.25, -2.5) {};
		\node [style=none] (bg3) at (-2.25, -2.5) {};
	\end{pgfonlayer}
	\begin{pgfonlayer}{edgelayer}
		\draw [style=background] (bg0.center)
			 to (bg1.center)
			 to (bg2.center)
			 to (bg3.center)
			 to cycle;
		\draw [style=thick, in=180, out=30] (63) to (65.center);
		\draw [style=thick, in=180, out=60] (63) to (64.center);
		\draw [style=thick, in=180, out=-30] (63) to (66.center);
		\draw [style=thick, in=180, out=-60] (63) to (67.center);
		\draw [style=thick, in=0, out=120] (63) to (68.center);
		\draw [style=thick, in=0, out=150] (63) to (69.center);
		\draw [style=thick, in=0, out=-150] (63) to (70.center);
		\draw [style=thick, in=0, out=-120] (63) to (71.center);
		\draw [style=dwphase] (lblphase6) to (63);
	\end{pgfonlayer}
\end{tikzpicture}

  \]
\end{lemma}

\begin{proof}
  We prove thickened green fusion by induction on the type \(n\in\N\) of the thick wires.
  \begin{itemize}
    \item \emph{Inductive claim:}
      The claim holds for any \(n \in \N\).
    \item \emph{Base case:}
      Trivial.
    \item \hypertarget{ind:scalable_fusion}{}\emph{Inductive hypothesis:}
      Fix \(n \in \N\) and assume the claim holds for wires of type \(n\).
    \item \emph{Inductive step:}
      Given \(X,Y \in \Herm[n+1]\), \(\vb x, \vb y\in \F_p(\delta)^{n+1}\), there exist 
      \(X',Y' \in \Herm[n]\), \(\vb x',\vb y', X_{1,\bullet}, Y_{1,\bullet} \in \F_p(\delta)^n\), and \(x_1, y_1, X_{1,1}\in \F_p(\delta)\) such that:
      \begin{equation*}
          X
        =

.
      \end{align*}
  \end{itemize}
  The scalable red spider fusion follows easily from the scalable green spider fusion.
\end{proof}

\begin{lemma}
  \label{gsa:lem:arrow_transpose}
  For any \(m,n \in \N\) and \(A \in \Matrices[n][m][\F_p(\delta)]\),
  \[
      %
.
  \]
\end{lemma}

\begin{proof}
  This follows from a straightforward induction on \(n\) and \(m\) of the matrix \(A\).
\end{proof}

\begin{lemma}
  \label{lem:scalable_box_loop}
  For any \(n \in \N\) and \(A \in \Matrices[n][n][\F_p(\delta)]\),
  \begin{align*}
      %

  \end{align*}
\end{lemma}

\begin{proof}
  We prove the claim by induction on the type \(n\) of the thick wires, equivalently the dimension of the matrix \(A\).
  \begin{itemize}
    \item \emph{Inductive claim:}
      The claim holds for any \(n \in \N\).
    \item \emph{Base case:}
      Trivial.
    \item \hypertarget{ind:scalable_box_loop}{}\emph{Inductive hypothesis:}
      Fix \(n = k\) and assume the claim holds for matrices in \(\Matrices[k][k][\F_p(\delta)]\).
    \item \emph{Inductive step:}
      Given any \(A \in\Matrices[(k+1)][(k+1)][\F_p(\delta)]\), there exist
      \(a \in \F_p(\delta)\),
      \(B \in \Matrices[k][1][\F_p(\delta)]\),
      \(C \in\Matrices[1][k][\F_p(\delta)]\), and \(D \in  \Matrices[k][k][\F_p(\delta)]\) such that:
      \begin{equation}
        A =
        \begin{bmatrix}
          a & C \\ B & D
        \end{bmatrix}
      \end{equation}
      Then
      \begin{align*}
        &
          \input{gsa_completeness/scalable_results/lemmas/scalable_box_loop/proof/0.tikz}
        \steq{\defref{def:matrix_arrow}\\ \defref{def:unlabelled_scalable} \\ \defref{def:scalable_spider}}
          \input{gsa_completeness/scalable_results/lemmas/scalable_box_loop/proof/1.tikz} 
        \steq{\lemref{lem:scalable_fusion}}
          \input{gsa_completeness/scalable_results/lemmas/scalable_box_loop/proof/2.tikz} 
        \steq{\indhyp{ind:scalable_box_loop} \\ \axiomref{D.7} \\ \lemref{gsa:lem:arrow_transpose}}
          \input{gsa_completeness/scalable_results/lemmas/scalable_box_loop/proof/3.tikz}
        \\
        &\steq{\defref{def:scalable_spider}\\ \lemref{lem:scalable_plus}}
          \input{gsa_completeness/scalable_results/lemmas/scalable_box_loop/proof/4.tikz}
        =
          \input{gsa_completeness/scalable_results/lemmas/scalable_box_loop/proof/5.tikz}.
      \end{align*}
  \end{itemize}
\end{proof}

It follows from \Cref{lem:gaa_functor}, together with \cite[Prop.~17, Eq.~20]{gsa}
\begin{proposition}
  \label{prop:matprop}
  For any matrix \(A\) the following are equivalent:
  \begin{center}
    \begin{tabular}{rlcrl}
      \((1)\) & \(A\) is injective & \quad\quad\quad & \((1)\) & \(A\) is surjective\\[9pt]
      \((2)\) & \input{gsa_completeness/scalable_results/matprop/inj3.tikz} = \input{gsa_completeness/scalable_results/matprop/inj3_2.tikz} & & \((2)\) & \input{gsa_completeness/scalable_results/matprop/surj3.tikz} = \input{gsa_completeness/scalable_results/matprop/surj3_2.tikz}\\[18pt]
      \((3)\) & \input{gsa_completeness/scalable_results/matprop/inj2.tikz} = \input{gsa_completeness/scalable_results/matprop/inj2_2.tikz} & & \((3)\) & \input{gsa_completeness/scalable_results/matprop/surj2.tikz} = \input{gsa_completeness/scalable_results/matprop/surj2_2.tikz}\\[18pt]
      \((4)\) & \input{gsa_completeness/scalable_results/matprop/inj1.tikz} = \input{gsa_completeness/scalable_results/matprop/inj1_2.tikz} & & \((4)\) & \input{gsa_completeness/scalable_results/matprop/surj1.tikz} = \input{gsa_completeness/scalable_results/matprop/surj1_2.tikz}
    \end{tabular}
  \end{center}
  Moreover, the following matrix rules hold:
  
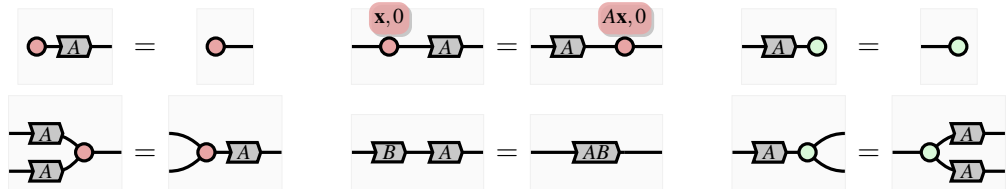
\begin{equation*}
    \begin{array}{c@{\;}c@{\;}c@{\qquad}c@{\;}c@{\;}c@{\qquad}c@{\;}c@{\;}c}
      \input{gsa_completeness/scalable_results/matprop/matrules/red.tikz}
        &=& \input{gsa_completeness/scalable_results/matprop/matrules/red_2.tikz}
      & \input{gsa_completeness/scalable_results/matprop/matrules/phase.tikz}
        &=& \input{gsa_completeness/scalable_results/matprop/matrules/phase_2.tikz}
      & \input{gsa_completeness/scalable_results/matprop/matrules/green.tikz}
        &=& \input{gsa_completeness/scalable_results/matprop/matrules/green_2.tikz} \\[2ex]
      \input{gsa_completeness/scalable_results/matprop/matrules/red_copy.tikz}
        &=& \input{gsa_completeness/scalable_results/matprop/matrules/red_copy_2.tikz}
      & \input{gsa_completeness/scalable_results/matprop/matrules/compose.tikz}
        &=& \input{gsa_completeness/scalable_results/matprop/matrules/compose_2.tikz}
      & \input{gsa_completeness/scalable_results/matprop/matrules/green_copy.tikz}
        &=& \input{gsa_completeness/scalable_results/matprop/matrules/green_copy_2.tikz}
    \end{array}
  \end{equation*}
\end{proposition}

\begin{lemma}
  \label{lem:scalable_arrow_phase}
  For any \(m,n \in \N\), any \(A \in \Matrices[m][n][\F_p(\delta)]\), any \(\vb x \in \F_p(\delta)^m\)
  and any \(X \in \Herm[m]\),
  \begin{align*}
      \begin{tikzpicture}
	\begin{pgfonlayer}{nodelayer}
		\node [style=none] (0) at (-1.5, -0.25) {};
		\node [style=rmatt] (1) at (-0.5, -0.25) {$A$};
		\node [style=gbn] (2) at (0.5, -0.25) {};
		\node [style=none] (3) at (1.5, -0.25) {};
		\node [style=bphase] (9) at (0.5, 0.65) {$\phasemat{\vb x}{X}$};
		\node [style=none] (bg0) at (-1.5, 0.75) {};
		\node [style=none] (bg1) at (1.5, 0.75) {};
		\node [style=none] (bg2) at (1.5, -0.75) {};
		\node [style=none] (bg3) at (-1.5, -0.75) {};
	\end{pgfonlayer}
	\begin{pgfonlayer}{edgelayer}
		\draw [style=background] (bg0.center)
			 to (bg1.center)
			 to (bg2.center)
			 to (bg3.center)
			 to cycle;
		\draw [style=thick, in=180, out=0] (0.center) to (1);
		\draw [style=thick, in=180, out=0] (1) to (2);
		\draw [style=thick, in=180, out=0] (2) to (3.center);
		\draw [style=dbphase] (9) to (2);
	\end{pgfonlayer}
\end{tikzpicture}

    =
      \begin{tikzpicture}
	\begin{pgfonlayer}{nodelayer}
		\node [style=none] (5) at (-2, -0.25) {};
		\node [style=gbn] (6) at (-0.5, -0.25) {};
		\node [style=rmatt] (7) at (0.75, -0.25) {$A$};
		\node [style=none] (8) at (2, -0.25) {};
		\node [style=bphase] (10) at (0, 0.65) {$\phasemat{A^\dag \vb x}{A^\dag XA}$};
		\node [style=none] (bg0) at (-2, 0.875) {};
		\node [style=none] (bg1) at (2, 0.875) {};
		\node [style=none] (bg2) at (2, -0.75) {};
		\node [style=none] (bg3) at (-2, -0.75) {};
	\end{pgfonlayer}
	\begin{pgfonlayer}{edgelayer}
		\draw [style=background] (bg0.center)
			 to (bg1.center)
			 to (bg2.center)
			 to (bg3.center)
			 to cycle;
		\draw [style=thick, in=180, out=0] (5.center) to (6);
		\draw [style=thick, in=180, out=0] (6) to (7);
		\draw [style=thick, in=180, out=0] (7) to (8.center);
		\draw [style=dbphase] (10) to (6);
	\end{pgfonlayer}
\end{tikzpicture}

  \end{align*}
\end{lemma}

\begin{proof}
  We prove the claim holds for arbitrary \(A\in \Matrices[m][n][\F_p(\delta)]\), \(\vb x \in \F_p(\delta)^m\)
  and any \(X \in \Herm[m]\), for all \(n,m\in \N\), first by induction on \(n \in \N\), followed by induction on \(m\) .

  \paragraph{Induction on \(n \in \N\).}
  \begin{itemize}
    \item \emph{Inductive claim}:
    We prove that the lemma holds for any \(n \in \N\), any \(A \in \Matrices[1][n][\F_p(\delta)]\), any \(x \in \F_p(\delta)\) and any \(y \in \F_p(\delta+\delta^{\minu 1})\). 

    \item \emph{Base case:}
    Take any \(x,y \in \F_p(\delta)\), noting that the only \(1\times 0\) matrix is the zero matrix:
    \begin{align*}
        \begin{tikzpicture}
	\begin{pgfonlayer}{nodelayer}
		\node [style=wn] (8) at (-0.75, 0) {};
		\node [style=bn] (9) at (0.25, 0) {};
		\node [style=none] (10) at (1.25, 0) {};
		\node [style=bphase] (11) at (0.25, 0.9) {$\phasemat{x}{y}$};
		\node [style=none] (12) at (-1.25, 1) {};
		\node [style=none] (13) at (1.25, 1) {};
		\node [style=none] (14) at (1.25, -0.5) {};
		\node [style=none] (15) at (-1.25, -0.5) {};
	\end{pgfonlayer}
	\begin{pgfonlayer}{edgelayer}
		\draw [style=background] (12.center) to (13.center) to (14.center) to (15.center) to cycle;
		\draw (8) to (10.center);
		\draw [style=dbphase] (11) to (9);
	\end{pgfonlayer}
\end{tikzpicture}
 
      \steq{\lemref{lem:push_pauli_state}} 
        \begin{tikzpicture}
	\begin{pgfonlayer}{nodelayer}
		\node [style=wn] (15) at (0, 0) {};
		\node [style=none] (17) at (0.75, 0) {};
		\node [style=none] (18) at (-0.5, 0.5) {};
		\node [style=none] (19) at (0.75, 0.5) {};
		\node [style=none] (20) at (0.75, -0.5) {};
		\node [style=none] (21) at (-0.5, -0.5) {};
	\end{pgfonlayer}
	\begin{pgfonlayer}{edgelayer}
		\draw [style=background] (18.center) to (19.center) to (20.center) to (21.center) to cycle;
		\draw (15) to (17.center);
	\end{pgfonlayer}
\end{tikzpicture}
 
      =
        \begin{tikzpicture}
	\begin{pgfonlayer}{nodelayer}
		\node [style=bn] (53) at (-0.5, 0) {};
		\node [style=none] (49) at (0.5, 0) {};
		\node [style=rmat] (50) at (0.5, 0) {$0$};
		\node [style=none] (52) at (1.25, 0) {};
		\node [style=none] (54) at (-1, 0.5) {};
		\node [style=none] (55) at (1.25, 0.5) {};
		\node [style=none] (56) at (1.25, -0.5) {};
		\node [style=none] (57) at (-1, -0.5) {};
	\end{pgfonlayer}
	\begin{pgfonlayer}{edgelayer}
		\draw [style=background] (54.center) to (55.center) to (56.center) to (57.center) to cycle;
		\draw (53) to (49.center);
		\draw (49.center) to (52.center);
	\end{pgfonlayer}
\end{tikzpicture}
 
      \steq{\axiomref{A.2}\\\axiomref{A.6}}
        \begin{tikzpicture}
	\begin{pgfonlayer}{nodelayer}
		\node [style=bn] (64) at (-0.5, 0) {};
		\node [style=none] (59) at (0.75, 0) {};
		\node [style=rmat] (60) at (0.75, 0) {$0$};
		\node [style=none] (61) at (1.5, 0) {};
		\node [style=bphase] (65) at (-0.5, 0.9) {$\phasemat{0^\dag x}{0^\dag y0}$};
		\node [style=none] (66) at (-1.125, 1.125) {};
		\node [style=none] (67) at (1.5, 1.125) {};
		\node [style=none] (68) at (1.5, -0.5) {};
		\node [style=none] (69) at (-1.125, -0.5) {};
	\end{pgfonlayer}
	\begin{pgfonlayer}{edgelayer}
		\draw [style=background] (66.center) to (67.center) to (68.center) to (69.center) to cycle;
		\draw (64) to (59.center);
		\draw (59.center) to (61.center);
		\draw [style=dbphase] (65) to (64);
	\end{pgfonlayer}
\end{tikzpicture}

    \end{align*}

    \item \hypertarget{ind:arrow_phase:n}{}\emph{Inductive hypothesis:}
    Fix some  \(n\in \N\). Assume that the lemma holds for all vectors  \(A \in \Matrices[1][n]\) and scalars \(x \in \F_p(\delta)\), \(y \in \F_p(\delta+\delta^{\minu 1})\).
  
    \item \emph{Inductive step:}
    Take any vector \(A\in \Matrices[1][(n+1)]\) and scalars \(x \in \F_p(\delta)\), \(y \in \F_p(\delta+\delta^{\minu 1})\).
    
    We proceed by case distinction on whether \(A\) is the zero matrix or not.

    First, if \(A\) is the zero vector, then:
    \begin{align*}
        \begin{tikzpicture}
	\begin{pgfonlayer}{nodelayer}
		\node [style=none] (0) at (-2.25, -0.125) {};
		\node [style=rmatt] (1) at (-0.25, -0.125) {$A$};
		\node [style=bn] (2) at (1.25, -0.125) {};
		\node [style=none] (3) at (2.25, -0.125) {};
		\node [style=bphase] (4) at (1.25, 0.775) {$\phasemat{x}{y}$};
		\node [style=wirelable] (9) at (-1.575, -0.375) {$n+1$};
		\node [style=none] (bg0) at (-2.25, 0.875) {};
		\node [style=none] (bg1) at (2.25, 0.875) {};
		\node [style=none] (bg2) at (2.25, -0.875) {};
		\node [style=none] (bg3) at (-2.25, -0.875) {};
	\end{pgfonlayer}
	\begin{pgfonlayer}{edgelayer}
		\draw [style=background] (bg0.center)
			 to (bg1.center)
			 to (bg2.center)
			 to (bg3.center)
			 to cycle;
		\draw [style=thick] (0.center) to (1);
		\draw [in=180, out=0] (1) to (2);
		\draw [in=180, out=0] (2) to (3.center);
		\draw [style=dbphase] (4) to (2);
	\end{pgfonlayer}
\end{tikzpicture}

      \steq{\defref{def:matrix_arrow}}
        \begin{tikzpicture}
	\begin{pgfonlayer}{nodelayer}
		\node [style=none] (11) at (-1.75, -0.25) {};
		\node [style=none] (14) at (1.75, -0.25) {};
		\node [style=gbn] (24) at (-0.75, -0.25) {};
		\node [style=wn] (25) at (0.25, -0.25) {};
		\node [style=bn] (26) at (1, -0.25) {};
		\node [style=bphase] (27) at (1, 0.65) {$\phasemat{x}{y}$};
		\node [style=none] (bg0) at (-1.75, 0.75) {};
		\node [style=none] (bg1) at (1.75, 0.75) {};
		\node [style=none] (bg2) at (1.75, -0.75) {};
		\node [style=none] (bg3) at (-1.75, -0.75) {};
	\end{pgfonlayer}
	\begin{pgfonlayer}{edgelayer}
		\draw [style=background] (bg0.center)
			 to (bg1.center)
			 to (bg2.center)
			 to (bg3.center)
			 to cycle;
		\draw [in=0, out=180] (14.center) to (26);
		\draw [in=0, out=180] (26) to (25);
		\draw [style=thick] (11.center) to (24);
		\draw [style=dbphase] (27) to (26);
	\end{pgfonlayer}
\end{tikzpicture}

      \steq{\lemref{lem:push_pauli_state}} 
        \begin{tikzpicture}
	\begin{pgfonlayer}{nodelayer}
		\node [style=none] (33) at (-1.125, 0) {};
		\node [style=none] (34) at (1.125, 0) {};
		\node [style=gbn] (35) at (-0.375, 0) {};
		\node [style=wn] (36) at (0.375, 0) {};
		\node [style=none] (bg0) at (-1.125, 0.5) {};
		\node [style=none] (bg1) at (1.125, 0.5) {};
		\node [style=none] (bg2) at (1.125, -0.5) {};
		\node [style=none] (bg3) at (-1.125, -0.5) {};
	\end{pgfonlayer}
	\begin{pgfonlayer}{edgelayer}
		\draw [style=background] (bg0.center)
			 to (bg1.center)
			 to (bg2.center)
			 to (bg3.center)
			 to cycle;
		\draw (34.center) to (36);
		\draw [style=thick] (33.center) to (35);
	\end{pgfonlayer}
\end{tikzpicture}

      \steq{\defref{def:matrix_arrow}\\ \axiomref{D.11} \\ \axiomref{D.16}}
        \begin{tikzpicture}
	\begin{pgfonlayer}{nodelayer}
		\node [style=none] (55) at (-2.125, -0.25) {};
		\node [style=none] (56) at (2.125, -0.25) {};
		\node [style=rmatt] (57) at (1.125, -0.25) {$A$};
		\node [style=gbn] (59) at (-0.625, -0.25) {};
		\node [style=bphase] (61) at (-0.625, 0.65) {$\phasemat{A^\dag x}{yA^\dag A}$};
		\node [style=none] (bg0) at (-2.125, 0.875) {};
		\node [style=none] (bg1) at (2.125, 0.875) {};
		\node [style=none] (bg2) at (2.125, -0.75) {};
		\node [style=none] (bg3) at (-2.125, -0.75) {};
	\end{pgfonlayer}
	\begin{pgfonlayer}{edgelayer}
		\draw [style=background] (bg0.center)
			 to (bg1.center)
			 to (bg2.center)
			 to (bg3.center)
			 to cycle;
		\draw [style=thick, in=180, out=0] (55.center) to (59);
		\draw [style=thick, in=180, out=0] (59) to (57);
		\draw (57) to (56.center);
		\draw [style=dbphase] (61) to (59);
	\end{pgfonlayer}
\end{tikzpicture}

    \end{align*}
      
    Otherwise, if \(A\) is not zero, there exists a nonzero element in \(a\) in \(A\). By permuting the rows of \(A\), assume without loss of generality that \(A= (a, A_\bullet)\), for some \(A_\bullet \in \Matrices[1][n]\). Then:
  
    \begin{align*}
      &
        \input{gsa_completeness/scalable_results/lemmas/scalable_arrow_phase/proof/eq2/0.tikz}
      \steq{\propref{prop:matprop}}
        \input{gsa_completeness/scalable_results/lemmas/scalable_arrow_phase/proof/eq2/1.tikz} 
      \steq{\defref{def:matrix_arrow}}
        \input{gsa_completeness/scalable_results/lemmas/scalable_arrow_phase/proof/eq2/2.tikz} 
      \steq{\axref{ax:dzx_fusion}}
        \input{gsa_completeness/scalable_results/lemmas/scalable_arrow_phase/proof/eq2/3.tikz}
      \\
      &\steq{\axiomref{D.21}}
        \input{gsa_completeness/scalable_results/lemmas/scalable_arrow_phase/proof/eq2/4.tikz} 
      \steq{\axiomref{D.5}}
        \input{gsa_completeness/scalable_results/lemmas/scalable_arrow_phase/proof/eq2/5.tikz} 
      \steq{\axref{ax:dzx_fusion}}
        \input{gsa_completeness/scalable_results/lemmas/scalable_arrow_phase/proof/eq2/6.tikz}
      \\
      &\steq{\axiomref{D.0} \\ \lemref{lem:compact_antipode}}
        \input{gsa_completeness/scalable_results/lemmas/scalable_arrow_phase/proof/eq2/7.tikz} 
      \steq{\lemref{lem:box_swapped}}
        \input{gsa_completeness/scalable_results/lemmas/scalable_arrow_phase/proof/eq2/8.tikz}
      \\
      &\steq{\defref{def:dzx_directed_fourier_box} \\ \axref{ax:dzx_fusion} \\ \lemref{lem:antipode_spider_symplectic}}
        \input{gsa_completeness/scalable_results/lemmas/scalable_arrow_phase/proof/eq2/9.tikz} 
      \steq{\axiomref{A.5}\\ \lemref{lem:box_inverse} \\ \lemref{lem:colour_inverted} \\ \axiomref{D.7}}
        \input{gsa_completeness/scalable_results/lemmas/scalable_arrow_phase/proof/eq2/10.tikz}\\
      &\steq{\indhyp{ind:arrow_phase:n} \\ \propref{prop:matprop}}
        \input{gsa_completeness/scalable_results/lemmas/scalable_arrow_phase/proof/eq2/11.tikz}
      \steq{\defref{def:matrix_arrow}}
         \input{gsa_completeness/scalable_results/lemmas/scalable_arrow_phase/proof/eq2/12.tikz} \\
      &\steq{\lemref{lem:scalable_fusion} \\ \defref{def:matrix_arrow} \\
      \defref{def:scalable_spider}}
        \input{gsa_completeness/scalable_results/lemmas/scalable_arrow_phase/proof/eq2/13.tikz}
      =
        \input{gsa_completeness/scalable_results/lemmas/scalable_arrow_phase/proof/eq2/14.tikz}.
    \end{align*}
  \end{itemize}

  \paragraph{Induction on \(m\).}
  \begin{itemize}
    \item \emph{Inductive claim:}
      We prove that the lemma holds for any \(n \in \N\), any \(A \in \Matrices[m][n][\F_p(\delta)]\), any \(\vb x \in \F_p(\delta)^m\) and any \(X \in \Herm[m]\).

    \item \emph{Base case:}
      The base case is trivial, because all \(A \in \Matrices[0][n][\F_p(\delta)]\), any \(\vb x \in \F_p(\delta)^0\) and \(X \in \Herm[0]\) are necessarily zero.

    \item \hypertarget{ind:arrow_phase:m}{}\emph{Inductive hypothesis:}
      Fix some \(n,m\in \N\). Assume that the lemma holds for all matrices in \(\Matrices[j][k]\), \(X \in \Herm[j]\) and vectors \( \vb x \in \F_p(\delta)^j\) such that \(j\leq m\) and \(k \leq n\).
  
    \item \emph{Inductive step:}
      Take arbitrary matrices \(A \in  \Matrices[(m+1)][n][\F_p(\delta)]\), \(X \in \Herm[m+1]\) and an arbitrary  vector \( \vb x \in \F_p(\delta)^{m+1}\).
    
    There exist scalars \(x_1,x_{1,1}, a\), vectors \(X_{1,\bullet}, A_{\bullet}\) and a Hermitian matrix \(X_{\bullet, \bullet}\) such that:
    \begin{equation*}
      \vb x = \begin{bmatrix} x_1 \\ \vb x_\bullet \end{bmatrix},
      \quad\quad
      X = \begin{bmatrix} X_{1,1} & X_{1,\bullet}^\dag \\ X_{1,\bullet} & X_{\bullet,\bullet} \end{bmatrix}
      \quad\quad\text{and}\quad\quad
      A = \begin{bmatrix} a \\ A_\bullet \end{bmatrix}.
    \end{equation*}
    Then
    \begin{align*}
      &
        \input{gsa_completeness/scalable_results/lemmas/scalable_arrow_phase/proof/eq3/0.tikz}
      \steq{\defref{eq:strict}}
        \input{gsa_completeness/scalable_results/lemmas/scalable_arrow_phase/proof/eq3/1.tikz} 
      \steq{\defref{def:matrix_arrow}}
        \input{gsa_completeness/scalable_results/lemmas/scalable_arrow_phase/proof/eq3/2.tikz} 
      \steq{\defref{def:scalable_spider}}
        \input{gsa_completeness/scalable_results/lemmas/scalable_arrow_phase/proof/eq3/3.tikz}
      \\
      &\steq{\indhyp{ind:arrow_phase:m} \\ \propref{prop:matprop}}
        \input{gsa_completeness/scalable_results/lemmas/scalable_arrow_phase/proof/eq3/4.tikz} 
      \steq{\lemref{gsa:lem:arrow_transpose}} 
        \input{gsa_completeness/scalable_results/lemmas/scalable_arrow_phase/proof/eq3/5.tikz} 
      \steq{\lemref{lem:scalable_fusion}\\ \defref{def:matrix_arrow}}
         \input{gsa_completeness/scalable_results/lemmas/scalable_arrow_phase/proof/eq3/6.tikz}
      \\
      &\steq{\lemref{lem:scalable_box_loop}\\ \lemref{lem:scalable_fusion}}
        \input{gsa_completeness/scalable_results/lemmas/scalable_arrow_phase/proof/eq3/7.tikz}.
    \end{align*}
  \end{itemize}

\end{proof}

\begin{lemma}
  \label{lem:scalable_arrow_label_black}
  \label{lem:scalable_arrow_label_white}
  For any \(m,n \in \N\),  \(X \in \Herm\), \(\vb{x} \in \F_p(\delta)^m\), and \(A \in \Matrices[m][n][\F_p(\delta)]\):
  \begin{itemize}
    \item If \(A\) is injective, then:
      
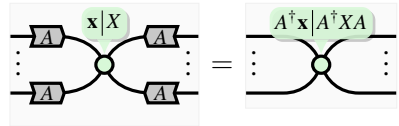
\begin{equation*}
        \begin{tikzpicture}
	\begin{pgfonlayer}{nodelayer}
		\node [style=none] (4) at (-6.25, 1.625) {};
		\node [style=none] (5) at (-1.25, 1.625) {};
		\node [style=none] (6) at (-1.25, -1.625) {};
		\node [style=none] (7) at (-6.25, -1.625) {};
		\node [style=ggn] (8) at (-3.75, 0) {};
		\node [style=none] (9) at (-1.25, 0.75) {};
		\node [style=none] (10) at (-1.25, -0.75) {};
		\node [style=none] (11) at (-6.25, 0.75) {};
		\node [style=none] (12) at (-6.25, -0.75) {};
		\node [style=lmatt] (13) at (-2.25, 0.75) {$A$};
		\node [style=lmatt] (14) at (-2.25, -0.75) {$A$};
		\node [style=rmatt] (15) at (-5.25, 0.75) {$A$};
		\node [style=rmatt] (16) at (-5.25, -0.75) {$A$};
		\node [style=none] (18) at (-1.5, 0.25) {$\vdots$};
		\node [style=none] (30) at (-6.25, 0.75) {};
		\node [style=none] (31) at (-6, 0.25) {$\vdots$};
		\node [style=bphase] (ph0) at (-3.75, 1.1) {$\phasemat{\vb{x}}{X}$};
	\end{pgfonlayer}
	\begin{pgfonlayer}{edgelayer}
		\draw [style=background] (6.center)
			 to (5.center)
			 to (4.center)
			 to (7.center)
			 to cycle;
		\draw [style=thick] (15) to (11.center);
		\draw [style=thick] (16) to (12.center);
		\draw [style=thick] (13) to (9.center);
		\draw [style=thick] (14) to (10.center);
		\draw [style=thick, in=180, out=-75] (8) to (14);
		\draw [style=thick, in=0, out=-105] (8) to (16);
		\draw [style=thick, in=360, out=105] (8) to (15);
		\draw [style=thick, in=180, out=75] (8) to (13);
		\draw [style=dbphase] (ph0) to (8);
	\end{pgfonlayer}
\end{tikzpicture}
      =
        \begin{tikzpicture}
	\begin{pgfonlayer}{nodelayer}
		\node [style=none] (0) at (1.25, 1.625) {};
		\node [style=none] (1) at (5.25, 1.625) {};
		\node [style=none] (2) at (5.25, -1.15) {};
		\node [style=none] (3) at (1.25, -1.15) {};
		\node [style=ggn] (19) at (3.25, 0) {};
		\node [style=none] (20) at (5.25, 0.75) {};
		\node [style=none] (21) at (5.25, -0.75) {};
		\node [style=none] (22) at (1.25, 0.75) {};
		\node [style=none] (23) at (1.25, -0.75) {};
		\node [style=none] (24) at (4.25, 0.75) {};
		\node [style=none] (25) at (4.25, -0.75) {};
		\node [style=none] (26) at (2.25, 0.75) {};
		\node [style=none] (27) at (2.25, -0.75) {};
		\node [style=none] (28) at (1.5, 0.25) {$\vdots$};
		\node [style=none] (29) at (5, 0.25) {$\vdots$};
		\node [style=bphase] (ph1) at (3.25, 1.1) {$\phasemat{A^\dag \vb{x}}{A^\dag XA}$};
	\end{pgfonlayer}
	\begin{pgfonlayer}{edgelayer}
		\draw [style=background] (2.center)
			 to (1.center)
			 to (0.center)
			 to (3.center)
			 to cycle;
		\draw [style=thick] (19)
			 to [in=360, out=105] (26.center)
			 to (22.center);
		\draw [style=thick] (19)
			 to [in=0, out=-105] (27.center)
			 to (23.center);
		\draw [style=thick] (19)
			 to [in=180, out=75] (24.center)
			 to (20.center);
		\draw [style=thick] (19)
			 to [in=180, out=-75] (25.center)
			 to (21.center);
		\draw [style=dbphase] (ph1) to (19);
	\end{pgfonlayer}
\end{tikzpicture}
      \end{equation*}
    \item If \(A\) is surjective, then:
      
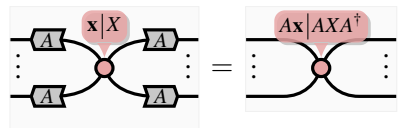
\begin{equation*}
        \begin{tikzpicture}
	\begin{pgfonlayer}{nodelayer}
		\node [style=none] (4) at (-6.25, 1.625) {};
		\node [style=none] (5) at (-1.25, 1.625) {};
		\node [style=none] (6) at (-1.25, -1.625) {};
		\node [style=none] (7) at (-6.25, -1.625) {};
		\node [style=gwn] (8) at (-3.75, 0) {};
		\node [style=none] (9) at (-1.25, 0.75) {};
		\node [style=none] (10) at (-1.25, -0.75) {};
		\node [style=none] (11) at (-6.25, 0.75) {};
		\node [style=none] (12) at (-6.25, -0.75) {};
		\node [style=rmatt] (13) at (-2.25, 0.75) {$A$};
		\node [style=rmatt] (14) at (-2.25, -0.75) {$A$};
		\node [style=lmatt] (15) at (-5.25, 0.75) {$A$};
		\node [style=lmatt] (16) at (-5.25, -0.75) {$A$};
		\node [style=none] (18) at (-1.5, 0.25) {$\vdots$};
		\node [style=none] (30) at (-6.25, 0.75) {};
		\node [style=none] (31) at (-6, 0.25) {$\vdots$};
		\node [style=wphase] (ph0) at (-3.75, 1.1) {$\phasemat{\vb{x}}{X}$};
	\end{pgfonlayer}
	\begin{pgfonlayer}{edgelayer}
		\draw [style=background] (6.center)
			 to (5.center)
			 to (4.center)
			 to (7.center)
			 to cycle;
		\draw [style=thick] (15) to (11.center);
		\draw [style=thick] (16) to (12.center);
		\draw [style=thick] (13) to (9.center);
		\draw [style=thick] (14) to (10.center);
		\draw [style=thick, in=180, out=-75] (8) to (14);
		\draw [style=thick, in=0, out=-105] (8) to (16);
		\draw [style=thick, in=360, out=105] (8) to (15);
		\draw [style=thick, in=180, out=75] (8) to (13);
		\draw [style=dwphase] (ph0) to (8);
	\end{pgfonlayer}
\end{tikzpicture}
      =
        \begin{tikzpicture}
	\begin{pgfonlayer}{nodelayer}
		\node [style=none] (0) at (1.25, 1.625) {};
		\node [style=none] (1) at (5.25, 1.625) {};
		\node [style=none] (2) at (5.25, -1.15) {};
		\node [style=none] (3) at (1.25, -1.15) {};
		\node [style=gwn] (19) at (3.25, 0) {};
		\node [style=none] (20) at (5.25, 0.75) {};
		\node [style=none] (21) at (5.25, -0.75) {};
		\node [style=none] (22) at (1.25, 0.75) {};
		\node [style=none] (23) at (1.25, -0.75) {};
		\node [style=none] (24) at (4.25, 0.75) {};
		\node [style=none] (25) at (4.25, -0.75) {};
		\node [style=none] (26) at (2.25, 0.75) {};
		\node [style=none] (27) at (2.25, -0.75) {};
		\node [style=none] (28) at (1.5, 0.25) {$\vdots$};
		\node [style=none] (29) at (5, 0.25) {$\vdots$};
		\node [style=wphase] (ph1) at (3.25, 1.1) {$\phasemat{A\vb{x}}{AXA^\dag}$};
	\end{pgfonlayer}
	\begin{pgfonlayer}{edgelayer}
		\draw [style=background] (2.center)
			 to (1.center)
			 to (0.center)
			 to (3.center)
			 to cycle;
		\draw [style=thick] (19)
			 to [in=360, out=105] (26.center)
			 to (22.center);
		\draw [style=thick] (19)
			 to [in=0, out=-105] (27.center)
			 to (23.center);
		\draw [style=thick] (19)
			 to [in=180, out=75] (24.center)
			 to (20.center);
		\draw [style=thick] (19)
			 to [in=180, out=-75] (25.center)
			 to (21.center);
		\draw [style=dwphase] (ph1) to (19);
	\end{pgfonlayer}
\end{tikzpicture}
      \end{equation*}
  \end{itemize}
\end{lemma}
\begin{proof}
  We begin with the first claim:
    \begin{align*}
      &
        \input{gsa_completeness/scalable_results/lemmas/scalable_arrow_label/proof/0.tikz}
      \steq{\lemref{lem:scalable_fusion}} 
        \input{gsa_completeness/scalable_results/lemmas/scalable_arrow_label/proof/1.tikz} 
      \steq{\propref{prop:matprop}}
        \input{gsa_completeness/scalable_results/lemmas/scalable_arrow_label/proof/2.tikz}
      \\
      &\steq{\lemref{lem:scalable_arrow_phase}} 
        \input{gsa_completeness/scalable_results/lemmas/scalable_arrow_label/proof/3.tikz} 
      \steq{\propref{prop:matprop}}
        \input{gsa_completeness/scalable_results/lemmas/scalable_arrow_label/proof/4.tikz} 
      \steq{\lemref{lem:scalable_fusion}} 
        \input{gsa_completeness/scalable_results/lemmas/scalable_arrow_label/proof/5.tikz}
    \end{align*}
    The second claim follows from the first claim, 
    \Cref{gsa:lem:arrow_transpose}, the fact that
    the adjoint of a surjective matrix is injective and a simple induction.
\end{proof}

Using the previous lemmas, it is easy to show that:
\begin{lemma}
  \label{lem:scalable_colour}
  For any \(m,n \in \N\),  \(X \in \Herm\), \(\vb{x} \in \F_p(\delta)^m\), and \(A \in \Matrices[m][n][\F_p(\delta)]\):
  \begin{itemize}
    \item If \(A\) is injective, then:
      \begin{equation*}


      \end{equation*}
    \item If \(A\) is surjective, then:
      \begin{equation*}
          %

      \end{equation*}
    \item If \(A\) is invertible, then:
      from green to red:
      \begin{equation*}
          %

      \end{equation*}
  \end{itemize}
\end{lemma}

\begin{lemma}
  \label{gsa:lem:gaussianeliminationi}
  Given a diagram in AP-form and an \emph{invertible} matrix \(A \in
  \Matrices[m][m][\F_p(\delta)]\),
  \begin{align*}
      %

  \end{align*}
\end{proof}

\begin{lemma}
  \label{gsa:lem:push_pauli_state_thick}
  The thick scalable version of \Cref{lem:push_pauli_state} holds:
  \begin{align*}
      %

  \end{align*}  
\end{lemma}

\begin{proof}
The proof is essentially the same as that of \Cref{lem:push_pauli_state}
\end{proof}

\begin{lemma}
  \label{lem:affine_symplectomorphisms}
  For any \(x,y \in \F_p(\delta)^n\) and \(X \in \Herm[m]\),
  \begin{align*}
      %

  \end{align*}
\end{lemma}

\begin{proof}
  We prove the claim by induction on the thickness \(n\) of scalable spiders.
  \begin{itemize}
    \item \emph{Inductive claim:} The claim holds for any \(n \in \N\).
    \item \emph{Base case:} Trivial.
    \item \hypertarget{ind:affine_symplectomorphisms}{}\emph{Inductive hypothesis:}
      Fix \(n \in \N\) and assume the claim holds for spiders of thickness \(n\).
    \item \emph{Inductive step:}

    \begin{align*}
      &
        %
,
  \end{align*}
  where we have written the matrix \(X\) in block form as
  \begin{equation*}
    X
    =
    %
  \end{equation*}
  for some \(A \in \Herm[m]\), \(B \in \Matrices[m][n]\) and \(C \in \Herm[n]\). The
  submatrix \(A\) can then be rewritten to the form
  \begin{equation*}
    A \mapsto Z^{\minu 1} %
 Z,
  \end{equation*}
  where \(D\) is symmetric and invertible, and \(Z\) is a unitary
  change-of-basis matrix that maps the domain of \(A\) to \(\ker{A}^\perp
  \oplus \ker{A}\). Graphically, we then have:
  \begin{align*}
    &
      %
,
  \end{align*}
  so that Schur complementing with respect to the submatrix \(D\) gives a
  diagram in AP-form.
\end{proof}

\propDZXReducedAPForm*

\proofof{prop:dzx_reduced_ap_form}
  Consider a diagram in AP-form described by the quadruple
  \((E,Y,\vb{x},\vb{y})\). Then, by \Cref{gsa:lem:gaussianeliminationi}
  we can act on \(E\) on the left by an invertible matrix \(A\). Thus,
  we can perform Gaussian elimination on \(E\). During this procedure,
  if at any point the resulting matrix \(AE\) has a zero row,
  this means that one of the internal vertices has no neighbours in
  the graph. This disconnected internal vertex then takes the form
  \(
    \begin{tikzpicture}
	\begin{pgfonlayer}{nodelayer}
		\node [style=bn] (6) at (0.625, 0) {};
		\node [style=bphase] (9) at (-0.375, 0) {$\phasemat{x}{0}$};
		\node [style=none] (bg0) at (-1.125, 0.375) {};
		\node [style=none] (bg1) at (1.125, 0.375) {};
		\node [style=none] (bg2) at (1.125, -0.375) {};
		\node [style=none] (bg3) at (-1.125, -0.375) {};
	\end{pgfonlayer}
	\begin{pgfonlayer}{edgelayer}
		\draw [style=background] (bg0.center)
			 to (bg1.center)
			 to (bg2.center)
			 to (bg3.center)
			 to cycle;
		\draw [style=rbphase] (9) to (6);
	\end{pgfonlayer}
\end{tikzpicture}

  \).
  If
  \(x \neq 0\)  then the diagram has empty semantics and we can use
  \Cref{gsa:prop:zero_normal_forms} to rewrite the diagram to normal
  form. Otherwise, \(x=0\) and we can eliminate the disconnected vertex
  using \axiomref{D.6} followed by \axiomref{D.4}.

  In the nonempty case, the resulting row reduced echelon form of \(E\), with zero rows removed,
  is of the form \(\begin{bmatrix} I_m & F \end{bmatrix}\), up to a permutation
  of its columns. Denoting this permutation by \(\varsigma\), we have
  \begin{align*}
      \begin{tikzpicture}
	\begin{pgfonlayer}{nodelayer}
		\node [style=gbn] (20) at (-0.162, 0) {};
		\node [style=none] (21) at (2.088, 0) {};
		\node [style=gbn] (22) at (4.088, 0.5) {};
		\node [style=none] (23) at (4.588, -0.5) {};
		\node [style=none] (24) at (3.588, -0.5) {};
		\node [style=bphase] (25) at (-2.463, 0) {$\phasemat{\vb{x} \\ \vb{y}}{0 & E \\ E^\dag & Y}$};
		\node [style=wirelable] (26) at (0.963, 0.25) {$m+n$};
		\node [style=none] (27) at (3.588, 0.5) {};
		\node [style=wirelable] (28) at (3.588, -0.25) {$n$};
		\node [style=wirelable] (29) at (3.588, 0.25) {$m$};
		\node [style=none] (bg0) at (-4.588, 1) {};
		\node [style=none] (bg1) at (4.588, 1) {};
		\node [style=none] (bg2) at (4.588, -1) {};
		\node [style=none] (bg3) at (-4.588, -1) {};
		\node [style=none] (30) at (-0.162, 0) {};
		\node [style=none] (31) at (2.088, 0) {};
		\node [style=none] (32) at (4.588, -0.5) {};
		\node [style=none] (33) at (3.588, -0.5) {};
	\end{pgfonlayer}
	\begin{pgfonlayer}{edgelayer}
		\draw [style=background] (bg0.center)
			 to (bg1.center)
			 to (bg2.center)
			 to (bg3.center)
			 to cycle;
		\draw [style=thick] (23.center)
			 to (24.center)
			 to [in=345, out=180] (21);
		\draw [style=thick] (20)
			 to (21.center)
			 to [in=-180, out=15] (27.center)
			 to (22);
		\draw [style=rbphase] (25) to (20);
		\draw [style=thick] (32.center)
			 to (33.center)
			 to [in=345, out=180] (31.center)
			 to (30.center);
	\end{pgfonlayer}
\end{tikzpicture}

    =
      \begin{tikzpicture}
	\begin{pgfonlayer}{nodelayer}
		\node [style=gbn] (4) at (-0.038, 0) {};
		\node [style=none] (5) at (2.213, 0) {};
		\node [style=gbn] (6) at (4.213, 0.5) {};
		\node [style=none] (8) at (3.713, -0.5) {};
		\node [style=bphase] (9) at (-3.038, 0) {$\phasemat{\vb{x} \\ \vb{a} \\ \vb{b}}{0 & I_m & F \\ I_m & A & B \\ F^\dag & B^\dag & C}$};
		\node [style=wirelable] (10) at (1.087, 0.25) {$m+n$};
		\node [style=none] (11) at (3.713, 0.5) {};
		\node [style=wirelable] (12) at (3.713, -0.25) {$n$};
		\node [style=wirelable] (13) at (3.713, 0.25) {$m$};
		\node [style=lmatt] (14) at (4.588, -0.5) {\(\varsigma\)};
		\node [style=none] (15) at (5.588, -0.5) {};
		\node [style=none] (bg0) at (-5.588, 1.25) {};
		\node [style=none] (bg1) at (5.588, 1.25) {};
		\node [style=none] (bg2) at (5.588, -1.25) {};
		\node [style=none] (bg3) at (-5.588, -1.25) {};
		\node [style=none] (16) at (-0.037, 0) {};
		\node [style=none] (17) at (2.214, 0) {};
		\node [style=none] (18) at (3.714, -0.5) {};
		\node [style=none] (19) at (4.589, -0.5) {};
	\end{pgfonlayer}
	\begin{pgfonlayer}{edgelayer}
		\draw [style=background] (bg0.center)
			 to (bg1.center)
			 to (bg2.center)
			 to (bg3.center)
			 to cycle;
		\draw [style=thick] (5)
			 to [in=180, out=345] (8.center)
			 to (14);
		\draw [style=thick] (6)
			 to (11.center)
			 to [in=15, out=-180] (5.center)
			 to (4);
		\draw [style=thick] (15.center) to (14);
		\draw [style=rbphase] (9) to (4);
		\draw [style=thick] (16.center)
			 to (17.center)
			 to [in=180, out=345] (18.center)
			 to (19.center);
	\end{pgfonlayer}
\end{tikzpicture}

  \end{align*}
  for matrices \(A,B,C,F\) and vectors \(\vb{a},\vb{b}\) of suitable
  dimensions. Let's ignore the permutation for now, then following
  \Cref{eq:scalable-AP-form-derivation} we have
  \begin{align*}
    &
      \input{gsa_completeness/scalable_results/eqs/ap_form_reduction_proof_2/0.tikz} 
    \steq{\lemref{eq:scalable-AP-form-derivation}}
      \input{gsa_completeness/scalable_results/eqs/ap_form_reduction_proof_2/1.tikz} 
    \steq{\defref{def:matrix_arrow}}
      \input{gsa_completeness/scalable_results/eqs/ap_form_reduction_proof_2/2.tikz}
    \\
    &\steq{\lemref{lem:scalable_fusion} \\ \defref{def:scalable_spider}}
      \input{gsa_completeness/scalable_results/eqs/ap_form_reduction_proof_2/3.tikz} 
    \steq{\propref{prop:matprop} \\ \lemref{gsa:lem:push_pauli_state_thick} \\ \lemref{lem:affine_symplectomorphisms}}
      \input{gsa_completeness/scalable_results/eqs/ap_form_reduction_proof_2/4.tikz}
    \steq{\lemref{lem:scalable_colour} \\ \lemref{lem:scalable_arrow_phase} \\ \lemref{lem:scalable_fusion}}
      \input{gsa_completeness/scalable_results/eqs/ap_form_reduction_proof_2/5.tikz}
    \\
    &\steq{\propref{prop:matprop} \\ \lemref{lem:scalable_arrow_phase}}
     \input{gsa_completeness/scalable_results/eqs/ap_form_reduction_proof_2/6.tikz}
    \steq{\lemref{lem:scalable_fusion}} 
      \input{gsa_completeness/scalable_results/eqs/ap_form_reduction_proof_2/7.tikz}
    \\
    &\steq{\lemref{lem:scalable_box_loop}}
      \input{gsa_completeness/scalable_results/eqs/ap_form_reduction_proof_2/8.tikz}
    \steq{\defref{def:matrix_arrow} \\ \defref{def:scalable_spider}}
      \input{gsa_completeness/scalable_results/eqs/ap_form_reduction_proof_2/9.tikz}
  \end{align*}
  where we have set
  \(\vb{s} \coloneqq F^\dag\vb{a} +
  \vb{b} +  (F^\dag A+B) \vb{x}\)
  and
  \(S \coloneqq F^\dag A F + B +
  BF + F^\dag B^\dag\). 
  Reintroducing the permutation,
  we see that this final diagram is in reduced AP-form, whose constraints can be easily verified.
\end{proof}

\begin{proposition}
  \label{gsa:prop:ap-unique}
  For any nonempty affine Lagrangian state there is exactly one equivalent diagram in reduced-AP form.
\end{proposition}
\begin{proof}
  Given a state, the data of the reduced AP form is exactly that of its canonical generating tableau, which is unique by \Cref{prop:generating_tableau}.
\end{proof}

  \endgroup
 
\end{document}